\documentclass[11pt]{article}
\usepackage[margin=1in]{geometry}
\usepackage{amsmath,amssymb,amsthm,mathtools,enumitem,booktabs,longtable}
\usepackage{xcolor,graphicx}
\usepackage{xr-hyper}
\usepackage[colorlinks,linkcolor=blue!55!black,citecolor=blue!55!black]{hyperref}

\newtheorem{theorem}{Theorem}[section]
\newtheorem{proposition}[theorem]{Proposition}
\newtheorem{lemma}[theorem]{Lemma}
\newtheorem{corollary}[theorem]{Corollary}
\theoremstyle{definition}
\newtheorem{definition}[theorem]{Definition}
\theoremstyle{remark}
\newtheorem{remark}[theorem]{Remark}
\newtheorem{example}[theorem]{Example}
\theoremstyle{plain}

\DeclareMathOperator{\spec}{Sp}\DeclareMathOperator{\dist}{dist}
\DeclareMathOperator{\Tr}{Tr}\DeclareMathOperator{\tr}{tr}
\DeclareMathOperator{\Ind}{Index}\DeclareMathOperator{\sgn}{sgn}
\DeclareMathOperator{\supp}{supp}\DeclareMathOperator{\Herm}{Herm}
\DeclareMathOperator{\conv}{conv}
\newcommand{\R}{\mathbb{R}}\newcommand{\C}{\mathbb{C}}\newcommand{\Z}{\mathbb{Z}}
\newcommand{\N}{\mathbb{N}}\newcommand{\T}{\mathbb{T}}
\newcommand{\eps}{\varepsilon}
\newcommand{\cG}{\mathcal{G}}\newcommand{\cK}{\mathcal{K}}
\newcommand{\cO}{\mathcal{O}}\newcommand{\cP}{\mathcal{P}}\newcommand{\cT}{\mathcal{T}}
\newcommand{\cS}{\mathcal{S}}\newcommand{\cW}{\mathcal{W}}
\newcommand{\BDI}{\mathrm{BDI}}
\newcommand{\ip}[2]{\langle #1,#2\rangle}
\newcommand{\avg}[1]{\langle #1\rangle}
\newcommand{\dbar}{\bar\partial}
\providecommand{\Rm}{\mathbb R}\providecommand{\Zm}{\mathbb Z}\providecommand{\Sm}{\mathbb S}
\providecommand{\EE}{\mathbb E}\providecommand{\PP}{\mathbb P}
\providecommand{\dint}{\displaystyle\int}
\providecommand{\ind}{\Ind}\providecommand{\one}{\mathbf 1}
\providecommand{\dd}{\mathfrak d}\providecommand{\dloc}{\delta_{\rm loc}}
\providecommand{\cA}{\mathcal A}\providecommand{\cJ}{\mathcal J}\providecommand{\cL}{\mathcal L}
\providecommand{\cR}{\mathcal R}\providecommand{\cU}{\mathcal U}
\title{\bf Homogenization theory of topological Anderson insulators: genuine gaps and pseudo-gaps}
\author{Guillaume Bal\thanks{Departments of Statistics and Mathematics and CCAM, University of
Chicago, Chicago, IL 60637; guillaumebal@uchicago.edu}}
\date{}

\begin{document}\maketitle

\begin{abstract}\noindent
A topological Anderson insulator is an insulator whose non-trivial topological phase is induced by disorder. We develop a homogenization theory of this phenomenon for constant-coefficient elliptic systems perturbed by fluctuations oscillating at scale $\eps$. Periodic fluctuations of every order below that of the operator, with appropriate amplitudes, do not average out: they produce an order-one correction $M$ given by a second-order (Born) formula, and the operator converges in the norm-resolvent sense to its homogenized limit. In particular, electric potentials and magnetic fields of zero mean, usually regarded as topologically trivial perturbations, generate corrections $M$ that change the topological phase. Fluctuations confined to a half-plane create an interface that carries the bulk-difference invariant of the homogenized operator for all $\eps$ sufficiently small, a property that is stable under local defects. For random fluctuations, rare regions typically fill the homogenized gap, which becomes a pseudo-gap, and the standard interface index need no longer be defined. We define an index along cuts deformed so as to avoid the rare regions, and prove that it equals almost surely the bulk-difference invariant of the operator homogenized in the Born approximation, for first-order alloys with vanishing low moments and for a class of second-order systems with scalar disorder at criticality that includes the BHZ model.
\end{abstract}

\noindent\textbf{Keywords:} Topological insulators; topological Anderson insulator; homogenization; operator error estimates; Peierls substitution; bulk--edge correspondence; random alloys; pseudo-gap.

\medskip\noindent\textbf{MSC codes:} 35B27, 35Q41, 47A53, 47B80, 82B44.

\section{Introduction}\label{sec:intro}

This paper considers topological insulators whose topological phase is induced by spatial fluctuations. Topological insulators are materials that are insulating in their bulk but carry robust, asymmetric transport along interfaces separating insulators in different topological phases. See, e.g., \cite{HasanKane10,ProdanSchulzBaldes16}. In two space dimensions, the bulk phases are characterized by Chern-type invariants, and a bulk-edge correspondence identifies their difference with a quantized interface current \cite{Hatsugai93,EG02,ProdanSchulzBaldes16}. For continuous Hamiltonians, this program has been carried out for large classes of elliptic operators \cite{Bal2019,Bal2022,QuinnBal,Drouot21}. The interface current is the index of a Fredholm {\em Toeplitz operator} built from an energy switch \cite{Bal2019}, and it equals the {\em bulk-difference invariant} $\BDI(H_-,H_+)$ of the two bulk operators $H_\pm$, computable from their symbols \cite{Bal2022,QuinnBal}. This guaranteed asymmetric interface transport may be interpreted as an obstruction to Anderson localization \cite{Bal2019CMS,ProdanSchulzBaldes16}.

A {\em topological Anderson insulator} is an insulator whose non-trivial topological phase is induced by disorder. Such phases were discovered numerically \cite{LCJS09} in models of HgTe/CdTe quantum wells described by the Bernevig--Hughes--Zhang (BHZ) Hamiltonian \cite{BHZ06}, and have since been observed in disordered atomic wires and photonic lattices \cite{Meier18,Stutzer18}. A standard explanation \cite{GWATB09} is an effective medium theory. In the self-consistent Born approximation, the disorder renormalizes the mass of the Dirac-type Hamiltonian, whose sign may then flip. The topology is thus produced by a renormalization of the bulk phase. A homogenization model of this mechanism for regularized Dirac operators with periodic fluctuations was proposed in \cite{BalDang24,BalDangCorr}. The aim of this paper is to analyze this effective medium picture for constant-coefficient elliptic systems of any order and for genuinely random fluctuations.

We first analyze periodic fluctuations, and random fluctuations that preserve a gap. We consider $H^\eps=P(D)+\mathcal V^\eps$, with $P$ elliptic of order $k$ with constant coefficients and $\mathcal V^\eps$ a fluctuation of order below $k$ oscillating at scale $\eps$, with amplitudes chosen so that all its terms act at the same scale. Periodic fluctuations produce an order-one correction $M$, given by a second-order (Born) formula, and $H^\eps$ converges in the norm-resolvent sense to $P+M$. For random alloys, $M$ is replaced by a slowly varying field of local effective matrices. Whether fluctuations change the topology depends on $M$. In particular, electric potentials and magnetic fields of zero mean, usually regarded as topologically trivial perturbations, generate corrections $M$ that may change the phase (section~\ref{sec:num}). When fluctuations are confined to a half-plane, the interface between the homogenized phase $P+M$ and the unperturbed phase $P$ carries the invariant $\BDI(P+M,P)$, as norm-resolvent convergence transfers the index.

We next analyze the pseudo-gap regime, restricting our scope to first- and second-order operators. When the law of the fluctuation charges a periodic pattern whose homogenized operator is gapless, rare regions in which the field resembles this pattern carry spectrum in the homogenized gap, which becomes a {\em pseudo-gap}. This is the case for uniform on-site disorder with a distribution containing $0$, as in \cite{GWATB09,LCJS09}: arbitrarily large regions with small couplings behave as puddles of the unperturbed phase, whose boundaries carry states in the gap. Such states may meet a straight cut infinitely often, and the Toeplitz operator is then not expected to be Fredholm. Since such regions are rare, they may be avoided. We define an interface index along {\em deformed cuts}, which coincide with a straight cut up to bounded modifications in any horizontal strip and stay away from the rare regions (Figure~\ref{fig:cuts}). This index is stable against bounded, possibly large, perturbations supported away from the cut, and a percolation argument shows that such cuts exist almost surely (a.s.). Our main results state that the index along deformed cuts is a.s. the bulk-difference invariant of the operator homogenized in the Born approximation: for first-order alloys with vanishing low moments, and for a class of second-order systems with scalar disorder at criticality. The latter regime contains the model of \cite{GWATB09}, whose Born correction diverges logarithmically.

\medskip\noindent{\bf Related work and outline.} Homogenization with large oscillatory potentials is analyzed in \cite[Ch.~1, \S12]{BLP78}. See also \cite{ACPSV04}. Operator-norm estimates were obtained by spectral methods \cite{BirmanSuslina04,BirmanSuslina06} and by shift methods \cite{ZhikovPastukhova05}. See also \cite{KLS12,Shen18}. Our correctors are built, as in \cite{ZhikovPastukhova05,KLS12}, on a solution smoothed at the oscillation scale. A point-free approach based on a reformulation as a first-order system is developed in \cite{DGJ26}, with applications to divergence-form operators and to the regularized Dirac operators of \cite{BalDang24}, for periodic and, in an annealed sense, stationary random potentials. Domain-wall modes go back to \cite{JackiwRebbi76}, and their derivation from periodic structures through effective Dirac operators is treated in \cite{FLW17,DrouotWeinstein20}, where the mass arises at first order from symmetry-breaking perturbations. Here, it arises at second order from mean-zero fluctuations. For random media, we refer to \cite{FOP82,Bal10,BalGu15} for large random potentials, to \cite{AllezChouk15,Labbe19,HairerLabbe15} for the continuum Anderson Hamiltonian and its renormalization at criticality, and to \cite{EGS05,BSS23,ProdanSchulzBaldes16} for invariants in mobility gaps. See also \cite{PRR26} for recent results on dynamical delocalization.

Section \ref{sec:setting} states the main results, and section \ref{sec:num} presents applications to Dirac, regularized Dirac and BHZ operators, with periodic and random fluctuations. The proofs are given in sections \ref{sec:hom} to \ref{sec:crit}, and appendix \ref{app:proba} collects probabilistic estimates. Supplementary material summarizes the setting (section~\ref{S-sec:Ssetting}), the main theoretical tools along with outlines of the proofs (section~\ref{S-sec:Stools}), and contains further examples, numerical illustrations, and extensions in later sections.

\section{Setting and main results}\label{sec:setting}

This section states the main results, first on periodic homogenization, next on random homogenization preserving a gap a.s., and finally on random homogenization in the presence of pseudo-gaps.

\paragraph{Notation.} We write $D=-i\nabla$, $\avg{\xi}=(1+|\xi|^2)^{1/2}$, $\T^d=\R^d/\Z^d$, and $\avg f=\int_{\T^d}f\,dy$. The space of $N\times N$ Hermitian matrices is $\Herm(N)$, with the real inner product $\tr(AB)$. For a $1$-periodic function $f$ we set $\hat f_q=\int_{\T^d}f(y)e^{-2\pi iq\cdot y}dy$, $q\in\Z^d$, and $f^\eps(x)=f(x/\eps)$. Throughout, $x\in\R^d$ is the physical variable and $y\in\T^d$ (or $\R^d$) the fast variable of the profiles, so that $y=x/\eps$. Resolvents are written $R(z)=(H-z)^{-1}$, with superscripts matching the operator. Pauli matrices are $\sigma_1,\sigma_2,\sigma_3$, with $\sigma_0=I$. We write $\one_\Omega$ for the indicator function of a set $\Omega$, $\cU_r=\{x:\dist(x,\cU)\le r\}$ for the $r$-neighborhood of a set $\cU$, and call {\em unit box} a translate $Q=m+[0,1)^2$, $m\in\Z^2$, with center $c_Q$. Constants $C,c$ may change from line to line. An index of notation, which also lists the symbols used with several meanings, is given in section~\ref{S-sec:Snotation} of the supplementary material.

\medskip
\noindent\textbf{(i) Operators and fluctuations.} The unperturbed operator is an elliptic system of order $k$ on $\R^d$, and the fluctuations are oscillatory perturbations of all orders below $k$. Let $P(D)=\sum_{|\alpha|\le k}A_\alpha D^\alpha$ with $A_\alpha\in\C^{N\times N}$, and write $P_k(\xi)=\sum_{|\alpha|=k}A_\alpha\xi^\alpha$ for its principal symbol. We assume:\\[1mm]
 \textup{[P1]} \emph{self-adjointness}: $P(\xi)\in\Herm(N)$ for all $\xi\in\R^d$; and
 \\[1mm]
 \textup{[P2]} \emph{ellipticity}: $|P_k(\xi)v|\ge c_P|\xi|^k|v|$ for all $\xi\in\R^d$ and $v\in\C^N$, for some $c_P>0$.
 \\[1mm]
Lower-order terms, including a constant ``mass'' $\Gamma\in\Herm(N)$, are part of $P$. Under [P1]--[P2], $P(D)$ is self-adjoint on $H^k(\R^d;\C^N)$ \cite[\S VIII.3]{ReedSimon1}. Since $P_k(e)$ is invertible for $e$ on the connected sphere $S^{d-1}$ ($d\ge2$), its signature is constant, and we call $P_k$ \emph{indefinite} or \emph{definite} accordingly. For odd $k$, $P_k(-\xi)=-P_k(\xi)$ forces $P_k$ to be indefinite.
\\[1mm]
 A general model of fluctuations is described as follows. Let $\rho:\R^d\to\R$ be an envelope, with $\rho\in W^{1,\infty}$ if $k=1$ and $\rho\in C_b^\infty$ if $k\ge2$. For a periodic matrix function $c$ and a multi-index $\alpha$, let $\mathrm{Sym}[c\,D^\alpha]:=\frac12\bigl(c\,D^\alpha+D^\alpha c^*\bigr)$. We consider symmetric perturbations of all orders below $k$:
\begin{equation}\label{eq:Heps}
H^\eps=P(D)+\mathcal V^\eps,\qquad \mathcal V^\eps=\sum_{|\alpha|\le k-1}\eps^{|\alpha|-k/2}\,\mathrm{Sym}\bigl[\rho\,c_\alpha(\tfrac x\eps)D^\alpha\bigr]+\sum_{|\alpha|\le k-2}\eps^{|\alpha|}\,\mathrm{Sym}\bigl[\rho^2d_\alpha(\tfrac x\eps)D^\alpha\bigr]+\mathcal R^\eps .
\end{equation}
Here $c_\alpha,d_\alpha\in C^{|\alpha|,\alpha_0}(\T^d;\C^{N\times N})$ for some $\alpha_0\in(0,1)$, and $\avg{c_\alpha}=0$ for all $\alpha$. The remainder $\mathcal R^\eps$ is a finite sum of terms $\eps^{|\alpha|+1}\mathrm{Sym}[\rho\,e(\frac x\eps)D^\alpha]$ and $\eps^{|\alpha|-k/2+1}\mathrm{Sym}[\rho\,f(\frac x\eps)D^\alpha]$, with $|\alpha|\le k-1$, coefficients of the same regularity, and $\avg f=0$. The envelopes $\rho$ and $\rho^2$ appear naturally in the Peierls substitution \eqref{eq:peierls} below, which explains the form of \eqref{eq:Heps}. Independent envelopes $\rho_{\alpha,c}$ and $\rho_{\alpha,d}$ for each term can be treated by the same proofs, with $\rho^2M$ in \eqref{eq:M} replaced by a slowly varying matrix, quadratic in the $\rho_{\alpha,c}$ and linear in $\rho_{0,d}$.

The amplitudes are those for which every term acts at the same scale as the zeroth-order term $\eps^{-k/2}c_0$: in the fast variable, $\eps^{|\alpha|-k/2}D^\alpha=\eps^{-k}\cdot\eps^{k/2}(\eps D)^\alpha$. This scale is dictated by the correction: a potential of size $\eps^{-a}$ produces a correction quadratic in the potential and inverse to the principal part, of order $\eps^{-2a}\cdot\eps^k$, which is $O(1)$ exactly for $a=k/2$. Define
\begin{equation}\label{eq:Weff}
W:=\tfrac12\bigl(c_0+c_0^*\bigr)+\tfrac12\sum_{1\le|\alpha|\le k-1}D_y^\alpha\bigl(c_\alpha^*\bigr),\qquad d^{\rm h}_0:=\tfrac12\bigl(d_0+d_0^*\bigr).
\end{equation}
Thus $W$ is the function obtained by applying $\sum_\alpha\mathrm{Sym}[c_\alpha D_y^\alpha]$ to constant vectors. It has mean zero but need not be Hermitian. We set
\begin{equation}\label{eq:M}
H^0=P(D)+\rho^2(x)M,\qquad M:=\avg{d_0^{\rm h}}-\avg{W^*P_k(D_y)^{-1}W}=\avg{d_0^{\rm h}}+\sum_{q\in\Z^d\setminus0}\hat W_q^*\,G(q)\,\hat W_q,
\end{equation}
with $G(q):=-P_k(2\pi q)^{-1}$. Both $H^\eps$ and $H^0$ are self-adjoint on $H^k$. Here $P_k(D_y)^{-1}$ acts on mean-zero periodic functions.

The basic case is a matrix-valued potential, $\mathcal V^\eps=\eps^{-k/2}\rho(x)W(x/\eps)$ with $W$ Hermitian, for which $M=-\avg{WP_k(D_y)^{-1}W}$. Both sums appear in the Peierls substitution: for $k\le2$,
\begin{equation}\label{eq:peierls}
P\bigl(D+\eps^{k/2-1}\rho A(\tfrac x\eps)\bigr)+\eps^{-k/2}\rho V(\tfrac x\eps)\,I,\qquad A\in C^{1,\alpha_0}(\T^d;\R^d),\ \avg A=0,\ \avg V=0,
\end{equation}
is of the form \eqref{eq:Heps} (Lemma~\ref{lem:peierls}). For $k=2$, its diamagnetic term $\rho^2P_2(A(\tfrac x\eps))$, quadratic in the vector potential, is the second sum. The electric potential $V$ and the vector potential $A$ contribute to $M$ at the same order. For $k=2$, the magnetic field $\mathrm{curl}\,A^\eps\sim\eps^{-1}$ and the electric potential $\eps^{-1}V^\eps$ moreover have the same amplitude. The exponent $k/2$ is the only one producing a non-trivial $O(1)$ limit, unless $M=0$. The correction $M$ is computed with the principal symbol only; lower-order terms of $P$ do not contribute. We set $r_1=\frac12$ and $r_k=1$ for $k\ge2$.

\medskip
\noindent\textbf{(ii) Periodic homogenization and attainable corrections.} For periodic fluctuations, $H^\eps$ converges to the homogenized operator $H^0$ with explicit rates.
\begin{theorem}[Norm-resolvent homogenization]\label{thm:hom}
Assume \textup{[P1]--[P2]} and the regularity of the coefficients in \eqref{eq:Heps}. There is $C$, depending on $P$, $\rho$ and the coefficients of $\mathcal V^\eps$, such that for all $z\in\C\setminus\R$ and $\eps\in(0,1]$,
\[
\big\|(H^\eps-z)^{-1}-(H^0-z)^{-1}\big\|_{L^2\to L^2}\le C\,\eps^{r_k}\,(1+|\Re z|)^2\,\bigl(1+|\Im z|^{-2}\bigr).
\]
\end{theorem}
\noindent The theorem is proved in section \ref{sec:hom}.
\begin{theorem}[Attainable corrections]\label{thm:univ}
Let $d\ge2$ and $N\ge2$, and consider $M(W)$ over all Hermitian trigonometric polynomials $W$ with $\avg W=0$. \textup{(i)} If $P_k$ is indefinite, $\{M(W)\}=\Herm(N)$; \textup{(ii)} if $P_k$ is positive (resp.\ negative) definite, $\{M(W)\}$ is the cone of negative (resp.\ positive) semidefinite matrices; and \textup{(iii)} if $k$ is odd, the modes along a single primitive lattice direction $q_0$ produce exactly the hyperplane $\{M:\tr(P_k(q_0)^{-1}M)=0\}$.
\end{theorem}
\noindent The correction is thus unrestricted when $P_k$ is indefinite, and sign-constrained when it is definite. In the latter case, the bulks $P+\Gamma$ are necessarily topologically trivial (see below). The theorem is proved in section \ref{sec:univ}. Explicit corrections are computed in Propositions~\ref{prop:EM} and~\ref{prop:regDirac}, and the sharpness of the rates of Theorem~\ref{thm:hom} is discussed in Remark~\ref{S-rem:sharp}.

\medskip
\noindent\textbf{(iii) Interface index and fluctuation-induced topology.} In two space dimensions, fluctuations confined to a half-plane generate an interface whose asymmetric transport is described by the index of a Fredholm operator. Let $d=2$ and let $\rho=\rho(x_2)$ be smooth with $\rho=0$ for $x_2\ge1$ and $\rho=1$ for $x_2\le-1$. The two bulk operators of $H^0$ are then
\[
H_+:=P(D)\qquad\text{and}\qquad H_-:=P(D)+M.
\]
We assume the common-gap condition
\begin{equation}\label{eq:gap}
\textup{[G]}\qquad (-g,g)\cap\bigl(\spec H_+\cup\spec H_-\bigr)=\emptyset\quad\text{for some }g>0.
\end{equation}

Fix a smooth non-decreasing energy switch $\varphi$ with $\varphi=0$ on $(-\infty,-g/4]$ and $\varphi=1$ on $[g/4,\infty)$, where $g$ is the gap parameter of the statement at hand, i.e., of [G], or of (A)--(B), [BG] and Theorem~\ref{thm:mainA} below, and put $U(H)=e^{2\pi i\varphi(H)}=I+S(H)$, where $S\in C_c^\infty((-g,g))$. With $\Pi=\mathbf 1_{\{x_1\ge0\}}$, we define the {\em Toeplitz operator}
\begin{equation}\label{eq:T}
T(H):=\Pi\,U(H)\,\Pi+I-\Pi .
\end{equation}
For elliptic interface operators with two gapped bulks in $(-g,g)$, $T(H)$ is Fredholm \cite{Bal2019,Bal2022}, and its index equals the interface current $2\pi\sigma_I(H)=\Tr\,2\pi i[H,\Pi_s]\varphi'(H)$ computed with a smooth spatial switch $\Pi_s(x_1)$. For $\Gamma_\pm\in\Herm(N)$ with $P+\Gamma_\pm$ gapped at $0$, the \emph{bulk-difference invariant} \cite{Bal2022} is
\begin{equation}\label{eq:ch}
\BDI(P+\Gamma_-,P+\Gamma_+):=c(P+\Gamma_-)-c(P+\Gamma_+),\ \ c(h):=\frac1{2\pi}\int_{\R^2}\Omega_h,\ \Omega_h:=i\,\tr\bigl(\Pi_h\,d\Pi_h\wedge d\Pi_h\bigr),
\end{equation}
with $\Pi_h$ the negative spectral projector of the symbol $h$. Only the difference is an integer.

\begin{proposition}[Bulk-edge correspondence {\cite{Bal2022,QuinnBal}}; see Lemma~\ref{lem:BDIc}]\label{thm:QB}
Let $H=P(D)+V(x)$ with $V\in C_b^\infty(\R^2;\Herm(N))$ equal to constant matrices $V_\pm$ for $\pm x_2\ge1$, and suppose $P+V_\pm$ satisfy \textup{[G]}. Then $T(H)$ is Fredholm, and $\Ind T(H)=\BDI(P+V_-,P+V_+)$. Moreover, $\Ind T(H)$ is invariant under norm-resolvent continuous deformations of $H$ that preserve \textup{[G]} for the bulks.
\end{proposition}
This result shows that the asymmetric transport invariant is captured by the index of the Toeplitz operator $T(H)$.
Homogenization estimates allow us to transfer the interface index from $H^0$ to $H^\eps$, showing that the asymmetric interface transport is observed for all $\eps$ {\em sufficiently small}, i.e., in the presence of sufficiently highly oscillatory perturbations.
\begin{theorem}[Fluctuation-induced interfaces]\label{thm:TAI}
Let $d=2$, and assume \textup{[P1]--[P2]} and \textup{[G]}. There is $\eps_0>0$ such that, for $\eps\in(0,\eps_0]$, $T(H^\eps)$ is Fredholm and $\Ind T(H^\eps)=\BDI(P+M,P)$.
\end{theorem}
\noindent When $P_k$ is indefinite, Theorem~\ref{thm:univ}(i) and Theorem~\ref{thm:TAI} show that every invariant $\BDI(P+M',P)$ with $P$ and $P+M'$ satisfying [G] is the index of $T(H^\eps)$ for some $W$ and all $\eps\le\eps_0$. When $P_k$ is definite, no fluctuation produces a non-trivial invariant: the eigenvalue branches of $P(\xi)+\Gamma$, $\Gamma\in\Herm(N)$, all tend to $+\infty$ (or all to $-\infty$), so $\spec(P+\Gamma)$ is a half-line, a gap at $0$ forces $\Pi_h\equiv0$ (or $I$), and $c(P+\Gamma)=0$. The invariant is evaluated for Dirac, regularized Dirac and BHZ operators in Propositions~\ref{prop:EM}, \ref{prop:regDirac} and~\ref{prop:regIndex} while an interface spectrum is shown in Figure~\ref{S-fig:first}(a). We now show that the fluctuation-induced index is stable against local defects of arbitrary size.
\begin{corollary}[Local defects]\label{cor:compact}
Under the assumptions of Theorem~\ref{thm:TAI}, let $V_c\in L^\infty(\R^2;\Herm(N))$ have compact support. Then, with $\eps_0$ as in Theorem~\ref{thm:TAI}, $T(H^\eps+V_c)$ is Fredholm and $\Ind T(H^\eps+V_c)=\BDI(P+M,P)$ for $\eps\in(0,\eps_0]$.
\end{corollary}
\noindent The threshold $\eps_0$ does not depend on $V_c$, and the same holds, a.s.\ and for every admissible cut, for the Toeplitz operators $\tilde T(H^\eps+V_c)$ along deformed cuts of Theorems~\ref{thm:mainB} and~\ref{thm:mainA} below. The defect is not homogenized. It may be large compared with the gap, and it may meet the interface and the cut. Theorem \ref{thm:TAI} and Corollary \ref{cor:compact} are proved in section \ref{sec:topo}.

Since the index changes along the family in $\eps$, the bulk gap must close somewhere. For $W\in L^\infty(\R^2;\Herm(N))$ and a continuous amplitude $\eps\mapsto a_\eps>0$ on $(0,\infty)$, let
\[
H^\eps_a=P(D)+a_\eps\rho(x_2)W(x/\eps),\qquad H^\eps_{a,B}=P(D)+a_\eps W(x/\eps)\quad\text{(lower bulk)}.
\]
When $0\notin\spec P\cup\spec H^\eps_{a,B}$, $T(H^\eps_a)$ is built with an energy switch whose transition lies in the common gap of $P$ and $H^\eps_{a,B}$ containing $0$. It is then Fredholm by Lemma~\ref{lem:fred}.
\begin{theorem}[Gap closing along the $\eps$-family]\label{thm:gapclose}
Let $d=2$, assume \textup{[P1]--[P2]} and $0\notin\spec P$, and let $0<\eps'<\eps''$ with $0\notin\spec H^{\eps'}_{a,B}\cup\spec H^{\eps''}_{a,B}$. If $\Ind T(H^{\eps'}_a)\ne\Ind T(H^{\eps''}_a)$, then $0\in\spec H^\eps_{a,B}$ for some $\eps\in(\eps',\eps'')$.
\end{theorem}
\noindent In particular, let $\mathcal V^\eps=\eps^{-k/2}\rho W(x/\eps)$ under the assumptions of Theorem~\ref{thm:TAI}. For $\eps\ge\eps_1$ large, $T(H^\eps)$ is Fredholm with index $0$, since $H^\eps$ is then close to $P$. Hence, if $\BDI(P+M,P)\ne0$, there is, after decreasing $\eps_0$ if necessary, some $\eps\in(\eps_0,\eps_1)$ with $0\in\spec\bigl(P+\eps^{-k/2}W(x/\eps)\bigr)$. For general fluctuations, the amplitudes $\eps^{|\alpha|-k/2}$ with $|\alpha|\geq k/2$ do not vanish as $\eps\to\infty$, and the index for large $\eps$ need not vanish. The same holds, configuration by configuration, for random fields whenever the two indices are defined and differ, e.g., for the gap-preserving laws of Theorem~\ref{thm:phasepurek}. The gap closing is a consequence of the change of index alone. In the pseudo-gap regime of blocks (v)--(vii), however, the hypothesis $0\notin\spec H^{\eps'}_{a,B}$ fails for small $\eps'$, and the theorem does not apply (Remark~\ref{S-rem:gapscope} of the supplementary material).

\smallskip

Finally, magnetic fields introduced by Peierls substitution, which are gauge covariant, cannot create topology in two classes of operators. Let $k\in\{1,2\}$, $g>0$, $J\in\Herm(N)$ a unitary involution and $s\in\{\pm1\}$. We consider the following two settings:
\begin{enumerate}[label=\textup{(\alph*)},leftmargin=*,nosep]
\item \emph{$J$-graded class.} $P(\xi)=P^{\rm o}(\xi)+\Gamma$, with $JP^{\rm o}(\xi)=-P^{\rm o}(\xi)J$ for all $\xi$, $J\Gamma=\Gamma J$ and $sJ\Gamma\ge g$.
\item \emph{$J$-definite class.} $P(\xi)=P^{\rm o}(\xi)+q(\xi)\Upsilon+\Gamma$, with $P^{\rm o}$ as in (a), $q(\xi)=\sum b_{ij}\xi_i\xi_j$ for a real positive definite matrix $(b_{ij})$, and $\Upsilon,\Gamma$ Hermitian matrices commuting with $J$ such that $sJ\Upsilon\ge0$ and $sJ\Gamma\ge g$.
\end{enumerate}
Heuristically, in class (a) the momentum enters only through $J$-odd terms, which mix the two eigenspaces of $J$ without contributing to $\Re\ip{H\psi}{J\psi}$, while the $J$-even mass separates them by at least $g$. In class (b), the additional $J$-even curvature has the same $J$-sign as the mass and reinforces it, so that no topological change occurs. A vector potential enters only through the momentum, and cannot overcome the mass in either class.
For real $a\in W^{1,\infty}(\R^2;\R^2)$, let $P(D^a)$ denote the Peierls substitution $\xi\mapsto D^a:=D+a(x)$, with $q(D^a)=\sum b_{ij}D^a_iD^a_j$ and $P_2(D^a)$ as in Lemma~\ref{lem:peierls}.
\begin{theorem}[Purely magnetic fluctuations]\label{thm:magno}
Let $d=2$ and let $P$ be in class (a) or (b). Then $\spec P(D^a)\cap(-g,g)=\emptyset$ for every real $a\in W^{1,\infty}(\R^2;\R^2)$. In particular, for $a=\eps^{k/2-1}\rho A(\tfrac x\eps)$ ($V=0$ in \eqref{eq:peierls}), $U(P(D^a))=I$ and $\Ind T(P(D^a))=0$ for all $\eps>0$. Under \textup{[P1]--[P2]} and \textup{[G]}, the correction $M$ of this purely magnetic fluctuation satisfies $\BDI(P+M,P)=0$.
\end{theorem}
\noindent The gap holds whatever the amplitude or profile of the vector potential. The $J$-graded class contains the Dirac operator with $J=\sigma_3$. For it, purely magnetic fluctuations give $M=0$ (Proposition~\ref{prop:EM}), consistent with the zero modes of Remark~\ref{S-rem:AC}. The $J$-definite class contains the regularized Dirac operator in its trivial phase $m\beta<0$. A quantization result for first-order systems is given in Remark~\ref{S-rem:quant} of the supplementary material.

\medskip
\noindent\textbf{(iv) Random alloys maintaining a genuine gap.} We now consider random fluctuations in $d=2$. An \emph{alloy} is $W_\omega=\sum_{j\in\Z^2}\omega_jw(\cdot-j)$, with real couplings $\omega=(\omega_j)$ satisfying $0\leq \kappa_- \leq |\omega_j|\le\kappa_+$ for all $j$, and a single-site profile $w\in C_c^\infty(\R^2;\Herm(N))$. For an operator of order $k$, we assume the {\em moment condition}
\begin{equation}\label{eq:moments}
\int_{\R^2}y^\alpha w(y)\,dy=0\qquad\text{for all multi-indices }\alpha\text{ with }|\alpha|\le2k-1.
\end{equation}
This condition is not generic. It holds, e.g., for $w=\Delta^kb$ with $b\in C_c^\infty(\R^2;\Herm(N))$. Let $H^\eps_\omega=P+\eps^{-k/2}\rho\,W_\omega(x/\eps)$ with an envelope $\rho\in C_b^\infty(\R^2;\R)$, and let $H^\eps_{\omega,B}$ denote the same operator with $\rho\equiv1$. With $\hat w(p)=\int_{\R^2}w(y)e^{-2\pi ip\cdot y}dy$ and $G$ as in \eqref{eq:M}, let
\begin{equation}\label{eq:Kq}
\mathbb K(q)=\sum_{\ell\in\Z^2,\ q+\ell\ne0}\hat w(q+\ell)^*\,G(q+\ell)\,\hat w(q+\ell),\qquad q\in\T^2 .
\end{equation}
Under \eqref{eq:moments}, $\mathbb K$ is bounded and Hermitian, and for periodic configurations $M(\omega)=\int\mathbb K\,d\mathfrak s_\omega$, where the structure factor $\mathfrak s_\omega$ is a nonnegative measure whose mass is the mean of $\omega_j^2$ (Lemma~\ref{lem:Kset}). Local effective matrices therefore lie in the sets
\[
\cK[\mu_1,\mu_2]:=[\mu_1,\mu_2]\cdot\overline\conv\{\mathbb K(q):q\in\T^2\},\quad 0\le\mu_1\le\mu_2,\qquad \cK:=\cK[\kappa_-^2,\kappa_+^2] .
\]
We set $\delta_1(\eps)=\eps^{1/4}$, $\delta_2(\eps)=\eps^{1/3}\log(1/\eps)$ and $\delta_k(\eps)=\eps^{1/3}$ for $k\ge3$.
\begin{theorem}[Uniform homogenization for alloys]\label{thm:unifhom}
Assume \textup{[P1]--[P2]} and \eqref{eq:moments}. There are $C$ and $\eps_*>0$ such that, for every configuration and every $\eps\in(0,\eps_*]$, there is a smooth field $M^\eps_\omega:\R^2\to\Herm(N)$ with
\[
\big\|(H^\eps_\omega-z)^{-1}-(P+\rho^2M^\eps_\omega-z)^{-1}\big\|\le C\delta_k(\eps)\,(1+|\Re z|)^2(1+|\Im z|^{-2}),\qquad z\notin\R ,
\]
and $\dist(M^\eps_\omega(x),\cK)\le C\delta_k(\eps)$ for all $x\in\R^2$.
\end{theorem}
\noindent The constants depend only on $(P,w,\kappa_+,\rho)$, and the bound is deterministic and uniform over all configurations. It also extends to $\Z^d$-alloys on $\R^d$. See Remark~\ref{S-rem:dimalloy}. The field $M^\eps_\omega$ is constructed in section \ref{sec:random}. It is a local effective matrix: $M^\eps_\omega(x)$ is the Born correction computed with the local structure factor of the couplings within a mesoscopic distance of $x$ (Lemma~\ref{lem:meso}). For i.i.d.\ couplings, $M^\eps_\omega(x)$ concentrates around the deterministic Born matrix $M_{\rm Born}$ of \eqref{eq:Mborn}, with deviations of probability $\exp(-ct^2\varpi^{-2})$ at mesoscale $\varpi$ (Lemma~\ref{lem:mean}), except in rare regions. Theorem~\ref{thm:unifhom} is stated configuration by configuration so as to apply also near these regions. See blocks (v)--(vii).

By Theorem~\ref{thm:unifhom}, for every configuration, $H^\eps_{\omega,B}$ is close to $P+M^\eps_\omega$, a field with values near $\cK$ that varies on the mesoscale. The two conditions below imply that $P+M(\cdot)$ is gapped for every field $M(\cdot)$ with values in $\cK$, so that no rare region fills the gap. The lower bound $\kappa_->0$ allows $\cK$ to avoid $0$ and to lie in a phase different from that of $P$. If $\kappa_-=0$, then $0\in\cK$ and the index in (ii) below vanishes. Laws supported in $\{\kappa_-\le|t|\le\kappa_+\}$ are then called \emph{gap-preserving}. Let $g>0$. The two conditions are:
\begin{enumerate}[label=\textup{(\Alph*)},leftmargin=*,nosep]
\item \emph{$J$-graded class.} $P$ is in the $J$-graded class (a) of block (iii), here for any order $k$, with $s=s_+$, and there is $s_-\in\{\pm1\}$ with $s_-J(\Gamma+M^{\rm e})\ge g$ for every $M\in\cK$, where $M^{\rm e}=\frac12(M+JMJ)$ is the part of $M$ that commutes with $J$.
\item \emph{Small effective set.} There are $\bar M\in\cK$ and $\bar r\ge0$ with $|M-\bar M|\le\bar r$ for every $M\in\cK$, $\dist(0,\spec P)\ge g$, and $\dist(0,\spec(P+\bar M))\ge g+\bar r$.
\end{enumerate}
\begin{theorem}[Gap-preserving fluctuations of any order]\label{thm:phasepurek}
Assume \textup{[P1]--[P2]}, \eqref{eq:moments}, and (A) or (B). There is $\eps_0>0$ such that, for every $\eps\in(0,\eps_0]$ and \emph{every} configuration, (i) $\spec(H^\eps_{\omega,B})\cap(-g/2,g/2)=\emptyset$, and (ii) if $\rho=\rho(x_2)$ is as in block (iii), then $T(H^\eps_\omega)$ is Fredholm with $\Ind T(H^\eps_\omega)=\BDI(P+\bar M,P)$, where $\bar M$ is any element of $\cK$ in case (A) and the matrix of (B) in case (B).
\end{theorem}
\noindent The theorem is proved in section \ref{sec:random}. In (A), every possible effective mass $\Gamma+M^{\rm e}$ has a definite $J$-sign, while the $J$-anticommuting part of $M$ is unrestricted: like the kinetic part and the vector potentials of Theorem~\ref{thm:magno}, it cannot close the gap (Lemma~\ref{lem:grade}). When $s_-$, the sign after perturbation, differs from $s_+$, the sign before perturbation, the disorder inverts the mass for every configuration. This is the case of the first-order alloy of Corollary~\ref{cor:bernoulli} and the second-order alloy of Proposition~\ref{S-prop:highex}(a). In (B), no structure of $P$ is needed, but the local effective matrices must vary less than the gap of $P+\bar M$. This is the case of the regularized Dirac operator with scalar disorder of Proposition~\ref{S-prop:highex}(b). Both conditions reappear in block (vi): (B) as the Born-gap condition [BG], with $\bar M=M_{\rm Born}$, the same margin $\bar r$, and the bound on $|M-\bar M|$ imposed outside rare regions only, and (A) as the graded criterion that implies [BG]. Common disorder models, such as uniform on-site disorder with a distribution containing $0$, are not gap-preserving. They are the subject of blocks (v)--(vii).

\medskip
\noindent\textbf{(v) Pseudo-gaps and deformed cuts.} For random fluctuations, the gap of the homogenized operator need not be a spectral gap of the bulk operator. Theorem \ref{thm:pseudogap} below shows that a genuine gap, as in block (iv), requires that the law charge no gapless pattern. Let $W(y,\omega)$ be a $\Z^d$-stationary, ergodic, Hermitian random field with a finite range of dependence, and realizations bounded in $C^{0,\alpha_0}$. The bulk operator $H^\eps_B=P(D)+\eps^{-k/2}W(x/\eps,\omega)$ has an almost-sure spectrum $\Sigma_\eps$ \cite{Pastur80}. The law of $W$ \emph{charges} a $1$-periodic pattern $W_*$ if
\[
\mathbb P\Big(\sup_{|y|\le L}|W(y)-W_*(y)|<\delta\Big)>0\qquad\text{for all }L,\delta>0 ,
\]
the support condition that identifies almost-sure spectra \cite{KirschMartinelli82,Kirsch08}. Recall $r_1=\frac12$ and $r_{k\geq2}=1$.
\begin{theorem}[Rare-region spectrum in the homogenized gap]\label{thm:pseudogap}
Assume \textup{[P1]--[P2]}.
\begin{enumerate}[label=\textup{(\roman*)},leftmargin=*,nosep]
\item If the law of $W$ charges a mean-zero pattern $W_*\in C^{0,\alpha_0}(\T^d)$ with $0\in\spec(P+M(W_*))$, there is $C$ such that, for every $\eps\in(0,1]$, a.s., $\dist(0,\spec H^\eps_B)\le C\eps^{r_k}$.
\item In $d=2$, such a pattern exists if the law charges every element of a path $t\mapsto W_t\in C^{0,\alpha_0}(\T^2)$, $t\in[0,1]$, of mean-zero patterns, continuous in $L^\infty$, with $P+M(W_0)$ and $P+M(W_1)$ gapped at $0$ and $\BDI(P+M(W_1),P+M(W_0))\ne0$.
\item Let $P_k$ be indefinite. Let $W=\sum_j\omega_jw(\cdot-j)$ with $w\in C_c^\infty(\R^d;\R)$, $\int w\ne0$, and i.i.d.\ couplings whose law $\mu_0$ satisfies $\supp\mu_0\supset[0,\kappa_0]$ or $\supp\mu_0\supset[-\kappa_0,0]$ for some $\kappa_0>0$. Let $a_\eps\to\infty$ as $\eps\to0$, and $\Sigma_\eps$ be the almost-sure spectrum of $P+a_\eps W(x/\eps)$. For every $R>0$ and $\delta'>0$, there is $\eps_0>0$ such that $\dist(E,\Sigma_\eps)\le\delta'$ for all $|E|\le R$ and $\eps\in(0,\eps_0]$.
\end{enumerate}
\end{theorem}
\noindent Consequently, under the hypotheses of (i) or of (iii), no deterministic operator $H^0_B$ gapped at $0$ is a norm-resolvent limit of $H^{\eps_n}_B$ along a sequence $\eps_n\to0$, a.s. Indeed, the homogenized gap is filled by spectrum carried by rare regions. Part (ii) is a convenient way to verify the hypothesis of (i), and is used in Corollary~\ref{cor:pseudoB}: if the law can interpolate between two homogenized phases with different invariants, the gap must close somewhere along the interpolation, leading to a pattern $W_*=W_{t_*}$ charged by the law. In (iii), the amplitude may be $a_\eps=\eps^{-k/2}$ or the critical amplitude of block (vii), and rare regions fill every compact energy window. In large regions where the couplings are nearly constant and small, the operator is nearly unperturbed, shifted by an arbitrary bounded constant.
\begin{definition}[Pseudo-gap]\label{def:pseudo}
Let $H^0_B$ be a deterministic self-adjoint operator with $\spec H^0_B\cap(-g,g)=\emptyset$. The interval $(-g,g)$ is a {\em pseudo-gap of $H^\eps_B$ relative to $H^0_B$} if for every $g'\in(0,g)$, $\Sigma_\eps\cap(-g',g')\not=\emptyset$ for all $\eps$ sufficiently small.
\end{definition}
\noindent The definition is used for the homogenized operators $H^0_B$ of Theorems~\ref{thm:mainB} and~\ref{thm:mainA}, which $H^\eps_B$ approximates outside rare regions (Steps~2--3 of the proof of Theorem~\ref{thm:mainB} and Step~2 of that of Theorem~\ref{thm:mainA}). A pseudo-gap in which the spectrum is localized is a {\em mobility gap}. We do not address localization here. In a pseudo-gap, the straight cut $\{x_1\geq0\}$ is replaced by a partition of the plane that avoids the rare regions.
\begin{definition}[Cuts]\label{def:cut}
A {\em cut}, or {\em deformed cut}, is a partition $\Rm^2=\cL\cup\cR$ into measurable sets such that, for every $X>0$, $\cR\cap\{|x_2|\le X\}$ differs from $\{x_1\geq0,\ |x_2|\le X\}$ by a bounded set. We set $\tilde\Pi=\one_\cR$, and call $\tilde T(H)=\tilde\Pi\,U(H)\,\tilde\Pi+I-\tilde\Pi$ the Toeplitz operator along the cut. Given $\cU\subset\Rm^2$ and $\dd>0$, the cut {\em avoids} $\cU$ at distance $\dd$ if $\dist(\cU\cap\cL,\cR)\geq\dd$ and $\dist(\cU\cap\cR,\cL)\geq\dd$.
\end{definition}
\begin{figure}[ht]
\centering
\includegraphics[width=\textwidth]{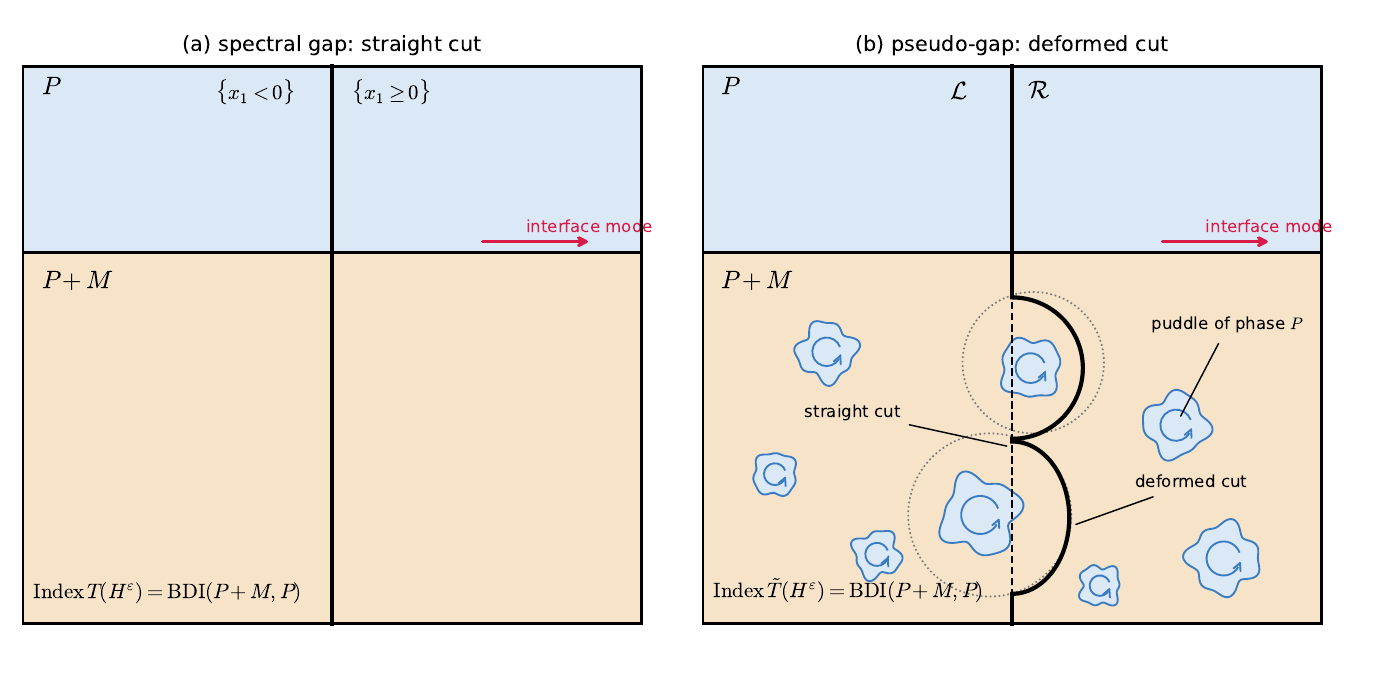}
\caption{(a) Spectral gap: the index is computed along the straight cut. (b) Pseudo-gap: rare regions, such as puddles of the phase $P$ whose boundaries carry states in the gap, may meet the straight cut infinitely often, while the deformed cut detours around them, with index $\BDI(P+M,P)$. See Corollary~\ref{cor:puddles}.}
\label{fig:cuts}
\end{figure}

\noindent The deformation is bounded in any horizontal strip but unbounded globally, so that $\tilde T(H)$ is not a compact perturbation of $T(H)$. Its index is nevertheless stable under bounded perturbations, however large, supported far from the cut, provided the unperturbed operator has gapped bulks and a resolvent with Combes--Thomas decay. This deterministic robustness result, Theorem \ref{thm:robust} of section \ref{sec:pseudo}, is the main basis of Theorems \ref{thm:mainB} and \ref{thm:mainA} below. It also applies to periodic fluctuations, and extends Corollary~\ref{cor:compact} to large defects with unbounded support kept away from a cut (Corollary~\ref{cor:farV}).

\medskip
\noindent\textbf{(vi) First-order alloys in the subcritical regime.} Let $P=\gamma(D)+\Gamma$ with $\gamma(\xi)=\xi_1\gamma_1+\xi_2\gamma_2$, let $\rho=\rho(x_2)$ be as in block (iii), and let $H^\eps=P+\eps^{-1/2}\rho\,W_\omega(x/\eps)$, where $w$ satisfies \eqref{eq:moments} with $k=1$ and the couplings $\omega_j$ are i.i.d., with a law $\mu_0$ supported in $[-\kappa_+,\kappa_+]$, whose support may contain $0$. The effective matrix of the randomly shifted alloy (section \ref{sec:random}) is the Born matrix, the mean-field correction of \cite{FOP82},
\begin{equation}\label{eq:Mborn}
M_{\rm Born}:=\mathrm{Var}(\omega)\int_{\T^2}\mathbb K(q)\,dq+(\EE\omega)^2\,\mathbb K(0).
\end{equation}
We assume the {\em Born gap} condition \textup{[BG]}: there are $g>0$ and $\bar r>0$ with $\dist(0,\spec P)\ge g$ and $\dist(0,\spec(P+M_{\rm Born}))\ge g+\bar r$.
The local effective matrices of Theorem~\ref{thm:unifhom} concentrate around $M_{\rm Born}$ (Lemma~\ref{lem:mean}). The rare regions are the mesoscopic boxes where they deviate from it by more than $\bar r/4$. Let $\cU_\omega$ be their union. The precise definition is given in section \ref{sec:pseudo}. With
\begin{equation}\label{eq:exponents}
r:=\eps^{-\nu_r},\qquad\dd:=\eps^{-\nu_\dd},\qquad 0<\nu_r<\nu_\dd<p_0/2,
\end{equation}
where $p_0$ is the exponent in the probability $O(\eps^{p_0})$ that a unit box meets $\cU_\omega$, here $p_0=1$, a cut is {\em admissible} if it avoids the $3r$-neighborhood $(\cU_\omega)_{3r}$ of the rare regions at distance $\dd/2$. The operator is modified only in $(\cU_\omega)_{3r}$ (section~\ref{sec:pseudo}).
\begin{theorem}[Subcritical first-order alloys]\label{thm:mainB}
Assume [BG]. There is $\eps_0>0$ such that for every $\eps\in(0,\eps_0]$, a.s., admissible cuts exist, and for every admissible cut, $\tilde T(H^\eps)$ is Fredholm with
\[
  \ind\tilde T(H^\eps)=\BDI(P+M_{\rm Born},P).
\]
\end{theorem}
\noindent The index is independent of the admissible cut, of the configuration and of $\bar r$, and the theorem is proved in section \ref{sec:pseudo}. Condition [BG] concerns only the typical local effective matrix: periodic patterns charged by the law whose effective matrix is gapless, or lies in another phase, are allowed, and by Theorem~\ref{thm:pseudogap} they occur in rare regions, which admissible cuts avoid. A simple sufficient condition for [BG] is (A) with $\cK$ replaced by $\{M_{\rm Born}\}$: by Lemma~\ref{lem:grade}, with gradings $s_+J$ and $s_-J$, $P$ and $P+M_{\rm Born}$ are then gapped on $(-g,g)$, so that [BG] holds with $g/2$ in place of $g$ and $\bar r=g/2$. For instance, the single site of Corollary~\ref{cor:bernoulli} with couplings uniform on $[-1,1]$ has $M_{\rm Born}=\frac23\tau\|v\|_{L^2}^2\sigma_3$. If $\tau<0$, $0<g\le m$ and $2|\tau|\|v\|_{L^2}^2\ge3(m+g)$, then [BG] holds, $(-g/2,g/2)$ is a pseudo-gap, and the index along admissible cuts is $-1$. The formula for $M_{\rm Born}$ follows from the computation of $\mathbb K$ in the proof of Corollary~\ref{cor:bernoulli}, whose in-plane part is odd in $q$ (see also Proposition~\ref{S-prop:uniform} of the supplementary material).

\medskip
\noindent\textbf{(vii) Scalar alloys in the critical regime.} We now consider operators of order two in dimension two, with scalar disorder, satisfying the following:
\begin{enumerate}[label=\textup{[C\arabic*]},leftmargin=*,nosep]
\item \emph{operator}: $P=A\,q(D)+P_1(D)+P_0$ with $A\in\Herm(N)$ invertible, possibly indefinite, $q(\xi)=\xi\cdot B\xi$ with $B$ real symmetric positive definite, $P_1(\xi)=\sum_jc_j\xi_j$ with $c_j\in\Herm(N)$, and $P_0\in\Herm(N)$;
\item \emph{fluctuation}: $W=\sum_{j}\omega_jw(\cdot-j)$ with $w\in C^\infty_c(\Rm^2;\Rm)$, and i.i.d.\ centered couplings with a law $\mu_0$ supported in $[-\kappa_+,\kappa_+]$ and $\EE\omega^2>0$;
\item \emph{criticality}: $s_0:=\EE\omega^2(\int w)^2>0$, and the amplitude is renormalized as
\begin{equation}\label{eq:Hcrit}
  H^\eps=P+\rho(x)\tilde V(x),\qquad\tilde V(x)=\lambda_\eps\eps^{-1}W(x/\eps),\qquad\lambda_\eps^2=\frac\ell{\log(1/\eps)},
\end{equation}
for an envelope $\rho\in C^\infty_b(\Rm^2;\Rm)$ and a disorder strength $\ell>0$.
\end{enumerate}
The BHZ operator $\xi\cdot\sigma+(m+b|\xi|^2)\sigma_3+c|\xi|^2$ with $|b|\not=|c|$ satisfies [C1] with $q=|\xi|^2$ and $A=b\sigma_3+c$. Without renormalization, the Born correction diverges logarithmically (section \ref{sec:random}): every dyadic shell of frequencies between $1$ and $1/\eps$ contributes equally to it. With $c_q:=(2\pi)^{-2}\int_{\Sm^1}q(e)^{-1}de$, the renormalized correction is $M_{\rm crit}:=-\,\ell\,s_0\,c_q\,A^{-1}$. For the BHZ model, this is $M_{\rm crit}=\frac{\ell s_0}{2\pi}\,\frac{c\,\sigma_0-b\,\sigma_3}{b^2-c^2}$.
These are the coefficients of the Born approximation of \cite{GWATB09}, with their logarithm absorbed in $\lambda_\eps$.

Homogenization holds for configurations that are small in the following sense. Write $\tilde V=\tilde V_{\rm low}+\tilde V_{\rm high}$ as the sum of its low- and high-frequency parts, and let $Y_0=-q(D)^{-1}\tilde V_{\rm high}$. The {\em local size} $\delta_Q$ of a configuration on a unit box $Q$, defined in \eqref{eq:deltaQ}, measures $Y_0$ and $\tilde V_{\rm low}$ in $L^\infty$ near $Q$, and $\nabla Y_0$ and the fluctuation $\cW$ of the corrector energy $|\nabla Y_0|_B^2=\nabla Y_0\cdot B\nabla Y_0$ about its mean in a local H\"older--Zygmund norm of negative order. When $\delta_Q$ is small on every unit box, $H^\eps$ is close to $P+\rho^2M_{\rm crit}$ in norm-resolvent sense (Theorem~\ref{thm:det}), and, for every fixed threshold $t\in(0,1]$, $\delta_Q<t$ with probability $1-O(\eps^{ct^2/\ell})$ (Lemma~\ref{lem:prob}). On $\R^2$, however, $\sup_Q\delta_Q\ge1$ a.s.\ for $\eps$ small (Lemma~\ref{S-lem:globalsize} of the supplementary material), so that rare regions must be removed. Let $\cU_\omega$ be the union of the {\em bad} unit boxes, on which $\delta_Q$ is not small or $\|\nabla Y_0\|_{L^\infty(Q^*)}$ exceeds a multiple of $\sqrt{\log(1/\eps)}$, defined precisely in Step 1 of the proof of Theorem~\ref{thm:mainA}. Admissible cuts are defined as in block (vi), with $r$ and $\dd$ as in \eqref{eq:exponents} and $p_0>0$ given by Lemma \ref{lem:probboxA}. The thresholds defining bad boxes and the exponents $\nu_r,\nu_\dd$ depend only on $(A,B,c_j,P_0,\rho,g,\varphi)$ and on the law of the couplings.
\begin{theorem}[Critical alloys]\label{thm:mainA}
Assume [C1]-[C3], let $H^\eps$ be the operator \eqref{eq:Hcrit} with renormalized amplitude $\lambda_\eps$ and $\rho=\rho(x_2)$ as in block \textup{(iii)}, and assume that $P$ and $P+M_{\rm crit}$ are gapped on $(-g,g)$. Then there is $\eps_0>0$ such that for every $\eps\in(0,\eps_0]$, a.s., admissible cuts exist, and for every admissible cut, $\tilde T(H^\eps)$ is Fredholm with $\ind\tilde T(H^\eps)=\BDI(P+M_{\rm crit},P)$.
\end{theorem}
\noindent When $A$ is indefinite and $\supp\mu_0$ contains $[0,\kappa_0]$ or $[-\kappa_0,0]$, Theorem~\ref{thm:pseudogap}(iii), with $a_\eps=\lambda_\eps\eps^{-1}$ and since $\int w\ne0$ by [C3], shows that the common gap $(-g,g)$ is a pseudo-gap of $H^\eps_B=P+\tilde V$ relative to $P+M_{\rm crit}$. Thus, the interface index survives a gap that is filled a.s. Theorem \ref{thm:mainA} is proved in section \ref{sec:crit}. It applies to the BHZ model of \cite{GWATB09} (Corollary~\ref{cor:BHZ}) and, with the extension of Remark~\ref{S-rem:Adisorder}, to disordered $p$-wave superconductors (Corollary~\ref{S-cor:pwave}).

\section{Applications}\label{sec:num}

This section examines which fluctuations change the topology. The invariant changes only when $M$ moves the bulk across a phase boundary, so that most corrections, even non-vanishing ones, leave the topology unchanged. Electric potentials and magnetic fields of zero mean, usually regarded as topologically trivial perturbations, nevertheless generate corrections $M$ that can cross this boundary. For Dirac-type operators, this boundary is a change of sign of the mass. For superconductors, it is a change of sign of the chemical potential (section~\ref{S-sec:Sexamples} of the supplementary material). We address the following questions. For periodic fluctuations: which components of the fluctuation produce a mass, and with which sign, for first-order (Proposition~\ref{prop:EM}) and second-order (Propositions~\ref{prop:regDirac} and~\ref{prop:regIndex}) Dirac-type operators. For random fluctuations: how an alloy can preserve a gap while inverting the mass (Corollary~\ref{cor:bernoulli}), why deformed cuts are needed (Corollary~\ref{cor:puddles}), that pseudo-gaps occur for natural laws (Corollary~\ref{cor:pseudoB}), and what the theory gives for the model of \cite{GWATB09} (Corollary~\ref{cor:BHZ}). In the periodic examples, the profiles depend on the fast variable $y_1$ only, so that the fluctuations depend on $x_1$ only, and $v$ denotes a fixed real mean-zero periodic profile with $-\Delta\phi_v=v$. Index maps, interface spectra, convergence rates and further examples are given in section~\ref{S-sec:Sexamples} of the supplementary material.

\paragraph{First-order Dirac operators.} For mean-zero, periodic $A:\T^2\to\R^2$, the Helmholtz decomposition $A=\nabla^\perp\psi+\nabla\zeta$, with $\nabla^\perp=(-\partial_2,\partial_1)$ and $\psi,\zeta$ mean-zero periodic, defines the \emph{stream function} $\psi$ of $A$. Then $\partial_1A_2-\partial_2A_1=\Delta\psi$ is the magnetic field, with Fourier coefficients $-4\pi^2|q|^2\hat\psi_q$. We write $e^\perp=(-e_2,e_1)$ for $e\in\R^2$. Let $P=\gamma(D)+\Gamma$ be first-order, with $d=2$, and let $W=V\,I+A_1\gamma_1+A_2\gamma_2$, with $V,A_j$ real, periodic and mean-zero. By Lemma~\ref{lem:peierls}, this is the electromagnetic perturbation \eqref{eq:peierls} for $k=1$. Let $\psi$ be the stream function of $A=(A_1,A_2)$, and $\mathcal Q(e):=i[\gamma(e)^{-1},\gamma(e^\perp)]$ for $|e|=1$, where we recall that $\gamma(e)=e_1\gamma_1+e_2\gamma_2$, invertible by \textup{[P2]}.
\begin{proposition}[Electric and magnetic fluctuations]\label{prop:EM} \textup{(i)} $M=-\sum_{q\ne0}\Re\bigl(\overline{\hat V_q}\,\hat\psi_q\bigr)\,\mathcal Q(\hat q)$; \textup{(ii)} for Clifford systems ($\gamma_i\gamma_j+\gamma_j\gamma_i=2\delta_{ij}$), $\mathcal Q\equiv2i\gamma_1\gamma_2$. For the Dirac operator ($\gamma=\sigma$), $M=2\avg{V\psi}\,\sigma_3$; and \textup{(iii)} if $P=D\cdot\sigma+m\sigma_3$ with $m(m+2\avg{V\psi})\ne0$, then $\Ind T(H^\eps)=\frac12\bigl(\sgn(m+2\avg{V\psi})-\sgn m\bigr)$ for $\eps\le\eps_0$.
\end{proposition}
\noindent Here $\avg{V\psi}=\int_{\T^2}V\psi\,dy$, and in (iii) $\rho=\rho(x_2)$ is as in block (iii) of section \ref{sec:setting}. By (i), $M=0$ if $W$ is purely electric ($A=0$) or purely magnetic ($V=0$), and the gradient part $\nabla\zeta$ of $A$ never contributes.

\begin{proof}[Proof of Proposition~\ref{prop:EM}]
(i) Fix $q$. Since $\hat A_q=2\pi i\bigl(|q|\hat\psi_q\,\hat q^\perp+\hat\zeta_q\,q\bigr)$, we can write $\hat A_q=s\hat q+t\hat q^\perp$ with $s=2\pi i|q|\hat\zeta_q$ and $t=2\pi i|q|\hat\psi_q$. Write $\gamma_q=\gamma(\hat q)$, $\gamma_\perp=\gamma(\hat q^\perp)$ and $v=\hat V_q$. By \eqref{eq:pair}, with the notation of section~\ref{sec:univ} and $G_1(\hat q)=-\gamma_q^{-1}$, the pair $\pm q$ contributes $-(2\pi|q|)^{-1}(X^*\gamma_q^{-1}X-X\gamma_q^{-1}X^*)$ with $X=v+s\gamma_q+t\gamma_\perp$. For $X=\sum_ax_aA_a$ with Hermitian $A_a$ and any Hermitian $G$,
\[
X^*GX-XGX^*=2i\sum_{a<b}\Im(\bar x_ax_b)\,(A_aGA_b-A_bGA_a).
\]
With $(A_a)=(I,\gamma_q,\gamma_\perp)$ and $G=\gamma_q^{-1}$, the pairs $(I,\gamma_q)$ and $(\gamma_q,\gamma_\perp)$ give $0$, and $(I,\gamma_\perp)$ gives $[\gamma_q^{-1},\gamma_\perp]$, with coefficient $2i\Im(\bar vt)=4\pi i|q|\Re(\bar v\hat\psi_q)$. In particular $s$, which comes from the gauge part $\nabla\zeta$ of $A$, drops out. Hence the pair contributes $-2\Re(\bar v\hat\psi_q)\mathcal Q(\hat q)$. The commutator is even in $q$, so unfolding the sum over $\Lambda_+$ into a sum over all $q\ne0$ gives the formula.

(ii) For Clifford systems, $\gamma(e)^{-1}=\gamma(e)$ and $\gamma(e)\gamma(e^\perp)=\gamma_1\gamma_2=-\gamma(e^\perp)\gamma(e)$, so $\mathcal Q\equiv2i\gamma_1\gamma_2$. For $\gamma=\sigma$ this gives $\mathcal Q=-2\sigma_3$. Then $M=2\sum_q\Re(\overline{\hat V_q}\hat\psi_q)\sigma_3=2\avg{V\psi}\sigma_3$ by Parseval.

(iii) The bulks $D\cdot\sigma+m\sigma_3$ and $D\cdot\sigma+(m+2\avg{V\psi})\sigma_3$ are gapped, and Theorem~\ref{thm:TAI} applies, with the Dirac value $c(m\sigma_3)=\frac12\sgn m$ of Lemma~\ref{lem:BDIc}.
\end{proof}

\noindent As an example, for $m>0$, $V(y)=a\cos(2\pi y_1)$ and $A(y)=(0,b\cos(2\pi y_1+\theta_0))$, we have $\psi=b\sin(2\pi y_1+\theta_0)/(2\pi)$ and $M=(ab\sin\theta_0/2\pi)\,\sigma_3$. Electric and magnetic fluctuations in phase ($\theta_0=0,\pi$) produce no mass, however large their amplitudes. The effect is maximal for $\theta_0=\pm\pi/2$, and the topology changes when $ab\sin\theta_0<-2\pi m$. Since $M$ is odd under time reversal, which maps $A$ to $-A$, changing the topology requires breaking time-reversal symmetry. As in the Haldane model \cite{Haldane88}, the magnetic field $\partial_1A_2-\partial_2A_1=\Delta\psi$ has zero mean, so the topology changes without any net flux. Pure mass fluctuations $V_3\sigma_3$ also give $M=0$ (Remark~\ref{rem:scalar}). Index maps and interface spectra for this example are shown in section~\ref{S-sec:Sexamples} of the supplementary material.

The same fluctuations give a gap-preserving alloy (Theorem~\ref{thm:phasepurek}, case (A)).
\begin{corollary}[A gap-preserving Dirac alloy]\label{cor:bernoulli}
Let $P=D\cdot\sigma+m\sigma_3$ with $m>0$, and let $W_\omega=\sum_j\omega_jw(\cdot-j)$ with $w=v\sigma_0+\tau\nabla^\perp v\cdot\sigma$, $v=\Delta b$, $\supp b\subset\{|y|<\frac12\}$ and $\omega_j=\pm1$. If $\tau<0$ and $2|\tau|\|v\|_{L^2}^2\ge m+g$ with $0<g\le m$, then (A) holds with $J=\sigma_3$, $s_+=1$, $s_-=-1$, and $\Ind T(H^\eps_\omega)=-1$ for every configuration and $\eps\le\eps_0$.
\end{corollary}
\begin{proof}
The single site satisfies \eqref{eq:moments} with $k=1$, since $\int v=0$, $\int\nabla^\perp v=0$ and $\int y_i\partial_jv=-\delta_{ij}\int v=0$. With $\hat w(p)=\hat v(p)(I+2\pi i\tau\,p^\perp\cdot\sigma)$, $G(p)=-(2\pi)^{-1}(p\cdot\sigma)/|p|^2$ and $(p\cdot\sigma)(p^\perp\cdot\sigma)=i|p|^2\sigma_3$, a direct computation gives
\[
\hat w(p)^*G(p)\hat w(p)=2\tau|\hat v(p)|^2\sigma_3-|\hat v(p)|^2\bigl(1-4\pi^2\tau^2|p|^2\bigr)\frac{p\cdot\sigma}{2\pi|p|^2}.
\]
Since $\supp v\subset\{|y|<\frac12\}$, the integer translates of $v$ have disjoint supports, and Poisson summation gives $\sum_\ell|\hat v(q+\ell)|^2=\|v\|_{L^2}^2$ for every $q$. Hence $\mathbb K(q)$ is $2\tau\|v\|_{L^2}^2\sigma_3$ plus an in-plane part, odd in $q$ since $|\hat v|$ is even, and, as $\kappa_\pm=1$, so is every element of $\cK$. The in-plane part anticommutes with $J=\sigma_3$ (see also Proposition~\ref{S-prop:bernoulli}). Hence (A) holds, and Theorem~\ref{thm:phasepurek} with $\bar M=2\tau\|v\|_{L^2}^2\sigma_3\in\cK$ and Lemma~\ref{lem:BDIc} give the index $\frac12(\sgn(m+2\tau\|v\|_{L^2}^2)-\sgn m)=-1$.
\end{proof}

\paragraph{Regularized Dirac and BHZ operators.} For second-order operators, electric and magnetic fluctuations decouple. We consider the regularized Dirac operator $P=D\cdot\sigma+(m+\beta\Delta)\sigma_3$, $\beta\ne0$, and the BHZ-type operator with symbol $\xi\cdot\sigma+(m+b|\xi|^2)\sigma_3+c|\xi|^2$, $|b|\ne|c|$. Since the principal parts are Laplacians, the natural quantities for $W=\sum_{a=0}^3V_a\sigma_a$ are the gradient correlations of the mean-zero periodic solutions $\phi_a$ of $-\Delta\phi_a=V_a$,
\[
\mathcal C_{ab}=\int_{\T^2}\nabla\phi_a\cdot\nabla\phi_b\,dy=\int_{\T^2}V_a\,\phi_b\,dy .
\]
\begin{proposition}[Regularized Dirac and BHZ operators]\label{prop:regDirac} \textup{(i)} For the regularized Dirac operator, $M=\beta^{-1}\big[(\mathcal C_{00}+\mathcal C_{33}-\mathcal C_{11}-\mathcal C_{22})\sigma_3+2\mathcal C_{03}\sigma_0+2\mathcal C_{31}\sigma_1+2\mathcal C_{32}\sigma_2\big]$; \textup{(ii)} for the BHZ-type operator and $W=V_0\sigma_0$, $M=\mathcal C_{00}\,(c\,\sigma_0-b\,\sigma_3)/(b^2-c^2)$; and \textup{(iii)} for the electromagnetic fluctuation \eqref{eq:peierls} with $k=2$ and $A=\nabla^\perp\psi$, and with $\mathcal C_{00}$ computed from $V_0=V$,
\[
M=\Bigl(\frac{\mathcal C_{00}}\beta-\beta\int_{\T^2}|\nabla\psi|^2\Bigr)\sigma_3\quad\text{and}\quad M=\mathcal C_{00}\,\frac{c\,\sigma_0-b\,\sigma_3}{b^2-c^2}+\Bigl(\int_{\T^2}|\nabla\psi|^2\Bigr)(b\,\sigma_3+c)
\]
for the regularized Dirac and the BHZ-type operators, respectively.
\end{proposition}
\noindent The coefficients $b/(b^2-c^2)$ and $c/(b^2-c^2)$ in (ii), and their signs, are those of the Born approximation \cite[Eq.~(5)]{GWATB09}, whose parameters $(\beta,\gamma)$ in \cite[Eq.~(1)]{GWATB09} are our $(b,c)$. Their condition $\beta>\gamma$ is our $|c|<|b|$. In (iii), the electric and magnetic contributions decouple. For the regularized Dirac operator, and for the BHZ-type operator with $|c|<|b|$, they shift the mass in opposite directions. For the profiles $V$ and $A$ of the first-order example, (iii) gives $m+M_3=m+a^2/(8\pi^2\beta)-\beta b^2/2$ for the regularized Dirac operator: the phase $\theta_0$ plays no role, and, for $\beta>0>m$, the index is $1$ if and only if $a^2/(8\pi^2)>\beta|m|+\beta^2b^2/2$, so that magnetic fluctuations raise the electric threshold.

\begin{proof}[Proof of Proposition~\ref{prop:regDirac}]
(i) Here $P_2(\xi)=-\beta|\xi|^2\sigma_3$, so $M=\beta^{-1}\sum_{a,b}\mathcal C_{ab}\,\sigma_a\sigma_3\sigma_b$ with $\mathcal C$ real symmetric. Symmetrizing in $(a,b)$, we have $\sigma_a\sigma_3\sigma_a=\pm\sigma_3$, with $+$ for $a\in\{0,3\}$ and $-$ for $a\in\{1,2\}$; the pair $(0,3)$ gives $2\sigma_0$; the pairs $(3,j)$, $j=1,2$, give $2\sigma_j$; and the pairs $(0,j)$ and $(1,2)$ give $0$.

(ii) We have $P_2(\xi)=|\xi|^2(b\sigma_3+c)$, so $G(q)=-(b\sigma_3+c)^{-1}/(4\pi^2|q|^2)$, and $(b\sigma_3+c)^{-1}=(c-b\sigma_3)/(c^2-b^2)$. Substituting into $M=\sum_q|\hat V_q|^2G(q)$, with $V=V_0$, and using $\sum_q|\hat V_q|^2/(4\pi^2|q|^2)=\int|\nabla\phi_0|^2=\mathcal C_{00}$, gives the formula. With $\mathcal E:=\mathcal C_{00}$, the mass shift is therefore $-\mathcal E\,b/(b^2-c^2)$ and the energy shift is $\mathcal E\,c/(b^2-c^2)$.

(iii) In both cases $a_{ij}=\delta_{ij}a$ with $a=-\beta\sigma_3$, respectively $a=b\sigma_3+c$. Since $\mathrm{div}\,A=0$, Lemma~\ref{lem:peierls} gives $W=V$ and $\avg{d_0}=a\avg{|A|^2}=a\int|\nabla\psi|^2$. The Born term is that of the scalar potential $V$, computed in (i) and (ii).
\end{proof}

\begin{proposition}[Index for regularized Dirac operators]\label{prop:regIndex}
Let $P$ be the regularized Dirac operator with $m\beta\ne0$, and $M=M_0\sigma_0+M_\perp\cdot\sigma+M_3\sigma_3$ as in Proposition~\ref{prop:regDirac}(i). If $|M_0|<\inf_\xi|d(\xi)|$, then \textup{[G]} holds and
\[
\BDI(P+M,P)=\tfrac12\bigl(\sgn m_{\rm eff}-\sgn m\bigr),\qquad m_{\rm eff}:=m+M_3-\beta|M_\perp|^2\ne0 .
\]
\end{proposition}
\noindent Here $d(\xi)=(\xi+M_\perp,\,m+M_3-\beta|\xi|^2)$, so that the lower bulk has eigenvalues $M_0\pm|d(\xi)|$.

\begin{proof}[Proof of Proposition~\ref{prop:regIndex}]
The lower bulk is $H_-=(D+M_\perp)\cdot\sigma+(m+M_3+\beta\Delta)\sigma_3+M_0$, with symbol $d(\xi)\cdot\sigma+M_0$ and eigenvalues $M_0\pm|d(\xi)|$. The components $d_1,d_2$ vanish only at $\xi=-M_\perp$, where $d_3=m_{\rm eff}$. Hence $\inf|d|>0$ if and only if $m_{\rm eff}\ne0$, and $|M_0|<\inf|d|$ gives $m_{\rm eff}\ne0$ and a gap of $H_-$ at $0$. Since $P$ is gapped ($m\ne0$), \textup{[G]} holds. By definition, $\BDI(P+M,P)=c(h_-)-c(h_+)$. Since $(d_1,d_2)$ vanishes only at $-M_\perp$, with $\det\partial_\xi(d_1,d_2)=1$, and $\sgn d_3(\infty)=-\sgn\beta$, \eqref{eq:degree} gives $c(h_-)=\frac12(\sgn m_{\rm eff}+\sgn\beta)$ and $c(h_+)=\frac12(\sgn m+\sgn\beta)$, which proves the claim.
\end{proof}

\begin{example}\label{ex:counter}
Take $m=\beta=1$, $V_0=V_2=0$ and $V_1=V_3=\lambda v$ for a fixed mean-zero trigonometric polynomial $v$, and let $-\Delta\phi_v=v$. Then $\mathcal C_{11}=\mathcal C_{33}$, so $m+M_3=m$, while $m_{\rm eff}=m-4\lambda^4\bigl(\int|\nabla\phi_v|^2\bigr)^2$. For $\lambda$ large, $m_{\rm eff}<0$ and $\BDI=-1$. The topology changes although $m+M_3=m$, through the vector part $M_\perp$ alone.
\end{example}
\noindent The vector part $M_\perp$ arises only from correlations of $V_3$ with $(V_1,V_2)$. In-plane fluctuations $V_1\sigma_1+V_2\sigma_2$ are matrix-valued potentials, not gauge-covariant magnetic fields. The latter enter through (iii) of Proposition~\ref{prop:regDirac}.

\paragraph{Mass puddles and deformed cuts.} Consider the Dirac operator $H_{\circ}=D\cdot\sigma+m(x_2)\sigma_3$, with $m$ smooth, $m=m_0>0$ for $x_2\geq1$ and $m=-m_0$ for $x_2\le-1$. The bulk operators $D\cdot\sigma\pm m_0\sigma_3$ are gapped on $(-m_0,m_0)$ and the interface $x_2=0$ carries an asymmetric transport with $\ind T(H_{\circ})=\frac12(\sgn(-m_0)-\sgn m_0)=-1$ by Lemma~\ref{lem:BDIc}. We consider fluctuations with {\em puddles}, bounded disjoint open sets $\cP_n$ in the lower half-plane $\{x_2<-1\}$ in which the mass takes the value $m_0$ of the upper phase:
\[
  H=H_{\circ}+V,\qquad V=2m_0\,\one_{\cP}\,\sigma_3,\qquad \cP=\cup_n\cP_n .
\]
\begin{corollary}\label{cor:puddles}
Let $\supp\varphi'\subset(-m_0,m_0)$. There is $\dd_0>0$, depending only on $m_0$ and $\varphi$, such that for every cut $(\cL,\cR)$ avoiding $\cP$ at distance $\dd\geq\dd_0$, the operator $\tilde T(H)$ is Fredholm and $\ind\tilde T(H)=-1$.
\end{corollary}
\begin{proof}
The operator $H_{\circ}$ is a reference operator in the sense of section~\ref{sec:pseudo} (conditions [R1]--[R2] before Theorem~\ref{thm:robust}): [R1] holds with $a=\frac12$ by the Combes--Thomas estimate of Lemma \ref{lem:CT}(i), and [R2] holds by Lemma \ref{lem:fred}. The perturbation $V$ is bounded by $K_V=\max(2m_0,1)$ and supported in $\cP$, at distance at least $\dd$ from the other side of the cut. Choose $\dd_0$ with $C_1K_V^3(\dd_0/2)^{-1}<\frac1{16}$. Theorem \ref{thm:robust} gives $\ind\tilde T(H)=\ind T(H_{\circ})=-1$.
\end{proof}
\noindent The puddles may be arbitrarily large and arbitrarily close to the interface away from the cut, and they may fill the bulk gap with spectrum: the boundary of a puddle is a bounded interface loop between two phases and carries states circulating along it, and large puddles carry such states at energies arbitrarily close to $0$. The straight cut may, however, fail. Let $E_{\frac12}$ be such that $\varphi(E_{\frac12})=\frac12$, and consider puddles centered on the line $x_1=0$ whose radius is chosen so that one of their states $\psi$ has energy $E_{\frac12}$ and is the only puddle state with energy in $\supp\varphi'$. By symmetry $\|\Pi\psi\|^2\approx\frac12$ and, heuristically, $U(H)\approx I-2\psi\otimes\psi^*$ near the puddle, so that $\Pi U(H)\Pi(\Pi\psi)\approx(1-2\|\Pi\psi\|^2)\Pi\psi\approx0$. Infinitely many such puddles along the straight cut would provide an approximate singular sequence for $T(H)$, which is then not expected to be Fredholm. The deformed cut avoids exactly this obstruction.

\paragraph{Pseudo-gaps for alloys with small couplings.} Consider the setting of block (vi) of section \ref{sec:setting}, with a law of couplings whose support contains small values, as for couplings uniformly distributed on an interval containing $0$.
\begin{corollary}\label{cor:pseudoB}
Assume [BG] and $\supp\mu_0\supset[0,\kappa_0]$, and that $P+\kappa_0^2\mathbb K(0)$ is gapped at $0$ with $\BDI(P+\kappa_0^2\mathbb K(0),P)\not=0$. Then, for every $g'>0$ and $\eps$ small, $\spec(P+\eps^{-\frac12}W_\omega(x/\eps))\cap(-g',g')\not=\emptyset$ a.s.
\end{corollary}
\noindent The spectral gap $(-g,g)$ of $P+M_{\rm Born}$ is thus a pseudo-gap of the bulk operator relative to $P+M_{\rm Born}$, while the index of Theorem \ref{thm:mainB} is defined, and non-trivial when $\BDI(P+M_{\rm Born},P)\ne0$. The two invariants $\BDI(P+\kappa_0^2\mathbb K(0),P)$ and $\BDI(P+M_{\rm Born},P)$ coincide when $\kappa_0^2\mathbb K(0)$ and $M_{\rm Born}$ can be joined by a path of matrices $M$ with $P+M$ gapped at $0$, e.g., under the graded criterion stated after Theorem \ref{thm:mainB}, applied along the segment between them.
\begin{proof}
Let $W_t$, $t\in[0,1]$, be the constant-coupling alloys with $\omega_j\equiv t\kappa_0$. They are mean-zero periodic patterns since $\int w=0$, and the law charges all of them since $[0,\kappa_0]\subset\supp\mu_0$. The structure factor of $W_t$ is $t^2\kappa_0^2\delta_0$, so that $M(W_t)=t^2\kappa_0^2\mathbb K(0)$. Hence $P+M(W_0)=P$ is gapped at $0$ by [BG], and $P+M(W_1)=P+\kappa_0^2\mathbb K(0)$ is gapped at $0$ with $\BDI(P+M(W_1),P)\not=0$ by assumption. The path $t\mapsto W_t$ is charged by the law, and Theorem \ref{thm:pseudogap}(i)--(ii) with $k=1$ gives $\dist(0,\spec(P+\eps^{-\frac12}W_\omega(x/\eps)))\le C\eps^{\frac12}$ a.s.
\end{proof}
\noindent The rare regions are here puddles of the unperturbed phase: regions where all couplings are small, so that the local effective matrix nearly vanishes. A unit box meets a mesoscopic region where the local effective matrix deviates from $M_{\rm Born}$ with probability at most $C\eps^{-3/2}\exp(-c\eps^{-1/2})$ (Lemma~\ref{lem:mesobad}), and puddles large enough to carry spectrum near $0$ are much rarer. Admissible cuts therefore exist a.s. (Lemma~\ref{lem:cuts}).

\paragraph{The BHZ model with on-site disorder.} The topological Anderson insulator of \cite{GWATB09,LCJS09} is modeled by the BHZ operator $P=\xi\cdot\sigma+(m+b|\xi|^2)\sigma_3+c|\xi|^2$ with uniform on-site disorder. In the setting of block (vii), we define the renormalized mass and energy shift
\[
  \bar m=m-\frac{\ell s_0b}{2\pi(b^2-c^2)},\qquad e_0=\frac{\ell s_0c}{2\pi(b^2-c^2)} .
\]
\begin{corollary}\label{cor:BHZ}
Let $|c|<|b|$, $m\not=0$ and $\bar m\not=0$, and assume that $P$ and $P+M_{\rm crit}$ are gapped on $(-g,g)$. Then the conclusions of Theorem \ref{thm:mainA} hold with $\ind\tilde T(H^\eps)=\frac12(\sgn\bar m-\sgn m)$. If moreover $\supp\mu_0\supset[0,\kappa_0]$ or $[-\kappa_0,0]$ for some $\kappa_0>0$, as for uniform on-site disorder, then $(-g,g)$ is a pseudo-gap relative to $P+M_{\rm crit}$.
\end{corollary}
\begin{proof}
The operator $P+M_{\rm crit}$ is the BHZ operator with $m$ replaced by $\bar m$ and an energy shift $e_0\sigma_0$. Both bulk symbols are of the form $h(\xi)=d(\xi)\cdot\sigma+f(\xi)\sigma_0$ with $d=(\xi_1,\xi_2,m'+b|\xi|^2)$ and $m'\in\{m,\bar m\}$. Since $|c|<|b|$, the two eigenvalues $f\pm|d|$ tend to $\pm\infty$, and by the gap assumption exactly one of them is negative for every $\xi$. Since $(d_1,d_2)=(\xi_1,\xi_2)$ vanishes only at $0$, where $d_3=m'$, and $\sgn d_3(\infty)=\sgn b$, \eqref{eq:degree} gives $c(h)=\frac12(\sgn m'-\sgn b)$ \cite{bernevig2013topological}. Hence $\BDI(P+M_{\rm crit},P)=\frac12(\sgn\bar m-\sgn m)$. Since $A=b\sigma_3+c$ is indefinite and $\int w\ne0$ by [C3], the last statement follows from Theorem \ref{thm:pseudogap}(iii) with $a_\eps=\lambda_\eps\eps^{-1}$.
\end{proof}
\noindent A trivial phase becomes topological when the renormalization inverts the mass, $\sgn\bar m\not=\sgn m$. The two phase boundaries are the mass inversion $\bar m=0$ and the closing of the common gap, which for the homogenized operator is a band-edge crossing. These are the curves A and B of \cite[Fig.~1]{GWATB09}, now with an index defined on $\Rm^2$ despite the pseudo-gap.

\noindent For $c=0$ and $b=-\beta$, $P$ is the regularized Dirac operator with scalar disorder, the setting of \cite{BalDang24}, and Corollary~\ref{cor:BHZ} applies with $e_0=0$ and $\bar m=m+\ell s_0/(2\pi\beta)$, both bulks being gapped when $m\bar m\ne0$. The index $\frac12(\sgn\bar m-\sgn m)$ is non-trivial exactly when $m\beta<0$ and $\ell s_0>2\pi|m\beta|$: scalar fluctuations generate a topological phase only when the mass and the regularization have opposite signs and the disorder is strong enough. This is the behavior predicted for electric fluctuations in \cite{BalDang24}, now established in the critical regime and in the presence of the pseudo-gap.

\noindent More generally, scalar disorder at criticality shifts the zeroth-order part of $P$ by $M_{\rm crit}=-\ell s_0c_qA^{-1}$, a direction fixed by the principal symbol alone, and the topology changes when $\BDI(P+M_{\rm crit},P)\not=0$. This applies beyond Dirac-type models. Disordered $p$-wave superconductors are treated in section~\ref{S-sec:Sexamples} of the supplementary material.

\section{Norm-resolvent homogenization}\label{sec:hom}

This section proves Theorem~\ref{thm:hom}. We use the notation \eqref{eq:Heps}--\eqref{eq:M}. Self-adjointness of $H^\eps$ on $H^k$ follows from the Kato--Rellich theorem \cite[Thm.~V.4.3]{Kato95}, \cite[Thm.~X.12]{ReedSimon2}, since $\mathcal V^\eps$ has order $<k$ and bounded coefficients. The only stability input is the obvious bound
\begin{equation}\label{eq:trivial}
\|(H^\eps-z)^{-1}\|\le|\Im z|^{-1},
\end{equation}
valid for any self-adjoint operator. Everything else is a two-scale construction. The corrector hierarchy is the classical multiscale expansion of periodic homogenization \cite[Ch.~1]{BLP78}, \cite[Ch.~1]{JKO94}. Its version for potentials of size $\eps^{-1}$ appears in \cite[Ch.~1, \S12]{BLP78} and \cite{ACPSV04}. To obtain \emph{operator-norm} estimates without extra regularity, we build the correctors on a smoothed version of the homogenized solution. This is a variant of the shift and Steklov-smoothing method of \cite{ZhikovPastukhova05}. We build an approximate solution $\psi_{\rm app}\in H^k$ such that, for $\psi^0=(H^0-z)^{-1}f$,
\[
\|(H^\eps-z)\psi_{\rm app}-f\|\le C\eps^{r_k}(1+|z|)\|\psi^0\|_{H^k}\qquad\text{and}\qquad\|\psi_{\rm app}-\psi^0\|\le C\eps^{r_k}\|\psi^0\|_{H^k}.
\]
These bounds, \eqref{eq:trivial} and the elliptic estimate $\|\psi^0\|_{H^k}\le C(1+|z|)(1+|\Im z|^{-1})\|f\|$ give the bound of Theorem~\ref{thm:hom} with $(1+|z|)^2$ in place of $(1+|\Re z|)^2$. The elliptic estimate holds because $H^0$ is a bounded perturbation of $P$, and $\|u\|_{H^k}\le C(\|P(D)u\|+\|u\|)$ follows from [P2] and Plancherel. This suffices for $|\Im z|\le1$, where $1+|z|\le2(1+|\Re z|)$. For $|\Im z|>1$, let $z_0=\Re z+i\,\sgn(\Im z)$. Then
\begin{equation}\label{eq:shiftz}
R^\eps(z)-R^0(z)=\bigl(1+(z-z_0)R^\eps(z)\bigr)\bigl(R^\eps(z_0)-R^0(z_0)\bigr)\bigl(1+(z-z_0)R^0(z)\bigr),
\end{equation}
where the outer factors have norm $\le1$, since $|\lambda-z_0|\le|\lambda-z|$ for $\lambda\in\R$, and the bound at $z_0$, where $1+|z_0|\le2(1+|\Re z|)$, gives the bound at $z$. Nothing depends on the dimension $d$.

\begin{lemma}[Cell problems]\label{lem:cell}
For every mean-zero $e\in C^{0,\alpha_0}(\T^d;\C^{N\times N})$, the equation $P_k(D_y)\phi=e$ has a unique mean-zero periodic solution, and $\phi\in C^{k,\alpha_0}$.
\end{lemma}
\begin{proof}
In Fourier, $\hat\phi_q=P_k(2\pi q)^{-1}\hat e_q$ for $q\ne0$. By [P2], $P_k(\xi)^{-1}$ is smooth and homogeneous of degree $-k$ on $\R^d\setminus0$. Multiplied by a smooth cutoff vanishing near $0$, which does not change its values on $2\pi(\Z^d\setminus0)$, it is a Fourier multiplier of order $-k$, which maps the H\"older--Zygmund space $C^{\alpha_0}_*$ into $C^{k+\alpha_0}_*$ \cite[Ch.~2]{BCD11}. Periodic functions are bounded functions on $\R^d$, to which this applies. Since $\alpha_0\notin\Z$, these spaces are $C^{0,\alpha_0}$ and $C^{k,\alpha_0}$.
\end{proof}

Let $S_\eps$ denote the Fourier multiplier by $\mathbf 1_{\{|\xi|\le1/\eps\}}$. By Plancherel, for $u\in H^k$, $t\ge0$ and $\eps\le1$,
\begin{equation}\label{eq:smooth}
\|(I-S_\eps)u\|_{H^s}\le\eps^{k-s}\|u\|_{H^k}\quad(0\le s\le k),\qquad\|S_\eps u\|_{H^{t}}\le C_t\eps^{-(t-k)_+}\|u\|_{H^k}.
\end{equation}
Smoothing at the oscillation scale is standard in operator estimates \cite{ZhikovPastukhova05,BirmanSuslina06,KLS12,Shen18}. Here it is a sharp Fourier truncation, which suffices in $L^2$.

For periodic matrix-valued $b$ and vector-valued $h$, the Leibniz rule gives
\begin{equation}\label{eq:leib}
P(D)[b^\eps h]=\sum_{|\alpha|\le k}\ \sum_{\nu\le\alpha}\tbinom\alpha\nu\,\eps^{-|\nu|}\bigl(A_\alpha D_y^\nu b\bigr)^{\eps}D^{\alpha-\nu}h .
\end{equation}
Here $\nu$ is a multi-index, and the matrices $A_\alpha$ act on the left of $D_y^\nu b$. The term with $|\nu|=k$ is $\eps^{-k}(P_k(D_y)b)^\eps h$. Every term with $|\nu|\ge1$ has a coefficient that is a $y$-derivative of a periodic function, and hence has mean zero.

\begin{proof}[Proof of Theorem~\ref{thm:hom} when $k=1$]
Let $k=1$, so $P=\gamma(D)+\Gamma$ with $\gamma(\xi)=\sum_j\xi_j\gamma_j$, and write $L=\gamma(D_y)$ on mean-zero periodic functions. Also, \eqref{eq:Heps} reads $\mathcal V^\eps=\eps^{-1/2}\rho W(x/\eps)+\mathcal R^\eps$, with $W$ Hermitian and $\|\mathcal R^\eps\|\le C\eps^{1/2}$. The remainder adds at most $C\eps^{1/2}\|\psi_{\rm app}\|$ to the residual below, so we may assume $\mathcal R^\eps=0$. The correctors are those of the formal two-scale expansion in powers of $\eps^{1/2}$. Let $\chi=-L^{-1}W$ and $\chi_2=-L^{-1}(W\chi-M)$. Both are in $C^{1,\alpha_0}$ by Lemma~\ref{lem:cell}. By \eqref{eq:M}, $\avg{W\chi}=M$, so $\chi_2$ is well defined. Set
\[
\psi_{\rm app}=\psi^0+\eps^{1/2}\rho\,\chi^\eps\psi^0+\eps\rho^2\chi_2^\eps\psi^0\in H^1 .
\]
Using $D_j[\chi^\eps]=\eps^{-1}(D_{y_j}\chi)^\eps$, $L\chi=-W$ and $L\chi_2=-(W\chi-M)$, a direct computation gives
\begin{align*}
(H^\eps-z)\psi_{\rm app}-f&=\eps^{1/2}\Big[\gamma_j\chi^\eps D_j(\rho\psi^0)+(\Gamma-z)\rho\chi^\eps\psi^0+\rho^3(W\chi_2)^\eps\psi^0\Big]\\
&\quad+\eps\Big[\gamma_j\chi_2^\eps D_j(\rho^2\psi^0)+(\Gamma-z)\rho^2\chi_2^\eps\psi^0\Big].
\end{align*}
Here the $O(\eps^{-1/2})$ terms cancel because $L\chi=-W$. The $O(1)$ terms combine to $\rho^2(W\chi-M+L\chi_2)^\eps\psi^0=0$, using $(H^0-z)\psi^0=f$.

The right-hand side is bounded by $C\eps^{1/2}(1+|z|)\|\psi^0\|_{H^1}$, with $C$ depending on $\|\chi\|_{C^1}$, $\|\chi_2\|_{C^1}$, $\|W\|_\infty$ and $\|\rho\|_{W^{1,\infty}}$. Moreover $\|\psi_{\rm app}-\psi^0\|\le C\eps^{1/2}\|\psi^0\|$. Since $\|\psi^0\|_{H^1}\le C(1+|z|)(1+|\Im z|^{-1})\|f\|$ holds from $\|\nabla u\|\le c_P^{-1}\|\gamma(D)u\|$, the claim follows as seen above Lemma \ref{lem:cell}.
\end{proof}

\begin{proof}[Proof of Theorem~\ref{thm:hom} for $k\geq2$]
No smoothing is needed at $k=1$ because the second corrector costs only one derivative of $\psi^0$. For $k\ge2$, the correctors must be built on a smoothed solution. We build $\psi_{\rm app}=\psi^0+\sum_{\mathfrak m}c_{\mathfrak m}$, where each corrector $c_{\mathfrak m}$ cancels one term of order $\le\frac12$ produced by $H^\eps-z$, and show that the procedure stops after finitely many steps, leaving only terms of order $\ge1$.

Let $v=S_\eps\psi^0$. A \emph{monomial} is an expression $\mathfrak m=\eps^ab^\eps Qv$, with $a\in\tfrac12\Z$, $b\in C^{0,\alpha_0}(\T^d)$ matrix-valued, and $Q$ a differential operator of order $q$ with $C_b^\infty$ coefficients. Its \emph{order} is $a$ and its \emph{weight} is $w(\mathfrak m)=a-q$. The order counts powers of $\eps$. The weight discounts the derivatives falling on $v$: beyond the $k$ derivatives controlled by $\psi^0\in H^k$, each costs a factor $\eps^{-1}$ thanks to the smoothing $S_\eps$. By \eqref{eq:smooth},
\begin{equation}\label{eq:mono}
\|\mathfrak m\|\le C\|b\|_\infty\,\eps^{\min(a,\,w+k)}\|\psi^0\|_{H^k}.
\end{equation}

\emph{Action of $H^\eps-z$.} Let $c=\eps^ab^\eps Qv$ with $b\in C^{k,\alpha_0}$. By \eqref{eq:leib}, $P(D)c$ consists of the \emph{leading} term $\eps^{a-k}(P_k(D_y)b)^\eps Qv$ and of \emph{$P$-terms} of order $a-|\nu|\ge a-k+1$ and weight $\ge w(c)-k$, with mean-zero coefficients when $|\nu|\ge1$. By the analogue of \eqref{eq:leib} for $\mathrm{Sym}[\rho\,c_\alpha^\eps D^\alpha]$, the remaining terms of $(H^\eps-z)c$, called \emph{potential terms}, have order $\ge a-k/2$ and weight $\ge w(c)-k/2$. Their part of order exactly $a-k/2$ is $\eps^{a-k/2}\rho\,(\mathcal T_yb)^\eps Qv$, with $\mathcal T_y:=\sum_\alpha\mathrm{Sym}[c_\alpha D_y^\alpha]$, where all derivatives of the first sum in \eqref{eq:Heps} fall on the fast variable. All other potential terms have order $\ge a+1-k/2$: this holds for the first sum when fewer derivatives fall on the fast variable and, since $k\ge2$, for the second sum, for $\mathcal R^\eps$ and for $-zc$.
\emph{The hierarchy.} Write $\mathcal V^\eps\psi^0=\mathcal V^\eps v+\mathcal V^\eps(I-S_\eps)\psi^0$. By \eqref{eq:smooth}, $\|\mathcal V^\eps(I-S_\eps)\psi^0\|\le C\eps^{k/2}\|\psi^0\|_{H^k}$, since a term of amplitude $\eps^{|\alpha|-k/2}$ with $|\mu|$ fast derivatives is bounded by $\eps^{|\alpha|-k/2-|\mu|}\|(I-S_\eps)\psi^0\|_{H^{|\alpha|-|\mu|}}$ (the second sum and $\mathcal R^\eps$ give $\eps^k$ and $\eps^{k/2+1}$). Applying the previous paragraph to $v$ (with $b=1$), the term of order $-k/2$ in $\mathcal V^\eps v$ is $\eps^{-k/2}\rho W^\eps v$, with $W=\mathcal T_y[I]$ given by \eqref{eq:Weff}. The first corrector is
\[
c_1=\eps^{k/2}\chi^\eps\rho v,\qquad\chi=-P_k(D_y)^{-1}W,
\]
of order and weight $k/2$. Its leading term cancels $\eps^{-k/2}\rho W^\eps v$. Two terms of order $0$ have nonzero mean: the potential term $\rho^2(\mathcal T_y\chi)^\eps v$ of $c_1$, and the potential term $\rho^2\bigl(\sum_\alpha\mathrm{Sym}[d_\alpha D^\alpha_y][I]\bigr)^\eps v$ of $v$ from the second sum. Since $\mathcal T_y$ is symmetric on $L^2(\T^d)$, $\avg{\mathcal T_y\chi}=\avg{(\mathcal T_y[I])^*\chi}=\avg{W^*\chi}$. Moreover $\avg{\mathrm{Sym}[d_\alpha D^\alpha_y][I]}=\avg{d_0^{\rm h}}$, since derivatives have mean zero. Hence the two mean terms sum to $\rho^2Mv$ with $M$ as in \eqref{eq:M}, and $(P-z)\psi^0+\rho^2Mv=f-\rho^2M(I-S_\eps)\psi^0$. Subtracting $\rho^2Mv$ leaves monomials of order $0$ with mean-zero coefficients.

Call a monomial \emph{to be corrected} if its order is $\le\frac12$ and it is not one of the terms just accounted for. Each monomial to be corrected, $\mathfrak m=\eps^ae^\eps Qv$, is cancelled by $c_{\mathfrak m}=-\eps^{a+k}(P_k(D_y)^{-1}e)^\eps Qv$, of order $a+k$ and weight $w(\mathfrak m)+k$. We verify in three steps that this procedure terminates with a small residual.

\emph{(a) Bookkeeping.} Monomials to be corrected are of three kinds: potential terms of $v$ other than $\eps^{-k/2}\rho W^\eps v$, of order $\ge1-k/2$ and weight $\ge-k/2$; the two order-$0$ terms above after subtraction of their means, of weight $\ge0$; and $P$-terms of correctors. By induction on the generation, all of them have order $\ge1-k/2$ and weight $\ge-k/2$, and all correctors have order and weight $\ge k/2$: the $P$-terms of a corrector of order $a_c\ge k/2$ and weight $\ge k/2$ have order $\ge a_c-k+1\ge1-k/2$ and weight $\ge-k/2$. The potential terms of correctors never need correction: those of $c_1$, of order $k/2$, have order $\ge1$ except $\rho^2(\mathcal T_y\chi)^\eps v$, already accounted for, and every other corrector has order $\ge1+k/2$, so that its potential terms have order $\ge1$. Finally, a corrector generated by a term of order $a$ has order $a+k$, and its $P$-terms have order $\ge a+1$. Orders of terms to be corrected thus increase by at least $1$ at each generation, starting from at least $1-k/2$ and remaining at most $\frac12$. The hierarchy therefore stops after at most $\lceil k/2\rceil$ generations of correctors beyond $c_1$, i.e., $1+\lceil k/2\rceil$ generations in total. The bound is attained in general, since the $P$-terms with $|\nu|=k-1$ raise the order by exactly $1$. For $k=2$, there is one generation beyond $c_1$, made of correctors of order $2$.

\emph{(b) Mean zero.} Every monomial to be corrected has a mean-zero coefficient in $C^{0,\alpha_0}$. A $P$-term with $\nu=0$ has the order $a_c\ge k/2\ge1$ of its corrector, so the $P$-terms to be corrected have $|\nu|\ge1$; potential terms of $v$ have coefficients $c_\alpha$, $f$, or their derivatives; and the order-$0$ terms have had their means removed. By Lemma~\ref{lem:cell}, all corrector coefficients lie in $C^{k,\alpha_0}$, and $\psi_{\rm app}\in H^k$.

\emph{(c) Residual.} The residual consists of $\mathcal V^\eps(I-S_\eps)\psi^0$, $-\rho^2M(I-S_\eps)\psi^0$, and monomials of order $a>\frac12$ and weight $\ge-k/2$. Orders lie in $\tfrac12\Z$, so $a\ge1$. Moreover $w+k\ge k/2\ge1$. By \eqref{eq:mono} and \eqref{eq:smooth}, the residual is $\le C\eps(1+|z|)\|\psi^0\|_{H^k}$. Since correctors have order and weight $\ge k/2$, $\|\psi_{\rm app}-\psi^0\|\le C\eps^{k/2}\|\psi^0\|_{H^k}\le C\eps\|\psi^0\|_{H^k}$. These are the bounds required before Lemma~\ref{lem:cell}.
\end{proof}

\section{Attainable corrections and Peierls substitution}\label{sec:univ}

This section proves Theorem \ref{thm:univ} and analyzes the form of the fluctuation obtained by Peierls substitution, which is used in the proof of Theorem \ref{thm:magno}.

\paragraph{Attainable sets.}

Recall $G(q)=-P_k(2\pi q)^{-1}$ from \eqref{eq:M}. For a Hermitian potential $W$, grouping $q$ and $-q$ in \eqref{eq:M}, and using $G(-q)=(-1)^kG(q)$ and $G(q)=(2\pi|q|)^{-k}G_1(\hat q)$, where $\hat q=q/|q|$ and $G_1(e):=-P_k(e)^{-1}$ for $|e|=1$, we obtain
\begin{equation}\label{eq:pair}
M=\sum_{q\in\Lambda_+}(2\pi|q|)^{-k}\begin{cases}\hat W_q^*G_1(\hat q)\hat W_q-\hat W_qG_1(\hat q)\hat W_q^*,&k\text{ odd},\\[2pt] \hat W_q^*G_1(\hat q)\hat W_q+\hat W_qG_1(\hat q)\hat W_q^*,&k\text{ even},\end{cases}
\end{equation}
where $\Lambda_+$ contains one element of each pair $\{q,-q\}$. For a Hermitian $G$, write $G^\perp=\{H\in\Herm(N):\tr(GH)=0\}$.

\begin{lemma}[Null cone]\label{lem:null}
Let $G\in\Herm(N)$ be invertible and indefinite, with $N\ge2$. Then $\mathrm{span}_\R\{vv^*:v^*Gv=0\}=G^\perp$.
\end{lemma}
\begin{proof}
The inclusion $\subseteq$ is clear, since $\tr(Gvv^*)=v^*Gv$. For the converse, it suffices to show that if $H$ is Hermitian with $v^*Hv=0$ on the null cone of $G$, then $H\in\R G$. By Sylvester's law of inertia we may assume $G=\mathrm{diag}(I_a,-I_b)$ with $a,b\ge1$. Denote the basis vectors by $e_i$ ($i\le a$) and $f_j$ ($j\le b$).
For all $\phi$, the vectors $e_i+e^{i\phi}f_j$ are null. This gives $H_{e_if_j}=0$ and $H_{e_ie_i}=-H_{f_jf_j}=:c$ for all $i,j$. For $i\ne i'$, the vectors $(e_i+e^{i\phi}e_{i'})/\sqrt2+f_j$ are also null, so that $\Re(e^{i\phi}H_{e_ie_{i'}})=0$ for all $\phi$, hence $H_{e_ie_{i'}}=0$. The negative block is treated similarly, and $H=cG$.
\end{proof}

\begin{lemma}[Single-direction cones]\label{lem:cones}
Let $G$ be invertible and indefinite, and let $\mathcal S^\mp_G=\{X^*GX\mp XGX^*:X\in\C^{N\times N}\}$. Then: (i) the finite sums of elements of $\mathcal S^-_G$ form $G^\perp$; and (ii) the finite sums of elements of $\mathcal S^+_G$ form $\Herm(N)$. If instead $G$ is definite, the finite sums of elements of $\mathcal S^+_G$ form the cone of semidefinite matrices with the sign of $G$.
\end{lemma}
\begin{proof}
(i) Since $\tr(GX^*GX)=\tr(GXGX^*)$, we have $\mathcal S^-_G\subset G^\perp$. Let $v$ be null and $u^*Gu=1$. Then $X=uv^*$ gives $vv^*$, and $X=vu^*$ gives $-vv^*$. We conclude (i) with Lemma~\ref{lem:null}.

(ii) For $u$ null, $X=uv^*$ gives $(v^*Gv)uu^*$, and $v^*Gv$ may take either sign. By Lemma~\ref{lem:null}, finite sums therefore contain $G^\perp$. The choice $X=I$ gives $M_+=2G$, with $\tr(GM_+)>0$. With $u_\pm^*Gu_\pm=\pm1$, the choice $X=u_-u_+^*$ gives $M_-$ with $\tr(GM_-)=-2$. Any $H$ can be written $\lambda M_\pm+(H-\lambda M_\pm)$ with $\lambda\ge0$ and the second term in $G^\perp$.

Finally, if $G<0$ (the case $G>0$ is symmetric), every element of $\mathcal S^+_G$ is $\le0$, and $X=tvv^*$ with $|v|=1$ gives $2t^2(v^*Gv)vv^*$, an arbitrary nonpositive multiple of $vv^*$.
\end{proof}

\begin{proof}[Proof of Theorem~\ref{thm:univ}]
\emph{(iii) and (i) for odd $k$.} The modes $q=jq_0$, $j\ge1$, all have $\hat q=\hat q_0$, and the positive factors $(2\pi j|q_0|)^{-k}$ can be absorbed into $\hat W_q$. Lemma~\ref{lem:cones}(i) gives (iii): every element of $G_1(\hat q_0)^\perp$, a space of dimension $N^2-1$, is a sum of at most $N^2-1$ elements of $\mathcal S^-_{G_1(\hat q_0)}$, each produced by one mode $q=jq_0$. For (i), take two lattice directions with non-proportional $G_1(\hat q)$. These exist as otherwise continuity would give $P_k(e)=p(e)A$ with $p$ scalar, odd and non-vanishing on the connected sphere $S^{d-1}$, which is impossible. Then $G_1(\hat q_1)^\perp+G_1(\hat q_2)^\perp=\Herm(N)$.

\emph{(i) for even $k$.} This follows from Lemma~\ref{lem:cones}(ii) with a single direction.

\emph{(ii).} This follows from the last statement of Lemma~\ref{lem:cones}.
\end{proof}

\begin{lemma}[Peierls substitution]\label{lem:peierls}
Let $k\in\{1,2\}$, and write $P_2(\xi)=\sum_{i,j}a_{ij}\xi_i\xi_j$ with $a_{ij}=a_{ji}$ if $k=2$. Then \eqref{eq:peierls}, with $D^a=D+a$, $a=\eps^{k/2-1}\rho A^\eps$, and $P_2(D^a):=\sum a_{ij}D^a_iD^a_j$, is of the form \eqref{eq:Heps}: for $k=1$ and $P=\gamma(D)+\Gamma$, with $c_0=V+A\cdot\gamma$ and no other terms; for $k=2$, with $c_0=V$, $c_{e_i}=2\sum_ja_{ij}A_j$, $d_0=\sum_{i,j}a_{ij}A_iA_j$, and the remainder term $\rho P_1(A^\eps)$ contributed by the first-order part $P_1(D^a)$. In particular, for $k=2$,
\[
W=V-i\sum_{i,j}a_{ij}\,\partial_{y_i}A_j,\qquad M=\Bigl\langle\sum_{i,j}a_{ij}A_iA_j\Bigr\rangle-\avg{W^*P_2(D_y)^{-1}W}.
\]
\end{lemma}
\begin{proof}
For $k=1$, $\gamma(D^a)=\gamma(D)+\eps^{-1/2}\rho A^\eps\cdot\gamma$. For $k=2$, expand $D^a_iD^a_j=D_iD_j+D_i\rho A^\eps_j+\rho A^\eps_iD_j+\rho^2A^\eps_iA^\eps_j$. Since $a_{ij}=a_{ji}$ is Hermitian and $A$ is real, $\sum a_{ij}(D_i\rho A_j^\eps+\rho A^\eps_iD_j)=\sum_i\mathrm{Sym}\bigl[2\rho\sum_ja_{ij}A^\eps_jD_i\bigr]$. The term $P_1(D^a)-P_1(D)=\rho P_1(A^\eps)$ has order $0$, amplitude $1=\eps^{0-1+1}$ and mean zero, so it is a remainder term. The formula for $W$ follows from \eqref{eq:Weff}.
\end{proof}

\begin{remark}\label{rem:scalar}
For purely scalar $W$, the correction $M$ vanishes for every odd-order operator, by \eqref{eq:pair}. More generally, it vanishes for odd $k$ whenever $W=f(y)B$ with $f$ real and scalar and $B$ a fixed Hermitian matrix, since then $\hat W_q=\hat f_qB$ and both terms in \eqref{eq:pair} equal $|\hat f_q|^2BG_1(\hat q)B$. This covers pure mass fluctuations $V_3\sigma_3$. When $M=0$, the amplitude $\eps^{-k/2}$ of \eqref{eq:Heps} produces no correction, and larger amplitudes become relevant. The cubic term in $W$ is of order one at amplitude $\eps^{-2k/3}$ (Remark~\ref{S-rem:cubic} of the supplementary material).
\end{remark}

\section{Topological consequences}\label{sec:topo}

This section proves Theorems \ref{thm:TAI}, \ref{thm:gapclose} and \ref{thm:magno}, and Corollary \ref{cor:compact}. We use the interface setting of block (iii) of section \ref{sec:setting}. The proofs rely on the Fredholm property of the Toeplitz operator \eqref{eq:T}, on its stability under norm-resolvent deformations, and on homogenization estimates of Theorem \ref{thm:hom}.

\paragraph{The Fredholm property.}

We use the Helffer--Sj\"ostrand formula \cite{HelfferSjostrand89}. For $f\in C_c^\infty(\R)$, let $\tilde f$ be an almost-analytic extension supported in $\{|\Im z|\le1\}$ with $|\dbar\tilde f(z)|\le C_n|\Im z|^n$. Then
\begin{equation}\label{eq:HS}
f(H)=\frac1\pi\int_\C\dbar\tilde f(z)\,(H-z)^{-1}\,d^2z .
\end{equation}

\begin{lemma}\label{lem:HSdecay}
Let $S\in C_c^\infty(\R)$ and $n\ge0$. If $F(z)$ satisfies $\|F(z)\|\le C|\Im z|^{-m}e^{-c|\Im z|\Lambda}$ for $\Lambda\ge0$, then
\[
\Big\|\int\dbar\tilde S(z)F(z)\,d^2z\Big\|\le C_{n,m}(1+\Lambda)^{-n}.
\]
\end{lemma}
\begin{proof}
Take $\tilde S$ with $|\dbar\tilde S|\le C|\Im z|^{n+m}$, and use $\int_0^1y^{n}e^{-cy\Lambda}dy\le C(1+\Lambda)^{-n-1}$.
\end{proof}

\noindent The following comparison is used repeatedly to transfer gaps through resolvent estimates.
\begin{lemma}[Spectral comparison]\label{lem:speccomp}
Let $H_1,H_2$ be self-adjoint, $E\in\spec H_2$ and $\eta,d>0$. If $\|(H_1-z)^{-1}-(H_2-z)^{-1}\|<\eta^{-1}-(d^2+\eta^2)^{-1/2}$ at $z=E+i\eta$, then $\dist(E,\spec H_1)<d$.
\end{lemma}
\begin{proof}
For self-adjoint $H$, $\|(H-z)^{-1}\|=\dist(z,\spec H)^{-1}$. Thus $\|(H_2-z)^{-1}\|=\eta^{-1}$, while $\dist(E,\spec H_1)\ge d$ would give $\|(H_1-z)^{-1}\|\le(d^2+\eta^2)^{-1/2}$.
\end{proof}

\begin{lemma}[Fredholm property]\label{lem:fred}
Let $d=2$ and $P$ satisfy \textup{[P1]--[P2]}. Let $H=P(D)+\mathcal V$, where $\mathcal V$ is a symmetric differential operator of order at most $k-1$ with bounded coefficients. Assume that the coefficients of $\mathcal V$ coincide with those of $\mathcal V_+$ for $x_2\ge X_0$ and with those of $\mathcal V_-$ for $x_2\le-X_0$, where $\mathcal V_\pm$ are of the same type (possibly with $x$-dependent coefficients), and that $H_\pm=P+\mathcal V_\pm$ satisfy $\spec H_\pm\cap J=\emptyset$ for an open interval $J\supset\supp\varphi'$. Then:
\\[1mm] (i) $[S(H),\Pi]$ is compact, $\mathbf 1_{A_0}f(H)$ is compact for every bounded measurable set $A_0\subset\R^2$ and every $f\in C_c(\R)$, $\|\mathbf 1_{\{|x_2|\ge X\}}S(H)\|+\|S(H)\mathbf 1_{\{|x_2|\ge X\}}\|\to0$ as $X\to\infty$, and $T(H)$ is Fredholm with $\mathfrak q(T(H))$ unitary in the Calkin algebra, where $\mathfrak q$ denotes the quotient map from bounded operators onto the Calkin algebra of bounded modulo compact operators. Consequently, any bounded $T'$ with $\|T'-T(H)\|<1$ is Fredholm with the same index;
\\[1mm]
(ii) $\Ind T(H)$ does not depend on the choice of $\varphi$ with $\supp\varphi'\subset J$;
\\[1mm]
(iii) $\Ind T(H)$ is invariant under norm-resolvent continuous deformations of $H$ within this class with $J$ fixed, under conjugation by unitaries commuting with $\Pi$, and under $H\mapsto\lambda H$, $\lambda>0$, with $\varphi$ replaced by $\varphi(\cdot/\lambda)$.
\end{lemma}
Lemma~\ref{lem:fred} is a variant of the Fredholm statements of \cite{Bal2019,Bal2022,QuinnBal}, with kernel estimates as in \cite{GerminetKlein03}. Above, $\mathcal V$ could be, e.g., a potential $V\in L^\infty(\R^2;\Herm(N))$, or $\mathcal V=P(D+a)-P(D)+V$ with $a\in W^{1,\infty}$ real and $k\le2$.

\begin{proof}
\emph{Combes--Thomas \cite{CombesThomas73}.} Let $\phi$ be a smoothed distance function with bounded derivatives and $|\nabla\phi|\le1$. Then $e^{b\phi}He^{-b\phi}=H+bB_b$, where $B_b$ is a differential operator of order $k-1$ whose coefficients are bounded uniformly for $|b|\le1$. Since $\mathcal V$ is of order at most $k-1$, the Kato--Rellich theorem gives $\mathrm{dom}(H)=H^k$ and the elliptic estimate $\|u\|_{H^k}\le C(\|Hu\|+\|u\|)$, and we deduce $\|B_b(H-z)^{-1}\|\le C(1+|z|)|\Im z|^{-1}$. Hence, $\|\mathbf 1_A(H-z)^{-1}\mathbf 1_B\|\le2|\Im z|^{-1}e^{-c|\Im z|\dist(A,B)/(1+|z|)}$ choosing $b=c|\Im z|/(1+|z|)$, and likewise for $H_\pm$. On $\supp\dbar\tilde S$, $|z|$ is bounded, so the Helffer--Sj\"ostrand formula and Lemma~\ref{lem:HSdecay} give, for all measurable $A,B\subset\R^2$,
\begin{equation}\label{eq:Sdecay}
\|\mathbf 1_AS(H)\mathbf 1_B\|\le C_n\bigl(1+\dist(A,B)\bigr)^{-n}.
\end{equation}

\emph{Localization in $x_2$.} We have $R-R_-=-R(\mathcal V-\mathcal V_-)R_-$, where $\mathcal V-\mathcal V_-$ has coefficients supported in $\{x_2\ge-X_0\}$ and is bounded from $H^k$ to $L^2$. For a potential, the Combes--Thomas bound for $H_-$ suffices. In general, it is needed from $L^2$ to $H^k$. Since $e^{b\phi}(H_--z)^{-1}e^{-b\phi}=(H_-+bB_b-z)^{-1}$, the elliptic estimate for the conjugated operator gives $\|\chi_A(H_--z)^{-1}\mathbf 1_B\|_{L^2\to H^k}\le C(1+|z|)|\Im z|^{-1}e^{-c|\Im z|\dist(A,B)/(1+|z|)}$ for smooth cutoffs $\chi_A$ supported in $A$ with bounded derivatives, and we take $\chi_A=1$ on $\{x_2\ge-X_0-\frac12\}$, supported in $A=\{x_2\ge-X_0-1\}$, so that $(\mathcal V-\mathcal V_-)\chi_A=\mathcal V-\mathcal V_-$. Hence, by Lemma~\ref{lem:HSdecay}, as in the proof of \eqref{eq:Sdecay}, $\|(S(H)-S(H_-))\mathbf 1_{\{x_2\le-X\}}\|\le C_n(X-X_0)^{-n}$, and $S(H_-)=0$ since $\spec H_-\cap\supp S=\emptyset$. The upper half-plane is treated with $H_+$, and the row versions follow by taking adjoints.

\emph{Compactness and Fredholm property.} For a bounded measurable set $A_0$ and $f\in C_c(\R)$, $\mathbf 1_{A_0}f(H)=\mathbf 1_{A_0}(H-i)^{-1}\cdot(H-i)f(H)$ is compact, since $\mathbf 1_{A_0}(H-i)^{-1}$ is compact by ellipticity and $(H-i)f(H)$ is bounded. For $X,Y>0$, let $Q=(-Y,0)\times(-X,X)$. Then
\[
(1-\Pi)S(H)\Pi=\mathbf 1_QS(H)\Pi+\mathbf 1_{\{x_1\le-Y\}}S(H)\mathbf 1_{\{x_1\ge0\}}+\mathbf 1_{\{-Y<x_1<0,\,|x_2|\ge X\}}S(H)\Pi .
\]
The first term is compact, the second has norm $\le C_nY^{-n}$ by \eqref{eq:Sdecay}, and the third has norm $\le\|\mathbf 1_{\{|x_2|\ge X\}}S(H)\|\to0$ as $X\to\infty$ by the localization in $x_2$. Hence $(1-\Pi)S(H)\Pi$ is a norm limit of compact operators, and so is $\Pi S(H)(1-\Pi)$. Thus $[S(H),\Pi]=(1-\Pi)S(H)\Pi-\Pi S(H)(1-\Pi)$ is compact. Since $(1-\Pi)U\Pi=(1-\Pi)[S(H),\Pi]$ is compact, both $T^*T-1=-\Pi U^*(1-\Pi)U\Pi$ and $TT^*-1$ are compact, so that $\mathfrak q(T)$ is unitary in the Calkin algebra and $T$ is Fredholm by Atkinson's theorem \cite[Chap.~5]{Douglas98}. If $\|T'-T\|<1$, then $\|\mathfrak q(T')-\mathfrak q(T)\|<1$, so that $\mathfrak q(T'_s)$, $T'_s=T+s(T'-T)$, is invertible for $s\in[0,1]$. Hence $T'$ is Fredholm, and its index equals that of $T$, since Fredholm operators form an open set on which the index is locally constant \cite[Ch.~IV, \S5]{Kato95}. This proves (i).

\emph{(ii) and (iii):} Along a deformation of $\varphi$ or a norm-resolvent continuous deformation of $H$, $U(H)$ moves norm-continuously \cite[Thm.~VIII.20]{ReedSimon1}, so that $T(H)$ stays Fredholm by (i) and its index is constant. Conjugation and rescaling leave $T$ unchanged up to unitary equivalence.
\end{proof}

\begin{lemma}\label{lem:BDIc}
Let $P$ satisfy \textup{[P1]--[P2]} with $d=2$ and any $k\ge1$, and let $H=P(D)+V(x)$ be as in Proposition~\ref{thm:QB}. Then, with $c$ as in \eqref{eq:ch}:
\\[1mm]
(i) $\Ind T(H)=c(P+V_-)-c(P+V_+)=\BDI(P+V_-,P+V_+)$, and $c(P+\Gamma)$ is locally constant along gapped deformations of $\Gamma\in\Herm(N)$;
\\[1mm]
(ii) for a two-band symbol $h=d\cdot\sigma+f\sigma_0$ with $|f|<|d|$, $c(h)=\frac1{4\pi}\int n\cdot(\partial_1n\times\partial_2n)$ with $n=d/|d|$, which is the degree of $n$ when $n$ extends continuously to $\R^2\cup\{\infty\}$. If moreover $(d_1,d_2)$ vanishes at a single point $\xi_0$, with $\det\partial_\xi(d_1,d_2)(\xi_0)>0$, then
\begin{equation}\label{eq:degree}
c(h)=\tfrac12\bigl(\sgn d_3(\xi_0)-\sgn d_3(\infty)\bigr).
\end{equation}
For the Dirac operator, $c(m\sigma_3):=c(D\cdot\sigma+m\sigma_3)=\tfrac12\sgn m$.
\end{lemma}
\begin{proof}
\emph{(i)} We follow \cite{Bal2022,QuinnBal}. A linear deformation fixing the bulks and Lemma~\ref{lem:fred}(iii) reduce to $V=\Gamma(x_2)$, with $\Gamma=\Gamma_\pm:=V_\pm$ for $\pm x_2\ge1$. The symbol of $H$ is an elliptic Hermitian symbol of order $k$ satisfying hypothesis (H1) of \cite{QuinnBal}, with bulk symbols $h_\pm(\xi)=P(\xi)+\Gamma_\pm$ gapped at $0$. For $z=i\omega$, let $G_\pm=(z-h_\pm)^{-1}$ and
\[
F_\pm(\omega,\xi)=\tr\bigl([G_\pm\partial_{\xi_1}h_\pm,\,G_\pm\partial_{\xi_2}h_\pm]\,G_\pm\bigr).
\]
By \cite[Thm.~2.3 and proof of Cor.~2.4]{QuinnBal}, $2\pi\sigma_I(H)$ is the limit as $t\to\infty$ of $\frac{i}{8\pi^2}\int_{(-t,t)^2}\int_\R(F_+-F_-)\,d\omega\,d\xi$. Since $|G_\pm|\le C(|\omega|+\avg\xi^k)^{-1}$, $G_+-G_-=G_+(\Gamma_+-\Gamma_-)G_-$ and $|\partial h_\pm|\le C\avg\xi^{k-1}$, we have $|F_+-F_-|\le C(|\omega|+\avg\xi^k)^{-4}\avg\xi^{2k-2}$, which is integrable on $\R^3$. By dominated convergence, the limit is the integral over $\R^3$, and we may integrate in $\omega$ first. Moreover $\Omega_{h_\pm}=O(\avg\xi^{-3})$, since $\Pi_{h_\pm}$ converges to its angular limit with radial derivative $O(\avg\xi^{-2})$.

Fix $\xi$, and write $h(\xi)=\sum_j\lambda_j\mathsf E_j$ for either bulk symbol, with spectral projectors $\mathsf E_j$. Expanding $G=\sum_j(z-\lambda_j)^{-1}\mathsf E_j$ and integrating in $\omega$ by residues as in \cite[Lemma~3.2]{Bal2022} and \cite[Prop.~2.6]{QuinnBal}, only pairs of eigenvalues on opposite sides of $0$ contribute. For such pairs $(j,j')$, $\mathsf E_j(\partial_lh)\mathsf E_{j'}=(\lambda_{j'}-\lambda_j)\mathsf E_j(\partial_l\mathsf E_{j'})\mathsf E_{j'}$ and $\mathsf E_j(\partial_l\mathsf E_{j'})\mathsf E_{j'}=\pm\mathsf E_j(\partial_l\Pi_h)\mathsf E_{j'}$, with $\Pi_h$ the negative spectral projector. Collecting terms gives the pointwise identity
\[
\int_\R F(\omega,\xi)\,d\omega=4\pi i\cdot i\,\tr\bigl(\Pi_h[\partial_{\xi_1}\Pi_h,\partial_{\xi_2}\Pi_h]\bigr)(\xi),
\]
whose right side is $4\pi i$ times the density of $\Omega_h$. Hence
\[
2\pi\sigma_I(H)=\frac{i}{8\pi^2}\cdot4\pi i\int_{\R^2}(\Omega_{h_+}-\Omega_{h_-})=c(h_-)-c(h_+),
\]
and $\Ind T(H)=2\pi\sigma_I(H)$ \cite{Bal2019,Bal2022}. Local constancy follows, since $c(P+\Gamma)-c(P+\Gamma_0)$ is an index, continuous in $\Gamma$.

\emph{(ii)} For $|f|<|d|$, the negative spectral projector is $\Pi_h=\frac12(1-n\cdot\sigma)$, independently of $f$, and $i\tr(\Pi_h[\partial_1\Pi_h,\partial_2\Pi_h])=\frac12n\cdot(\partial_1n\times\partial_2n)$, and we obtain the integral formula. Under the additional assumption, the poles $n=(0,0,\pm1)$ have preimages only at $\xi_0$ and at infinity, and \eqref{eq:degree} counts them with orientation. For the Dirac operator, $n=(\xi,m)/|(\xi,m)|$ covers the hemisphere $\{\sgn(m)\,n_3>0\}$ once, with orientation $\sgn m$, so that \eqref{eq:degree} is not needed. In particular, the Dirac domain wall with masses $m_\pm$ has index $\tfrac12(\sgn m_--\sgn m_+)$.
\end{proof}

\begin{proof}[Proof of Theorem~\ref{thm:TAI}]
Since $\rho\in C_b^\infty$, the operator $H^0=P+\rho^2M$ satisfies the hypotheses of Proposition~\ref{thm:QB}, so $T(H^0)$ is Fredholm with index $\BDI(P+M,P)$. By Theorem~\ref{thm:hom} and the Helffer--Sj\"ostrand formula applied to $S=e^{2\pi i\varphi}-1$,
\[
\|U(H^\eps)-U(H^0)\|\le\frac1\pi\int|\dbar\tilde S|\,C\eps^{r_k}(1+|\Re z|)^2(1+|\Im z|^{-2})\,d^2z\le C'\eps^{r_k}.
\]
Hence $\|T(H^\eps)-T(H^0)\|\le C'\eps^{r_k}$, and the claim follows from Lemma~\ref{lem:fred}(i) as soon as $C'\eps^{r_k}<1$.
\end{proof}

\begin{proof}[Proof of Corollary~\ref{cor:compact}]
Let $H=P(D)+\mathcal V$, with $\mathcal V$ a symmetric differential operator of order at most $k-1$ with bounded coefficients, be such that $T(H)$ is Fredholm. We show that $T(H+V_c)$ is Fredholm with $\Ind T(H+V_c)=\Ind T(H)$. The corollary follows by taking $H=H^\eps$, for which $T(H^\eps)$ is Fredholm with index $\BDI(P+M,P)$ for $\eps\le\eps_0$ by Theorem~\ref{thm:TAI}, and similarly for $\tilde T$, using Theorems~\ref{thm:mainB} and~\ref{thm:mainA}.

By the Kato--Rellich theorem, $\mathrm{dom}(H)=H^k$ with equivalent norms. The perturbation $V_c$ is compact relative to $H$. Indeed, $V_c(H-z)^{-1}=V_c\one_{\supp V_c}(H-z)^{-1}$, and $\one_{\supp V_c}(H-z)^{-1}$ is compact by ellipticity. By the resolvent identity, $(H+V_c-z)^{-1}-(H-z)^{-1}=-(H+V_c-z)^{-1}V_c(H-z)^{-1}$ is compact, with norm at most $\|V_c\|_\infty|\Im z|^{-2}$. By \eqref{eq:HS}, $U(H+V_c)-U(H)=S(H+V_c)-S(H)$ is a norm-convergent integral of compact operators, since $|\dbar\tilde S(z)|\le C|\Im z|^2$, and is therefore compact. Hence $T(H+V_c)-T(H)=\Pi\big(U(H+V_c)-U(H)\big)\Pi$ is compact, and the Fredholm property and the index are preserved under compact perturbations \cite[Ch.~IV, \S5]{Kato95}. The same argument applies with $\tilde\Pi$ in place of $\Pi$ for the deformed cuts of block (v) of section \ref{sec:setting}.
\end{proof}

\begin{proof}[Proof of Theorem~\ref{thm:gapclose} and of the statement following it]
The family $\eps\mapsto H^\eps_a$ is not norm-resolvent continuous, since $W(x/\eps)$ changes on all scales. It becomes so after the dilation $y=x/\eps$, and the index is then locally constant as long as the bulk is gapped. Let $(\mathcal U_\eps u)(y)=\eps u(\eps y)$, a unitary operator on $L^2(\R^2)$ commuting with $\Pi$, and write $P=\sum_{j\le k}P_j$ with $P_j$ homogeneous of degree $j$. Then
\[
\tilde H(\eps):=\eps^k\,\mathcal U_\eps H^\eps_a\,\mathcal U_\eps^{-1}=\sum_{j\le k}\eps^{k-j}P_j(D_y)+\eps^ka_\eps\,\rho(\eps y_2)W(y),
\]
with bulks $\eps^k\mathcal U_\eps P\,\mathcal U_\eps^{-1}$ and $\tilde H_B(\eps)=\eps^k\mathcal U_\eps H^\eps_{a,B}\mathcal U_\eps^{-1}$. Both $\tilde H$ and $\tilde H_B$ are norm-resolvent continuous on $(0,\infty)$: the coefficients of the differential part and $\eps^ka_\eps$ are continuous in $\eps$, and $\sup_y|\rho(\eps y_2)-\rho(\eps_\circ y_2)|\le\|\rho'\|_\infty|\eps-\eps_\circ|/\min(\eps,\eps_\circ)$ since $\rho$ is constant outside $[-1,1]$. If $0\notin\spec H^\eps_{a,B}$ for all $\eps\in[\eps',\eps'']$, continuity and compactness give $g_*>0$ such that both bulks of $\tilde H(\eps)$ are gapped on $(-g_*,g_*)$ for all these $\eps$ \cite[Thm.~VIII.23]{ReedSimon1}. With a fixed switch whose transition lies in $(-g_*/2,g_*/2)$, $T(\tilde H(\eps))$ is then Fredholm and norm-continuous in $\eps$ by Lemma~\ref{lem:fred}, so its index is constant on $[\eps',\eps'']$. By Lemma~\ref{lem:fred}(ii)--(iii), it equals $\Ind T(H^\eps_a)$ at $\eps=\eps'$ and $\eps=\eps''$, a contradiction.

For the statement following the theorem, let $a_\eps=\eps^{-k/2}$, so that $H^\eps_a=H^\eps$ and $H^\eps_{a,B}=H^\eps_B$. The second resolvent identity gives $\|(H^\eps-z)^{-1}-(P-z)^{-1}\|\le|\Im z|^{-2}\eps^{-k/2}\|\rho W\|_\infty$, hence $\|U(H^\eps)-U(P)\|\le C\eps^{-k/2}$ by \eqref{eq:HS}. Since $U(P)=I$ by \textup{[G]}, Lemma~\ref{lem:fred}(i) gives $\Ind T(H^{\eps_1})=0$ for $\eps_1$ large, while $\Ind T(H^{\eps_0})=\BDI(P+M,P)\ne0$ by Theorem~\ref{thm:TAI}. The bulk $H^{\eps_1}_B$ is gapped at $0$ since $\eps_1^{-k/2}W$ is small, and so is $H^{\eps_0}_B$ after decreasing $\eps_0$, which preserves the conclusion of Theorem~\ref{thm:TAI}, by Theorem~\ref{thm:hom} with $\rho\equiv1$ and Lemma~\ref{lem:speccomp} with $H_2=H^{\eps_0}_B$ and $H_1=P+M$. The hypotheses of Theorem~\ref{thm:gapclose} thus hold with $(\eps',\eps'')=(\eps_0,\eps_1)$, and $0\in\spec H^\eps_B$ for some $\eps\in(\eps_0,\eps_1)$, as claimed.
\end{proof}

\paragraph{Magnetic fields.} The gaps are obtained from the grading $J$ by the following lemma.
\begin{lemma}[Graded gap]\label{lem:grade}
Let $H$ be self-adjoint, and let $J$ be a unitary involution with $J\,\mathrm{dom}(H)\subset\mathrm{dom}(H)$. Suppose $H=H_{\rm o}+H_{\rm e}$ on $\mathrm{dom}(H)$, with $H_{\rm o},H_{\rm e}$ symmetric on $\mathrm{dom}(H)$, $JH_{\rm o}=-H_{\rm o}J$ and $JH_{\rm e}=H_{\rm e}J$. If $\ip{JH_{\rm e}\psi}{\psi}\ge g\|\psi\|^2$ for all $\psi\in\mathrm{dom}(H)$, then $\spec H\cap(-g,g)=\emptyset$.
\end{lemma}
\begin{proof}
For $\psi\in\mathrm{dom}(H)$, symmetry and anticommutation give $\ip{H_{\rm o}\psi}{J\psi}=-\overline{\ip{H_{\rm o}\psi}{J\psi}}$, which is purely imaginary. Since $J$ is self-adjoint and commutes with $H_{\rm e}$, $\ip{H_{\rm e}\psi}{J\psi}=\ip{JH_{\rm e}\psi}{\psi}\ge g\|\psi\|^2$. Hence, for $|E|<g$,
$
\|(H-E)\psi\|\,\|\psi\|\ge\Re\ip{(H-E)\psi}{J\psi}\ge(g-|E|)\|\psi\|^2,
$
and $E\notin\spec H$ since $H$ is self-adjoint.
\end{proof}
\noindent For instance, for $H=(D+a)\cdot\sigma+\mu(x)\sigma_3+V(x)$ with $a,\mu,V$ real, bounded and measurable, Lemma~\ref{lem:grade} with $J=\pm\sigma_3$, $H_{\rm o}=(D+a)\cdot\sigma$ and $H_{\rm e}=\mu\sigma_3+V$ shows that $\spec H\cap(-g,g)=\emptyset$ whenever $\pm\mu-|V|\ge g$ almost everywhere.

\begin{proof}[Proof of Theorem~\ref{thm:magno}]
The operators $D^a_j$ are scalar, so $P^{\rm o}(D^a)$ anticommutes with $J$, and $q(D^a)\Upsilon$ and $\Gamma$ commute with it. In case (a), apply Lemma~\ref{lem:grade} with grading $sJ$, $H_{\rm o}=P^{\rm o}(D^a)$ and $H_{\rm e}=\Gamma$. In case (b), take $H_{\rm e}=q(D^a)\Upsilon+\Gamma$. Write $(b_{ij})=L^TL$. Then $q(D^a)=\sum_l\bigl(\sum_iL_{li}D^a_i\bigr)^2\ge0$. Since $q(D^a)$ commutes with the constant matrix $sJ\Upsilon\ge0$, we get $sJH_{\rm e}\ge sJ\Gamma\ge g$. This gives the gap for every $a$. Since $\supp S\subset(-g,g)$, it follows that $S(P(D^a))=0$, so that $U(P(D^a))=I$, $T(P(D^a))=I$ and $\Ind T(P(D^a))=0$. For a purely magnetic fluctuation, $H^\eps=P(D^a)$ with $a=\eps^{k/2-1}\rho A^\eps$, so that $\BDI(P+M,P)=\Ind T(H^\eps)=0$ for $\eps\le\eps_0$ by Theorem~\ref{thm:TAI}.
\end{proof}

\section{Random homogenization with a genuine gap}\label{sec:random}

This section proves Theorems \ref{thm:unifhom} and \ref{thm:phasepurek}. We first identify the effective matrix of stationary random fields. The uniform estimate for alloys is obtained by building the correctors of section \ref{sec:hom} from mesoscale averages, and the gap-preserving result follows by combining it with Lemma~\ref{lem:grade}.

\paragraph{The effective matrix and the critical dimension.}

Let $W$ be $\R^d$-stationary with finite range of dependence (e.g., a $\Z^d$-stationary field made $\R^d$-stationary by a uniform random shift in the unit cell). Its spectral density is the array $s(q)=(s_{ba,cd}(q))_{a,b,c,d\le N}$ defined by
\[
\mathbb E\bigl[\overline{W_{ba}(y)}\,W_{cd}(y')\bigr]=\int_{\R^d}e^{2\pi iq\cdot(y'-y)}s_{ba,cd}(q)\,dq ,
\]
which is real-analytic since the covariance has compact support. For a matrix $G\in\C^{N\times N}$ we define the contraction $G\diamond s\in\C^{N\times N}$ by $(G\diamond s)_{ad}=\sum_{b,c}G_{bc}\,s_{ba,cd}$. For a randomly shifted periodic pattern, $s=\sum_q\overline{(\hat W_q)_{ba}}(\hat W_q)_{cd}\,\delta_q$, with $\delta_q$ the Dirac mass at $q$, and $G\diamond s$ reproduces $\sum_q\hat W_q^*G(q)\hat W_q$. For scalar $W=V\,I$ with spectral density $s_V$, $G\diamond s=s_VG$. Applying \eqref{eq:M}, in its $L$-periodic form with $q\in L^{-1}\Z^d$, to $L\Z^d$-periodizations and letting $L\to\infty$ gives, formally, the random effective matrix for alloys
\[
M=\int_{\R^d}G(q)\diamond s(q)\,dq.
\]
Lemma~\ref{lem:mean}(i) gives a rigorous version.
Since $s$ is smooth and $|G(q)|\sim|q|^{-k}$, the integral is finite if and only if $\int_{|q|<1}|s(q)||q|^{-k}dq<\infty$. Three regimes follow. If $s(0)\ne0$, this holds for $d>k$, diverges logarithmically for $d=k$, and diverges like a power for $d<k$.
If $W=D_y^\nu\Psi$ with $|\nu|=j_0$ and $\Psi$ stationary (divergence-type disorder, e.g., an alloy with mean-zero single site, $j_0=1$), then $s(q)=O(|q|^{2j_0})$ and the condition becomes $k<d+2j_0$. In two dimensions, first-order systems are thus subcritical, while regularized Dirac and BHZ operators are critical unless $s(0)=0$.
Theorems \ref{thm:unifhom} and \ref{thm:mainB} concern the subcritical regime, in which the fluctuation part of $\eps^{-k/2}W(x/\eps)$, whose power spectrum in the macroscopic variable is $\eps^{d-k}s(\eps\xi)$, vanishes and the deterministic Born term $M$ survives. At criticality, $d=k=2$, this power spectrum is a mollified white noise and the Born term grows like $s(0)\log(1/\eps)$, which is why the amplitude is renormalized in \eqref{eq:Hcrit}. Theorems \ref{thm:det} and \ref{thm:mainA} treat this regime for scalar disorder. In the supercritical regime $d<k$, homogenization fails \cite{Bal10,BalGu15}, the limit is stochastic, and we do not consider this setting.

\begin{lemma}[Local effective matrices]\label{lem:Kset}
The function $\mathbb K$ of \eqref{eq:Kq} is bounded and Hermitian. For an $L$-periodic configuration with $L\in\N$, $M(\omega)=\int_{\T^2}\mathbb K\,d\mathfrak s_\omega$, where
\[
\hat\omega_q=L^{-2}\sum_{j\in(\Z/L\Z)^2}\omega_je^{-2\pi iq\cdot j},\qquad \mathfrak s_\omega=\sum_{q\in L^{-1}\Z^2/\Z^2}|\hat\omega_q|^2\delta_q ,
\]
with $\delta_q$ the Dirac mass at $q\in\T^2$. The structure factor $\mathfrak s_\omega$ is a nonnegative measure of total mass equal to the mean of $\omega_j^2$ over a period. Consequently $M(\omega)\in\cK$ if $\kappa_-\le|\omega_j|\le\kappa_+$.\end{lemma}
\begin{proof}
Boundedness of $\mathbb K$ follows from $|\hat w(p)|^2|G(p)|\le C|p|^{3k}$ near $p=0$, a consequence of \eqref{eq:moments}, and its Hermiticity from that of $G$. For $L$-periodic $W$, \eqref{eq:M} holds with $q\in L^{-1}\Z^2\setminus0$ and Fourier coefficients taken as averages over the supercell. This follows by rescaling to a $1$-periodic pattern. On the supercell, $\hat W_{q+\ell}=\hat\omega_q\hat w(q+\ell)$ for $q\in L^{-1}\Z^2/\Z^2$ and $\ell\in\Z^2$, and $\hat w(0)=0$. Substituting into \eqref{eq:M} gives $M(\omega)=\sum_q|\hat\omega_q|^2\,\mathbb K(q)=\int_{\T^2}\mathbb K\,d\mathfrak s_\omega$. By Parseval, the mass of $\mathfrak s_\omega$ is $\sum_q|\hat\omega_q|^2=L^{-2}\sum_j\omega_j^2$, the sum running over one period. It lies in $[\kappa_-^2,\kappa_+^2]$ if $\kappa_-\le|\omega_j|\le\kappa_+$. Then $M(\omega)$ is this mass times the average of $\mathbb K$ against a probability measure, which lies in $\overline\conv\{\mathbb K(q)\}$, so that $M(\omega)\in\cK$.
\end{proof}

We now prove Theorem~\ref{thm:unifhom}. The proof of Theorem~\ref{thm:hom} used periodicity twice: to solve the cell problems, and to define the mean $M$. For alloys, cell problems are solved by explicit lattice sums, and the mean is replaced by a mesoscale average $M_\omega$. Throughout, $d=2$, $k\ge1$, $P$ satisfies [P1]--[P2], $w$ satisfies \eqref{eq:moments}, and $|\omega_j|\le\kappa_+$. All constants depend only on $(P,w,\kappa_+,\rho)$. We write $\hat a=\mathcal Fa$ for the Fourier transform on $\R^2$, $\hat a(p)=\int a(y)e^{-2\pi ip\cdot y}dy$, and $\mathcal F^{-1}$ for its inverse. For $s>0$, let $\mathcal S^s$ be the class of $a\in C^\infty(\R^2\setminus0;\C^{N\times N})$ with
\[
|\partial^\alpha a(q)|\le C_\alpha|q|^{s-|\alpha|}\ (0<|q|\le1),\qquad |\partial^\alpha a(q)|\le C_{\alpha,n}\avg q^{-n}\ (|q|\ge1),
\]
for all $\alpha$ and $n$. By \eqref{eq:moments} and Taylor's formula, $\hat w\in\mathcal S^{2k}$. Multiplication by $(2\pi q)^\nu$ maps $\mathcal S^s$ into $\mathcal S^{s+|\nu|}$, multiplication by a constant matrix preserves $\mathcal S^s$, and multiplication by $P_k(2\pi q)^{-1}$ maps $\mathcal S^s$ into $\mathcal S^{s-k}$ when $s>k$. We fix a radial $\chi_0\in C_c^\infty(\R^2)$ equal to $1$ near $0$.
\begin{lemma}[Lattice sums]\label{lem:latsum}
If $a\in\mathcal S^s$ with $s>0$, then $|\partial^\alpha\mathcal F^{-1}a(y)|\le C_\alpha\avg y^{-2-s}$. For bounded $(\omega_j)$, then $\sum_j\omega_j\mathcal F^{-1}a(\cdot-j)$ is smooth with $\|\partial^\alpha\sum_j\omega_j\mathcal F^{-1}a(\cdot-j)\|_\infty\le C_\alpha\sup_j|\omega_j|$.
\end{lemma}
\begin{proof}
Split $a=a\,(1-\chi_0)+\sum_{l\ge1}a\,\phi(2^l\cdot)$, with $\phi$ supported in an annulus. The first part is Schwartz. Rescaling the $l$-th piece to the unit annulus and integrating by parts gives $|\mathcal F^{-1}[a\phi(2^l\cdot)](y)|\le C_n2^{-l(2+s)}(1+2^{-l}|y|)^{-n}$. Summing over $l$ gives $C\avg y^{-2-s}$, and $q^\alpha a\in\mathcal S^{s+|\alpha|}$ handles derivatives. Finally, $\sum_j\avg{y-j}^{-2-s}$ is bounded since $s>0$.
\end{proof}

\medskip\noindent\textbf{Mesoscale averages.} Let $\tilde K\ge0$ be a Schwartz function with $\hat{\tilde K}\in C_c^\infty$ even, $\hat{\tilde K}(0)=1$, and $\sqrt{\tilde K}$ Lipschitz with rapid decay. For instance, $\tilde K=|\mathcal F^{-1}\psi|^2$ with $\psi\in C_c^\infty$ even, suitably normalized. Then $|1-\hat{\tilde K}(p)|\le C|p|^2$. Put $K_{\varpi}(y)=\varpi^2\tilde K(\varpi y)$. For $\varpi$ small, the Poisson summation formula gives $\sum_jK_{\varpi}(y-j)=1$ for all $y$. Let $\chi_w=\mathcal F^{-1}(G\hat w)$ be the corrector of a single site, so that $\hat\chi_w\in\mathcal S^k$ and $\chi:=\sum_j\omega_j\chi_w(\cdot-j)$ solves $P_k(D_y)\chi=-W$. By Lemma~\ref{lem:latsum}, applied to $\hat\chi_w$ and to $2\pi iq\,\hat\chi_w\in\mathcal S^{k+1}$, $|\chi_w(y)|\le C\avg y^{-2-k}$ and $|\nabla \chi_w(y)|\le C\avg y^{-3-k}$, and $\chi$ is bounded with all its derivatives. We define the correction density $\Xi$, its mesoscale average $A$, and the local effective matrix $M_\omega$ by
\begin{equation}\label{eq:Mloc}
\Xi=W\chi,\qquad A=K_{\varpi}*\Xi,\qquad M_\omega=\tfrac12(A+A^*).
\end{equation}
For periodic configurations, the average of $\Xi$ over a period is $M(\omega)$.

\begin{lemma}[Mesoscale effective matrix]\label{lem:meso}
For $\varpi$ small enough: (i) $\|A-M_\omega\|_\infty\le C\varpi$; and (ii) for each $y$, $\dist\bigl(M_\omega(y),\,\mu(y)\cdot\overline\conv\{\mathbb K(q)\}\bigr)\le C\varpi$, with $\mu(y)=\sum_jK_{\varpi}(y-j)\omega_j^2\in[\inf_j\omega_j^2,\sup_j\omega_j^2]$.
\end{lemma}
\begin{proof}
\emph{Step 1: reduction to a lattice sum.} We write $\Xi=\sum_{j,l}\omega_j\omega_l\,w(\cdot-j)\chi_w(\cdot-l)$ and $k_0(m)=\int w(z)\chi_w(z-m)\,dz$. Since $\int w=0$, then $k_0(m)=\int w(z)\bigl(\chi_w(z-m)-\chi_w(-m)\bigr)dz$,
so that $|k_0(m)|\le C\avg m^{-3-k}$. Since $w$ has compact support, $|K_{\varpi}(y-y')-K_{\varpi}(y-j)|\le C\varpi^3(1+\varpi|y-j|)^{-n}$ for $y'\in\supp w(\cdot-j)$ and any fixed $n>2$. Hence $A(y)=B(y)+E_1(y)$ with $ B(y)=\sum_{j,m}K_{\varpi}(y-j)\,\omega_j\omega_{j+m}\,k_0(m)$ where $\|E_1\|_\infty\le C\varpi^3\sum_j(1+\varpi|y-j|)^{-n}\sum_m\avg m^{-2-k}\le C\varpi$.

\emph{Step 2: almost-Hermiticity.} The kernel of $P_k(D)^{-1}$ satisfies $\mathsf G(y)^*=\mathsf G(-y)$, because its symbol is Hermitian. Hence $k_0(-m)=k_0(m)^*$. Reindexing $(j,m)\mapsto(j+m,-m)$ in $B^*$ gives
\[
B-B^*=\sum_{j,m}\bigl(K_{\varpi}(y-j)-K_{\varpi}(y-j-m)\bigr)\omega_j\omega_{j+m}k_0(m).
\]
Since $\sum_j|K_{\varpi}(y-j)-K_{\varpi}(y-j-m)|\le C\varpi|m|$ and $\sum_m|m||k_0(m)|<\infty$, we get $\|B-B^*\|_\infty\le C\varpi$. This proves (i).

\emph{Step 3: structure-factor representation.} Let $a_j=\sqrt{K_{\varpi}(y-j)}\,\omega_j$. The same estimate, with $\sqrt{K_{\varpi}}$ in place of $K_{\varpi}$, shows that $B$ differs by $O(\varpi)$ from
\[
\sum_{j,m}a_ja_{j+m}k_0(m)=\int_{\T^2}|\hat a(q)|^2\,\mathbb K(q)\,dq,\qquad\hat a(q)=\sum_ja_je^{-2\pi iq\cdot j}.
\]
Here we used the autocorrelation identity $\sum_ja_ja_{j+m}=\int|\hat a|^2e^{2\pi iq\cdot m}dq$ and the Poisson summation formula $\sum_mk_0(m)e^{2\pi iq\cdot m}=\mathbb K(q)$. The latter follows from $\widehat{k_0}(p)=\hat w(p)\hat\chi_w(-p)$ and $\hat w(-p)=\hat w(p)^*$. The measure $|\hat a|^2dq$ is nonnegative with mass $\sum_ja_j^2=\mu(y)$. This proves (ii).
\end{proof}

\begin{lemma}[Quadratic corrector]\label{lem:mult}
Let $0<\varpi_0\le\varpi\le1$ and $\nu$a multi-index with $|\nu|\le k-1$. Define
\[
m_\nu(q)=(2\pi q)^\nu P_k(2\pi q)^{-1}\bigl(1-\chi_0(q/\varpi_0)\bigr)\bigl(1-\hat{\tilde K}(q/\varpi)\bigr).
\]
For the quadratic corrector $\chi_2$ of the proof of Theorem~\ref{thm:unifhom}, $D_y^\nu\chi_2=-m_\nu(D_y)\Xi$.
The kernel of this multiplier has $L^1$ norm at most $CB_\nu$, where
\begin{gather*}
B_\nu=\varpi^{-1}\ \ (|\nu|=k-1),\qquad B_\nu=\varpi^{-2}\bigl(1+\log(\varpi/\varpi_0)\bigr)\ \ (|\nu|=k-2),\\
B_\nu=\varpi_0^{|\nu|-k+2}\varpi^{-2}\ \ (|\nu|\le k-3).
\end{gather*}
The kernel of the multiplier $\chi_0(q/\varpi_0)(1-\hat{\tilde K}(q/\varpi))$ has $L^1$ norm at most $C(\varpi_0/\varpi)^2$.
\end{lemma}
\begin{proof}
Decompose dyadically into pieces supported where $|q|\sim2^l$. Each factor satisfies symbol estimates at scale $2^l$: $|\partial^\alpha f|\le C_\alpha2^{-l|\alpha|}\sup_{|q|\sim2^l}|f|$. For $1-\hat{\tilde K}(q/\varpi)$ with $2^l\le\varpi$ this uses $|\partial^\alpha(1-\hat{\tilde K}(q/\varpi))|\le C\varpi^{-2}|q|^{2-|\alpha|}$. Hence each piece has kernel $L^1$ norm $\le C\sup_{|q|\sim2^l}|m_\nu|\le C2^{l(|\nu|-k)}\min(1,2^{2l}\varpi^{-2})$. Summing over $2^l\gtrsim\varpi_0$ gives $B_\nu$. The sum converges at high frequencies because $|\nu|<k$. The second multiplier is $\chi_0(p)(1-\hat{\tilde K}(p\varpi_0/\varpi))$ after rescaling $q=\varpi_0 p$, and all its $p$-derivatives are $O((\varpi_0/\varpi)^2)$.
\end{proof}

\begin{proof}[Proof of Theorem~\ref{thm:unifhom}]
Set $\varpi_0=\eps^{1/2}$, and $\varpi=\eps^{1/4}$ if $k=1$, $\varpi=\eps^{1/3}$ if $k\ge2$. Put $M^\eps_\omega(x)=M_\omega(x/\eps)$, $\psi^0=(P+\rho^2M^\eps_\omega-z)^{-1}f$, and $v=S_\eps\psi^0$, with $S_\eps$ the smoothing operator of \eqref{eq:smooth}. Since $\|M_\omega\|_\infty\le\|\Xi\|_\infty\le C$, the elliptic estimate $\|\psi^0\|_{H^k}\le C(1+|z|)(1+|\Im z|^{-1})\|f\|$ is uniform in $\omega$.

\emph{Linear correctors.} We run the hierarchy of the proof of Theorem~\ref{thm:hom} for $k\ge2$, correcting only monomials of order $\le0$ instead of $\le\frac12$, starting from $c_1=\eps^{k/2}\chi^\eps\rho v$. For $k=1$ no further linear corrector arises. Every corrector coefficient produced, except one, is linear in $\omega$, as are the coefficients of their $P$-terms. These coefficients have the form $\sum_j\omega_j\mathcal F^{-1}a(\cdot-j)$ with $a\in\mathcal S^s$ and we call $s$ its degree. A $P$-term in which $|\nu|$ derivatives fall on a coefficient of degree $s$ and order $a$ has degree $s+|\nu|$ and order $a-|\nu|$, and a linear monomial of degree $s$ to be corrected is cancelled by the corrector with coefficient $\sum_j\omega_j\mathcal F^{-1}b(\cdot-j)$, $b=-P_k(2\pi\cdot)^{-1}a$, of degree $s-k$ and order increased by $k$. Hence degree plus order is invariant along the hierarchy, and equals $k+\frac k2$, its value for $c_1$. A corrector generated from a monomial of order $a\le0$ to be corrected therefore has degree $\frac k2-a\ge\frac k2>0$. By Lemma~\ref{lem:latsum}, all linear coefficients are bounded with all derivatives, uniformly in the configuration.

Steps (a) and (c) of that proof involve only orders, weights and the bound \eqref{eq:mono}, which uses $\|b\|_\infty$. Step (b) used the mean-zero property only to solve the cell problems, and these are solved here explicitly. Hence the hierarchy terminates. Its residual consists of $\eps^{-k/2}\rho W^\eps(I-S_\eps)\psi^0$, $\rho^2M^\eps_\omega(I-S_\eps)\psi^0$, monomials of positive order with bounded coefficients, and the terms treated next. The first three have total norm $\le C\eps^{1/2}(1+|z|)\|\psi^0\|_{H^k}$. As in Theorem~\ref{thm:hom}, here $(P-z)\psi^0+\rho^2M^\eps_\omega v=f-\rho^2M^\eps_\omega(I-S_\eps)\psi^0$.

\emph{The quadratic corrector.} The only non-linear monomial to be corrected is $\rho^2(\Xi-M_\omega)^\eps v$, which remains from the potential term $\rho^2\,\Xi^\eps v$ of $c_1$. Let $\chi_2=-P_k(D_y)^{-1}(1-\chi_0(D_y/\varpi_0))(\Xi-A)$, as a Fourier multiplier, and $c_{\chi_2}=\eps^k\chi_2^\eps\rho^2v$. The leading term of $c_{\chi_2}$ cancels $\rho^2(\Xi-M_\omega)^\eps v$ up to
\[
\rho^2\bigl[(A-M_\omega)+\chi_0(D_y/\varpi_0)(\Xi-A)\bigr]^\eps v,
\]
of norm $\le C(\varpi+(\varpi_0/\varpi)^2)\|\psi^0\|$ by Lemmas~\ref{lem:meso}(i) and~\ref{lem:mult}. Since $\Xi-A=(1-\hat{\tilde K}(D/\varpi))\Xi$ and $\|\Xi\|_\infty\le C$, Lemma~\ref{lem:mult} also gives $\|D^\nu\chi_2\|_\infty\le CB_\nu$ for $|\nu|\le k-1$. For $|\nu|=k$, write $D^\nu=D^{\nu'}D_{y_i}$ with $|\nu'|=k-1$. Then $D^\nu\chi_2$ is, up to sign, the multiplier of Lemma~\ref{lem:mult} with index $\nu'$ applied to the bounded function $D_{y_i}\Xi$. Hence $\chi_2\in C^k_b$ for each $\eps$, and $c_{\chi_2}\in H^k$. The other terms of $(H^\eps-z)c_{\chi_2}$ are the $P$-terms, bounded by $C\eps^{k-|\nu|}B_\nu\|\psi^0\|_{H^k}$ since at most $k-|\nu|$ derivatives fall on $\rho^2v$, the potential term $\eps^{k/2}\rho^3(W\chi_2)^\eps v$, bounded by $C\eps^{k/2}B_0\|\psi^0\|$, and the $z$-term, bounded by $C\eps^k|z|B_0\|\psi^0\|$. With the above choice of $\varpi_0$ and $\varpi$, every one of these bounds is $\le C\delta_k(\eps)$. For instance, $\eps^{k/2}B_0=\eps^{1/4}$ for $k=1$, $\eps^{k/2}B_0=\eps\varpi^{-2}(1+\log(\varpi/\varpi_0))$ for $k=2$, and $\eps^{k/2}B_0=\eps^{1/3}$ for $k\ge3$. The windowing errors are $\varpi+(\varpi_0/\varpi)^2\le C\delta_k(\eps)$. Also $\|\psi_{\rm app}-\psi^0\|\le C(\eps^{k/2}+\eps^kB_0)\|\psi^0\|_{H^k}$.

\emph{Conclusion.} The resolvent bound follows as in Theorem~\ref{thm:hom}, including the reduction \eqref{eq:shiftz}. If $\kappa_-\le|\omega_j|\le\kappa_+$, then $\mu(y)=\sum_jK_{\varpi}(y-j)\omega_j^2\in[\kappa_-^2,\kappa_+^2]$, so $\dist(M_\omega(y),\cK)\le C\varpi$. The envelope enters the hierarchy only through the coefficients of the operators $Q$, which involve finitely many derivatives of $\rho$, so any $\rho\in C_b^\infty$ is allowed. Finally, a potential $V_{\rm s}\in C_b^\infty$ added to both operators enters the equation for $\psi^0=(P+V_{\rm s}+\rho^2M^\eps_\omega-z)^{-1}f$, for which the elliptic estimate is unchanged. Applied to a corrector, it produces monomials with coefficients $V_{\rm s}\,b^\eps$ and order at least $\frac k2>0$. Such monomials require no further corrector, and they enter the residual directly, with norm $O(\eps^{k/2})\le C\delta_k(\eps)$.
\end{proof}

\begin{remark}[Rates and background potentials]\label{rem:Vs}
The choices $\varpi_0=\eps^{1/2}$ and $\varpi=\eps^{1/4}$ ($k=1$), $\varpi=\eps^{1/3}$ ($k\ge2$) balance the windowing error $\varpi+(\varpi_0/\varpi)^2$ against the size $\eps^{k/2}B_0$ of the quadratic corrector. For periodic configurations, $M_\omega$ is constant once $\varpi$ is small, and Theorem~\ref{thm:hom} gives the better rate $\eps^{r_k}$. The last paragraph of the proof shows that Theorem~\ref{thm:unifhom} holds, with the same rate, if a potential $V_{\rm s}\in C_b^\infty(\R^2;\Herm(N))$ is added to both $H^\eps_\omega$ and $P+\rho^2M^\eps_\omega$, with $C$ depending also on finitely many derivatives of $V_{\rm s}$. This version is in fact used in section \ref{sec:pseudo}.
\end{remark}

\paragraph{Gap-preserving laws.} For alloys, the spectrum of $H^\eps_B$ is a.s.\ deterministic \cite{Pastur80}, and equals the closure of the union of the spectra over all supported configurations, those with all couplings in the support of their law \cite{KirschMartinelli82}. The argument uses only Weyl sequences and independence, and so applies to the present operators. A spectral gap of uniform size for every supported configuration is therefore equivalent to an almost-sure gap.

\begin{proof}[Proof of Theorem~\ref{thm:phasepurek}]
Let $\delta_\eps=C\delta_k(\eps)$ be the distance bound of Theorem~\ref{thm:unifhom}. The first step, which treats (A) and (B) separately, constructs a convex set $\mathcal M_\eps$ of measurable fields $M(x)$ containing $M^\eps_\omega$ and $\bar M$, such that, for every $M\in\mathcal M_\eps$, $P+M$ is gapped on $(-g+\delta_\eps,g-\delta_\eps)$, and $P$ is gapped on $(-g,g)$. The rest of the proof uses only these properties.

\emph{Gaps in case (A).} Let $\mathcal M_\eps$ be the set of fields with values within $\delta_\eps$ of $\cK$, which is convex since $\cK$ is. For $M\in\mathcal M_\eps$, write $M=M^{\rm o}+M^{\rm e}$, and apply Lemma~\ref{lem:grade} to $P+M(x)$ with grading $s_-J$, $H_{\rm o}=P^{\rm o}(D)+M^{\rm o}(x)$ and $H_{\rm e}=\Gamma+M^{\rm e}(x)$. Here $J$ preserves $H^k$, and $P^{\rm o}(\xi)$ is Hermitian, as the $J$-anticommuting part of $P(\xi)$. The map $M\mapsto M^{\rm e}$ has norm $\le1$, so $s_-J(\Gamma+M^{\rm e}(x))\ge g-\delta_\eps$, and $P+M$ is gapped on $(-g+\delta_\eps,g-\delta_\eps)$. The bulk $P$ is gapped on $(-g,g)$ by Lemma~\ref{lem:grade} with grading $s_+J$.

\emph{Gaps in case (B).} Let $\mathcal M_\eps$ be the set of fields with $|M(x)-\bar M|\le\bar r+\delta_\eps$, a convex set. It contains every field with values within $\delta_\eps$ of $\cK$, since $|M-\bar M|\le\bar r$ on $\cK$. For $M\in\mathcal M_\eps$, $P+M$ is gapped on $(-g+\delta_\eps,g-\delta_\eps)$ by norm perturbation from $\dist(0,\spec(P+\bar M))\ge g+\bar r$, and $P$ is gapped on $(-g,g)$ by assumption.

\emph{(i)} We apply this with $M=M^\eps_\omega$, the field of Theorem~\ref{thm:unifhom} with $\rho\equiv1$. If $E\in\spec H^\eps_{\omega,B}$ with $|E|<g/2$, then $\dist(E,\spec(P+M^\eps_\omega))\ge g/2-\delta_\eps\ge g/4$ for $\eps$ small. Lemma~\ref{lem:speccomp} at $z=E+ig/8$, with $d=g/4$, then contradicts the resolvent estimate of Theorem~\ref{thm:unifhom}, uniformly in $\omega$, since $\eta^{-1}-(d^2+\eta^2)^{-1/2}=8(1-5^{-1/2})/g$.

\emph{(ii)} Let $H^0_\omega=P+\rho^2M^\eps_\omega$, and deform along $M_s=(1-s)M^\eps_\omega+s\bar M\in\mathcal M_\eps$, $s\in[0,1]$. The lower bulk $P+M_s$ stays gapped on $(-g/2,g/2)$ for $\eps$ small, and Lemma~\ref{lem:fred} shows that $T(P+\rho^2M_s)$ is Fredholm with constant index. At $s=1$, the index is $\BDI(P+\bar M,P)$ by Proposition~\ref{thm:QB}. Finally, Theorem~\ref{thm:unifhom} and the Helffer--Sj\"ostrand formula give $\|U(H^\eps_\omega)-U(H^0_\omega)\|<1$ for $\eps$ small, uniformly in $\omega$. Lemma~\ref{lem:fred}(i) then transfers the conclusion to $T(H^\eps_\omega)$. In case (A), every $\bar M\in\cK$ is allowed, and $\BDI(P+\bar M,P)$ therefore does not depend on $\bar M$.
\end{proof}

\section{Pseudo-gaps and deformed cuts}\label{sec:pseudo}

This section proves Theorems \ref{thm:pseudogap}, \ref{thm:robust} and \ref{thm:mainB}. We obtain spectrum near $0$ from approximate eigenfunctions (quasimodes) of gapless periodic patterns, placed in rare regions where the random field is close to the pattern, and spectrum in every compact energy window from constant configurations with small couplings. We then prove the robustness of the index along deformed cuts, construct admissible cuts by a percolation argument, introduce the surgery producing reference operators, and combine these ingredients with Theorem \ref{thm:unifhom} to prove Theorem \ref{thm:mainB}.

\paragraph{Rare-region spectrum and pseudo-gaps.}

\begin{proof}[Proof of Theorem~\ref{thm:pseudogap}(i)--(ii) and of the consequence stated after the theorem]
\emph{(ii) reduces to (i).} The map $t\mapsto M(W_t)$ is continuous. If $P+M(W_t)$ were gapped at $0$ for every $t\in[0,1]$, then $c(P+M(W_t))$ would be constant by Lemma~\ref{lem:BDIc}(i), and $\BDI(P+M(W_1),P+M(W_0))$ would vanish. Hence $0\in\spec(P+M(W_{t_*}))$ for some $t_*$, and the law charges $W_*=W_{t_*}$.

\emph{Proof of (i).} Let $H^\eps_*=P+\eps^{-k/2}W_*(x/\eps)$. By Theorem~\ref{thm:hom}, $\|R^\eps_*(i\eta)-R^0_*(i\eta)\|\le C\eps^{r_k}\eta^{-2}$ for $\eta\in(0,1]$. Since $0\in\spec(P+M(W_*))$, Lemma~\ref{lem:speccomp} with $E=0$, $d=2\eta$ and $\eta=2C\eps^{r_k}/c$, $c=1-5^{-1/2}$, gives $d_0:=\dist(0,\spec H^\eps_*)<2\eta$ when $\eta\le1$. For the remaining $\eps\le1$, $d_0\le\dist(0,\spec P)+\eps^{-k/2}\|W_*\|_\infty$ is bounded, so that $d_0\le C'\eps^{r_k}$ for all $\eps\in(0,1]$. Choose $E_\eps\in\spec H^\eps_*$ with $|E_\eps|\le C'\eps^{r_k}$, and $\psi\in C_c^\infty(B_{R})$ with $\|(H^\eps_*-E_\eps)\psi\|\le\eps^{r_k}\|\psi\|$. This is possible by the spectral theorem \cite[\S VIII.3]{ReedSimon1}, since $C_c^\infty$ is a core. By stationarity, independence at distance and the second Borel--Cantelli lemma, a.s.\ some $j\in\Z^d$ satisfies $\sup_{|y-j|\le R/\eps+1}|W(y)-W_*(y-j)|<\eps^{k/2+r_k}$. Then $\psi_j=\psi(\cdot-\eps j)$ satisfies $\|(H^\eps_B-E_\eps)\psi_j\|\le2\eps^{r_k}\|\psi\|$, so that $\dist(0,\spec H^\eps_B)\le(C'+2)\eps^{r_k}$.

\emph{Consequence.} If $H^{\eps_n}_B$ converged in norm-resolvent sense to an operator gapped at $0$, Lemma~\ref{lem:speccomp} would keep $\spec H^{\eps_n}_B$ a fixed distance away from $0$, which contradicts (i) and (iii).
\end{proof}

\begin{proof}[Proof of Theorem~\ref{thm:pseudogap}(iii)]
\emph{Step 1: a comparison estimate.} Let $W_{\rm per}:=\sum_{j\in\Z^d}w(\cdot-j)-\int w$, the periodic part of the lattice sum, a smooth, $1$-periodic, mean-zero function, and for $\lambda'\in\R$ let $H_{\lambda'}^\eps=P+\lambda'W_{\rm per}(x/\eps)$. We claim that for every $B_0,R_1>0$ and $\eta\in(0,1]$,
\begin{equation}\label{eq:smallosc}
\sup_{|\lambda'|\le B_0,\ |\Re z|\le R_1,\ |\Im z|=\eta}\big\|(H^\eps_{\lambda'}-z)^{-1}-(P-z)^{-1}\big\|\le C\eps^{r_k},
\end{equation}
with $C$ depending on $(P,w,B_0,R_1,\eta)$. Indeed, $H^\eps_{\lambda'}$ is of the form \eqref{eq:Heps} with $\rho\equiv1$ and the matrix-valued profile $\eps^{k/2}\lambda'W_{\rm per}$, whose correction $\eps^k\lambda'^2M(W_{\rm per})$ changes the resolvent of $P$ by $O(\eps^k\eta^{-2})$. Theorem~\ref{thm:hom} gives \eqref{eq:smallosc}, uniformly in $|\lambda'|\le B_0$, since its constants are uniform for coefficients in bounded sets. This corresponds to alloys with constant $\omega_j=\omega_*$.

\emph{Step 2: spectrum of $H^\eps_{\lambda'}$.} Let $E_1\in\spec P$ with $|E_1|\le R_1-1$, and let $\delta'\in(0,1]$. We claim that $\dist(E_1,\spec H^\eps_{\lambda'})\le\delta'$ for $\eps\le\eps_1(\delta',B_0,R_1)$, uniformly in $|\lambda'|\le B_0$. Otherwise Lemma~\ref{lem:speccomp}, with $H_2=P$, $E=E_1$ and $\eta=d=\delta'$, shows that the left-hand side of \eqref{eq:smallosc} with $\eta=\delta'$ is at least $(1-2^{-1/2})/\delta'$, which is impossible for $\eps$ small.

\emph{Step 3: constant configurations.} Suppose $\supp\mu_0\supset[0,\kappa_0]$. The case $[-\kappa_0,0]$ is symmetric. Let $\sigma=\sgn\int w$. For $t\ge0$, the constant configuration $\omega_j\equiv\omega_*:=t/(a_\eps|\int w|)$ is supported, i.e., has all couplings in $\supp\mu_0$, as soon as $\omega_*\le\kappa_0$, which holds for $t\le T_0$, with $T_0=T_0(R)$ fixed in Step 4, and $\eps$ small. For it,
\[
a_\eps\omega_*\sum_jw(x/\eps-j)=\sigma t+\lambda'W_{\rm per}(x/\eps),\qquad\lambda'=\frac{\sigma t}{\int w},
\]
so that $H_{\omega_*}:=P+a_\eps\omega_*\sum_jw(x/\eps-j)=\sigma t+H^\eps_{\lambda'}$.

\emph{Step 4: conclusion.} Since $P_k$ is indefinite, the eigenvalues of the symbol of $P$ tend to $+\infty$ along some branches and to $-\infty$ along others, so $\spec P$ is unbounded above and below and has bounded gaps. Hence there is $T_0=T_0(R)$ such that for every $|E|\le R$ there are $t\in[0,T_0]$ and $E_1\in\spec P$ with $E=\sigma t+E_1$ and $|E_1|\le R+T_0$. Apply Step 2 with $B_0=T_0/|\int w|$ and $R_1=R+T_0+1$: $\dist(E,\spec H_{\omega_*})=\dist(E_1,\spec H^\eps_{\lambda'})\le\delta'$ for $\eps$ small, uniformly in $|E|\le R$. Since the constant configuration $\omega_*$ is supported, $\spec H_{\omega_*}\subset\Sigma_\eps$ \cite{KirschMartinelli82}.
\end{proof}

\paragraph{Robustness of the index along deformed cuts.} Let now $d=2$ and $H_{\circ}=P(D)+V_{\circ}$ with $P$ of order $k\in\{1,2\}$ and $V_{\circ}\in L^\infty(\Rm^2;\Herm(N))$. We call $H_{\circ}$ a {\em reference operator} if:
\\[1mm]
 \textup{[R1]} \emph{decay}: there are $C_0\geq1$, $a\in(0,1]$ and $R_0>0$ with $\supp\varphi'\subset(-R_0/2,R_0/2)$ such that for all measurable $\Omega,\Omega'$ and $0<|\Im z|\le1$, $|\Re z|\le R_0$, $\|\one_\Omega(H_{\circ}-z)^{-1}\one_{\Omega'}\|\le C_0|\Im z|^{-1}e^{-a|\Im z|\dist(\Omega,\Omega')}$; and
 \\[1mm]
 \textup{[R2]} \emph{interface}: the conclusions of Lemma \ref{lem:fred}(i) hold for $H_{\circ}$. By Lemma~\ref{lem:fred}, this is the case when $V_{\circ}=V_\pm$ for $\pm x_2\ge X_0$ and $P+V_\pm$ are gapped on an open interval containing $\supp\varphi'$.

 We recall that deformed cuts are defined in Definition \ref{def:cut} and that $\tilde T(H)=\tilde\Pi\,U(H)\,\tilde\Pi+I-\tilde\Pi$ is the Toeplitz operator defined along the deformed cut.

\begin{theorem}[Robustness]\label{thm:robust}
Let $H_{\circ}$ be a reference operator and let $(\cL,\cR)$ be a cut. Let $V\in L^\infty(\Rm^2;\Herm(N))$ with $\|V\|_\infty\le K_V$, $K_V\geq1$, be such that the cut avoids $\supp V$ at distance $\dd$. For every $n\geq1$, there is $C_n$ depending only on $(C_0,R_0,\varphi,n)$ such that, if $\vartheta_1:=C_nK_V^3(a\dd)^{-n}<\frac1{16}$, then $\tilde T(H_{\circ}+V)$ is Fredholm and $\ind\tilde T(H_{\circ}+V)=\ind T(H_{\circ})$.
\end{theorem}
\noindent The perturbation $V$ need not be small, and its support may cover most of the plane, including parts of the interface away from the cut. Condition [R1] holds for every bounded potential when $k=1$, and follows from an $H^1$ resolvent bound when $k=2$ (Lemma~\ref{lem:CT} below).

\begin{proof}[Proof of Theorem~\ref{thm:robust}]
Let $H=H_{\circ}+V$, $U_{\circ}=U(H_{\circ})$, the relative unitary $U_{\rm rel}:=U_{\circ}^*U(H)$ and $\cJ:=\tilde\Pi U_{\rm rel}\tilde\Pi+1-\tilde\Pi$. Then
\begin{equation}\label{eq:factor}
\tilde T(H)=\tilde T(H_{\circ})\,\cJ+\tilde\Pi U_{\circ}(1-\tilde\Pi)[U_{\rm rel},\tilde\Pi].
\end{equation}
Step 1 shows that $\mathfrak q(\tilde T(H_{\circ}))$ is invertible with the index of $T(H_{\circ})$, Steps 2--3 that $[U_{\rm rel},\tilde\Pi]$ is compact plus small, and Step 4 that $\cJ$ is Fredholm of index $0$. Only the decay [R1] of the resolvent of $H_{\circ}$ is used. This matters for $k=2$, since the large defect $V$ destroys $H^1$ bounds behind Lemma~\ref{lem:CT}(ii).

All Helffer--Sj\"ostrand integrals are over $\supp\dbar\tilde S\subset\{|\Re z|\le R_0,\ |\Im z|\le1\}$, where [R1] applies. By \eqref{eq:HS} and Lemma~\ref{lem:HSdecay}, if $\|F(z)\|\le CK|\Im z|^{-m}e^{-a|\Im z|\Lambda}$ on this set, then
\begin{equation}\label{eq:HSa}
\Big\|\int\dbar\tilde S(z)F(z)\,d^2z\Big\|\le C_nK(a\Lambda)^{-n}.
\end{equation}
Let $\cO_\cL=\supp V\cap\cL$ and $\cO_\cR=\supp V\cap\cR$, so that $\dist(\cO_\cL,\cR)\ge\dd$ and $\dist(\cO_\cR,\cL)\ge\dd$, and let $V_\cL=V\one_{\cO_\cL}$, $V_\cR=V\one_{\cO_\cR}$, $H_\cL=H_{\circ}+V_\cL$ and $H_\cR=H_{\circ}+V_\cR$, with resolvents $R,R_{\circ},R_\cL,R_\cR$ at $z$, all of norm $\le|\Im z|^{-1}$. Write $S_{\circ}=S(H_{\circ})$, so that $U_{\circ}=1+S_{\circ}$.

\emph{Step 1: the deformed cut for $H_{\circ}$.} Let $\dloc(X):=\|\one_{\{|x_2|>X\}}S_{\circ}\|+\|S_{\circ}\one_{\{|x_2|>X\}}\|$, which tends to $0$ as $X\to\infty$ by [R2], and fix $X$ with $4\|S\|_\infty\dloc(X)<\frac1{16}$. By Definition~\ref{def:cut}, $\tilde\Pi-\Pi=\one_{A_+}-\one_{A_-}+\one_{\{|x_2|>X\}}(\tilde\Pi-\Pi)$ with $A_\pm$ bounded. Let $B\in\{S_{\circ},S_{\circ}^*\}$. The terms of $B(\tilde\Pi-\Pi)$ and $(\tilde\Pi-\Pi)B$ containing $\one_{A_\pm}$ are compact by [R2], applied to $S$ and $\bar S$, and the remaining ones have norm $\le\dloc(X)$. Since $[S_{\circ},\Pi]$ is compact by [R2], we obtain: \textup{(a)} $[B,\tilde\Pi]=[B,\Pi]+[B,\tilde\Pi-\Pi]$ is compact plus an operator of norm $\le2\dloc(X)$; and \textup{(b)} $\tilde T(H_{\circ})-T(H_{\circ})=(\tilde\Pi-\Pi)S_{\circ}\tilde\Pi+\Pi S_{\circ}(\tilde\Pi-\Pi)$ is compact plus an operator of norm $\le2\dloc(X)<\frac12$. Since $\mathfrak q(T(H_{\circ}))$ is unitary by [R2], $\mathfrak q(\tilde T(H_{\circ}))$ is invertible with inverse of norm $\le2$, and $\Ind\tilde T(H_{\circ})=\Ind T(H_{\circ})$.

\emph{Step 2: decomposition of $S(H)-S_{\circ}$.} The resolvent identities $R-R_\cL=-RV_\cR R_\cL$, $R_\cL=R_{\circ}-R_{\circ}V_\cL R_\cL$, $R-R_\cR=-RV_\cL R_\cR$ and $R_\cR=R_{\circ}-R_{\circ}V_\cR R_\cR$ give
\[
R-R_\cL-R_\cR+R_{\circ}=RV_\cR R_{\circ}V_\cL R_\cL+RV_\cL R_{\circ}V_\cR R_{\circ}-RV_\cL R_{\circ}V_\cR R_\cR V_\cR R_{\circ} .
\]
Since $\dist(\cO_\cL,\cO_\cR)\ge\dd$, [R1] gives $\|V_\cL R_{\circ}V_\cR\|+\|V_\cR R_{\circ}V_\cL\|\le2C_0K_V^2|\Im z|^{-1}e^{-a|\Im z|\dd}$. By \eqref{eq:HSa}, $S(H)-S_{\circ}=\Sigma_\cL+\Sigma_\cR+E$ with $\Sigma_\bullet=S(H_\bullet)-S_{\circ}$ and $\|E\|\le C_nK_V^3(a\dd)^{-n}$.

\emph{Step 3: almost commutation with $\tilde\Pi$.} By \eqref{eq:HS}, $\Sigma_\cL=-\frac1\pi\int\dbar\tilde S\,R_\cL V_\cL R_{\circ}=-\frac1\pi\int\dbar\tilde S\,R_{\circ}V_\cL R_\cL$. The first form shows that $\Sigma_\cL\tilde\Pi$ involves $V_\cL R_{\circ}\one_\cR$, the second that $\tilde\Pi\Sigma_\cL$ involves $\one_\cR R_{\circ}V_\cL$. Since $\dist(\cO_\cL,\cR)\ge\dd$, [R1] and \eqref{eq:HSa} give $\|\Sigma_\cL\tilde\Pi\|+\|\tilde\Pi\Sigma_\cL\|\le C_nK_V(a\dd)^{-n}$, and symmetrically $\|\Sigma_\cR(1-\tilde\Pi)\|+\|(1-\tilde\Pi)\Sigma_\cR\|\le C_nK_V(a\dd)^{-n}$. Since $[\Sigma,\tilde\Pi]=(1-\tilde\Pi)\Sigma\tilde\Pi-\tilde\Pi\Sigma(1-\tilde\Pi)$ for every bounded $\Sigma$, and $\|[E,\tilde\Pi]\|\le2\|E\|$,
\begin{equation}\label{eq:commsmall}
\big\|[S(H)-S_{\circ},\tilde\Pi]\big\|\le\vartheta_1=C_nK_V^3(a\dd)^{-n}.
\end{equation}

\emph{Step 4: factorization.} Since $U_{\rm rel}=1+U_{\circ}^*(S(H)-S_{\circ})$,
$
[U_{\rm rel},\tilde\Pi]=[S_{\circ}^*,\tilde\Pi]\,(S(H)-S_{\circ})+U_{\circ}^*\,[S(H)-S_{\circ},\tilde\Pi]
$,
so that, by Step 1(a), \eqref{eq:commsmall} and $\|S(H)-S_{\circ}\|\le2\|S\|_\infty$, $[U_{\rm rel},\tilde\Pi]$ is compact plus an operator of norm $\le\vartheta:=\vartheta_1+4\|S\|_\infty\dloc(X)<\frac18$. The identity \eqref{eq:factor} follows from $\tilde T(H)=\tilde\Pi U_{\circ}U_{\rm rel}\tilde\Pi+1-\tilde\Pi$ and $(1-\tilde\Pi)U_{\rm rel}\tilde\Pi=(1-\tilde\Pi)[U_{\rm rel},\tilde\Pi]$, and its last term has norm $\le\vartheta$ in the Calkin algebra. Moreover $\cJ\cJ^*-1=-\tilde\Pi[U_{\rm rel},\tilde\Pi](1-\tilde\Pi)[U_{\rm rel},\tilde\Pi]^*\tilde\Pi$, and similarly for $\cJ^*\cJ-1$, so that $\mathfrak q(\cJ)$ is invertible with $\|\mathfrak q(\cJ)^{-1}\|\le2$. Its index is $0$: the bounds of Steps~2--4 hold uniformly along $s\mapsto H_{\circ}+sV$, $s\in[0,1]$, and the corresponding operators $\cJ_s$ are Fredholm, depend norm-continuously on $s$ by \eqref{eq:HS}, and satisfy $\cJ_0=1$.

\emph{Conclusion.} By Steps~1(b) and~4, $\mathfrak q(\tilde T(H))=\mathfrak q(\tilde T(H_{\circ}))\mathfrak q(\cJ)+\beta$ with $\|\beta\|\le\vartheta$ and $\|(\mathfrak q(\tilde T(H_{\circ}))\mathfrak q(\cJ))^{-1}\|\le4$. Since $4\vartheta<1$, $\tilde T(H)$ is Fredholm, and, the index being locally constant and additive under products \cite[Ch.~IV, \S5]{Kato95},
$
\Ind\tilde T(H)=\Ind\tilde T(H_{\circ})+\Ind\cJ=\Ind T(H_{\circ})
$.
\end{proof}
We verify [R1] by Combes--Thomas estimates. Let $\phi$ be a smoothed and truncated version of $\dist(\cdot,\Omega')$, with $|\nabla\phi|\le1$, bounded second derivatives, $\phi\ge\dist(\Omega,\Omega')-2$ on $\Omega$ and $\phi\le2$ on $\Omega'$, the truncation being at a level larger than $\dist(\Omega,\Omega')$. For $0<b\le1$ and any operator $Z$,
\begin{equation}\label{eq:CTreduce}
\|\one_\Omega Z\one_{\Omega'}\|\le e^{4}e^{-b\dist(\Omega,\Omega')}\|e^{b\phi}Ze^{-b\phi}\| .
\end{equation}
Since $\phi$ is bounded and Lipschitz, $e^{\pm b\phi}$ are bounded and preserve $H^1$ and $H^2$. If $e^{b\phi}He^{-b\phi}=H+bB$ with $\|bB(H-z)^{-1}\|\le\frac12$, then $e^{b\phi}(H-z)e^{-b\phi}=(1+bB(H-z)^{-1})(H-z)$ on $\mathrm{dom}(H)$, and a Neumann series gives
\begin{equation}\label{eq:CTN}
\|e^{b\phi}(H-z)^{-1}e^{-b\phi}\|\le2|\Im z|^{-1}.
\end{equation}
\begin{lemma}[Decay of resolvents]\label{lem:CT}
Let $H=P(D)+V$ with $V\in L^\infty(\R^2;\Herm(N))$.
\\[1mm]
(i) If $P=\gamma(D)$, with $\gamma(\xi)=\xi_1\gamma_1+\xi_2\gamma_2$, $\gamma_i\in\Herm(N)$ and $\|\gamma\|=\sup_{|e|=1}\|\gamma(e)\|$, let $a_\gamma=1/(2\max(\|\gamma\|,1))$. Then for all $z$ with $0<|\Im z|\le1$ and measurable $\Omega,\Omega'$,
\[
\|\one_\Omega(H-z)^{-1}\one_{\Omega'}\|\le2e^4|\Im z|^{-1}e^{-a_\gamma|\Im z|\dist(\Omega,\Omega')} ,
\]
so that \textup{[R1]} holds with $C_0=2e^4$ and $a=a_\gamma$, independently of $V$.
\\[1mm]
(ii) If $P$ is of order two as in \textup{[C1]}, let $\eta\in(0,1]$, $R_0>0$ and $\Theta\ge1$, and suppose that $\|(H-z_0)^{-1}\|_{L^2\to H^1}\le\Theta$ for all $z_0$ with $|\Im z_0|=\eta$ and $|\Re z_0|\le R_0$. Then \textup{[R1]} holds with $C_0=2e^4$ and $a=c/\Theta$, where $c>0$ depends only on $P$.
\end{lemma}
\begin{proof}
In both cases $e^{b\phi}He^{-b\phi}=H+bB_b$ for $0<b\le1$, since the potential commutes with $e^{b\phi}$, and it suffices to choose $b=a|\Im z|$ with $\|bB_b(H-z)^{-1}\|\le\frac12$, and to conclude with \eqref{eq:CTN} and \eqref{eq:CTreduce}.
\\[1mm]
\emph{(i)} Here $B_b=i\gamma(\nabla\phi)$, and $\|bB_b(H-z)^{-1}\|\le b\|\gamma\||\Im z|^{-1}\le\frac12$ for $b=a_\gamma|\Im z|\le1$.
\\[1mm]
\emph{(ii)} Here $B_b$ is a first-order operator whose coefficients are bounded in terms of $(A,B,c_j)$, uniformly in $b\le1$, so that $\|B_bu\|\le C_P\|u\|_{H^1}$ with $C_P\ge1$. Let $0<|\Im z|\le1$ and $|\Re z|\le R_0$, and $z_0=\Re z+i\eta\,\sgn\Im z$. From $(H-z)^{-1}=(H-z_0)^{-1}\bigl(1+(z-z_0)(H-z)^{-1}\bigr)$ and $|z-z_0|\le1$, $\|(H-z)^{-1}\|_{L^2\to H^1}\le\Theta(1+|\Im z|^{-1})\le2\Theta|\Im z|^{-1}$. Hence $b=|\Im z|/(4C_P\Theta)$ works.
\end{proof}

\paragraph{Existence of admissible cuts.}

We use a Peierls bound \cite{Grimmett99}: the number of $*$-connected sets (connected for the adjacency $0<\|x-y\|_\infty\le1$) of $n$ sites of $\Z^2$ containing a given site is at most $64^n$, since such a set has a spanning tree, which is traversed by a walk of length $2(n-1)$, crossing each edge twice, in a graph of degree $8$. Suppose that each site of $\Z^2$ is bad with probability at most $p'$, and that the events that sites at mutual $\infty$-distance larger than $7$ are bad are independent. Any $n$ sites contain at least $n/225$ sites at mutual distance larger than $7$, chosen greedily since each choice excludes at most $15^2$ sites. Hence a given set of $n$ sites is entirely bad with probability at most $(p')^{n/225}$, and the $*$-cluster of bad sites containing a given site has at least $m$ sites with probability at most $(p'')^m$, $p'':=64(p')^{1/225}$. In particular, all such clusters are finite a.s.\ if $p''<1$.

\begin{lemma}\label{lem:cuts}
Let the couplings $\omega_m$ be independent, let $r=\eps^{-\nu_r}$ and $\dd=\eps^{-\nu_\dd}$ with $0<\nu_r<\nu_\dd<p_0/2$ for some $p_0\in(0,1]$, and let $\cU\subset\R^2$ be a random set. If, for every unit box $Q$, the event that $Q$ meets $\cU$ has probability at most $C\eps^{p_0}$ and depends only on the couplings $\omega_m$ with $|\eps m-c_Q|\le2r$, then there is $\eps_0>0$ such that, for every $\eps\le\eps_0$, a.s., some cut avoids $\cU_{3r}=\{x:\dist(x,\cU)\le3r\}$ at distance $\dd/2$.
\end{lemma}
\noindent In this paper, $\cU$ is the union of the bad mesoboxes of first-order alloys, defined below, or of the bad unit boxes at criticality, defined in section \ref{sec:crit}. The proof is a coarse-graining argument of a standard type in percolation theory \cite{Grimmett99}.
\begin{proof}
Tile $\R^2$ by coarse boxes $B$ of side $\dd$, indexed by $\Z^2$, and call $B$ bad if some unit box within distance $2\dd$ of $B$ meets $\cU$. This event depends only on the couplings within distance $2\dd+2+2r\le3\dd$ of $B$, so that the events that coarse boxes at index distance larger than $7$ are bad are independent. Moreover $\PP(B\text{ bad})\le C\dd^2\eps^{p_0}= C\eps^{p_0-2\nu_\dd}=:p'\to0$. By the Peierls bound, for $\eps$ small, $p''=64(p')^{1/225}\le\frac12$, and a.s.\ all $*$-connected clusters of bad coarse boxes are finite. If $u$ is a unit box meeting $\cU$, every coarse box within distance $2\dd$ of $u$ is bad. These boxes cover the $2\dd$-neighbourhood of $u$ and are $*$-connected, so that the bad cluster $\mathcal C(u)$ containing them satisfies $\dist(u,\R^2\setminus\mathcal C(u))\ge2\dd$.

Let $\mathcal C_1,\mathcal C_2,\dots$ be the bad clusters that meet the strip $\{|x_1|\le\dd\}$. Set $\cR:=\{x_1\ge0\}\setminus\bigcup_i\mathcal C_i$ and $\cL=\R^2\setminus\cR$. Let $u$ be a unit box meeting $\cU$. If $\mathcal C(u)=\mathcal C_i$ for some $i$, then $u\subset\cL$ and $\dist(u,\cR)\ge\dist(u,\R^2\setminus\mathcal C_i)\ge2\dd$. Otherwise $\mathcal C(u)$ does not meet the strip, so that $u\subset\{|x_1|>\dd\}$, and $\dist(u,\mathcal C_i)\ge2\dd$ for every $i$, since the clusters $\mathcal C_i$ and $\mathcal C(u)$ are distinct. If $u\subset\{x_1>0\}$ then $u\subset\cR$ and $\dist(u,\cL)\ge\dd$; if $u\subset\{x_1<0\}$ then $u\subset\cL$ and $\dist(u,\cR)\ge\dd$. Since $3r\ll\dd$, the cut avoids $\cU_{3r}$ at distance $\dd/2$. It remains to show that only finitely many $\mathcal C_i$ meet a given window $\{|x_2|\le X\}$; since each $\mathcal C_i$ is bounded, $\cR$ then differs from $\{x_1\ge0\}$ there by a bounded set. Finiteness of the clusters is not enough, since the window is an infinite horizontal strip. The strip $\{|x_1|\le\dd\}$ meets at most four coarse boxes in each row of the coarse lattice. A $*$-connected cluster that contains a coarse box of the strip in row $n$ and meets the window contains at least $m_n:=\max(0,|n|-X/\dd-1)$ coarse boxes, which by the Peierls bound has probability at most $(p'')^{m_n}$. Hence the expected number of clusters meeting both the strip and the window is at most $4\sum_{n\in\Z}(p'')^{m_n}<\infty$, and, a.s., only finitely many of them do, simultaneously for all integers $X$.
\end{proof}

\paragraph{Reference operators by surgery.} Reference operators are obtained from $H^\eps$ by interpolating, in the vicinity of a set $\cU$ containing the rare regions, between the fluctuation and its homogenized limit. Let $\cU\subset\R^2$ and $r\ge4$. Fix a cutoff $\theta=\theta_{\cU,r}\in C^\infty(\R^2;[0,1])$ with
\[
\theta=0\ \text{on }\cU_{2r},\qquad\theta=1\ \text{outside }\cU_{3r},
\]
obtained by mollifying the indicator of the complement of $\cU_{5r/2}$ at scale $r/4$. Then $\|\partial^\alpha\theta\|_\infty\le C_\alpha r^{-|\alpha|}\le C_\alpha$ for every multi-index $\alpha$, with $C_\alpha$ independent of $\cU$ and $r$. Let $\tilde V$ be the fluctuation and $M_\star\in\Herm(N)$ constant. Define
\begin{equation}\label{eq:Hg}
H_{\circ}:=P+\rho\,\theta\,\tilde V+(1-\theta^2)\rho^2M_\star,\qquad V:=H^\eps-H_{\circ}=\rho(1-\theta)\tilde V-(1-\theta^2)\rho^2M_\star .
\end{equation}
Then $V$ is supported in $\cU_{3r}$, and $\|V\|_\infty\le\|\rho\|_\infty\|\tilde V\|_\infty+\|\rho\|_\infty^2|M_\star|$. Write $H_{\circ,B}$ for the same construction with $\rho\equiv1$.

\paragraph{Bad mesoboxes for first-order alloys.} We make precise the objects introduced in block (vi) of section \ref{sec:setting}, with $r$ and $\dd$ as in \eqref{eq:exponents}, $p_0=1$, and $g,\bar r$ the constants of [BG]. Let $\varpi=\eps^{1/4}$ and let $M_\omega$ be the local effective matrix \eqref{eq:Mloc}. A {\em mesobox} is a square $Q'$ of side $\eps\varpi^{-1}$ aligned with $\eps\varpi^{-1}\Zm^2$, with fast-variable image $\eps^{-1}Q'$. Let $M^{\rm tr}_{Q'}$ be the local effective matrix computed with the couplings $\omega_j$ at distance at most $\eps^{-1-\nu_r}$ from $\eps^{-1}Q'$, the other couplings being replaced by $0$. The mesobox is {\em bad} if $|M^{\rm tr}_{Q'}-M_{\rm Born}|>\bar r/4$ somewhere on $\eps^{-1}Q'$. Let $\cU_\omega$ be the union of the bad mesoboxes. Bad mesoboxes are rare because the local effective matrices concentrate around $M_{\rm Born}$ (Lemma~\ref{lem:mean}, which holds for every $k\ge1$ and every $\varpi$ small, with constants depending only on $(P,w,\kappa_+)$). We write $K_j(y)=K_\varpi(y-j)$, so that $\sum_jK_j(y)=1$, $\sup_jK_j(y)\le C\varpi^2$ and $\sum_jK_j(y)^{1/2}\le C\varpi^{-1}$, uniformly in $y$, and $k_0(m)=\int w(z)\chi_w(z-m)\,dz$ as in the proof of Lemma \ref{lem:meso}, so that $\sum_m(1+|m|)|k_0(m)|<\infty$ and $\sum_mk_0(m)e^{2\pi iq\cdot m}=\mathbb K(q)$. In particular $\mathbb K$ is continuous, $\int_{\T^2}\mathbb K=k_0(0)$ and $\mathbb K(0)=\sum_mk_0(m)$.
\begin{lemma}[Born mean and concentration; cf.~\cite{FOP82,Bal10}]\label{lem:mean}
(i) For every configuration and every $y$, $M_\omega(y)=\frac12(F(y)+F(y)^*)+O(\varpi)$, with
\[
F(y)=\sum_{j,l\in\Z^2}a_ja_l\,k_0(l-j),\qquad a_j=K_j(y)^{1/2}\omega_j .
\]
For i.i.d.\ couplings, $|\EE M_\omega(y)-M_{\rm Born}|\le C\varpi$.
\\[1mm]
(ii) For i.i.d.\ couplings, every $y$ and every $t\in(0,1]$, $\PP\bigl(|M_\omega(y)-\EE M_\omega(y)|>t+C\varpi\bigr)\le C\exp(-ct^2\varpi^{-2})$.
\end{lemma}
\begin{proof}
\emph{(i)} The representation is Steps 1--3 of the proof of Lemma \ref{lem:meso}. Since $\EE\omega_j\omega_l=v\delta_{jl}+m_1^2$, with $m_1=\EE\omega$ and $v={\rm Var}\,\omega$,
\[
\EE F(y)=v\,k_0(0)\sum_jK_j+m_1^2\sum_mk_0(m)\sum_j(K_jK_{j+m})^{1/2}.
\]
Since $\sqrt{K_\varpi}$ is Lipschitz with constant $C\varpi^2$, $|(K_jK_{j+m})^{1/2}-K_j|\le C\varpi^2|m|K_j^{1/2}$, so that $\sum_j(K_jK_{j+m})^{1/2}=1+O(\varpi|m|)$. Hence $\EE F=v\int_{\T^2}\mathbb K+m_1^2\mathbb K(0)+O(\varpi)=M_{\rm Born}+O(\varpi)$, and $M_{\rm Born}$ is Hermitian.

\emph{(ii)} By (i), it suffices to bound the deviations of $F(y)$. Write $\omega_j=m_1+\tilde\omega_j$, with $\tilde\omega_j$ i.i.d., centered and bounded by $2\kappa_+$. Then $F-\EE F=F_2+F_1$, with
\[
F_2=\sum_{j,l}(K_jK_l)^{1/2}k_0(l-j)\bigl(\tilde\omega_j\tilde\omega_l-v\delta_{jl}\bigr),\qquad F_1=m_1\sum_j\tilde\omega_j\sum_l(K_jK_l)^{1/2}\bigl(k_0(l-j)+k_0(j-l)\bigr).
\]
Fix one of the $2N^2$ real entries $\Re(\cdot)_{ab}$, $\Im(\cdot)_{ab}$. For $F_2$ it is $\ip{\tilde\omega}{\mathsf A\tilde\omega}-\EE\ip{\tilde\omega}{\mathsf A\tilde\omega}$ for a real matrix with $|\mathsf A_{jl}|\le(K_jK_l)^{1/2}|k_0(l-j)|$. By the Schur test, $\|\mathsf A\|\le\sup_jK_j\sum_m|k_0(m)|\le C\varpi^2$, and $\|\mathsf A\|_{HS}^2\le\sup_lK_l\sum_m|k_0(m)|^2\le C\varpi^2$. The Hanson--Wright inequality \cite[Theorem~1.1]{RV13}, applied to finite truncations, which converge in $L^2(\PP)$ since $\|\mathsf A\|_{HS}<\infty$ and $\sum_j|\mathsf A_{jj}|\le C\sum_jK_j<\infty$, gives the tail $2\exp(-c\min(t^2\varpi^{-2},t\varpi^{-2}))=2\exp(-ct^2\varpi^{-2})$ for $t\le1$. For $F_1$, the entry is $\sum_jb_j\tilde\omega_j$ with $|b_j|\le2|m_1|K_j^{1/2}\sup_lK_l^{1/2}\sum_m|k_0(m)|\le C\varpi K_j^{1/2}$, so that $\sum_jb_j^2\le C\varpi^2$, and Hoeffding's inequality \cite[Theorem~2.8]{BLM13} gives the same tail. A union bound over the entries concludes.
\end{proof}
\begin{lemma}[Bad mesoboxes]\label{lem:mesobad}
\textup{(a)} For every mesobox $Q'$ and every configuration, $|M_\omega-M^{\rm tr}_{Q'}|\le C\eps^{1+\nu_r}$ on $\eps^{-1}Q'$. \textup{(b)} A unit box meets a bad mesobox with probability at most $C\eps^{-3/2}\exp(-c\min(\bar r,1)^2\eps^{-1/2})\le\eps$ for $\eps$ small, and this event depends only on the couplings within distance $2\eps^{-\nu_r}$ of the box. \textup{(c)} On the fast-variable image of every good mesobox, $|M_\omega-M_{\rm Born}|\le\bar r/2$.
\end{lemma}
\begin{proof}
(a) Let $L=\eps^{-1-\nu_r}$ and $y\in\eps^{-1}Q'$. In $A(y)=\int K_\varpi(y-y')\Xi(y')dy'$, see \eqref{eq:Mloc}, the region $|y'-y|\ge L/2$ contributes $O((\varpi L)^{-n})$, since $\Xi$ and $\Xi^{\rm tr}$ are bounded and $\varpi L=\eps^{-3/4-\nu_r}\to\infty$. For $|y'-y|<L/2$, $W(y')$ is unchanged, and $|\chi(y')-\chi^{\rm tr}(y')|\le\kappa_+\sum_{|l-y'|\ge L/2}|\chi_w(y'-l)|\le CL^{-1}$, since $|\chi_w(z)|\le C\avg z^{-2-k}$ by Lemma \ref{lem:latsum}. Hence $|A-A^{\rm tr}|\le CL^{-1}$ on $\eps^{-1}Q'$, and likewise for $M_\omega$.

(b) By \eqref{eq:Mloc}, $\nabla M_\omega=\frac12\bigl((\nabla K_\varpi)*\Xi+((\nabla K_\varpi)*\Xi)^*\bigr)$, so that $|\nabla M_\omega|\le C\varpi$. Let $t=\min(\bar r,1)/32$, and take a net of $\eps^{-1}Q'$, which has side $\varpi^{-1}$, with mesh $t/(C\varpi)$. It has at most $Ct^{-2}$ points. If $Q'$ is bad, then by (a), for $\eps$ small, $|M_\omega-M_{\rm Born}|>\bar r/8$ somewhere on $\eps^{-1}Q'$, hence $|M_\omega-M_{\rm Born}|>\bar r/8-t\ge3t$ at a point of the net, and by Lemma \ref{lem:mean}(i), $|M_\omega-\EE M_\omega|>3t-C\varpi\ge t+C\varpi$ there. Lemma \ref{lem:mean}(ii) and a union bound over the net bound the probability by $Ct^{-2}\exp(-ct^2\varpi^{-2})$. A unit box meets at most $C\eps^{-3/2}$ mesoboxes. The dependence statement follows from the definition of $M^{\rm tr}_{Q'}$, since the fast distance $\eps^{-1-\nu_r}$ is the macroscopic distance $\eps^{-\nu_r}$.

(c) On a good mesobox, $|M^{\rm tr}_{Q'}-M_{\rm Born}|\le\bar r/4$, and (a) adds at most $C\eps^{1+\nu_r}\le\bar r/4$.
\end{proof}

\begin{proof}[Proof of Theorem~\ref{thm:mainB}]
We use the exponents \eqref{eq:exponents}, and the reference operator $H_{\circ}$ defined by \eqref{eq:Hg} with $\tilde V=\eps^{-\frac12}W_\omega(x/\eps)$ and $M_\star=M_{\rm Born}$. Constants depend on $(P,w,\kappa_+,g,\bar r,\rho,\varphi)$.

\emph{Step 1: bad set and cuts.} By Lemma \ref{lem:mesobad}(b), a unit box meets a bad mesobox with probability at most $\eps$, and this event depends only on the couplings within distance $2\eps^{-\nu_r}=2r$ of the box. Lemma \ref{lem:cuts}, with $p_0=1$, gives a.s., for $\eps\le\eps_0$, an admissible cut.

\emph{Step 2: reference operator.} Let $\theta=\theta_{\cU_\omega,r}$, and let $H_{\circ}$ and $V=H^\eps-H_{\circ}$ be given by \eqref{eq:Hg} with $M_\star=M_{\rm Born}$. Since $\theta=0$ on $(\cU_\omega)_{2r}$, the mesoboxes meeting $\supp\theta$ are good, and $|M_\omega-M_{\rm Born}|\le\bar r/2$ on their fast-variable images by Lemma \ref{lem:mesobad}(c). Apply Theorem \ref{thm:unifhom}, in the form of Remark \ref{rem:Vs}, with the envelope $\tilde\rho=\theta\rho$ and the potential $V_{\rm s}=(1-\theta^2)\rho^2M_{\rm Born}$, whose derivatives are bounded independently of $\cU_\omega$, $r$ and $\eps$. Since $H_{\circ}=P+\eps^{-\frac12}\tilde\rho W_\omega(x/\eps)+V_{\rm s}$, this gives the resolvent estimate of Theorem \ref{thm:unifhom}, with rate $C\eps^{1/4}$, for $H_{\circ}$ against $P+V_{\rm s}+\tilde\rho^2M^\eps_\omega=P+\rho^2M_{\circ}$, with $M_{\circ}=\theta^2M^\eps_\omega+(1-\theta^2)M_{\rm Born}$, and for $H_{\circ,B}$ against $P+M_{\circ}$. Only the values of $M_\omega$ on the fast-variable image of $\supp\theta$ enter, so $|M_{\circ}-M_{\rm Born}|\le\bar r/2$ everywhere.

\emph{Step 3: gaps.} Let $M(x)$ be measurable with $|M(x)-M_{\rm Born}|\le\bar r/2$. Multiplication by $M-M_{\rm Born}$ has norm $\le\bar r/2$, so that $\dist(0,\spec(P+M))\ge g+\bar r/2$ by [BG]. Applying this to $M_{\circ}$, $H_{\circ,B}$ is gapped on $(-g/2,g/2)$ for $\eps$ small: for $E\in\spec H_{\circ,B}$ with $|E|<g/2$, $\dist(E,\spec(P+M_{\circ}))>g/2$, and Lemma~\ref{lem:speccomp} at $z=E+ig/8$, with $d=g/2$, contradicts the resolvent estimate of Step~2. The upper bulk $P$ is gapped on $(-g,g)$ by [BG].

\emph{Step 4: index of the reference operator.} We deform $M_{\circ}$ along $M_s=(1-s)M_{\circ}+sM_{\rm Born}$, $s\in[0,1]$. The ball $\{|M-M_{\rm Born}|\le\bar r/2\}$ is convex, so the lower bulk of $P+\rho^2M_s$ stays gapped on $(-g,g)$ by Step 3. By Lemma \ref{lem:fred}, $\ind T(P+\rho^2M_s)$ is constant, and at $s=1$ it equals $\BDI(P+M_{\rm Born},P)$ by Proposition~\ref{thm:QB}. By the resolvent estimate of Step~2 and the Helffer--Sj\"ostrand formula, $\|U(H_{\circ})-U(P+\rho^2M_{\circ})\|\le C\eps^{1/4}<1$ for $\eps$ small, so Lemma \ref{lem:fred} gives $\ind T(H_{\circ})=\BDI(P+M_{\rm Born},P)$.

\emph{Step 5: robustness.} [R1] holds for $H_{\circ}$ with $C_0=2e^4$ and $a=a_\gamma$, independent of $\eps$, by Lemma~\ref{lem:CT}(i), and [R2] holds by Lemma \ref{lem:fred}, since the bulks of $H_{\circ}$ are $P$ and $H_{\circ,B}$, both gapped on $(-g/2,g/2)\supset\supp S$. The defect $V$ is supported in $(\cU_\omega)_{3r}$ and $\|V\|_\infty\le K_V:=C\eps^{-1/2}$. An admissible cut avoids $(\cU_\omega)_{3r}\supset\supp V$ at distance $\dd/2$, so that Theorem~\ref{thm:robust} applies with $\dd/2$ in place of $\dd$. Since $\vartheta_1=C_nK_V^3(a_\gamma\dd/2)^{-n}\le C_n\eps^{-3/2+n\nu_\dd}<\frac1{16}$ for $n\nu_\dd>3/2$ and $\eps$ small, it gives $\ind\tilde T(H^\eps)=\ind T(H_{\circ})=\BDI(P+M_{\rm Born},P)$.
\end{proof}
\paragraph{Distant defects for periodic fluctuations.}
\begin{corollary}[Distant defects, periodic fluctuations]\label{cor:farV}
Assume the hypotheses of Theorem~\ref{thm:TAI} with $k=1$, and let $K_V\ge1$. There are $\eps_0>0$, and $\dd_0$ depending only on $(\gamma,g,\varphi,K_V)$, such that for every $\eps\in(0,\eps_0]$, every cut $(\cL,\cR)$ and every $V\in L^\infty(\R^2;\Herm(N))$ with $\|V\|_\infty\le K_V$ whose support the cut avoids at distance $\dd_0$, $\tilde T(H^\eps+V)$ is Fredholm and $\Ind\tilde T(H^\eps+V)=\BDI(P+M,P)$. In particular, this holds along the straight cut for every such $V$ supported in $\{|x_1|\ge\dd_0\}$.
\end{corollary}
\begin{proof}
Take $H_\circ=H^\eps$, which for $k=1$ is $\gamma(D)$ plus a bounded potential. By Lemma~\ref{lem:CT}(i), [R1] holds with $C_0=2e^4$, $a=a_\gamma$ and $R_0=2g$, uniformly in $\eps$. By Theorem~\ref{thm:hom} with $\rho\equiv1$ and Lemma~\ref{lem:speccomp}, the lower bulk of $H^\eps$ is gapped on $(-3g/4,3g/4)$ for $\eps$ small, and $P$ is gapped on $(-g,g)$ by [G], so that [R2] holds by Lemma~\ref{lem:fred}. Choose $\dd_0$ with $C_1K_V^3(a_\gamma\dd_0)^{-1}<\frac1{16}$. Theorems~\ref{thm:robust} and~\ref{thm:TAI} give $\Ind\tilde T(H^\eps+V)=\Ind T(H^\eps)=\BDI(P+M,P)$.
\end{proof}
\noindent The threshold $\dd_0$ does not depend on $\eps$, because the decay [R1] of first-order operators does not depend on the potential. For $k=2$, a matrix potential and $P$ as in [C1], the same argument applies with Lemma~\ref{lem:CT}(ii), but $\dd_0$ then depends on $\eps$ through the $H^1$ bound on the resolvent of $H^\eps$.

\section{Proofs of the results at criticality}
\label{sec:crit}

This section states and proves the critical homogenization Theorem~\ref{thm:det} and proves Theorem \ref{thm:mainA}. The probabilistic estimates are proved in appendix~\ref{app:proba}. The proof of Theorem \ref{thm:det} removes the singular part of the potential by an exponential transform, which leaves the renormalized effective matrix and a perturbation that is small in negative H\"older norms when the configuration is small. Theorem \ref{thm:mainA} then follows as Theorem \ref{thm:mainB}, with the surgery performed near the bad unit boxes.

In this section and in appendix~\ref{app:proba}, $\hat f(\xi)=\int f(x)e^{-ix\cdot\xi}dx$. We use the dyadic Littlewood--Paley decomposition $\{\Delta_j\}_{j\geq-1}$ of \cite[\S2.2]{BCD11}, with $\Delta_j=\phi_j(D)$ and $\phi_j=\phi(2^{-j}\cdot)$ supported in $\{\frac342^j\le|\xi|\le\frac832^j\}$ for $j\geq0$, and denote by $k_j$ the kernel of $\Delta_j$. The H\"older--Zygmund spaces $C^s=B^s_{\infty,\infty}(\R^2)$ have norms $\|f\|_{C^s}=\sup_j2^{js}\|\Delta_jf\|_{L^\infty}$. More generally, $B^s_{p,r}$ denotes the Besov space with norm $\|(2^{js}\|\Delta_jf\|_{L^p})_j\|_{\ell^r}$, and $H^s=B^s_{2,2}$, $s\in\R$, is the $L^2$ Sobolev space. For a unit box $Q=m+[0,1)^2$ we set $Q^*=m+[-1,2]^2$. We fix a partition of unity $\sum_{m\in\Z^2}\psi(\cdot-m)\equiv1$ with $\psi\in C^\infty_c((-1,2)^2;[0,1])$ and set $\psi_Q=\psi(\cdot-m)$. By the localization principle \cite[Theorem 2.4.7]{Triebel92},
\begin{equation}\label{eq:locprinc}
\|u\|_{C^s}\le C_s\sup_Q\|\psi_Qu\|_{C^s}\qquad(u\in C^s,\ s\in\R).
\end{equation}

\begin{lemma}\label{lem:prod}
Let $\kappa\in(0,\frac14)$. Then
\[
\|fg\|_{H^{-2\kappa}}\le C\|f\|_{C^{-\kappa}}\|g\|_{H^{1-2\kappa}},\qquad
\|fh\|_{C^{-\kappa}}\le C\|f\|_{C^{-\kappa}}\|h\|_{C^{1-\kappa}} .
\]
\end{lemma}
\begin{proof}
Write $fg=T_fg+T_gf+R(f,g)$ (Bony's decomposition, \cite[Ch.~2]{BCD11}), with $T_fg=\sum_jS_{j-1}f\,\Delta_jg$. By \cite[Theorem~2.82]{BCD11}, applied with the low-frequency factor $f\in B^{-\kappa}_{\infty,\infty}$, $\|T_fg\|_{H^{1-3\kappa}}\le C\|f\|_{C^{-\kappa}}\|g\|_{H^{1-2\kappa}}$. By \cite[Theorem~2.85]{BCD11}, since $-\kappa+(1-2\kappa)>0$, $\|R(f,g)\|_{H^{1-3\kappa}}\le C\|f\|_{C^{-\kappa}}\|g\|_{H^{1-2\kappa}}$. In $T_gf$ the rough factor $f$ carries the high frequencies, and its summability index $\infty$ is inherited: since $\Delta_jT_gf=\sum_{|j'-j|\le4}\Delta_j(S_{j'-1}g\,\Delta_{j'}f)$,
\[
\|\Delta_jT_gf\|_{L^2}\le C\sum_{|j'-j|\le4}\|S_{j'-1}g\|_{L^2}\|\Delta_{j'}f\|_{L^\infty}\le C2^{j\kappa}\|g\|_{L^2}\|f\|_{C^{-\kappa}},
\]
so that $T_gf\in B^{-\kappa}_{2,\infty}\subset H^{-2\kappa}$, with norm $\le C\|g\|_{L^2}\|f\|_{C^{-\kappa}}$. Since $1-3\kappa>0$, this proves the first bound. For the second, $\|T_hf\|_{C^{-\kappa}}\le C\|h\|_{L^\infty}\|f\|_{C^{-\kappa}}$, while $T_fh$ and $R(f,h)$ lie in $C^{1-2\kappa}$, by \cite[Theorems~2.82 and~2.85]{BCD11}, since $1-2\kappa>0$.
\end{proof}

\paragraph{The local size.} We make precise the objects introduced in block (vii) of section \ref{sec:setting}. Let $\chi_0\in C^\infty_c(\Rm^2)$ be radial and equal to $1$ near $0$. With $\tilde V$ as in \eqref{eq:Hcrit}, we set $\tilde V_{\rm low}=\chi_0(D)\tilde V$, $\tilde V_{\rm high}=(1-\chi_0(D))\tilde V$ and $Y_0=-q(D)^{-1}\tilde V_{\rm high}$. With $A$, $B$ and $q$ as in [C1], $|v|_B^2:=v\cdot Bv$ and $Y_\eps=\eps[0,1)^2$, we define
\begin{equation}\label{eq:ceps}
  c_\eps:=\frac1{|Y_\eps|}\dint_{Y_\eps}\EE|\nabla Y_0(x)|_B^2\,dx,\qquad\cW:=|\nabla Y_0|_B^2-c_\eps,\qquad M_\eps:=-c_\eps A^{-1}.
\end{equation}
For $\kappa\in(0,\frac14)$ fixed throughout and a unit box $Q$, the local size of a configuration is
\begin{equation}\label{eq:deltaQ}
  \delta_Q:=\|Y_0\|_{L^\infty(Q^*)}+\|\tilde V_{\rm low}\|_{L^\infty(Q^*)}+\|\psi_Q\nabla Y_0\|_{C^{-\kappa}}+\|\psi_Q\cW\|_{C^{-\kappa}},\qquad\delta:=\sup_Q\delta_Q .
\end{equation}
\begin{theorem}[Critical homogenization]\label{thm:det}
Assume [C1]--[C3], and let $R\ge1$. Let $\theta\in C^\infty(\R^2;[0,1])$ have bounded derivatives up to order two, let
\[
H_\theta:=P+\rho\,\theta\,\tilde V+(1-\theta^2)\rho^2M_\eps ,
\]
and let $\delta_\theta:=\sup\{\delta_Q:\dist(Q,\supp\theta)\le7\}$. There are $\delta_*>0$ and $C$, depending only on $(A,B,c_j,P_0,\rho,R)$, on the law of the couplings and on the bounds on the derivatives of $\theta$, such that if $\delta_\theta\le\delta_*$, then for all $z$ with $|\Re z|\le R$ and $\Im z\ne0$,
\[
  \big\|(H_\theta-z)^{-1}-(P+\rho^2M_{\rm crit}-z)^{-1}\big\|\le C\bigl(1+|\Im z|^{-2}\bigr)\Big(\delta_\theta+\frac1{\log(1/\eps)}\Big).
\]
Moreover, with $N_\theta:=\sup\{\|\nabla Y_0\|_{L^\infty(Q^*)}:\dist(Q,\supp\theta)\le7\}$, $\|(H_\theta-z_0)^{-1}\|_{L^2\to H^1}\le C(1+N_\theta)$ for $|\Im z_0|=1$ and $|\Re z_0|\le R$. The same holds with $\rho\equiv1$ and $P+M_{\rm crit}$ in place of $P+\rho^2M_{\rm crit}$.
\end{theorem}
\noindent Theorem \ref{thm:det} is deterministic: its constants do not depend on $\eps$ or on the configuration, the law entering only through Lemma~\ref{lem:ceps}, and $A$ may be indefinite. For $\theta\equiv1$, $H_\theta=H^\eps$ and $\delta_\theta=\delta$. However, $\delta\geq1$ a.s.\ on $\R^2$ for $\eps$ small (Lemma~\ref{S-lem:globalsize} of the supplementary material). This is why Theorem \ref{thm:det} is stated with a cutoff $\theta$: with $\theta$ vanishing near the unit boxes of large local size, $H_\theta$ is the reference operator \eqref{eq:Hg} used in the proof of Theorem \ref{thm:mainA}.

\paragraph{The renormalization constant.}
\begin{lemma}\label{lem:ceps}
With $s_0$, $c_q$ and $M_{\rm crit}$ as in block (vii) of section \ref{sec:setting}, $c_\eps=\ell s_0c_q+O(1/\log(1/\eps))$, and hence $|M_\eps-M_{\rm crit}|\le C/\log(1/\eps)$.
\end{lemma}
\begin{proof}
Averaging over $Y_\eps$ replaces the covariance of $W$ by its cell average, with spectral density $s=\EE\omega^2|\hat w|^2$, so that $s(0)=s_0$ and
\[
c_\eps=(2\pi)^{-2}\int\frac{|\xi|_B^2}{q(\xi)^2}(1-\chi_0(\xi))^2\lambda_\eps^2s(\eps\xi)\,d\xi=(2\pi)^{-2}\lambda_\eps^2\int\frac{(1-\chi_0(\xi))^2s(\eps\xi)}{q(\xi)}\,d\xi ,
\]
since $|\xi|_B^2=q(\xi)$. In polar coordinates $\xi=te$, $|e|=1$, this is $(2\pi)^{-2}\lambda_\eps^2\int_{S^1}q(e)^{-1}\int_0^\infty(1-\chi_0(t))^2s(\eps te)\,\frac{dt}t\,de$. The function $s=\EE\omega^2|\hat w|^2$ is Schwartz. Hence the inner integral is $s_0\log(1/\eps)+O(1)$, uniformly in $e$, and $c_\eps=\lambda_\eps^2\log(1/\eps)\,s_0c_q+O(\lambda_\eps^2)=\ell s_0c_q+O(1/\log(1/\eps))$.
\end{proof}

\paragraph{The exponential transform.} Following \cite{HairerLabbe15}, who treated the scalar case $-\Delta$, we remove the singular part of the potential by conjugation. With $\theta$ and $H_\theta$ as in Theorem \ref{thm:det}, set
$\tilde\rho:=\theta\rho$, $Y:=\tilde\rho Y_0$ and $X:=A^{-1}Y$. The matrices $X(x)$ commute with one another and with $A$, since they are scalar multiples of $A^{-1}$.

\begin{lemma}[Exponential transform]\label{lem:conj}
\textup{(i)} For smooth $f$, $e^{-X}(H_\theta-z)e^{X}f=(P+\rho^2M_\eps-z)f+\cA f$, where $\cA=\sum_{i=1}^{6}\cA_i$, with
\begin{align*}
\cA_1&=-2\tilde\rho(B\nabla Y_0)\cdot\nabla-2Y_0(B\nabla\tilde\rho)\cdot\nabla,&
\cA_2&=\tilde\rho\tilde V_{\rm low}-2(B\nabla\tilde\rho)\cdot\nabla Y_0-(\nabla\cdot B\nabla\tilde\rho)Y_0,\\
\cA_3&=-A^{-1}\tilde\rho^2\cW,&
\cA_4&=-A^{-1}\bigl(2\tilde\rho Y_0(B\nabla\tilde\rho)\cdot\nabla Y_0+Y_0^2|\nabla\tilde\rho|_B^2\bigr),\\
\cA_5&=\sum_j\bigl(e^{-X}c_je^{X}-c_j\bigr)D_j+e^{-X}c_j(D_jX)e^{X},&
\cA_6&=e^{-X}P_0e^{X}-P_0 .
\end{align*}
\textup{(ii)} If $\delta_\theta\le1$, then $\|e^{\pm X}-1\|_{L^\infty}\le C_0'\delta_\theta$, $\|e^{\pm X}\|_{C^{1-\kappa}}\le C_0'$, and, for $u\in H^{2-2\kappa}$,
\[
\|\cA u\|_{H^{-2\kappa}}\le C_0'\,\delta_\theta\,\|u\|_{H^{2-2\kappa}},
\]
with $C_0'$ depending on the bounds on the derivatives of $\theta$.
\end{lemma}
\begin{proof}
\emph{(i)} Since the values of $X$ commute,
$
\partial_i\partial_j(e^Xf)=e^X\bigl(\partial_i\partial_jf+\partial_iX\partial_jf+\partial_jX\partial_if+(\partial_i\partial_jX)f+\partial_iX\partial_jXf\bigr)
$.
With $q(D)=-\sum B_{ij}\partial_i\partial_j$, this gives
\[
e^{-X}q(D)e^Xf=q(D)f-2(B\nabla X)\cdot\nabla f+(q(D)X)f-(\nabla X\cdot B\nabla X)f .
\]
Multiply by $A$, and use $X=A^{-1}Y$, the commutation of $A$ with $X$, and that $Y$ is scalar. The result is
\[
Aq(D)f-2(B\nabla Y)\cdot\nabla f+(q(D)Y)f-A^{-1}|\nabla Y|^2_Bf .
\]
Now $q(D)Y=\tilde\rho\,q(D)Y_0+[q(D),\tilde\rho]Y_0=-\tilde\rho\tilde V_{\rm high}-2(B\nabla\tilde\rho)\cdot\nabla Y_0-(\nabla\cdot B\nabla\tilde\rho)Y_0$. The term $-\tilde\rho\tilde V_{\rm high}$ cancels the singular part of the potential $\tilde\rho\tilde V=\tilde\rho\tilde V_{\rm high}+\tilde\rho\tilde V_{\rm low}$. Expand
\[
|\nabla Y|_B^2=\tilde\rho^2(c_\eps+\cW)+2\tilde\rho Y_0(B\nabla\tilde\rho)\cdot\nabla Y_0+Y_0^2|\nabla\tilde\rho|_B^2 ,
\]
and note that $-A^{-1}\tilde\rho^2c_\eps=\theta^2\rho^2M_\eps$. Together with the term $(1-\theta^2)\rho^2M_\eps$ of $H_\theta$, which commutes with $X$ since $M_\eps=-c_\eps A^{-1}$, this gives $\rho^2M_\eps$. Finally, for the first-order part, $e^{-X}c_jD_j(e^Xf)=e^{-X}c_je^XD_jf+e^{-X}c_j(D_jX)e^Xf$.

\emph{(ii)} Every coefficient of $\cA$ is a sum of terms $\chi f$ or $\chi fh$, where $f\in\{\nabla Y_0,\cW,Y_0,\tilde V_{\rm low}\}$, $\chi$ is a product of at most two factors among $\tilde\rho$, $\nabla\tilde\rho$, $\nabla^2\tilde\rho$, with $\nabla^2\tilde\rho$ occurring only when $f\in\{Y_0,\tilde V_{\rm low}\}$, and $h$ is $1$, $Y_0$, or a smooth function of $X$. We first show that such coefficients are controlled by $\delta_\theta$. Let $\chi'$ be a smooth function equal to $1$ on $\supp\chi$ and supported in its $1$-neighbourhood, with bounded derivatives. We use \eqref{eq:locprinc} with $s\in\{-\kappa,1-\kappa\}$. For $u=\chi f$, $\psi_Q\chi f$ vanishes unless $Q^*$ meets $\supp\chi$, and then $\psi_Q\chi f=\sum_{Q'}(\psi_Q\chi)(\psi_{Q'}f)$, a sum over the at most $25$ unit boxes $Q'$ with $Q'^*\cap Q^*\ne\emptyset$. Since $\supp\chi\subset\supp\theta$, each such $Q'$ satisfies $\dist(Q',\supp\theta)\le4\sqrt2$. The same argument with $\chi'$ in place of $\chi$, whose support lies in the $1$-neighbourhood of $\supp\theta$, gives $\dist(Q',\supp\theta)\le4\sqrt2+1\le7$. By Lemma~\ref{lem:prod} (second inequality, with $C^1\subset C^{1-\kappa}$), $\|(\psi_Q\chi)(\psi_{Q'}f)\|_{C^{-\kappa}}\le C\|\chi\|_{C^1}\|\psi_{Q'}f\|_{C^{-\kappa}}$, and $\|\psi_{Q'}f\|_{C^{-\kappa}}\le C\|f\|_{L^\infty(Q'^*)}$ when $f\in\{Y_0,\tilde V_{\rm low}\}$. Hence $\|\chi f\|_{C^{-\kappa}}\le C\delta_\theta$, and $\|\chi f\|_{L^\infty}\le C\delta_\theta$ for $f\in\{Y_0,\tilde V_{\rm low}\}$. The same argument with $s=1-\kappa$, together with $\|u\|_{C^{1-\kappa}}\le C(\|u\|_{L^\infty}+\|\nabla u\|_{C^{-\kappa}})$ \cite[Ch.~2]{BCD11}, gives $\|\chi'Y_0\|_{C^{1-\kappa}}\le C\delta_\theta$ and $\|e^{\pm X}\|_{C^{1-\kappa}}+\delta_\theta^{-1}\|e^{\pm X}-1\|_{L^\infty}\le Ce^{C\delta_\theta}$. This gives the first two bounds of (ii).

We now bound the terms one by one, using Lemma~\ref{lem:prod} with $g=\nabla u$ or $g=u$ in $H^{1-2\kappa}$, and write $\delta$ for $\delta_\theta\le1$, so that the factors $e^{C\delta}$ and $1+\delta$ below are bounded.
For $\cA_1$, the coefficient $\tilde\rho B\nabla Y_0$ lies in $C^{-\kappa}$ with norm $\le C\delta$, and $Y_0B\nabla\tilde\rho$ lies in $L^\infty$ with norm $\le C\delta$. For $\cA_2$, $\tilde\rho\tilde V_{\rm low}$ and $(\nabla\cdot B\nabla\tilde\rho)Y_0$ are bounded by $C\delta$ in $L^\infty$, and $(B\nabla\tilde\rho)\cdot\nabla Y_0$ by $C\delta$ in $C^{-\kappa}$. The coefficient of $\cA_3$ is bounded by $C\delta$ in $C^{-\kappa}$. In $\cA_4$, $Y_0\nabla Y_0$ is a product of a $C^{1-\kappa}$ and a $C^{-\kappa}$ function, hence bounded by $C\delta^2$ in $C^{-\kappa}$, while $Y_0^2|\nabla\tilde\rho|^2_B$ is bounded by $C\delta^2$ in $L^\infty$. The first part of $\cA_5$ has an $L^\infty$ coefficient bounded by $Ce^{C\delta}\delta$ acting on $D_ju\in H^{1-2\kappa}$. In the second part, $D_jX=A^{-1}D_jY$, so that the coefficient $(D_jY)\,e^{-X}c_jA^{-1}e^X$ is a $C^{-\kappa}$ scalar times a $C^{1-\kappa}$ matrix, of norm $\le Ce^{C\delta}(1+\delta)\delta$ in $C^{-\kappa}$. Finally, $\cA_6$ is bounded by $Ce^{C\delta}\delta$ in $L^\infty$.
It remains to note that multiplication by an $L^\infty$ function maps $H^{1-2\kappa}\subset L^2$ into $L^2\subset H^{-2\kappa}$, and $\|\nabla u\|_{H^{1-2\kappa}}\le\|u\|_{H^{2-2\kappa}}$.
\end{proof}

\begin{proof}[Proof of Theorem \ref{thm:det}]
The conjugation of Lemma~\ref{lem:conj} removes the singular term $\tilde\rho\tilde V_{\rm high}$ exactly, since $q(D)Y_0=-\tilde V_{\rm high}$, and the mean of the corrector energy produces the shift $\rho^2M_\eps$. What remains, $\cA$, is small in negative H\"older--Zygmund norms. Let $R_M(z)=(P+\rho^2M_\eps-z)^{-1}$. By the elliptic estimate for $P+\rho^2M_\eps$, whose potential is smooth and bounded uniformly in $\eps$ (Lemma~\ref{lem:ceps}), and by duality and interpolation, $\|R_M(z)\|_{H^{-2\kappa}\to H^{2-2\kappa}}\le C(1+|z|)^2(1+|\Im z|^{-1})$, uniformly in $\eps$. We first take $|\Re z|\le R$ and $|\Im z|=1$, so that $|z|\le R+1$, and let $\mathcal K=\cA R_M(z)$ on $H^{-2\kappa}$. Choose $\delta_*\le1$ with $C_1'\delta_*\le\frac12$, where $C_1'$, depending on $R$, is such that, by this bound and Lemma~\ref{lem:conj}(ii), $\|\mathcal K\|\le C_1'\delta_\theta$ whenever $\delta_\theta\le1$. Then $\tilde R:=R_M(1+\mathcal K)^{-1}:H^{-2\kappa}\to H^{2-2\kappa}$ is bounded and satisfies $(P+\rho^2M_\eps+\cA-z)\tilde R=I$.

For $f\in L^2$, put $v=\tilde Re^{-X}f\in H^{2-2\kappa}$ and $u=e^Xv$. For fixed $\eps$, $X$ is smooth with bounded derivatives, and Lemma~\ref{lem:conj}(i) extends by density to $v\in H^{2-2\kappa}$, in the sense of distributions. Hence $(H_\theta-z)u=e^X(P+\rho^2M_\eps+\cA-z)v=f$, with $u\in L^2$. By elliptic regularity ($H_\theta$ is $P$ plus a bounded smooth potential), $u\in H^2=\mathrm{dom}(H_\theta)$, so $u=(H_\theta-z)^{-1}f$. Therefore
\[
(H_\theta-z)^{-1}-R_M=(e^X-1)\tilde Re^{-X}+\tilde R(e^{-X}-1)+(\tilde R-R_M).
\]
The first two terms are $O(\delta_\theta)$ on $L^2$, since $\|e^{\pm X}-1\|_\infty\le C\delta_\theta$ by Lemma~\ref{lem:conj}(ii) and $\|\tilde R\|_{L^2\to L^2}\le C$. The third equals $-R_M\mathcal K(1+\mathcal K)^{-1}$, of norm $\le C\delta_\theta$ from $L^2\subset H^{-2\kappa}$ to $H^{2-2\kappa}\subset L^2$. Finally, the resolvent identity and Lemma~\ref{lem:ceps} give $\|R_M-(P+\rho^2M_{\rm crit}-z)^{-1}\|\le C/\log(1/\eps)$. This proves the bound for $|\Im z|=1$. For general $z$ with $|\Re z|\le R$ and $\Im z\ne0$, let $z_0=\Re z+i\,\sgn(\Im z)$. In \eqref{eq:shiftz}, written for $H_\theta$ and $P+\rho^2M_{\rm crit}$, the outer factors have norm $\le\sup_{\lambda\in\R}|\lambda-z_0|/|\lambda-z|\le\max(1,|\Im z|^{-1})$. This gives the factor $1+|\Im z|^{-2}$.

For the $H^1$ bound, $(H_\theta-z_0)^{-1}=e^X\tilde Re^{-X}$ with $\tilde R$ built at $z_0$, where $|\Im z_0|=1$ and $|\Re z_0|\le R$ as above, $\|\tilde R\|_{L^2\to H^{2-2\kappa}}\le C$ and $H^{2-2\kappa}\subset H^1$, and
\[
\|\nabla X\|_\infty\le C\bigl(\|\nabla\tilde\rho\|_\infty\|Y_0\|_{L^\infty(\supp\tilde\rho)}+\|\nabla Y_0\|_{L^\infty(\supp\tilde\rho)}\bigr)\le C(\delta_\theta+N_\theta) .
\]
The case $\rho\equiv1$ is identical.
\end{proof}

\paragraph{Bad boxes.} The bad unit boxes of block (vii) of section \ref{sec:setting} are defined in Step 1 of the proof of Theorem \ref{thm:mainA} through truncated versions of the local size $\delta_Q$ and of $\|\nabla Y_0\|_{L^\infty(Q^*)}$. By the next lemma, proved in appendix \ref{app:proba}, these have finite range and are rarely large.
\begin{lemma}\label{lem:probboxA}
There are $c,C_1>0$, depending only on $(q,w,\kappa_+,\EE\omega^2,\ell)$ and on the fixed choices of $\kappa$, $\psi$ and $\chi_0$, with the following property. Let $\delta_*\in(0,1]$, $\nu_r>0$ and $r=\eps^{-\nu_r}$, and set $p_0:=\min\big(c\,\delta_*^2/\ell,1\big)$. There are $C,\eps_1>0$, depending moreover on $\delta_*$ and $\nu_r$, and, for every unit box $Q$, random variables $\delta_Q^{\rm tr}$ and $N_Q^{\rm tr}$, measurable with respect to $\{\omega_m:|\eps m-c_Q|\le r\}$, such that for $\eps\le\eps_1$:
\textup{(i)} $|\delta_Q-\delta_Q^{\rm tr}|+\big|\|\nabla Y_0\|_{L^\infty(Q^*)}-N_Q^{\rm tr}\big|\le\eps^{10}$ deterministically; and
\textup{(ii)}
\[
  \PP\big(\delta_Q^{\rm tr}>\delta_*/2\ \mbox{ or }\ N_Q^{\rm tr}>\tfrac12C_1\sqrt{\log(1/\eps)}\big)\le C\eps^{p_0}.
\]
\end{lemma}
\begin{proof}[Proof of Theorem \ref{thm:mainA}]
\emph{Step 1: parameters, bad boxes and cuts.} We apply Theorem \ref{thm:det} with $R:=g+1$ and with $\theta$ ranging over the cutoffs $\theta_{\cU,r}$ of \eqref{eq:Hg}, whose derivatives are bounded independently of $\cU$ and $r$, and let $\delta_*\le1$ be its threshold, decreased if necessary for Steps 2 and 3. Thus $\delta_*$ depends only on $(A,B,c_j,P_0,\rho,g,\varphi)$ and on the law of the couplings, not on $\eps$ or on the configuration. Lemma \ref{lem:probboxA} with this $\delta_*$ fixes $p_0$ and $C_1$, independently of $\nu_r$. We then fix $\nu_r,\nu_\dd$ as in \eqref{eq:exponents} and $n>3/\nu_\dd$. A unit box $Q$ is {\em bad} if $\delta_Q^{\rm tr}>\delta_*/2$ or $N_Q^{\rm tr}>\frac12C_1\sqrt{\log(1/\eps)}$, with the truncated variables of Lemma \ref{lem:probboxA}, and $\cU_\omega$ is the union of the bad unit boxes. By that lemma, $Q$ meets $\cU_\omega$ only if one of the at most $9$ unit boxes adjacent to or equal to $Q$ is bad, which has probability at most $C\eps^{p_0}$ and depends only on the couplings with $|\eps m-c_Q|\le r+2\le2r$, and every good unit box satisfies $\delta_Q\le\delta_*$ and $\|\nabla Y_0\|_{L^\infty(Q^*)}\le C_1\sqrt{\log(1/\eps)}$. Lemma \ref{lem:cuts} then gives admissible cuts a.s.\ for $\eps\le\eps_0$.

\emph{Step 2: reference operator and bulk gap.} We set $\theta=\theta_{\cU_\omega,r}$ and build $H_{\circ}$, $H_{\circ,B}$ by \eqref{eq:Hg} with $M_\star=M_\eps$, so that $H_\circ$ is the operator $H_\theta$ of Theorem \ref{thm:det}. Since $\theta=0$ on $(\cU_\omega)_{2r}$ and $2r\ge8>7$, every unit box $Q$ with $\dist(Q,\supp\theta)\le7$ is good, so that $\delta_\theta\le\delta_*$ and $N_\theta\le C_1\sqrt{\log(1/\eps)}$. By Theorem \ref{thm:det}, for $|\Re z|\le R$ and $\Im z\ne0$,
\[
\|(H_{\circ}-z)^{-1}-(P+\rho^2M_{\rm crit}-z)^{-1}\|+\|(H_{\circ,B}-z)^{-1}-(P+M_{\rm crit}-z)^{-1}\|\le C\bigl(1+|\Im z|^{-2}\bigr)\Bigl(\delta_*+\frac1{\log(1/\eps)}\Bigr).
\]
Since $P+M_{\rm crit}$ is gapped on $(-g,g)$, Lemma~\ref{lem:speccomp} at $z=E+i$, with $|E|<g/2$ and $d=g/2$, shows that $H_{\circ,B}$ is gapped on $(-g/2,g/2)$ once $2C(\delta_*+1/\log(1/\eps))<1-(1+g^2/4)^{-1/2}$, which holds for $\delta_*$ and $\eps$ small.

\emph{Step 3: index and robustness.} We follow Steps 4--5 of the proof of Theorem~\ref{thm:mainB}. Since $M_{\rm crit}$ is constant, no deformation is needed: the bound of Step 2, which holds on $\supp\dbar\tilde S\subset\{|\Re z|\le g/2\}$ and is integrable against $|\dbar\tilde S(z)|\le C|\Im z|^2$, and the Helffer--Sj\"ostrand formula give $\|U(H_{\circ})-U(P+\rho^2M_{\rm crit})\|\le C(\delta_*+1/\log(1/\eps))<1$, so that $\ind T(H_{\circ})=\BDI(P+M_{\rm crit},P)$ by Lemma \ref{lem:fred} and Proposition~\ref{thm:QB}. Condition [R2] holds as there, by Step 2. For the second-order operator $H_\circ$, [R1] follows from the $H^1$ bound of Theorem \ref{thm:det}, with $N_\theta\le C_1\sqrt{\log(1/\eps)}$, and Lemma~\ref{lem:CT}(ii) with $\eta=1$: $C_0=2e^4$, $R_0=g+1$ and $a=c/\sqrt{\log(1/\eps)}$. Now $K_V=C\lambda_\eps\eps^{-1}$, so that $\vartheta_1\le C_n\eps^{-3}\bigl(\eps^{\nu_\dd}\sqrt{\log(1/\eps)}\bigr)^{n}<\frac1{16}$ for $n\nu_\dd>3$ and $\eps$ small, and Theorem~\ref{thm:robust} gives $\ind\tilde T(H^\eps)=\ind T(H_{\circ})=\BDI(P+M_{\rm crit},P)$.
\end{proof}

\appendix
\section{Probabilistic estimates}\label{app:proba}

This appendix proves Lemma \ref{lem:probboxA}. Its main ingredient is the following bound on the probability that a unit box has a large local size \eqref{eq:deltaQ}.
\begin{lemma}[Good configurations]\label{lem:prob}
There are $c,C>0$, and for every $t\in(0,1]$ there is $\eps_t>0$, such that for every unit box $Q$ and every $\eps\le\eps_t$, $\PP(\delta_Q\geq t)\le C\exp(-ct^2/\lambda_\eps^2)=C\,\eps^{ct^2/\ell}$.
\end{lemma}
The tools are standard: Hoeffding's and the Hanson--Wright inequalities for linear and quadratic functionals of the couplings, and Bernstein's inequality for band-limited functions. For Gaussian noise, bounds of this type for the renormalized square of $\nabla Y_0$ in $C^{-\kappa}$ are classical \cite{HairerLabbe15}. Here the couplings are bounded and i.i.d., so that Gaussian moment bounds are replaced by Hoeffding and Hanson--Wright inequalities, and we need tails that are exponential in $\lambda_\eps^{-2}$, uniformly in $\eps$, for the percolation argument of section \ref{sec:pseudo}.

Throughout this appendix, $t\in(0,1]$ is the threshold of Lemma~\ref{lem:prob} and $\eps\le\eps_t$, where $\eps_t$ is small enough that $\lambda_\eps\le1$ and that, in each application of the second part of Lemma~\ref{lem:sampling} below, the parameter $a$ is $\ge1$. Constants $C,c$ depend only on $(q,w,\kappa_+,\EE\omega^2,\ell,\kappa,\psi,\chi_0)$, and not on $t$, which appears explicitly in all bounds. Only $\eps_t$ depends on $t$.

Suprema of band-limited functions are controlled by their values on a lattice, by Bernstein's inequality \cite[Lemma 2.1]{BCD11}. We use the following weighted version.
\begin{lemma}\label{lem:sampling}
There are universal constants $\delta_0,C$ such that the following holds. Let $R\ge1$, $c\in\R^2$, $\iota\in\{-1,1\}$, and let $u\in L^\infty(\R^2)$ have $\supp\hat u\subset B(0,R)$ and $\sup_z|u(z)|\avg{z-c}^{-\iota}<\infty$. Then
\[
\sup_{z\in\R^2}|u(z)|\avg{z-c}^{-\iota}\le4\sup_{y\in(\delta_0/R)\Z^2}|u(y)|\avg{y-c}^{-\iota} .
\]
Moreover, $\sum_{y\in h\Z^2}e^{-a\avg{y-c}}\le Ch^{-2}e^{-a/2}$ for $h\in(0,1]$ and $a\ge1$.
\end{lemma}
\begin{proof}
Let $G$ be the left-hand side. Let $\Theta\in\mathcal S(\R^2)$ with $\hat\Theta=1$ on $B(0,1)$, and $\Theta_R=R^2\Theta(R\cdot)$. Then $u=\Theta_R*u$, so, using $\avg{w-c}^\iota\le\avg{z-c}^\iota(1+|z-w|)\le\avg{z-c}^\iota(1+R|z-w|)$ for $\iota=\pm1$,
\[
|\nabla u(z)|\le\int R^3|\nabla\Theta(R(z-w))|\,|u(w)|\,dw\le C_\Theta R\,G\,\avg{z-c}^\iota ,\qquad C_\Theta:=\int|\nabla\Theta|(1+|v|)\,dv .
\]
Let $\delta_0=\min(1,1/(4C_\Theta))$. Every $z$ is within $\delta_0/R$ of a lattice point $y$, and $\avg{\cdot-c}^\iota\le2\avg{z-c}^\iota$ on the segment $[y,z]$, while $\avg{y-c}^\iota\le2\avg{z-c}^\iota$. Hence $|u(z)|\le|u(y)|+2C_\Theta\delta_0G\avg{z-c}^\iota$, and dividing by $\avg{z-c}^\iota$, $G\le2\sup_y|u(y)|\avg{y-c}^{-\iota}+2C_\Theta\delta_0G\le2\sup_y|u(y)|\avg{y-c}^{-\iota}+G/2$. For the sum, $e^{-a\avg{y-c}}\le e^{-a/2}e^{-|y-c|/2}$ since $a\ge1$ and $\avg{y-c}\ge\max(1,|y-c|)$.
\end{proof}

\paragraph{Representation of the fields.} For finitely supported $u:\Z^2\to\R$, let $\mathsf Su:=\eps^{-1}\sum_mu_m\,w(\cdot/\eps-m)$. Then $\tilde V=\lambda_\eps \mathsf S\omega$, a locally finite sum. With the Fourier multipliers
\[
\tau_0(\xi)=-\frac{1-\chi_0(\xi)}{q(\xi)},\qquad\tau_1(\xi)=i\xi\,\tau_0(\xi),\qquad\tau_{\rm low}=\chi_0 ,
\]
we have $Y_0=\lambda_\eps\tau_0(D)\mathsf S\omega$, $Z:=\nabla Y_0=\lambda_\eps\tau_1(D)\mathsf S\omega$ and $\tilde V_{\rm low}=\lambda_\eps\chi_0(D)\mathsf S\omega$. These multipliers map real functions to real functions, $|\tau_0|\le C\avg\xi^{-2}$ and $|\tau_1|\le C\avg\xi^{-1}$, and $\tau_0,\tau_1$ vanish for $|\xi|\le c_\chi$, where $\chi_0=1$ on $B(0,c_\chi)$. Their kernels $K_0,K_1,K_{\rm low}$ are smooth away from $0$ and satisfy $|K(z)|\le C_n|z|^{-n}$ for $|z|\ge1$ \cite[Ch.~VI, \S4]{Stein93}. Near $0$, $|K_0(z)|\le C\log(2/|z|)$, $|K_1(z)|\le C|z|^{-1}$, and $K_{\rm low}$ is Schwartz. In particular $K_0,K_1\in L^1$. Since $\sum_m|w(\cdot/\eps-m)|\le C$ and $|\omega_m|\le\kappa_+$, we have the deterministic bounds $\|\mathsf S\omega\|_{L^\infty}\le C\eps^{-1}$ and $\|Z\|_{L^\infty}\le C\eps^{-1}$.

The synthesis operator $\mathsf S$ satisfies
\begin{equation}\label{eq:synth}
\|\mathsf Su\|_{L^2}\le C\|u\|_{\ell^2},\qquad \mathsf S=(1+\eps^2|D|^2)^{-n}\mathsf S_n,\quad\|\mathsf S_n\|\le C_n,
\end{equation}
where $\mathsf S_n$ is the synthesis operator built on $w_n=(1-\Delta)^nw\in C_c^\infty$ in place of $w$. Indeed, every point lies in at most $C$ of the supports of $w(\cdot/\eps-m)$, so that $|\mathsf Su|^2\le C\eps^{-2}\sum_m|u_m|^2|w(\cdot/\eps-m)|^2$ pointwise, and integration gives the first bound. The factorization follows from $(1+\eps^2|D|^2)^n[w(\cdot/\eps-m)]=w_n(\cdot/\eps-m)$, and the first bound applies to $w_n$.

\begin{lemma}[Linear and quadratic functionals]\label{lem:conc}
\textup{(a)} \emph{Linear functionals.} Let $\tau$ be one of the multipliers $\tau_0,\tau_1,\tau_{\rm low}$, or a product of one of them with a bounded multiplier, and let $\Lambda$ be a real finite measure with $\int|\tau|^2(1+\eps^2|\xi|^2)^{-2n}|\hat\Lambda|^2d\xi<\infty$ for some $n\ge0$. Then $\ip{\Lambda}{\lambda_\eps\tau(D)\mathsf S\omega}=\sum_ma_m\omega_m$, with an absolutely convergent series and
\[
\sum_ma_m^2\le C_n\lambda_\eps^2\int_{\R^2}|\tau(\xi)|^2(1+\eps^2|\xi|^2)^{-2n}|\hat\Lambda(\xi)|^2\,d\xi .
\]
We call any $\sigma^2\ge\sum_ma_m^2$ a {\em variance proxy} of the functional. The right-hand side above is one. Since the couplings are independent, centered and bounded by $\kappa_+$, every variance proxy satisfies $\PP(|\ip{\Lambda}{\lambda_\eps\tau(D)\mathsf S\omega}|\ge s)\le2\exp(-s^2/(2\kappa_+^2\sigma^2))$ for $s>0$.

\noindent\textup{(b)} \emph{Quadratic functionals.} Let $\Lambda$ be real, bounded and compactly supported. Then $\ip{\Lambda}{|Z|_B^2}=\ip{\omega}{\mathsf A\omega}$, where
\[
\mathsf A=\lambda_\eps^2\mathsf S^*\mathcal X\mathsf S,\qquad\mathcal X=\sum_{a,b}B_{ab}\,\tau_{1,a}(D)^*M_\Lambda\,\tau_{1,b}(D),
\]
$M_\Lambda$ is multiplication by $\Lambda$, and $\tau_{1,a}$ are the components of $\tau_1$. The operator $\mathsf A$ is real, symmetric and trace-class, $\|\mathsf A\|_{HS}\le C\lambda_\eps^2\|\mathcal X\|_{HS}$ and $\|\mathsf A\|\le C\lambda_\eps^2\|\mathcal X\|$, and for every $s>0$
\[
\PP\bigl(|\ip\omega{\mathsf A\omega}-\EE\ip\omega{\mathsf A\omega}|\ge s\bigr)\le2\exp\bigl(-\mathcal E_{\rm HW}(\mathsf A,s)\bigr),\qquad \mathcal E_{\rm HW}(\mathsf A,s):=c\min\Bigl(\frac{s^2}{\kappa_+^4\|\mathsf A\|_{HS}^2},\frac{s}{\kappa_+^2\|\mathsf A\|}\Bigr).
\]
We call $\mathcal E_{\rm HW}$ the {\em Hanson--Wright exponent}.
\end{lemma}
\begin{proof}
(a) We have $a=\lambda_\eps \mathsf S^*\bar\tau(D)\Lambda=\lambda_\eps \mathsf S_n^*(1+\eps^2|D|^2)^{-n}\bar\tau(D)\Lambda$, and the bound follows from \eqref{eq:synth} and Plancherel. Absolute convergence follows from the decay of the kernels. Hoeffding's inequality \cite[Theorem~2.8]{BLM13} for the finitely many terms with $|m|\le M$, followed by $M\to\infty$, gives the tail bound.

(b) The identity is a rewriting of $\int\Lambda\,Z\cdot BZ$. Writing $M_\Lambda=M_{|\Lambda|^{1/2}}M_{\sgn\Lambda}M_{|\Lambda|^{1/2}}$ and using \eqref{eq:synth} with $n=1$, the trace norm of $\mathsf A$ is at most $C\lambda_\eps^2\|\Lambda\|_{L^1}\int|\tau_1|^2(1+\eps^2|\xi|^2)^{-2}d\xi<\infty$. The norm bounds follow from \eqref{eq:synth}. The tail bound is the Hanson--Wright inequality \cite[Theorem~1.1]{RV13} for the independent centered variables $\omega_m$, whose sub-Gaussian norms are $\le C\kappa_+$, applied to the truncations $\mathsf A_M=\one_{|m|\le M}\mathsf A\one_{|m|\le M}$. Since $\mathsf A_M\to \mathsf A$ in trace norm, the centered forms converge in $L^2(\PP)$, their variance being $\le C\kappa_+^4\|\mathsf A-\mathsf A_M\|_{HS}^2$, and the means converge since $\sum_m|(\mathsf A-\mathsf A_M)_{mm}|\le\|\mathsf A-\mathsf A_M\|_{\rm tr}$. The bound passes to the limit.
\end{proof}

\begin{lemma}\label{lem:blocks}
For a unit box $Q$ with center $c$, $j\ge-1$ and $y\in\R^2$, let $\Lambda_{j,y}(x)=\psi_Q(x)k_j(y-x)$, so that $\Delta_j(\psi_Qf)(y)=\ip{\Lambda_{j,y}}{f}$. Then, uniformly in $Q$, $j$ and $y$, \textup{(a)} $\int|\tau_1|^2|\hat\Lambda_{j,y}|^2\,d\xi\le C\avg{y-c}^{-6}$; and \textup{(b)} the operator $\mathsf A_{j,y}$ of Lemma~\ref{lem:conc}(b) with $\Lambda=\Lambda_{j,y}$ satisfies
\[
\|\mathsf A_{j,y}\|\le C\lambda_\eps^2\sqrt{2+j}\,\avg{y-c}^{-3},\qquad\|\mathsf A_{j,y}\|_{HS}\le C\lambda_\eps\sqrt{2+j}\,\avg{y-c}^{-3}.
\]
\end{lemma}
\begin{proof}
\emph{Near field, $|y-c|<6$.} Here $\avg{y-c}\le7$, and decay in $y$ is not needed. The Fourier transform of $x\mapsto k_j(y-x)$ has modulus $\phi_j$, and the transform of a product is $(2\pi)^{-2}$ times the convolution of the transforms, so that $|\hat\Lambda_{j,y}|\le\Psi_j:=(2\pi)^{-2}\phi_j*|\hat\psi|$. If $\dist(\xi,\supp\phi_j)=d$, then $\Psi_j(\xi)\le(2\pi)^{-2}\int_{|\zeta'|\ge d}|\hat\psi(\zeta')|\,d\zeta'\le C_n(1+d)^{-n}$, since $0\le\phi_j\le1$ and $\hat\psi$ is Schwartz. For $j\ge0$, the set $\mathcal N_j=\{\xi:\dist(\xi,\supp\phi_j)\le2^{j-2}\}$ is contained in $\{2^{j-1}\le|\xi|\le3\cdot2^j\}$, of measure $\le C2^{2j}$, where $\Psi_j\le C$ and $\avg\xi^{-2}\le C2^{-2j}$. Outside $\mathcal N_j$, $\Psi_j(\xi)\le C_n(2^j+|\xi|)^{-n}$. Hence $\int\avg\xi^{-2}\Psi_j^2\,d\xi\le C$, also for $j=-1$, where $\Psi_{-1}\le C_n\avg\xi^{-n}$. This gives (a), since $|\tau_1|\le C\avg\xi^{-1}$. For (b), the kernel of $\tau_{1,a}(D)^*M_\Lambda\tau_{1,b}(D)$ in the Fourier variables is $(2\pi)^{-2}\overline{\tau_{1,a}(\xi)}\hat\Lambda(\xi-\xi')\tau_{1,b}(\xi')$, so that, with the same splitting,
\[
\|\mathcal X_{j,y}\|_{HS}^2\le C\int\!\!\int\avg\xi^{-2}\avg{\xi'}^{-2}\Psi_j(\xi-\xi')^2\,d\xi\,d\xi'\le C\int\Psi_j(\eta)^2\avg\eta^{-2}\log(2+|\eta|)\,d\eta\le C(2+j),
\]
using $\int\avg\xi^{-2}\avg{\xi-\eta}^{-2}d\xi\le C\avg\eta^{-2}\log(2+|\eta|)$. Lemma~\ref{lem:conc}(b) gives $\|\mathsf A_{j,y}\|\le\|\mathsf A_{j,y}\|_{HS}\le C\lambda_\eps^2\sqrt{2+j}$, and $\lambda_\eps^2\le\lambda_\eps$.

\emph{Far field, $|y-c|\ge6$.} For $x\in\supp\psi_Q\subset B(c,3)$, $|y-x|\ge|y-c|/2\ge\avg{y-c}/3$, so that $|k_j(z)|\le C2^{2j}(1+2^j|z|)^{-8}$ gives $\|\Lambda_{j,y}\|_{L^1}+\|\Lambda_{j,y}\|_{L^\infty}\le C\avg{y-c}^{-8}$. Since $|\tau_1|\le C$, Plancherel gives $\int|\tau_1|^2|\hat\Lambda_{j,y}|^2\le C\|\Lambda_{j,y}\|_{L^2}^2\le C\|\Lambda_{j,y}\|_{L^1}\|\Lambda_{j,y}\|_{L^\infty}$, hence (a). For (b), $\|\mathsf A_{j,y}\|\le C\lambda_\eps^2\|\mathcal X_{j,y}\|\le C\lambda_\eps^2\|\Lambda_{j,y}\|_{L^\infty}$. For the Hilbert--Schmidt norm, write $M_\Lambda=M_{|\Lambda|^{1/2}}M_{\sgn\Lambda\,|\Lambda|^{1/2}}$ and $\mathsf S=(1+\eps^2|D|^2)^{-1}\mathsf S_1$ as in \eqref{eq:synth}: the factor $M_{|\Lambda|^{1/2}}\tau_{1,b}(D)(1+\eps^2|D|^2)^{-1}$ has Hilbert--Schmidt norm squared $\le C\|\Lambda\|_{L^1}\int|\tau_1|^2(1+\eps^2|\xi|^2)^{-2}d\xi\le C\|\Lambda\|_{L^1}\log(1/\eps)$, and the remaining factors have norm $\le C\|\Lambda\|_{L^\infty}^{1/2}$. Hence $\|\mathsf A_{j,y}\|_{HS}\le C\lambda_\eps^2(\log(1/\eps)\|\Lambda\|_{L^1}\|\Lambda\|_{L^\infty})^{1/2}\le C\lambda_\eps\sqrt\ell\,\avg{y-c}^{-8}$, since $\lambda_\eps^2\log(1/\eps)=\ell$.
\end{proof}

\begin{proof}[Proof of Lemma \ref{lem:prob}]
Fix $t\in(0,1]$ and a unit box $Q$ with center $c=c_Q$, and note that $\avg{z-c}\le3$ on $Q^*$. Set $\Gamma_j=\delta_02^{-j-2}\Z^2$. The Fourier transform of $\Delta_jf$ is supported in $B(0,2^{j+2})$, so Lemma~\ref{lem:sampling} applies to it with $R=2^{j+2}$. In each step, a union bound over $\Gamma_j$ uses the second part of Lemma~\ref{lem:sampling}, and the sum over $j\ge-1$ uses $\sum_{j\ge-1}2^{2j}e^{-a2^{\gamma_0j}}\le C_{\gamma_0}e^{-a/4}$ for $a\ge1$ and fixed $\gamma_0\in(0,2]$ ($\gamma_0=1$, $2\kappa$, $\kappa/2$ in Steps 1, 3, 4).

\emph{Step 1: $Y_0$.} We write $Y_0=\sum_{j\ge-1}g_j$ with $g_j=\Delta_jY_0$. The series converges uniformly since $Y_0\in C^1_b$ for fixed $\eps$. By Lemma~\ref{lem:conc}(a) with $\Lambda=k_j(y-\cdot)$, whose transform has modulus $\phi_j$, $g_j(y)$ has variance proxy $\sigma_j^2\le C\lambda_\eps^2\int|\tau_0|^2\phi_j^2\le C\lambda_\eps^22^{-2j}$. Let $s_j=c_1t2^{-(j+1)/2}$ with $c_1$ such that $12\sum_js_j\le t/4$. Then
\[
\PP\bigl(\exists y\in\Gamma_j:|g_j(y)|>s_j\avg{y-c}\bigr)\le\sum_{y\in\Gamma_j}2e^{-cs_j^2\avg{y-c}/\sigma_j^2}\le C2^{2j}e^{-ct^22^{j}/\lambda_\eps^2}.
\]
On the complement of these events, Lemma~\ref{lem:sampling} with $\iota=1$ gives $\|g_j\|_{L^\infty(Q^*)}\le12s_j$, hence $\|Y_0\|_{L^\infty(Q^*)}\le t/4$. The total probability of the exceptional events is $\le C\exp(-ct^2/\lambda_\eps^2)$.

\emph{Step 2: $\tilde V_{\rm low}$.} Its Fourier transform is supported in $\supp\chi_0$, and $\tilde V_{\rm low}(y)$ has variance proxy $\le C\lambda_\eps^2\int|\chi_0|^2$. The argument of Step~1, applied to the single band-limited function $\tilde V_{\rm low}$ in place of the blocks $g_j$, with $R\ge1$ such that $\supp\chi_0\subset B(0,R)$, gives $\|\tilde V_{\rm low}\|_{L^\infty(Q^*)}\le t/4$ outside an event of probability $\le C\exp(-ct^2/\lambda_\eps^2)$.

\emph{Step 3: $\psi_QZ$.} Let $g_j=\Delta_j(\psi_QZ)$, so that $g_j(y)=\ip{\Lambda_{j,y}}{Z}$ with $\Lambda_{j,y}$ as in Lemma~\ref{lem:blocks}. Since $Z$ is bounded for fixed $\eps$ and $k_j$ is a Schwartz function, $\avg{z-c}g_j(z)$ is bounded. By Lemmas~\ref{lem:conc}(a) and~\ref{lem:blocks}(a), each component of $g_j(y)$ has variance proxy $\le C\lambda_\eps^2\avg{y-c}^{-6}$. With $s'_j=t2^{\kappa j}/32$,
\[
\PP\bigl(\exists y\in\Gamma_j:|g_j(y)|>s'_j\avg{y-c}^{-1}\bigr)\le\sum_{y\in\Gamma_j}4e^{-ct^22^{2\kappa j}\avg{y-c}^4/\lambda_\eps^2}\le C2^{2j}\exp\bigl(-c\,t^2\,2^{2\kappa j}/\lambda_\eps^2\bigr),
\]
and on the complement Lemma~\ref{lem:sampling} with $\iota=-1$ gives $\|g_j\|_{L^\infty}\le\sup_z\avg{z-c}|g_j(z)|\le4s'_j=\frac t82^{\kappa j}$. Hence, outside an event of probability $\le C\exp(-ct^2/\lambda_\eps^2)$, $\|\psi_QZ\|_{C^{-\kappa}}\le t/8$.

\emph{Step 4: $\psi_Q\cW$.} We write $\cW=\cW_{\rm r}+\cW_{\rm d}$ with $\cW_{\rm r}=|Z|_B^2-\EE|Z|_B^2$ and $\cW_{\rm d}=\EE|Z|_B^2-c_\eps$. The law of $\omega$ is invariant under shifts of $\Z^2$, so $\cW_{\rm d}$ is deterministic and $\eps$-periodic, with mean zero over $Y_\eps$ by definition of $c_\eps$. Each component $Z_a(x)=\sum_ma_m\omega_m$ satisfies $\EE Z_a(x)^2=\EE\omega^2\sum_ma_m^2$, which by Lemma~\ref{lem:conc}(a) with $\Lambda=\delta_x$ and $n=1$ is $\le C\lambda_\eps^2\log(1/\eps)=C\ell$. Since also $c_\eps\le C\ell$, $\|\cW_{\rm d}\|_{L^\infty}\le C\ell$. For a bounded, mean-zero, $1$-periodic function $f$, the Fourier transform of $f(\cdot/\eps)$ is supported in $(2\pi/\eps)\Zm^2\setminus\{0\}$, so that $\Delta_jf(\cdot/\eps)=0$ unless $\frac832^j\ge2\pi/\eps$, and $\|f(\cdot/\eps)\|_{C^{-\kappa}}\le C\eps^\kappa\|f\|_{L^\infty}$. Hence, by Lemma \ref{lem:prod}, $\|\psi_Q\cW_{\rm d}\|_{C^{-\kappa}}\le C\eps^\kappa\ell\le t/16$ for $\eps\le\eps_t$.

For $\cW_{\rm r}$, let $g_j=\Delta_j(\psi_Q\cW_{\rm r})$. As in Step 3, $\avg{z-c}g_j(z)$ is bounded. For $y\in\Gamma_j$, $g_j(y)=\ip\omega{\mathsf A_{j,y}\omega}-\EE\ip\omega{\mathsf A_{j,y}\omega}$, with the bounds on $\|\mathsf A_{j,y}\|$ and $\|\mathsf A_{j,y}\|_{HS}$ of Lemma~\ref{lem:blocks}(b). Let $s''_j=t2^{\kappa j}/64$. By Lemma~\ref{lem:conc}(b) with $s=s''_j\avg{y-c}^{-1}$, the event $|g_j(y)|>s''_j\avg{y-c}^{-1}$ has probability at most $2\exp(-\mathcal E_{\rm HW}(\mathsf A_{j,y},s''_j\avg{y-c}^{-1}))$. Using $t\le1$, $2^{2\kappa j}/(2+j)\ge c\,2^{\kappa j/2}$ and $2^{\kappa j}/\sqrt{2+j}\ge c\,2^{\kappa j/2}$ for $j\ge-1$, we obtain
\[
\mathcal E_{\rm HW}\bigl(\mathsf A_{j,y},s''_j\avg{y-c}^{-1}\bigr)\ge c\min\Bigl(\frac{t^22^{2\kappa j}\avg{y-c}^4}{\lambda_\eps^2(2+j)},\frac{t\,2^{\kappa j}\avg{y-c}^2}{\lambda_\eps^2\sqrt{2+j}}\Bigr)\ge c\avg{y-c}\frac{t^2\,2^{\kappa j/2}}{\lambda_\eps^2}.
\]
The union bound over $\Gamma_j$ and the sum over $j$ give an exceptional probability $\le C\exp(-ct^2/\lambda_\eps^2)$, outside of which Lemma~\ref{lem:sampling} with $\iota=-1$ gives $\|g_j\|_{L^\infty}\le4s''_j$, so that $\|\psi_Q\cW_{\rm r}\|_{C^{-\kappa}}\le t/16$.

\emph{Conclusion.} All the exceptional probabilities above are $\le C\exp(-ct^2/\lambda_\eps^2)$, with $C$ independent of $t$. Outside their union, the four terms of $\delta_Q$ are bounded by $t/4$, $t/4$, $t/8$ and $t/8$, so $\delta_Q<t$.
\end{proof}

\begin{proof}[Proof of Lemma \ref{lem:probboxA}]
\emph{(i) Truncation.} Each of the fields $Y_0$, $Z$ and $\tilde V_{\rm low}$ is of the form $F=\lambda_\eps\tau(D)\mathsf S\omega$, with $\mathsf S$ the synthesis operator \eqref{eq:synth} and $\tau\in\{\tau_0,\tau_1,\tau_{\rm low}\}$. Let $F^{\rm tr}$ be defined with $\omega_m$ replaced by $\omega_m\one_{\{|\eps m-c_Q|\le r\}}$. For $x\in Q^*$ and $|\eps m-c_Q|>r$, the support of $w(\cdot/\eps-m)$ lies at distance $\ge\frac12|\eps m-c_Q|$ from $x$ for $\eps$ small, so the bounds $|K(z)|\le C_n|z|^{-n}$, $|z|\ge1$, on the kernels $K_0,K_1,K_{\rm low}$ give
\[
|F(x)-F^{\rm tr}(x)|\le C\lambda_\eps\kappa_+\eps^{-1}\sum_{|\eps m-c_Q|>r}\eps^2|\eps m-c_Q|^{-n}\le C\eps^{-1}r^{2-n}\le\eps^{12}
\]
for $n$ large, depending on $\nu_r$. Moreover $\||Z|_B^2-|Z^{\rm tr}|_B^2\|_{L^\infty(Q^*)}\le C\eps^{12}\cdot\eps^{-1}=C\eps^{11}$, by the deterministic bound $\|Z\|_{L^\infty}\le C\eps^{-1}$, which also holds for $Z^{\rm tr}$. Define $\delta_Q^{\rm tr}$ and $N_Q^{\rm tr}$ from the truncated fields, with the same constant $c_\eps$. Since $\|\psi_Qf\|_{C^{-\kappa}}\le C\|f\|_{L^\infty(Q^*)}$, the differences are at most $C\eps^{11}\le\eps^{10}$ for $\eps$ small. This proves (i). By construction, the truncated fields, hence $\delta_Q^{\rm tr}$ and $N_Q^{\rm tr}$, depend only on the couplings $\omega_m$ with $|\eps m-c_Q|\le r$.

\emph{(ii) Probability.} Since $\eps^{10}\le\delta_*/4$ and $\eps^{10}\le C_1\sqrt{\log(1/\eps)}/4$ for $\eps$ small, (i) shows that the event in (ii) is contained in $\{\delta_Q>\delta_*/4\}\cup\{\|\nabla Y_0\|_{L^\infty(Q^*)}>C_1\sqrt{\log(1/\eps)}/4\}$. We bound the probabilities of these two events. The event $\{\delta_Q>\delta_*/4\}$ has probability $\le Ce^{-c\delta_*^2/\lambda_\eps^2}=C\eps^{c\delta_*^2/\ell}$ by Lemma \ref{lem:prob} with $t=\delta_*/4$, since $\lambda_\eps^2=\ell/\log(1/\eps)$. For the event $\{\|Z\|_{L^\infty(Q^*)}>C_1\sqrt{\log(1/\eps)}/4\}$, with $Z=\nabla Y_0$, let $J=\lceil\log_2(1/\eps)\rceil$ and split $Z=Z_{\le}+Z_>$ with $Z_\le=\sum_{j\le J}\Delta_jZ$. The high part is bounded deterministically: by \eqref{eq:synth} with $n=1$, $\Delta_jZ=\lambda_\eps\,\Delta_j\tau_1(D)(1+\eps^2|D|^2)^{-1}\mathsf S_1\omega$, where $\|\mathsf S_1\omega\|_{L^\infty}\le C\eps^{-1}$ as for $\mathsf S\omega$, and for $j>J$ the multiplier $\phi_j\tau_1(1+\eps^2|\xi|^2)^{-1}$ has a kernel of $L^1$ norm $\le C2^{-j}(\eps2^j)^{-2}$: after the change of variables $\xi=2^j\eta$, it becomes a function supported in $\supp\phi_0$ whose derivatives up to order $3$ are bounded by $C2^{-j}(\eps2^j)^{-2}$, and the $L^1$ norm of a kernel is invariant under this rescaling. Since $\eps2^J\ge1$, $\|Z_>\|_{L^\infty}\le C\eps^{-1}\sum_{j>J}2^{-j}(\eps2^j)^{-2}\le C\eps^{-1}2^{-J}\le C$. The Fourier transform of $Z_\le$ is supported in $B(0,2^{J+2})$ and, by Lemma~\ref{lem:conc}(a) with $\Lambda=\delta_y$ and the multiplier $\tau_1\sum_{j\le J}\phi_j$, $Z_\le(y)$ has variance proxy $C\lambda_\eps^2\int_{c_\chi\le|\xi|\le2^{J+2}}|\xi|^{-2}d\xi\le C\ell$. Lemma~\ref{lem:sampling} with $\iota=1$, mesh $\delta_02^{-J-2}\ge c\eps$ and thresholds $s\avg{y-c}$, $s=C_1\sqrt{\log(1/\eps)}/96$, gives $\|Z_\le\|_{L^\infty(Q^*)}\le C_1\sqrt{\log(1/\eps)}/8$ outside an event of probability $\le C\eps^{-2}\eps^{cC_1^2/\ell}\le C\eps$ for $C_1^2\ge3\ell/c$. Hence $\|Z\|_{L^\infty(Q^*)}\le C_1\sqrt{\log(1/\eps)}/4$ outside this event, for $\eps$ small. Since $p_0\le1$, both probabilities are $\le C\eps^{p_0}$, which proves (ii). The constant $c$ in $p_0$ is that of Lemma~\ref{lem:prob} divided by $16$, and $C_1$ is fixed by the condition $C_1^2\ge3\ell/c$ above. Neither depends on $\delta_*$ or $\nu_r$.
\end{proof}

\section*{Acknowledgments}
This work was supported in part by NSF grant DMS-2606421 and ONR grant N00014-26-1-2017.

\section*{Use of generative AI}
Generative AI (Claude Opus 5.5) was used in the preparation of this work: to carry out the numerical evaluations and prepare the figures, and to write the scripts reproducing them; to help draft and revise the text and the proofs; to audit the manuscript for internal consistency and gaps in its arguments and proofs; and to help with bibliographical data. All such output was checked by the author, who is alone accountable for the correctness of all statements and proofs.

\clearpage
\setcounter{section}{0}\setcounter{subsection}{0}\setcounter{equation}{0}\setcounter{figure}{0}\setcounter{table}{0}
\renewcommand{\thesection}{S\arabic{section}}\renewcommand{\theHsection}{S\arabic{section}}
\numberwithin{equation}{section}
\renewcommand{\theequation}{S\arabic{section}.\arabic{equation}}\renewcommand{\theHequation}{S\arabic{section}.\arabic{equation}}
\renewcommand{\thefigure}{S\arabic{figure}}\renewcommand{\theHfigure}{S\arabic{figure}}
\renewcommand{\thetable}{S\arabic{table}}\renewcommand{\theHtable}{S\arabic{table}}
\begin{center}
{\Large\bf Supplementary material}\\[6pt]
{\large Homogenization theory of topological Anderson insulators: genuine gaps and pseudo-gaps}\\[6pt]
Guillaume Bal
\end{center}
\phantomsection\addcontentsline{toc}{section}{Supplementary material}

This supplementary material recalls the setting and the main formulas of the paper (section~\ref{sec:Ssetting}), collects the tools used in the proofs together with concise outlines of the main proofs (section~\ref{sec:Stools}), and contains examples and numerical illustrations (section~\ref{sec:Sexamples}), extensions and refinements, including the quantization of first-order bulk invariants (section~\ref{sec:Sremarks}), and possible generalizations (section~\ref{sec:Sdiscussion}). Section~\ref{sec:Snotation} is an index of the notation of the main paper and of this supplement, including the symbols used with more than one meaning. Sections~\ref{sec:Ssetting} and~\ref{sec:Stools} primarily summarize background material in the fields the paper draws on: homogenization, index theory, random operators, percolation, concentration inequalities, and paraproduct calculus.

\section{Setting and results from the main paper}\label{sec:Ssetting}

\paragraph{Overview.} A topological insulator is gapped in its bulk but carries robust currents along interfaces between materials in different topological phases. A topological Anderson insulator is one whose non-trivial phase is created by disorder. In the effective medium picture \cite{S:GWATB09}, disorder renormalizes the mass of a Dirac-type Hamiltonian, and the phase changes with the sign of the renormalized mass. The main paper has three layers. \emph{Homogenization}: a rapidly oscillating periodic fluctuation of a constant-coefficient elliptic operator $P(D)$, with amplitudes chosen so that all its terms act at the same scale, does not average out, and the perturbed operator converges in norm-resolvent sense to $P+M$, with $M$ a constant matrix given by a second-order (Born) formula. \emph{Topology}: in two dimensions, a fluctuation confined to a half-plane creates an interface between the phases $P$ and $P+M$, whose index is transferred from the limit to the perturbed operator. In particular, electric potentials and magnetic fields of zero mean, usually regarded as topologically trivial, generate corrections $M$ that change the phase. \emph{Randomness}: for random fluctuations, rare regions typically fill the homogenized gap, which becomes a pseudo-gap. The index is then defined along cuts deformed around the rare regions, and equals the bulk-difference invariant of the operator homogenized in the Born approximation, both for first-order alloys with vanishing low moments and for a class of second-order operators with scalar disorder at criticality, which includes the BHZ model with on-site disorder. We note two technical difficulties. The principal symbol may be indefinite, which is the topologically interesting case and rules out coercivity, energy estimates and maximum principles, so that all estimates are spectral or Fourier-analytic. And at criticality the Born correction diverges logarithmically, so that the amplitude must be renormalized and the analysis carried out in spaces of negative regularity.

The rest of this section recalls, without proofs, the statements and formulas used in the rest of the supplement.
Notation follows section~\ref{M-sec:setting} of the main paper, and all symbols are listed in section~\ref{sec:Snotation}.

\paragraph{Operators and correction.} The unperturbed operator is $P(D)=\sum_{|\alpha|\le k}A_\alpha D^\alpha$, with $A_\alpha\in\C^{N\times N}$ and principal symbol $P_k(\xi)=\sum_{|\alpha|=k}A_\alpha\xi^\alpha$. It satisfies \textup{[P1]} $P(\xi)\in\Herm(N)$ for all $\xi\in\R^d$, and \textup{[P2]} $|P_k(\xi)v|\ge c_P|\xi|^k|v|$ for some $c_P>0$. For $d\ge2$ the signature of $P_k(e)$, $|e|=1$, is constant, and $P_k$ is called definite or indefinite accordingly. It is indefinite when $k$ is odd. The basic fluctuation is a matrix-valued potential,
\begin{equation}\label{eq:SHeps}
H^\eps=P(D)+\eps^{-k/2}\rho(x)\,W(x/\eps),\qquad W\in C^{0,\alpha_0}(\T^d;\Herm(N)),\quad\avg W=0,
\end{equation}
with a smooth bounded envelope $\rho$. The homogenized operator and the correction are
\begin{equation}\label{eq:SM}
H^0=P(D)+\rho^2M,\qquad M=-\avg{W\,P_k(D_y)^{-1}W}=\sum_{q\in\Z^d\setminus0}\hat W_q^*\,G(q)\,\hat W_q,\qquad G(q)=-P_k(2\pi q)^{-1},
\end{equation}
where $P_k(D_y)^{-1}$ acts on mean-zero periodic functions. Pairing the modes $q$ and $-q$, with $\Lambda_+$ containing one element of each pair and $G_1(e)=-P_k(e)^{-1}$ for $|e|=1$,
\begin{equation}\label{eq:Spair}
M=\sum_{q\in\Lambda_+}(2\pi|q|)^{-k}\Bigl(\hat W_q^*G_1(\hat q)\hat W_q+(-1)^k\,\hat W_qG_1(\hat q)\hat W_q^*\Bigr).
\end{equation}
We use repeatedly the elementary identity, valid for Hermitian $A_a$ and $G$ and complex $x_a$,
\begin{equation}\label{eq:Sbilin}
X=\sum_ax_aA_a\quad\Longrightarrow\quad X^*GX-XGX^*=2i\sum_{a<b}\Im(\bar x_ax_b)\,\bigl(A_aGA_b-A_bGA_a\bigr).
\end{equation}
The main paper also treats fluctuations of all orders below $k$. We only use the electromagnetic case (Peierls substitution): for $k\le2$,
\begin{equation}\label{eq:Speierls}
H^\eps=P\bigl(D+\eps^{k/2-1}\rho A(\tfrac x\eps)\bigr)+\eps^{-k/2}\rho V(\tfrac x\eps),
\end{equation}
with $A\in C^{1,\alpha_0}(\T^d;\R^d)$, $V\in C^{0,\alpha_0}(\T^d;\R)$, $\avg A=0$ and $\avg V=0$.
For $k=1$, this is \eqref{eq:SHeps} with $W=V+A\cdot\gamma$ when $P=\gamma(D)+\Gamma$. For $k=2$ and $P_2(\xi)=\sum_{i,j}a_{ij}\xi_i\xi_j$, the correction is (Lemma~\ref{M-lem:peierls} of the main paper)
\begin{equation}\label{eq:SMpeierls}
M=\Bigl\langle\sum_{i,j}a_{ij}A_iA_j\Bigr\rangle-\avg{W^*P_2(D_y)^{-1}W},\qquad W=V-i\sum_{i,j}a_{ij}\,\partial_{y_i}A_j .
\end{equation}
Theorem~\ref{M-thm:hom} of the main paper states that, for all $z\notin\R$ and $\eps\in(0,1]$,
\begin{equation}\label{eq:Shom}
\big\|(H^\eps-z)^{-1}-(H^0-z)^{-1}\big\|_{L^2\to L^2}\le C\eps^{r_k}(1+|\Re z|)^2(1+|\Im z|^{-2}),\qquad r_1=\tfrac12,\quad r_k=1\ (k\ge2).
\end{equation}
By Theorem~\ref{M-thm:univ} of the main paper, for $d,N\ge2$, the set of corrections $M(W)$ over Hermitian mean-zero trigonometric polynomials $W$ is $\Herm(N)$ when $P_k$ is indefinite, and the cone of negative (resp.\ positive) semidefinite matrices when $P_k$ is positive (resp.\ negative) definite. In the definite case, the spectrum of every constant-coefficient bulk $P+\Gamma$ is a half-line, so that a gap at $0$ lies outside the spectrum, both bulk invariants vanish, and $\BDI(P+M,P)=0$ for every $M$. Topology requires an indefinite principal symbol.

\paragraph{Interface index.} Let $d=2$, and let $\rho=\rho(x_2)$ be smooth, with $\rho=0$ for $x_2\ge1$ and $\rho=1$ for $x_2\le-1$. The two bulk operators of $H^0$ are $H_+=P(D)$ and $H_-=P(D)+M$, and we assume the common-gap condition
\begin{equation}\label{eq:Sgap}
\textup{[G]}\qquad(-g,g)\cap\bigl(\spec H_+\cup\spec H_-\bigr)=\emptyset\quad\text{for some }g>0 .
\end{equation}
With a smooth non-decreasing switch $\varphi$, equal to $0$ on $(-\infty,-g/4]$ and to $1$ on $[g/4,\infty)$, $g$ being the gap parameter of the statement at hand, and $\Pi=\one_{\{x_1\ge0\}}$, the Toeplitz operator is
\begin{equation}\label{eq:ST}
T(H)=\Pi\,U(H)\,\Pi+I-\Pi,\qquad U(H)=e^{2\pi i\varphi(H)} .
\end{equation}
For a constant Hermitian symbol $h(\xi)$ gapped at $0$, with negative spectral projector $\Pi_h(\xi)$, let
\begin{equation}\label{eq:Sc}
c(h)=\frac1{2\pi}\int_{\R^2}\Omega_h,\qquad\Omega_h=i\,\tr\bigl(\Pi_h\,d\Pi_h\wedge d\Pi_h\bigr).
\end{equation}
For $\Gamma_\pm\in\Herm(N)$ with $P+\Gamma_\pm$ gapped at $0$, the bulk-difference invariant is defined by
\begin{equation}\label{eq:SBDI}
\BDI(P+\Gamma_-,P+\Gamma_+)=c(P+\Gamma_-)-c(P+\Gamma_+),
\end{equation}
as in \cite{S:Bal2022}, and $c$ is locally constant along gapped deformations of $\Gamma$. By the bulk-edge correspondence (Proposition~\ref{M-thm:QB} of the main paper), $\Ind T(H)=\BDI(P+V_-,P+V_+)$ for interface operators $H=P(D)+V(x)$ with $V=V_\pm$ constant for $\pm x_2\ge1$. Theorem~\ref{M-thm:TAI} of the main paper states that, under [P1]--[P2] and [G], $T(H^\eps)$ is Fredholm with $\Ind T(H^\eps)=\BDI(P+M,P)$ for all $\eps$ small enough. For two-band symbols $h=d\cdot\sigma+f\sigma_0$, $d=(d_1,d_2,d_3)$, gapped at $0$, that is with $|f|<|d|$, the negative projector is $\frac12(1-n\cdot\sigma)$ with $n=d/|d|$, independently of $f$, and
\begin{equation}\label{eq:Stwoband}
c(h)=\frac1{4\pi}\int_{\R^2}n\cdot(\partial_1n\times\partial_2n)\,d\xi .
\end{equation}
If $n$ extends continuously to $\R^2\cup\{\infty\}$ and $(d_1,d_2)$ vanishes at a single point $\xi_0$, with $\det\partial_\xi(d_1,d_2)(\xi_0)>0$, then
\begin{equation}\label{eq:Sdegree}
c(h)=\tfrac12\bigl(\sgn d_3(\xi_0)-\sgn d_3(\infty)\bigr).
\end{equation}
For the Dirac operator, $c(\xi\cdot\sigma+m\sigma_3)=\frac12\sgn m$, since $n$ covers the hemisphere $\{\sgn(m)\,n_3>0\}$ once, with orientation $\sgn m$.

\paragraph{Formulas for Dirac-type operators.} For $P=\gamma(D)+\Gamma$ of first order in $d=2$, $\gamma(\xi)=\xi_1\gamma_1+\xi_2\gamma_2$, and $W=V\,I+A_1\gamma_1+A_2\gamma_2$, we write $A=\nabla^\perp\psi+\nabla\zeta$ with $\psi,\zeta$ periodic and mean zero. The function $\psi$ is the stream function of $A$. Proposition~\ref{M-prop:EM} of the main paper gives, with $\mathcal Q(e)=i[\gamma(e)^{-1},\gamma(e^\perp)]$,
\begin{equation}\label{eq:SEM}
M=-\sum_{q\ne0}\Re\bigl(\overline{\hat V_q}\,\hat\psi_q\bigr)\,\mathcal Q(\hat q),\qquad M=2\avg{V\psi}\,\sigma_3\quad\text{for }\gamma=\sigma .
\end{equation}
For the regularized Dirac operator $P=D\cdot\sigma+(m+\beta\Delta)\sigma_3$, $\beta\ne0$, and $W=\sum_{a=0}^3V_a\sigma_a$, let $-\Delta\phi_a=V_a$ with $\phi_a$ periodic and mean zero, and $\mathcal C_{ab}=\int_{\T^2}\nabla\phi_a\cdot\nabla\phi_b\,dy$. Proposition~\ref{M-prop:regDirac} of the main paper gives
\begin{equation}\label{eq:SregDirac}
M=\beta^{-1}\bigl[(\mathcal C_{00}+\mathcal C_{33}-\mathcal C_{11}-\mathcal C_{22})\sigma_3+2\mathcal C_{03}\sigma_0+2\mathcal C_{31}\sigma_1+2\mathcal C_{32}\sigma_2\bigr],
\end{equation}
and, for the electromagnetic fluctuation \eqref{eq:Speierls} with $k=2$, $A=\nabla^\perp\psi$ and $V_0=V$, $M=\bigl(\mathcal C_{00}/\beta-\beta\int_{\T^2}|\nabla\psi|^2\bigr)\sigma_3$. For the BHZ-type operator with symbol $\xi\cdot\sigma+(m+b|\xi|^2)\sigma_3+c|\xi|^2$, $|b|\ne|c|$, the same proposition gives
\begin{equation}\label{eq:SBHZ}
\begin{gathered}
M=\mathcal C_{00}\,\frac{c\,\sigma_0-b\,\sigma_3}{b^2-c^2}\quad(W=V\sigma_0),\\
M=\mathcal C_{00}\,\frac{c\,\sigma_0-b\,\sigma_3}{b^2-c^2}+\Bigl(\int_{\T^2}|\nabla\psi|^2\Bigr)(b\sigma_3+c)\quad\text{(electromagnetic)} .
\end{gathered}
\end{equation}
Finally, for the regularized Dirac operator with $m\beta\ne0$ and $M=M_0\sigma_0+M_\perp\cdot\sigma+M_3\sigma_3$, Proposition~\ref{M-prop:regIndex} of the main paper states that if $|M_0|<\inf_\xi|d(\xi)|$, with $d(\xi)=(\xi+M_\perp,\,m+M_3-\beta|\xi|^2)$, then [G] holds and
\begin{equation}\label{eq:Smeff}
\BDI(P+M,P)=\tfrac12\bigl(\sgn m_{\rm eff}-\sgn m\bigr),\qquad m_{\rm eff}=m+M_3-\beta|M_\perp|^2\ne0 .
\end{equation}

\paragraph{Gaps from gradings.} Several gaps in the paper must hold for every configuration of the disorder, or for every field $M(x)$ with values in a set $\cK$, although the fluctuations are not small. Norm perturbation theory cannot give such gaps. They come instead from a grading, a unitary involution $J$ that splits the operator into a part anticommuting with $J$, which cannot close a gap however large it is, and a part commuting with $J$ that has a definite sign relative to $J$, which opens it. We recall Lemma~\ref{M-lem:grade} of the main paper: if $J$ is a unitary involution preserving $\mathrm{dom}(H)$ and $H=H_{\rm o}+H_{\rm e}$ with $H_{\rm o},H_{\rm e}$ symmetric, $JH_{\rm o}=-H_{\rm o}J$, $JH_{\rm e}=H_{\rm e}J$ and
\begin{equation}\label{eq:Sgrade}
\ip{JH_{\rm e}\psi}{\psi}\ge g\|\psi\|^2\quad\text{for all }\psi\in\mathrm{dom}(H),\qquad\text{then}\qquad\spec H\cap(-g,g)=\emptyset .
\end{equation}
The proof takes one line. For $\psi\in\mathrm{dom}(H)$, the anticommutation makes $\ip{H_{\rm o}\psi}{J\psi}$ purely imaginary, so that, for $|E|<g$, $\|(H-E)\psi\|\,\|\psi\|\ge\Re\ip{(H-E)\psi}{J\psi}\ge(g-|E|)\|\psi\|^2$. The bound involves neither the size of $H_{\rm o}$ nor the spatial variations of $H_{\rm e}$, only the sign condition. For instance, for $H=(D+a)\cdot\sigma+\mu(x)\sigma_3+V(x)$ with $a$, $\mu$, $V$ real, bounded and measurable, the grading $J=\pm\sigma_3$ makes the kinetic part and every vector potential $a$ odd, and the mass and the electric potential even. The gap $(-g,g)$ then persists whenever $\pm\mu-|V|\ge g$ almost everywhere, however large or rough $a$, $\mu$ and $V$ are.

The lemma is used in three ways. First, it proves Theorem~\ref{M-thm:magno}: vector potentials enter only $J$-odd terms or, in class (b), a $J$-even term with the sign of the mass, so that purely magnetic fluctuations cannot close the gap and leave the index trivial. Second, under condition (A) of Theorem~\ref{M-thm:phasepurek}, it gives the gap of $P+M(x)$ for every field $M(x)$ with values near $\cK$, which is how the random operator keeps a genuine gap in every configuration; this covers Corollary~\ref{M-cor:bernoulli} and Proposition~\ref{prop:highex}(a). Third, it gives a simple sufficient condition for the Born-gap condition [BG] of Theorem~\ref{M-thm:mainB}, with the gradings $s_+J$ and $s_-J$ for $P$ and $P+M_{\rm Born}$. In each case, the topological change comes from the sign alone: when the fluctuation inverts the $J$-sign of the even part, $s_-\ne s_+$, the mass is inverted, while the $J$-odd part of the fluctuation, however large, plays no role.

\paragraph{Alloys.} In $d=2$, an alloy is $W_\omega=\sum_{j\in\Z^2}\omega_jw(\cdot-j)$, with real couplings $|\omega_j|\le\kappa_+$ and a single-site profile $w\in C_c^\infty(\R^2;\Herm(N))$ satisfying the moment condition
\begin{equation}\label{eq:Smoments}
\int_{\R^2}y^\alpha w(y)\,dy=0\qquad\text{for all }|\alpha|\le2k-1 .
\end{equation}
With $\hat w$ the Fourier transform of $w$ on $\R^2$, define, for $q\in\T^2$ and $0\le\kappa_-\le\kappa_+$,
\begin{equation}\label{eq:SK}
\mathbb K(q)=\sum_{\ell\in\Z^2,\,q+\ell\ne0}\hat w(q+\ell)^*G(q+\ell)\hat w(q+\ell),\qquad\cK=[\kappa_-^2,\kappa_+^2]\cdot\overline\conv\{\mathbb K(q):q\in\T^2\}.
\end{equation}
For $L$-periodic configurations, $M(\omega)=\int_{\T^2}\mathbb K\,d\mathfrak s_\omega$, where the structure factor $\mathfrak s_\omega$ is a nonnegative measure whose mass is the mean of $\omega_j^2$ over a period. Hence $M(\omega)\in\cK$ when $\kappa_-\le|\omega_j|\le\kappa_+$ (Lemma~\ref{M-lem:Kset} of the main paper). For i.i.d.\ couplings and $k=1$, the typical local effective matrix is the Born matrix
\begin{equation}\label{eq:SMborn}
M_{\rm Born}=\mathrm{Var}(\omega)\int_{\T^2}\mathbb K(q)\,dq+(\EE\omega)^2\,\mathbb K(0).
\end{equation} For a unitary involution $J$, let $M^{\rm e}=\frac12(M+JMJ)$. Theorem~\ref{M-thm:phasepurek} of the main paper assumes one of the following conditions:
\begin{enumerate}[label=\textup{(\Alph*)},leftmargin=*,nosep]
\item $P(\xi)=P^{\rm o}(\xi)+\Gamma$ with $JP^{\rm o}(\xi)=-P^{\rm o}(\xi)J$ and $J\Gamma=\Gamma J$, and there are signs $s_\pm$ with $s_+J\Gamma\ge g$ and $s_-J(\Gamma+M^{\rm e})\ge g$ for every $M\in\cK$;
\item there are $\bar M\in\cK$ and $\bar r\ge0$ with $|M-\bar M|\le\bar r$ for all $M\in\cK$, $\dist(0,\spec P)\ge g$ and $\dist(0,\spec(P+\bar M))\ge g+\bar r$.
\end{enumerate}
It concludes that for $\eps$ small and every configuration with $\kappa_-\le|\omega_j|\le\kappa_+$, the bulk operator $P+\eps^{-k/2}W_\omega(x/\eps)$ has no spectrum in $(-g/2,g/2)$, and the interface operator $H^\eps_\omega=P+\eps^{-k/2}\rho W_\omega(x/\eps)$ satisfies $\Ind T(H^\eps_\omega)=\BDI(P+\bar M,P)$, with $\bar M$ any element of $\cK$ in case (A).

\paragraph{Critical scalar disorder.} Theorem~\ref{M-thm:mainA} of the main paper concerns $d=k=2$ and $P=A\,q(D)+P_1(D)+P_0$, with $A\in\Herm(N)$ invertible, $q(\xi)=\xi\cdot B\xi$ for a real positive definite $B$, $P_1(\xi)=\sum_jc_j\xi_j$ with $c_j\in\Herm(N)$, and $P_0\in\Herm(N)$. The disorder is a scalar alloy $W=\sum_j\omega_jw(\cdot-j)$, $w\in C_c^\infty(\R^2;\R)$, with i.i.d.\ centered couplings, and
\begin{equation}\label{eq:Scrit}
H^\eps=P+\rho\,\lambda_\eps\eps^{-1}W(x/\eps),\quad\lambda_\eps^2=\frac\ell{\log(1/\eps)},\quad s_0=\EE\omega^2\Bigl(\int w\Bigr)^2>0,\quad M_{\rm crit}=-\ell s_0c_qA^{-1},
\end{equation}
with $c_q=(2\pi)^{-2}\int_{\Sm^1}q(e)^{-1}de$. If $P$ and $P+M_{\rm crit}$ are gapped on $(-g,g)$, then for $\eps$ small, almost surely (a.s.), the Toeplitz operators along admissible deformed cuts are Fredholm with index $\BDI(P+M_{\rm crit},P)$.

\section{Tools and outlines of the proofs}\label{sec:Stools}

This section collects the tools used in the main paper, with short proofs or sketches, and concise outlines of the proofs of the main results. It follows the logic of the paper: periodic homogenization (section~\ref{sec:Stools-hom}), the spectral tools and the interface index (section~\ref{sec:Stools-fred}), random homogenization (section~\ref{sec:Stools-random}), the analysis at criticality (section~\ref{sec:Stools-crit}), the probabilistic estimates (section~\ref{sec:Stools-prob}), and the deformed cuts that combine them in Theorems~\ref{M-thm:mainB} and~\ref{M-thm:mainA} (section~\ref{sec:Stools-perc}). Standard references are \cite{S:ReedSimon1,S:Kato95} for spectral theory, \cite{S:BLP78} for homogenization, \cite{S:Kirsch08} for random operators, \cite{S:BCD11} for Littlewood--Paley theory, and \cite{S:Grimmett99} for percolation.

\subsection{Elliptic systems and periodic homogenization}\label{sec:Stools-hom}\label{sec:Stools-ell}

The Fourier transform conjugates $P(D)$ to multiplication by the Hermitian matrix $P(\xi)$, and ellipticity identifies the maximal domain of this multiplication operator with $\mathcal FH^k$, since $|P(\xi)v|\ge\frac{c_P}2|\xi|^k|v|$ for $|\xi|$ large. Hence $P(D)$ is self-adjoint on $H^k$ \cite[\S VIII.3]{S:ReedSimon1}, and $\spec P(D)$ is the closure of the union over $\xi$ of the eigenvalues of $P(\xi)$. The same lower bound gives $\|u\|_{H^k}\le C(\|P(D)u\|+\|u\|)$, and interpolation absorbs bounded potentials and differential perturbations of order $<k$, so that, for such perturbations $H$ of $P(D)$,
\begin{equation}\label{eq:Sellres}
\|(H-z)^{-1}\|_{L^2\to H^k}\le C(1+|z|)(1+|\Im z|^{-1}).
\end{equation}
The same interpolation shows, by the Kato--Rellich theorem \cite[Theorem~X.12]{S:ReedSimon2}, that the perturbed operators are self-adjoint on $H^k$. Since $P_k(e)$ is invertible for $e$ on the connected sphere $S^{d-1}$, $d\ge2$, its signature is constant. For odd $k$, $P_k(-e)=-P_k(e)$ forces it to be balanced, so that $P_k$ is indefinite. For instance, the regularized Dirac operator $D\cdot\sigma+(m+\beta\Delta)\sigma_3$ has the indefinite principal symbol $-\beta|\xi|^2\sigma_3$, and the BHZ operator has principal symbol $|\xi|^2(b\sigma_3+c)$, indefinite if and only if $|c|<|b|$.

\paragraph{The scaling.} In the fast variable $y=x/\eps$, $P_k(D_x)=\eps^{-k}P_k(D_y)$. The first-order effect of a mean-zero potential of size $\eps^{-a}$ averages out, and the second-order effect is $\eps^{-a}(\eps^{-k}P_k(D_y))^{-1}\eps^{-a}=\eps^{k-2a}WP_k(D_y)^{-1}W$, of order one exactly for $a=k/2$. Potentials of this size are classical in homogenization \cite[Ch.~1, \S12]{S:BLP78}, \cite{S:ACPSV04}. A derivative falling on the fast variable costs $\eps^{-1}$, which explains the amplitudes $\eps^{|\alpha|-k/2}$ of the terms of order $|\alpha|$ in \eqref{M-eq:Heps}.

\paragraph{The two-scale computation for $k=1$.} For $H^\eps=L(D)+\Gamma+\eps^{-1/2}\rho W(x/\eps)$ with $L(\xi)=\gamma_1\xi_1+\gamma_2\xi_2$, the ansatz $\psi_{\rm app}=\psi^0+\eps^{1/2}\rho\chi^\eps\psi^0+\eps\rho^2\chi_2^\eps\psi^0$ and the identity $L(D)[b^\eps u]=\eps^{-1}(L(D_y)b)^\eps u+b^\eps L(D)u$ give, at order $\eps^{-1/2}$, the first cell problem $L(D_y)\chi=-W$, solvable since $\avg W=0$. At order $1$, they give $(L(D)+\Gamma-z)\psi^0+\rho^2(W\chi+L(D_y)\chi_2)^\eps\psi^0$, which equals $(L(D)+\Gamma+\rho^2M-z)\psi^0=f$ once $M=\avg{W\chi}$ and $L(D_y)\chi_2=-(W\chi-M)$. The remainder is of order $\eps^{1/2}$. This is the classical two-scale expansion \cite{S:BLP78,S:JKO94}. In Fourier series, $M=-\avg{WL(D_y)^{-1}W}$, which is \eqref{eq:SM}.

\paragraph{From approximate solutions to resolvent estimates.} If $\psi_{\rm app}\in H^k=\mathrm{dom}(H^\eps)$ satisfies $\|(H^\eps-z)\psi_{\rm app}-f\|\le r$ and $\|\psi_{\rm app}-\psi^0\|\le r'$, with $\psi^0=(H^0-z)^{-1}f$, then $\|(H^\eps-z)^{-1}f-\psi^0\|\le|\Im z|^{-1}r+r'$, since $H^\eps$ is self-adjoint. A norm-resolvent estimate requires this for every $f\in L^2$, with $r$ and $r'$ proportional to $\|f\|$. For such $f$, the only control on $\psi^0$ is the elliptic estimate \eqref{eq:Sellres}, $\|\psi^0\|_{H^k}\le C(1+|z|)(1+|\Im z|^{-1})\|f\|$. Thus, the function $\psi^0$ has $k$ derivatives in $L^2$, and in general no more. Residual bounds of the form $r,r'\le C\eps^{r_k}(1+|z|)\|\psi^0\|_{H^k}$ therefore give the norm-resolvent bound \eqref{eq:Shom}, in the spirit of the operator estimates of \cite{S:BirmanSuslina04,S:ZhikovPastukhova05}, while bounds involving higher norms of $\psi^0$ would not.

\paragraph{Outline of the proof of Theorem~\ref{M-thm:hom}.} For $k=1$, we have just given the proof. Its remainder involves $\chi$, $\chi_2$ and the first derivatives of $\rho\psi^0$ only, so that $\psi^0\in H^1$ suffices, and the rate is $\eps^{1/2}$. For $k\ge2$, the correctors are monomials $\eps^ab^\eps Q\psi^0$, products of a periodic coefficient $b$ and of a differential operator $Q$ of order $q$ acting on $\psi^0$. The hierarchy described below produces monomials with $q>k$, which $\|\psi^0\|_{H^k}$ does not control. The correctors are therefore built on the smoothed function $v=S_\eps\psi^0$, where $S_\eps$ is the Fourier multiplier by $\one_{\{|\xi|\le1/\eps\}}$, as in \cite{S:ZhikovPastukhova05}. The two bounds \eqref{M-eq:smooth},
\[
\|(I-S_\eps)u\|_{H^s}\le\eps^{k-s}\|u\|_{H^k}\quad(0\le s\le k),\qquad\|S_\eps u\|_{H^t}\le C_t\,\eps^{-(t-k)_+}\|u\|_{H^k},
\]
express the trade-off. Each derivative of $v$ beyond the $k$ derivatives of $\psi^0$ costs a factor $\eps^{-1}$. This is also the cost of a derivative falling on a coefficient that oscillates at scale $\eps$, so that the cutoff at the oscillation frequency $1/\eps$ does not alter the counting of powers of $\eps$. The counting is recorded by the weight $a-q$ of a monomial, and \eqref{M-eq:mono} bounds the monomial by $\eps^{\min(a,\,a-q+k)}\|\psi^0\|_{H^k}$. Replacing $\psi^0$ by $v$ costs only two terms, $\mathcal V^\eps(I-S_\eps)\psi^0$ and $\rho^2M(I-S_\eps)\psi^0$, of norms $O(\eps^{k/2})$ and $O(\eps^k)$ relative to $\|\psi^0\|_{H^k}$, by the first bound.

The smoothing acts on the solution only. The homogenized operator $H^0=P+\rho^2M$, which generates the limiting dynamics, is not modified. The correctors are algebraic: they enter $\psi_{\rm app}$ but no evolution equation, so that no truncated operator is ever inverted or exponentiated. A sharp cutoff suffices, since only Sobolev norms of $v$ are used, and no commutator or composition with $S_\eps$ is needed. The same device is used for alloys in Theorem~\ref{M-thm:unifhom}.

The expansion is organized as a hierarchy. Each term produced by $H^\eps-z$ acting on a corrector is classified by its order, the power of $\eps$ it carries, and by its weight. The terms of order at most $\frac12$, which all have mean-zero periodic coefficients except those producing $M$, are cancelled by new correctors obtained by solving cell problems. The orders of the terms to be corrected increase by at least $1$ from one generation to the next, from $1-k/2$ to at most $\frac12$, so that the hierarchy terminates after at most $\lceil k/2\rceil$ generations beyond the first corrector, a bound attained in general: two generations of correctors in total for $k=1,2$, three for $k=3,4$, and $1+\lceil k/2\rceil$ for general $k$. The remaining monomials have order at least $1$ and weight at least $-k/2$, hence norm $O(\eps)$ relative to $\|\psi^0\|_{H^k}$ by \eqref{M-eq:mono}, and we obtain the rate $\eps^{r_k}$. For $k=2$, e.g., the first corrector $\eps\chi^\eps\rho v$, with $P_2(D_y)\chi=-W$, produces terms of order $0$. Their mean is $M$, and the rest is cancelled by correctors of order $2$, after which the hierarchy stops. Only the invertibility of $P_k(2\pi q)$ for $q\ne0$, the estimate \eqref{eq:Sellres} and the bound $\|(H^\eps-z)^{-1}\|\le|\Im z|^{-1}$ are used. None requires a sign, which is why indefinite principal symbols are allowed.

\paragraph{Which fluctuations invert the mass.} The topological change requires the correction $M$ to invert the mass of $P$, and its structure is constrained by the symmetries of $P$ and of the fluctuation. Without such constraints, every correction is attainable (Theorem~\ref{M-thm:univ}): by \eqref{eq:Spair}, $M$ is a sum of terms $X^*GX\pm XGX^*$ with $G=G_1(\hat q)$ and $X=\hat W_q$ arbitrary, and for invertible indefinite $G$ the null vectors $v^*Gv=0$ span the hyperplane $\{H:\tr(GH)=0\}$ (Lemma~\ref{M-lem:null}), so that rank-one fluctuations in one lattice direction produce this hyperplane, and two directions produce every Hermitian matrix when $k$ is odd. For the Dirac operator, a single direction $\hat q$ produces exactly the matrices orthogonal to $\hat q\cdot\sigma$, so that a cosine potential in $y_1$ can produce a mass but no $\sigma_1$ component. Symmetries of $P$ preserved by the fluctuation are obviously inherited by the correction: If $\mathsf U$ is unitary and $R$ is an orthogonal map preserving $\Z^d$ with $\mathsf U P_k(\xi)\mathsf U^*=P_k(R\xi)$ and $\mathsf U W(y)\mathsf U^*=W(Ry)$, in law if $W$ is random, then $\mathsf U\hat W_q\mathsf U^*=\hat W_{Rq}$ and $\mathsf U G(q)\mathsf U^*=G(Rq)$, so that $\mathsf U M\mathsf U^*=M$ by \eqref{eq:SM}. The same holds for an antiunitary $\Theta$ with $\Theta P_k(\xi)\Theta^{-1}=P_k(-\xi)$ and $\Theta W\Theta^{-1}=W$. In $d=2$, invariance of $P$ and $W$ under time reversal or under a reflection ($\det R=-1$) thus forces the invariants \eqref{eq:Sc} of both bulks to vanish, so that $\BDI(P+M,P)=0$: fluctuations can change the index only if they break every such symmetry of $P$. Gradings are more informative, and complement the graded gap \eqref{eq:Sgrade}: there, the part of the operator anticommuting with $J$ cannot close the gap; here, the parities of the fluctuation fix the parity, and the $J$-sign, of the correction.
\begin{lemma}[Gradings]\label{lem:grading}
Let $J$ be a unitary involution, and write $W=W_{\rm e}+W_{\rm o}$ with $W_{\rm e}$ commuting and $W_{\rm o}$ anticommuting with $J$. By \eqref{eq:SM}, $M=M_{\rm ee}+M_{\rm oo}+M_{\rm eo}$, with
\begin{gather*}
M_{\rm ee}=\sum_{q\ne0}\hat W_{{\rm e},q}^*G(q)\hat W_{{\rm e},q},\qquad M_{\rm oo}=\sum_{q\ne0}\hat W_{{\rm o},q}^*G(q)\hat W_{{\rm o},q},\\
M_{\rm eo}=\sum_{q\ne0}\bigl(\hat W_{{\rm e},q}^*G(q)\hat W_{{\rm o},q}+\hat W_{{\rm o},q}^*G(q)\hat W_{{\rm e},q}\bigr).
\end{gather*}
\textup{(i)} If $JP_k=-P_kJ$, then $M_{\rm ee}$ and $M_{\rm oo}$ anticommute with $J$, and $M_{\rm eo}$ commutes with $J$. The part of $M$ commuting with $J$ thus comes only from the cross term $M_{\rm eo}$.
\textup{(ii)} If $JP_k=P_kJ$ and $sJP_k\ge0$ for some $s\in\{\pm1\}$, then $M_{\rm ee}$ and $M_{\rm oo}$ commute with $J$, with $sJM_{\rm ee}\le0\le sJM_{\rm oo}$, and $M_{\rm eo}$ anticommutes with $J$. The part of $M$ commuting with $J$ is thus $M_{\rm ee}+M_{\rm oo}$, a sum of two terms of opposite $J$-signs.
\end{lemma}
\begin{proof}
$G(q)=-P_k(2\pi q)^{-1}$ anticommutes, respectively commutes, with $J$. A product $X^*GY$ with $X,Y$ of definite parity has parity given by the product of the parities. In case (ii), $sJ$ commutes with $P_k$, so $sJG=-(sJP_k(2\pi q))^{-1}\le0$. Then $sJ\hat W_{\rm e}^*G\hat W_{\rm e}=\hat W_{\rm e}^*(sJG)\hat W_{\rm e}\le0$, while $sJ\hat W_{\rm o}^*G\hat W_{\rm o}=-\hat W_{\rm o}^*(sJG)\hat W_{\rm o}\ge0$.
\end{proof}
\noindent For the Dirac operator ($J=\sigma_3$, case (i)), $W_{\rm e}$ and $W_{\rm o}$ are the $\{\sigma_0,\sigma_3\}$ and $\{\sigma_1,\sigma_2\}$ components of the fluctuation. The terms $M_{\rm ee}$ and $M_{\rm oo}$ contribute only to the vector potential $M_\perp$, while the mass and energy shifts, which commute with $J$, come only from the correlations between the two groups of components in $M_{\rm eo}$, as in \eqref{eq:SEM} and \eqref{eq:SMgen}. For the regularized Dirac operator ($J=\sigma_3$, $s=-\sgn\beta$, case (ii)), the roles are exchanged. The electric and mass components, through $M_{\rm ee}$, push the mass toward $\sgn\beta$, i.e., out of the trivial phase $m\beta<0$. The off-diagonal components, through $M_{\rm oo}$, push it the other way, as does the diamagnetic term $\avg{\sum a_{ij}A_iA_j}$ of \eqref{eq:SMpeierls}, and the correlations between the two groups of components, in $M_{\rm eo}$, produce only the vector potential $M_\perp$. Compare \eqref{eq:SregDirac}.

\subsection{Spectral tools and interface indices}\label{sec:Stools-fred}

\paragraph{Spectral comparison and quasimodes.} For self-adjoint $H$, $\|(H-z)^{-1}\|=1/\dist(z,\spec H)$. Lemma~\ref{M-lem:speccomp} turns this into a comparison principle: if two resolvents are close at $z=E+i\eta$ and $E$ lies in the spectrum of one operator, then the other has spectrum near $E$. It is used both to transfer gaps and spectrum. If $\|(H-\lambda)\psi\|\le\delta\|\psi\|$ for some $\psi\ne0$, then $\dist(\lambda,\spec H)\le\delta$. Conversely, such quasimodes exist near every point of the spectrum and may be taken compactly supported, by the spectral theorem and since $C_c^\infty$ is a core \cite[\S VIII.3]{S:ReedSimon1}. Lemma~\ref{M-lem:speccomp} is a quantitative form of the fact that the spectrum cannot suddenly expand under norm-resolvent convergence: if $H_n\to H$ in norm-resolvent sense and $(a',b')\cap\spec H=\emptyset$, then $[a'',b'']\cap\spec H_n=\emptyset$ for $a'<a''<b''<b'$ and $n$ large \cite[Thm.~VIII.23]{S:ReedSimon1}. This qualitative form is used, with a compactness argument, to obtain the uniform gap in the proof of Theorem~\ref{M-thm:gapclose}.

\paragraph{Helffer--Sj\"ostrand formula \cite{S:HelfferSjostrand89}.} For $f\in C_c^\infty(\R)$ and an almost analytic extension $\tilde f$, with $|\dbar\tilde f(z)|\le C_n|\Im z|^n$, we have $f(H)=\frac1\pi\int\dbar\tilde f(z)(H-z)^{-1}d^2z$. Resolvent bounds with powers of $|\Im z|^{-1}$ therefore become bounds on $f(H)$. In particular (Lemma~\ref{M-lem:HSdecay}), exponential off-diagonal decay of the resolvent at a rate proportional to $|\Im z|$ becomes polynomial off-diagonal decay of $f(H)$. The formula is applied only to the compactly supported functions $S=e^{2\pi i\varphi}-1$ and $\varphi'$, and this is essential in two ways. First, norm-resolvent convergence $H_n\to H$ implies $\|f(H_n)-f(H)\|\to0$ for $f\in C_0(\R)$ \cite[Thm.~VIII.20]{S:ReedSimon1}, and more generally for continuous $f$ with equal limits at $\pm\infty$, but not for a switch $\varphi$, whose limits differ. For instance, if $H$ and $H_n$ are multiplication by $\xi$ on $L^2(\R)$, except that $H_n$ is multiplication by $-\xi$ for $|\xi|\ge n$, then the resolvents converge in norm while $\|\varphi(H_n)-\varphi(H)\|=1$. This is why the Toeplitz operator is built from $U(H)=I+S(H)$ rather than from $\varphi(H)$. Second, the resolvent estimates of the main paper, such as \eqref{eq:Shom} and the Combes--Thomas bounds, degrade as $|\Re z|$ grows, and they are only needed on $\supp\dbar\tilde S$, which is bounded.

\paragraph{Combes--Thomas estimates.} We obtain off-diagonal decay of resolvents by conjugation with $e^{b\phi}$ \cite{S:CombesThomas73,S:GerminetKlein03}, $\phi$ a smoothed distance function. For first-order operators $\gamma(D)+V$, the conjugation adds the bounded term $ib\gamma(\nabla\phi)$, and a Neumann series gives decay at rate $a|\Im z|$ for every bounded $V$ (Lemma~\ref{M-lem:CT}(i)). For second-order operators, the conjugation adds a first-order term, and the argument requires an $H^1$ bound on the resolvent, on which the rate then depends (Lemma~\ref{M-lem:CT}(ii)).

\paragraph{Fredholm operators and the Toeplitz operator.} A bounded operator is Fredholm if and only if it is invertible modulo compact operators (Atkinson's theorem \cite{S:Douglas98}). Fredholm operators form an open set on which the index is locally constant and invariant under compact perturbations \cite[Ch.~IV, \S5]{S:Kato95}. Quantitatively, if $T$ is Fredholm and $\|T'-T\|<\|\mathfrak q(T)^{-1}\|^{-1}$, the norm being that of the Calkin algebra, then $\mathfrak q(T'_s)=\mathfrak q(T)\bigl(1+s\,\mathfrak q(T)^{-1}(\mathfrak q(T')-\mathfrak q(T))\bigr)$, with $T'_s=T+s(T'-T)$, is invertible for $s\in[0,1]$ by the Neumann series. Each $T'_s$ is then Fredholm by Atkinson's theorem, and $\Ind T'=\Ind T$ since the index is constant along norm-continuous paths of Fredholm operators. For the Toeplitz operator $T=\Pi U\Pi+I-\Pi$ with $U=U(H)$ unitary, $\mathfrak q(T)$ is unitary because $TT^*-I$ and $T^*T-I$ are compact. Indeed, $TT^*-I=-\Pi U(I-\Pi)U^*\Pi$ and $T^*T-I=-\Pi U^*(I-\Pi)U\Pi$, while $\Pi U(I-\Pi)=[\Pi,U](I-\Pi)$ and $\Pi U^*(I-\Pi)=[\Pi,U^*](I-\Pi)$ are compact, since $[U(H),\Pi]=[S(H),\Pi]$ is compact by Lemma~\ref{M-lem:fred}(i). Then $\|\mathfrak q(T)^{-1}\|=\|\mathfrak q(T)^*\|=1$, and the radius is $1$, which is the statement of Lemma~\ref{M-lem:fred}(i). It is sharp. For $T=I$, the operator $T'=0$ lies at distance $1$. The same holds with $T'$ replaced by $T'+K$, $K$ compact. In the proofs, only small perturbations are needed: the differences $\|U(H^\eps)-U(H^0)\|$ tend to $0$, and the robustness argument of Theorem~\ref{M-thm:robust} uses explicit bounds $\|\mathfrak q(\cdot)^{-1}\|\le2$ and $4$ instead. The unilateral shift, the compression of the bilateral shift to a half-line, has index $-1$ \cite{S:Douglas98}. The index of $T(H)$ in \eqref{eq:ST} is of this type and counts the net number of states transported across the cut \cite{S:Bal2019,S:ProdanSchulzBaldes16}. If $H$ has no spectrum on $\supp(e^{2\pi i\varphi}-1)$, then $T(H)=I$. A non-zero index requires spectrum in the gap, carried by the interface. Lemma~\ref{M-lem:fred} shows that $T(H)$ is Fredholm with $\mathfrak q(T(H))$ unitary when $[U(H),\Pi]$ is compact, which follows from the decay of $U(H)-I$ along the cut and away from the interface, where $H$ is close to a gapped bulk operator.

\paragraph{The bulk-difference invariant.} For bulks with constant coefficients, the index depends only on the two bulks and equals the invariant \eqref{eq:SBDI} \cite{S:Bal2022,S:QuinnBal}. Lemma~\ref{M-lem:BDIc} recalls the proof: an integral formula for the interface current \cite{S:QuinnBal}, in which the spectral variable is integrated out by residues \cite{S:Bal2022}. For two-band symbols, \eqref{eq:Stwoband} identifies $c(h)$ with the pull-back of the area form of the sphere by $n=d/|d|$, and \eqref{eq:Sdegree} counts preimages of a pole. For the Dirac domain wall $D\cdot\sigma+m(x_2)\sigma_3$ with $m=\pm m_0$ for $\pm x_2>1$, the index is $c(-m_0\sigma_3)-c(m_0\sigma_3)=-1$, and the interface carries one chiral mode \cite{S:JackiwRebbi76}. Note that $c(h)$ is invariant under translations of $\xi$, but the corresponding gauge transformation, applied to an interface operator, shifts both bulks.

\paragraph{Outline of the proof of Theorem~\ref{M-thm:gapclose}.} We obtain Theorem~\ref{M-thm:TAI} directly from the Helffer--Sj\"ostrand formula, \eqref{eq:Shom} and Proposition~\ref{M-thm:QB}; for definite $P_k$, the spectrum of every constant-coefficient bulk $P+\Gamma$ is a half-line, so that a gap at $0$ forces both bulk invariants to vanish. Theorem~\ref{M-thm:magno} is the first application of the graded gap \eqref{eq:Sgrade}. For Theorem~\ref{M-thm:gapclose}, the family $\eps\mapsto H^\eps$ is not norm-resolvent continuous, but it becomes so after the dilation $x=\eps y$. Along a continuous family with gapped bulks the index cannot change. For $\eps$ large, $H^\eps$ is close to $P$, whose Toeplitz operator is the identity, so that a non-zero index for small $\eps$ forces a gap closing at an intermediate $\eps$ (the case of regularized Dirac operators with periodic fluctuations is \cite{S:BalDang24}).

\subsection{Random homogenization}\label{sec:Stools-random}

\paragraph{Random operators.} For a stationary ergodic field, the spectrum of $P+W(\cdot,\omega)$ is a.s.\ a deterministic set $\Sigma$ \cite{S:Pastur80,S:Kirsch08}, since translations act by unitary conjugation and events such as $\{\spec H_\omega\cap(a,b)\ne\emptyset\}$ are invariant. For alloys, $\Sigma$ is the closure of the union of the spectra over all supported configurations \cite{S:KirschMartinelli82}: a compactly supported quasimode of a supported configuration is, with positive probability, a quasimode of $H_\omega$ after translation, and by independence this happens a.s.\ somewhere. Similarly, if the law charges a periodic pattern, regions where the field resembles the pattern exist a.s., by the second Borel--Cantelli lemma, however small their probability.

\paragraph{The random Born term and criticality.} For a stationary field with spectral density $s$, the correction becomes $M=\int G(q)\diamond s(q)\,dq$, the Born term of \cite{S:FOP82}. Since $|G(q)|\sim|q|^{-k}$, the integral converges at $q=0$ if and only if $\int_{|q|<1}|q|^{-k}|s(q)|dq<\infty$. When $s(0)\ne0$, this holds if and only if $d>k$. For $d=k$, each dyadic shell of frequencies between $\eps$ and $1$ contributes equally, and the integral diverges logarithmically. The moment condition \eqref{eq:Smoments} makes $\hat w(q)=O(|q|^{2k})$ and removes the divergence, while scalar second-order disorder in $d=2$ with $\int w\ne0$ is critical.

\paragraph{Lattice sums and mesoscale averages.} A Fourier multiplier that vanishes to order $s>0$ at the origin and decays rapidly has a kernel decaying like $\avg y^{-2-s}$ (Lemma~\ref{M-lem:latsum}, see \cite[Ch.~VI]{S:Stein93}). Under \eqref{eq:Smoments}, the corrector profile decays fast enough for the corrector of an alloy to be an absolutely convergent lattice sum, for every configuration. Since non-periodic configurations have no global mean, the effective matrix is replaced by a local average over a disk of radius $\varpi^{-1}$ in the fast variable. The radius depends on $\eps$ as follows. In Theorem~\ref{M-thm:unifhom}, $\varpi=\eps^{1/4}$ for $k=1$ and $\varpi=\eps^{1/3}$ for $k\ge2$, together with a low-frequency cutoff $\varpi_0=\eps^{1/2}$ in the quadratic corrector: these choices balance the windowing error $\varpi+(\varpi_0/\varpi)^2$ against the size of that corrector, and give the rate $\delta_k(\eps)$. Theorem~\ref{M-thm:mainB}, with $k=1$, uses $\varpi=\eps^{1/4}$, so that its mesoboxes have side $\eps\varpi^{-1}=\eps^{3/4}$ in the macroscopic variable (Figure~\ref{fig:scales}(a)). At criticality, no mesoscale average is used, and bad regions are defined on unit boxes by the local size $\delta_Q$. Lemma~\ref{M-lem:meso} shows that this local effective matrix is, up to $O(\varpi)$, the Born matrix computed with the local structure factor of the couplings, hence lies near the set $\cK$ of \eqref{eq:SK}. Figure~\ref{fig:alloy} shows this field for one realization, at fixed values of $\varpi$.

\begin{figure}[ht]
\centering\includegraphics[width=\textwidth]{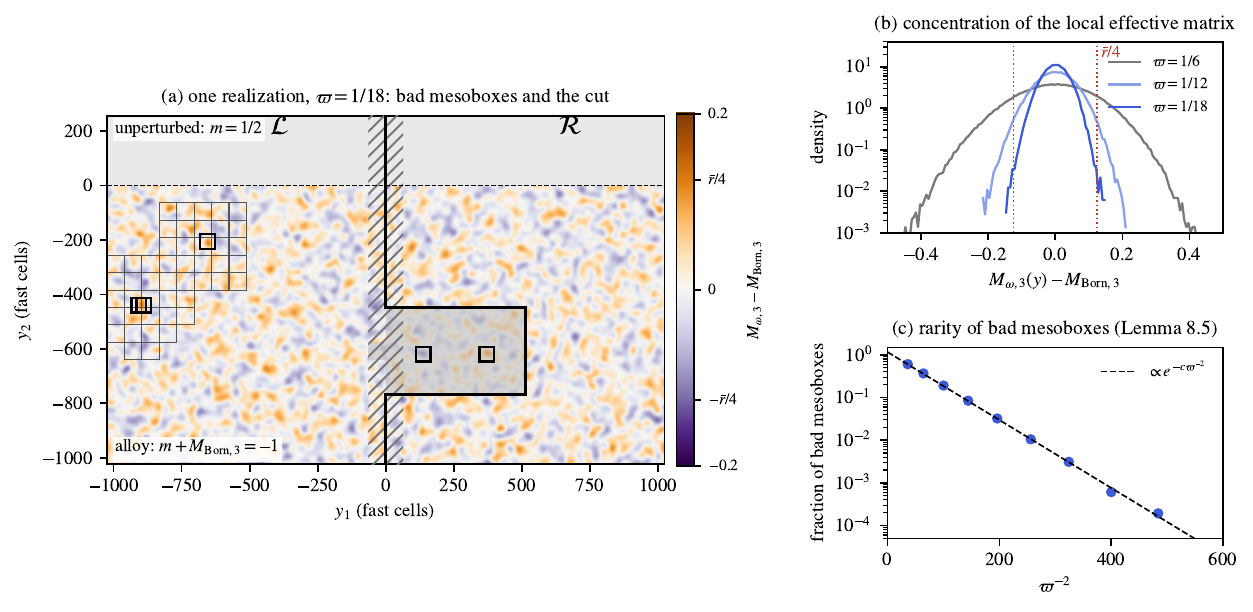}
\caption{Local effective matrix of a first-order alloy, for the Dirac operator $P=D\cdot\sigma+m\sigma_3$, $m=\frac12$, the single site of Proposition~\ref{prop:uniform}, couplings i.i.d.\ and uniform on $[-1,1]$, and $\tau$ chosen so that $m+M_{{\rm Born},3}=-1$, with $g=\bar r=\frac12$. For this single site, Lemma~\ref{M-lem:meso} gives $M_\omega=2\tau\|v\|^2\mu\,\sigma_3+O(\varpi)$ with $\mu=\sum_jK_\varpi(\cdot-j)\omega_j^2$. A direct computation from \eqref{M-eq:Mloc} shows that the error is about $5\cdot10^{-3}$ for $\varpi=\frac14$ and below $10^{-3}$ for $\varpi\le\frac1{12}$, far below $\bar r/4$. (a) A realization with few bad mesoboxes, chosen for legibility (compare panel (c)): deviation $M_{\omega,3}-M_{{\rm Born},3}$ in the lower half-plane, where $\rho$ is taken equal to $1$. Lengths are in fast cells, and the scales are compressed. Bad mesoboxes, where the deviation exceeds $\bar r/4$, are framed (enlarged threefold). Coarse boxes of side $\dd=64$ within distance $2\dd$ of them are outlined. The cluster meeting the strip $|y_1|\le\dd$ (hatched) is removed with its layer (grey), as in Lemma~\ref{M-lem:cuts}, and the thick line is the boundary of $\cR$. (b) Distribution of the deviation at three mesoscales. Its standard deviation is $0.64\,\varpi$ (Lemma~\ref{M-lem:mean}(ii)). (c) Fraction of bad mesoboxes over twelve realizations of $2048^2$ cells, against $\varpi^{-2}$, with the exponential law of Lemma~\ref{M-lem:mesobad}.}\label{fig:alloy}
\end{figure}

\paragraph{Outline of the proof of Theorem~\ref{M-thm:unifhom}.} We repeat the proof of Theorem~\ref{M-thm:hom} with cell problems solved by lattice sums and the mean replaced by the local effective matrix. The corrector hierarchy is linear in the couplings, except for one quadratic corrector that compensates the deviation of the local product $W\chi$ from its mesoscale average. This corrector involves the inverse principal symbol at low frequencies, which is cut off at a frequency $\varpi_0$. The rate $\delta_k(\eps)$ results from balancing the error of the mesoscale average, of order $\varpi+(\varpi_0/\varpi)^2$, against the size of the quadratic corrector, which grows as $\varpi$ and $\varpi_0$ decrease. All constants depend only on a bound on the couplings, so that the estimate is uniform over configurations.

\paragraph{Outline of the proof of Theorem~\ref{M-thm:pseudogap}.} Under condition (A), Theorem~\ref{M-thm:phasepurek} is the second application of the graded gap \eqref{eq:Sgrade}; under (B), the gap follows by norm perturbation. The gap is transferred to the random operator by Theorem~\ref{M-thm:unifhom}, and the index by a deformation of the local effective matrix to a constant. In Theorem~\ref{M-thm:pseudogap}, (i) a periodic operator with a gapless pattern has spectrum near $0$ by Theorem~\ref{M-thm:hom}, and a compactly supported quasimode can be transplanted to a rare region where the random field is close to the pattern; (ii) along a path of charged patterns joining two phases, the gap must close somewhere, since the index cannot change along a gapped path; (iii) in regions where all couplings are nearly equal and small, the operator is nearly the unperturbed one shifted by a constant, because a periodic mean-zero potential of bounded amplitude oscillating at scale $\eps$ is negligible in resolvent sense. Such shifts reach every energy, since the spectrum of $P$ contains two half-lines.

\subsection{Analysis at criticality}\label{sec:Stools-crit}

\paragraph{Spaces of negative regularity.} With the Littlewood--Paley decomposition \cite[Ch.~2]{S:BCD11} $f=\sum_{j\ge-1}\Delta_jf$, the H\"older--Zygmund norm $\|f\|_{C^s}=\sup_j2^{js}\|\Delta_jf\|_{L^\infty}$ measures distributions for $s<0$. A function oscillating at frequency $1/\eps$ with amplitude $1$ has $C^{-\kappa}$ norm of order $\eps^\kappa$: negative norms see oscillations as small. This is why the corrector gradient $\nabla Y_0$, of order $\sqrt{\log(1/\eps)}$ in $L^\infty$, can be small in $C^{-\kappa}$. Products of a distribution in $C^{-\kappa}$ with a function of positive regularity are handled by Bony's paraproduct decomposition \cite[Ch.~2]{S:BCD11}. The resonant part requires the sum of the regularities to be positive, and the part in which the rough factor carries the high frequencies forces a small loss of regularity, which is why section~\ref{M-sec:crit} works on the scale $H^{-2\kappa}\to H^{2-2\kappa}$ (Lemma~\ref{M-lem:prod}).

\paragraph{Logarithmic renormalization.} The corrector $Y_0=-q(D)^{-1}\tilde V_{\rm high}$ of the high-frequency part of the critical potential has energy $\EE|\nabla Y_0|_B^2$ proportional to $\lambda_\eps^2\int s(\eps\xi)/q(\xi)\,d\xi$, over frequencies $|\xi|\gtrsim1$. In polar coordinates this is $\lambda_\eps^2s(0)\log(1/\eps)$ times an angular factor, up to $O(\lambda_\eps^2)$, since $s(\eps\xi)\approx s(0)$ for $|\xi|\lesssim1/\eps$. The amplitude $\lambda_\eps^2=\ell/\log(1/\eps)$ exactly compensates the logarithm (Lemma~\ref{M-lem:ceps}). This is the analogue of the Wick renormalization of $|\nabla Y|^2$ in singular SPDEs, and in particular in the construction of the continuum Anderson Hamiltonian \cite{S:AllezChouk15,S:HairerLabbe15,S:Labbe19}. Figure~\ref{fig:renorm} shows the equal contribution of the dyadic shells and the slow, $O(1/\log(1/\eps))$, convergence of $c_\eps$.
\begin{figure}[ht]
\centering\includegraphics[width=\textwidth]{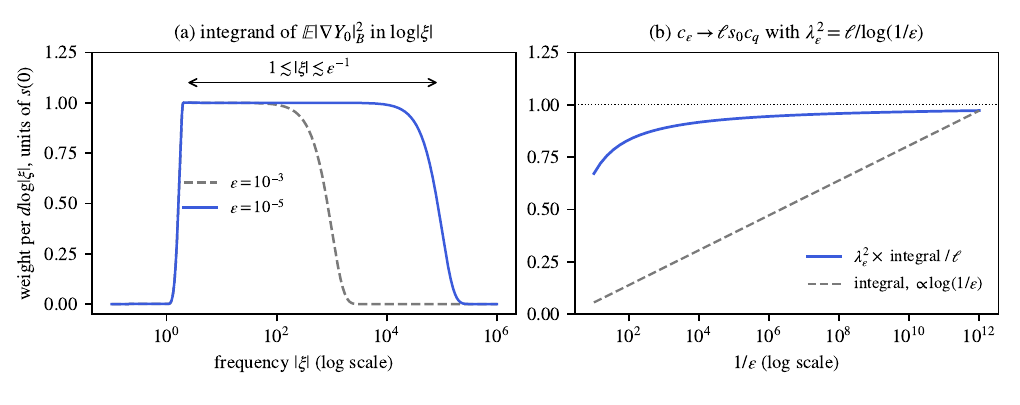}
\caption{Logarithmic renormalization (Lemma~\ref{M-lem:ceps}). (a) Integrand of $\EE|\nabla Y_0|_B^2$ per unit of $\log|\xi|$, in units of $s(0)$, for $\eps=10^{-3}$ and $10^{-5}$: the cutoff $1-\chi_0$ removes $|\xi|\lesssim1$, the spectral density $s(\eps\xi)$ removes $|\xi|\gtrsim\eps^{-1}$, and every dyadic shell in between contributes equally. (b) The integral (dashed, rescaled) grows like $\log(1/\eps)$, while $\lambda_\eps^2$ times the integral, divided by $\ell$ (solid), tends to $1$, that is $c_\eps\to\ell s_0c_q$, with an error $O(1/\log(1/\eps))$, still visible at $\eps=10^{-12}$.}\label{fig:renorm}
\end{figure}
\paragraph{The exponential transform.} At criticality, the corrector expansion of Theorems~\ref{M-thm:hom} and~\ref{M-thm:unifhom} does not close. The natural first corrector $Y_0=-q(D)^{-1}\tilde V_{\rm high}$ has a gradient of order $\sqrt{\log(1/\eps)}$ in $L^\infty$, and its energy $|\nabla Y_0|_B^2$ has a mean $c_\eps$ that collects all frequencies between $1$ and $1/\eps$. A hierarchy of correctors would require $\nabla Y_0$ to be small pointwise. Following \cite{S:HairerLabbe15}, the first corrector is instead placed in an exponent. With $Y=\tilde\rho Y_0$, $\tilde\rho=\theta\rho$, and $X=A^{-1}Y$, whose values commute with each other and with $A$,
\[
e^{-X}Aq(D)e^{X}=Aq(D)-2(B\nabla Y)\cdot\nabla+q(D)Y-A^{-1}|\nabla Y|_B^2 .
\]
This identity has three effects (Lemma~\ref{M-lem:conj}(i)). The term $q(D)Y$ contains $-\tilde\rho\tilde V_{\rm high}$, and cancels the singular part of the potential exactly. The quadratic term contains $-A^{-1}\tilde\rho^2c_\eps=\tilde\rho^2M_\eps$, which, with the background $(1-\theta^2)\rho^2M_\eps$, gives the homogenized operator $P+\rho^2M_\eps$, and $M_\eps\to M_{\rm crit}$ (Lemma~\ref{M-lem:ceps}). What remains is a first-order operator $\cA$ whose coefficients involve $\nabla Y_0$ and the fluctuation $\cW=|\nabla Y_0|_B^2-c_\eps$, which are small in $C^{-\kappa}$ but not in $L^\infty$, and $Y_0$, $\tilde V_{\rm low}$ and $e^{\pm X}-1$, which are small in $L^\infty$. All of them are controlled by the local size $\delta_Q$. The choice of $X$ is forced. Only the term $A\,(q(D)X)$ of the conjugation can absorb a potential, so that the cancellation of $\tilde\rho\tilde V_{\rm high}$ requires $A\,q(D)X=-\tilde\rho\tilde V_{\rm high}$ up to the commutator with $\tilde\rho$, and we obtain $X=A^{-1}\tilde\rho Y_0$, the solution of the cell problem of the principal part with the singular potential as source, i.e., the critical analogue of the first corrector $\chi=-P_k(D_y)^{-1}W$, with $\chi\leftrightarrow A^{-1}Y_0$. Since the potential is scalar, $X$ must be a scalar multiple of $A^{-1}$. Its values then commute with each other and with $A$, which is why [C1] requires the form $A\,q(D)$ and scalar disorder, and for matrix-valued disorder no such $X$ exists (section~\ref{sec:Sdiscussion}). To first order, $e^X\approx1+X$, so that the transform is this corrector resummed into an exponent. The linear ansatz $\psi^0+X\psi^0$ would leave the product $\tilde V_{\rm high}X\psi^0$, too singular to control. The exponential instead collects all quadratic terms into $-A^{-1}|\nabla Y|_B^2$, whose mean is the logarithmically divergent Born term, made finite by $\lambda_\eps^2$, and whose fluctuation $\cW$ is small in $C^{-\kappa}$, consistently with $\tilde V_{\rm high}Y_0=-(q(D)Y_0)Y_0$, equal to $|\nabla Y_0|_B^2$ up to a divergence. In the periodic and subcritical regimes, where the corrector has bounded derivatives, the linear ansatz suffices. Only the high frequencies are removed, since $q(D)^{-1}$ is unbounded at $\xi=0$ while $\tilde V_{\rm low}$ is bounded, and the cutoff $\tilde\rho=\theta\rho$ confines the transform to good boxes, where $Y_0$ is small and $e^{\pm X}$ is close to the identity. It requires no positivity, and the direction $-A^{-1}$ of the shift is set by the principal symbol alone. This is the mechanism behind the topological transition in the BHZ model.

\paragraph{Outline of the proof of Theorem~\ref{M-thm:det}.} By the paraproduct bounds of Lemma~\ref{M-lem:prod}, $\cA$ is small, of order $\delta_\theta$, as a map from $H^{2-2\kappa}$ to $H^{-2\kappa}$ (Lemma~\ref{M-lem:conj}(ii)), while $P+\rho^2M_\eps$ is invertible between these spaces, by the elliptic estimate, duality and interpolation. A Neumann series inverts the conjugated operator, and undoing the conjugation costs only $e^{\pm X}=1+O(\delta_\theta)$ in $L^\infty$. The proof thus never needs $\nabla Y_0$ to be small pointwise, only in negative norms, a smallness that holds with high probability on each unit box (Lemma~\ref{M-lem:prob}). The renormalization constant then replaces $M_\eps$ by $M_{\rm crit}$ up to $O(1/\log(1/\eps))$. Since the local size is large somewhere a.s.\ (Lemma~\ref{lem:globalsize}), the theorem is applied to $H_\theta$, with $\theta$ vanishing near the bad unit boxes, and the robustness theorem handles the rest.

\paragraph{Large local size on $\Rm^2$.} With $\theta\equiv1$, Theorem~\ref{M-thm:det} of the main paper requires $\delta=\sup_Q\delta_Q\le\delta_*$, with $\delta_Q$ the local size \eqref{M-eq:deltaQ}, a condition on all unit boxes. The following lemma shows that it never holds on the whole plane, which is why rare regions must be removed in Theorem~\ref{M-thm:mainA}.
\begin{lemma}\label{lem:globalsize}
Assume [C1]--[C3] of the main paper. For $\eps$ small, a.s., $\delta\ge1$. More precisely, $\delta_Q\ge1$ for infinitely many unit boxes $Q$.
\end{lemma}
\begin{proof}
Recall that $\int w\ne0$, since $s_0>0$. The law is centered and non-degenerate, so its support contains a point $a\ne0$ with $a\int w>0$. Let $\eta_1>0$ and $R_1\ge1$, and consider the event $\mathcal E_Q$ that $|\omega_m-a|\le\eta_1$ for all $m$ with $|\eps m-c_Q|\le R_1$. Since $\sum_mw(\cdot/\eps-m)=\int w+W_{\rm per}(\cdot/\eps)$ with $W_{\rm per}$ periodic and mean zero, and the kernel $K_{\rm low}$ of $\chi_0(D)$ is a Schwartz function with $\int K_{\rm low}=\chi_0(0)=1$,
\[
\tilde V_{\rm low}(c_Q)=\lambda_\eps\eps^{-1}\Bigl(a\int w+(K_{\rm low}*W_{\rm per}(\cdot/\eps))(c_Q)+O\bigl(\eta_1+\kappa_+\textstyle\int_{|z|\ge R_1/2}|K_{\rm low}|\bigr)\Bigr)
\]
on $\mathcal E_Q$. The convolution term vanishes for $\eps$ small, because $\hat K_{\rm low}=\chi_0$ has compact support and $W_{\rm per}(\cdot/\eps)$ has no Fourier modes below $2\pi/\eps$. Choosing $\eta_1$ small and $R_1$ large, $\tilde V_{\rm low}(c_Q)\ge\frac12\lambda_\eps\eps^{-1}a\int w\ge1$ on $\mathcal E_Q$ for $\eps$ small. The probability of $\mathcal E_Q$ is positive and bounded below uniformly in $Q$, and the events $\mathcal E_Q$ for boxes at mutual distance $>2R_1$ are independent. By the second Borel--Cantelli lemma, a.s.\ $\mathcal E_Q$ occurs for infinitely many $Q$, and then $\delta\ge\delta_Q\ge|\tilde V_{\rm low}(c_Q)|\ge1$.
\end{proof}

\subsection{Probabilistic estimates}\label{sec:Stools-prob}

All probabilistic arguments of the paper serve one purpose: to show that the regions where the homogenized description fails are rare and sparse, so that deformed cuts can avoid them.

\paragraph{Concentration.} Linear functionals of independent bounded couplings satisfy Hoeffding's inequality \cite[Theorem~2.8]{S:BLM13}, and quadratic forms the Hanson--Wright inequality \cite[Theorem~1.1]{S:RV13}, with Gaussian tails governed by the Hilbert--Schmidt norm of the form for moderate deviations and exponential tails governed by its operator norm for large ones. For the local effective matrix at mesoscale $\varpi$, the operator norm of the form is $O(\varpi^2)$ and its Hilbert--Schmidt norm $O(\varpi)$, and deviations of size $t$ have probability $\le Ce^{-ct^2\varpi^{-2}}$ (Lemma~\ref{M-lem:mean}(ii)). Suprema over boxes are controlled by union bounds over nets.

\paragraph{Estimates at criticality.} Lemma~\ref{M-lem:cuts} needs bad unit boxes of probability $O(\eps^{p_0})$ with $p_0>0$. Since $\lambda_\eps^2=\ell/\log(1/\eps)$, this amounts to tails $\exp(-ct^2/\lambda_\eps^2)=\eps^{ct^2/\ell}$ for $\{\delta_Q\ge t\}$, uniformly in $\eps$: a deviation of fixed size of a field of amplitude $\lambda_\eps$ has a cost that is Gaussian in $\lambda_\eps$, hence only polynomial in $\eps$, which is why bad regions are most frequent at criticality. With $t=\delta_*/4$, this gives $p_0=\min(c\delta_*^2/\ell,1)$ (Lemma~\ref{M-lem:probboxA}). Such tails, rather than the moment bounds that suffice for Gaussian noise \cite{S:HairerLabbe15}, are proved in Lemma~\ref{M-lem:prob}. Three features make this delicate. First, $\delta_Q$ involves $C^{-\kappa}$ norms, i.e., suprema over all dyadic scales $j$ and all points of the plane, of fields that depend on all couplings. Each block $\Delta_j(\psi_Q\cdot)$ is band-limited, so its supremum is controlled by its values on a lattice of mesh of order $2^{-j}$ (Lemma~\ref{M-lem:sampling}). The localization by $\psi_Q$ makes the variances decay away from $Q$, which makes the union bound over the infinite lattice summable, and the factor $2^{\kappa j}$ in the norm makes the sum over $j$ converge. Second, the renormalized square $\cW$ is quadratic in the couplings, and its blocks are controlled by the Hanson--Wright inequality. Their Hilbert--Schmidt norms grow like $\lambda_\eps^2\sqrt j$, through the logarithmic convolution of $\avg\xi^{-2}$ with itself in $d=2$, and this growth is again absorbed by $2^{\kappa j}$ (Lemma~\ref{M-lem:blocks}). The deterministic part $\EE|\nabla Y_0|_B^2-c_\eps$ is periodic at scale $\eps$ with mean zero, hence of size $\eps^\kappa$ in $C^{-\kappa}$. Third, percolation needs finite range: the fields are truncated to the couplings within distance $r$ of $Q$, at a deterministic cost $\eps^{10}$, and the gradient bound $\|\nabla Y_0\|_{L^\infty}\le C_1\sqrt{\log(1/\eps)}$ is proved by sampling at frequencies below $1/\eps$, the higher frequencies being bounded deterministically (Lemma~\ref{M-lem:probboxA}).

\paragraph{Rare regions.} The probability that a given macroscopic unit region is bad depends on the mechanism generating the rare regions. It is of order $\exp(-c\eps^{-1/2})$ for mesoscale deviations of the local effective matrices from the Born matrix \eqref{eq:SMborn} in the subcritical first-order case, of order $\exp(-c\eps^{-1})$ for macroscopic mean fields when $k=1$ and $\int w\not=0$, of order $\exp(-c\eps^{-2})$ for gapless periodic patterns of macroscopic size, and only of order $\eps^{p_0}$ at criticality. In all cases, bad regions occur a.s.\ on $\Rm^2$ while bad boxes are so rare that their clusters are a.s.\ finite, which is all the deformed-cut construction requires.

\paragraph{Percolation.} If sites are bad with small probability, independently beyond a fixed range, then a given set of $n$ sites contains a fixed fraction of mutually distant sites, and is entirely bad with probability exponentially small in $n$. Since the number of $*$-connected sets of $n$ sites containing the origin is at most $64^n$ \cite{S:Grimmett99}, bad clusters are then finite a.s., with exponentially small probability of being large. Percolation needs bad events that depend on finitely many couplings. The fields that define bad regions depend on all couplings through decaying kernels, and are therefore replaced by truncated versions that differ from them by a negligible deterministic amount.

\subsection{Deformed cuts and the proofs of Theorems~\ref{M-thm:mainB} and~\ref{M-thm:mainA}}\label{sec:Stools-perc}

\paragraph{Why deformed cuts.} In a genuine gap, we compute the index along the straight cut $\{x_1=0\}$, as for periodic fluctuations. In a pseudo-gap, rare regions in which the local effective matrix lies in another phase, such as puddles of the phase $P$ inside the region of phase $P+M$, are bounded by interfaces that carry states at energies in the gap, and a straight cut meets infinitely many of them almost surely. The robustness argument, which needs all bad regions to stay far from the cut, then fails, and the straight-cut index is no longer controlled. The cut is therefore chosen for each configuration, so as to avoid the bad regions at a large distance $\dd$ (Figure~\ref{M-fig:cuts} of the main paper and Figure~\ref{fig:scales}(b)). Bad regions occur a.s.\ but are sparse (section~\ref{sec:Stools-prob}). Two results make this work: the index along an admissible cut does not depend on the cut, nor on bounded, possibly large, perturbations supported away from it (Theorem~\ref{M-thm:robust} and Remark~\ref{rem:awaycut}), and admissible cuts exist almost surely (Lemma~\ref{M-lem:cuts}). Deformed cuts are needed in general, as Corollary~\ref{M-cor:puddles} of the main paper illustrates for the Dirac operator with puddles of the opposite mass.

\paragraph{Scales.} All lengths in Lemma~\ref{M-lem:cuts}, Theorem~\ref{M-thm:robust} and Theorems~\ref{M-thm:mainB}--\ref{M-thm:mainA} of the main paper are measured in the macroscopic variable $x$, in which the couplings sit at the points $\eps m$. The relevant scales are
\begin{center}
\begin{tabular}{lcc}
\toprule
object & size in $x$ & size in fast cells\\
\midrule
fast cell (one coupling $\omega_m$) & $\eps$ & $1$\\
bad mesobox (Theorem~\ref{M-thm:mainB}) & $\eps\varpi^{-1}=\eps^{3/4}$ & $\eps^{-1/4}$\\
unit box (bad box in Theorem~\ref{M-thm:mainA}) & $1$ & $\eps^{-1}$\\
surgery radius $r$, range $2r$ of bad events & $\eps^{-\nu_r}$ & $\eps^{-1-\nu_r}$\\
distance $\dd$ of the cut & $\eps^{-\nu_\dd}$ & $\eps^{-1-\nu_\dd}$\\
typical spacing of bad unit boxes & $\gtrsim\eps^{-p_0/2}$ & $\gtrsim\eps^{-1-p_0/2}$\\
\bottomrule
\end{tabular}
\end{center}
Here $\eps^{p_0}$ bounds the probability that a given unit box meets the bad set, i.e., the density of bad unit boxes per unit area in $x$, so that their typical spacing is $\eps^{-p_0/2}$; $p_0=1$ in Theorem~\ref{M-thm:mainB}, and $p_0=\min(c\delta_*^2/\ell,1)$ in Theorem~\ref{M-thm:mainA} (Lemma~\ref{M-lem:probboxA}). Bad sets have size at most $1$ in $x$, while $r$ and $\dd$ tend to infinity. The surgery and the cut thus act on super-macroscopic scales, around bad regions that are sparse, since $\nu_\dd<p_0/2$. The distance $\dd$ must be large in $x$ because the decay [R1] of the resolvent, with rate of order one, is measured in $x$, and the robustness error $C_nK_V^3(a\dd/2)^{-n}$ of Theorem~\ref{M-thm:robust} must absorb $K_V^3$, with $K_V\sim\eps^{-1/2}$ in Theorem~\ref{M-thm:mainB} and $K_V\sim\eps^{-1}$ in Theorem~\ref{M-thm:mainA}. Here $n\ge1$ is arbitrary, with a constant $C_n$ depending on $n$; outside Theorems~\ref{M-thm:mainB} and~\ref{M-thm:mainA}, $n=1$ suffices. The resolvent of the reference operator decays across the cut like $e^{-a|\Im z|\dd}$, at a rate that vanishes on the real axis, and the Helffer--Sj\"ostrand formula turns this decay into the polynomial decay $(a\dd)^{-n}$ of $S(H)$, for every $n$, because $\dbar\tilde S$ vanishes to every order on the real axis (Lemma~\ref{M-lem:HSdecay}). The factor $K_V^3$ comes from the resolvent expansion in Step~2 of the proof of Theorem~\ref{M-thm:robust}, whose remainder contains the defect three times. Since $\dd=\eps^{-\nu_\dd}$, it suffices to fix $n$ with $n\nu_\dd>\frac32$ in Theorem~\ref{M-thm:mainB}, and $n\nu_\dd>3$ in Theorem~\ref{M-thm:mainA}, where $a=c/\sqrt{\log(1/\eps)}$ only adds a logarithmic factor. Figure~\ref{fig:scales} shows these scales and the construction of the cut.

\begin{figure}[ht]
\centering\includegraphics[width=\textwidth]{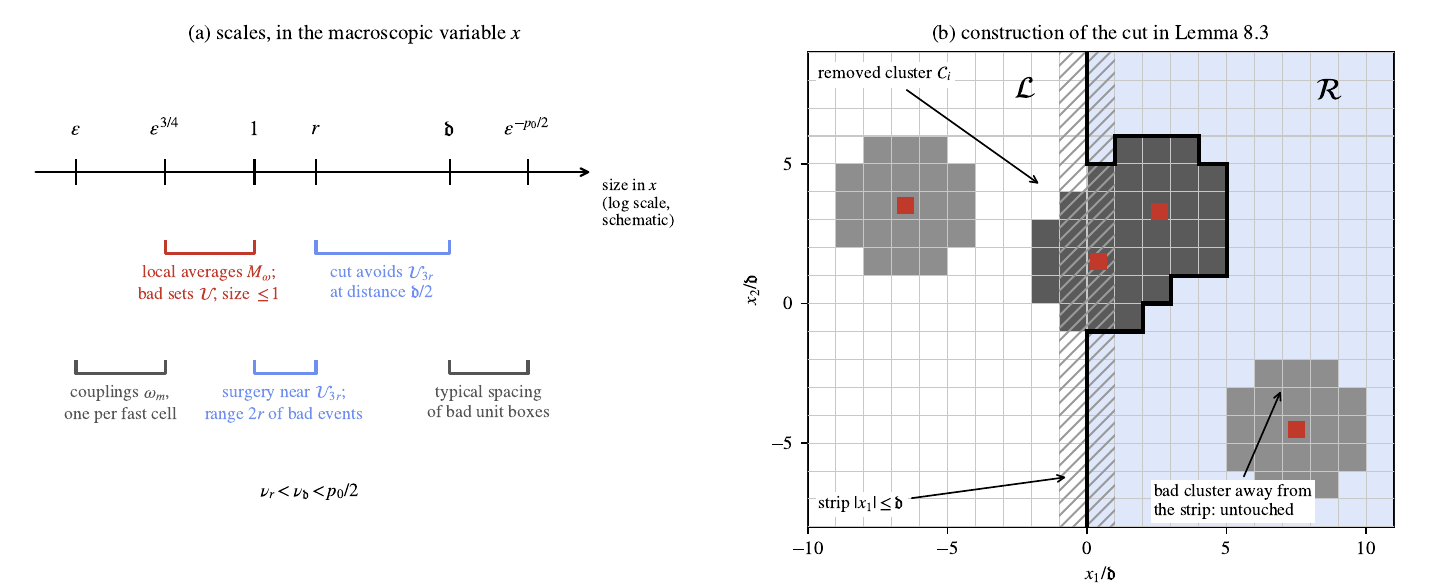}
\caption{(a) Scales in the macroscopic variable $x$ (schematic logarithmic axis). (b) Construction of the cut in Lemma~\ref{M-lem:cuts}, in units of the coarse box side $\dd$. Red: bad unit boxes, enlarged. A coarse box is bad if a unit box within distance $2\dd$ of it meets the bad set. Bad coarse boxes are dark grey if their $*$-cluster meets the strip $|x_1|\le\dd$ (hatched) and medium grey otherwise. The cluster meeting the strip is removed from $\cR$. Each bad unit box lies at distance at least $2\dd$ from the complement of its cluster. The thick line is the boundary of $\cR$ (blue): the straight cut, except for a bounded detour around the removed cluster. Bad clusters away from the strip are left on their side.}\label{fig:scales}
\end{figure}

\paragraph{Outline of the proof of Lemma~\ref{M-lem:cuts}.} The argument is a coarse-graining of a standard type \cite{S:Grimmett99}. We tile the plane by coarse boxes of side $\dd$, much larger than the range of dependence, and a coarse box is called bad if a bad region lies nearby. Bad coarse boxes are rare and finitely dependent, so that their clusters are finite. Starting from the half-plane $\{x_1\ge0\}$, the clusters that meet a vertical strip around the straight cut are removed. A bad region makes every coarse box within distance $2\dd$ of it bad, so that it lies deep inside its cluster. Every bad region is therefore far from the other side of the cut, and the exponential tails of the cluster sizes show that only finitely many clusters are removed in each horizontal strip.

\paragraph{Outline of the proof of Theorem~\ref{M-thm:robust}.} We compare the resolvent of a perturbation $V$ supported on both sides of the cut, away from it, with the resolvents of the perturbations supported on one side only. The parts of $V$ on the two sides interact only through the resolvent of the reference operator across a distance at least $\dd$, which is small by the decay [R1]. Each one-sided perturbation almost commutes with the projection on one side of the cut. As a result, the Toeplitz operator of the perturbed operator factors, up to compact and small operators, as the Toeplitz operator of the reference operator times an operator of index zero. Only the decay of the resolvent of the reference operator is used. That of the perturbed operator may be poorly behaved, since the perturbation is large.

\paragraph{Outline of the proof of Theorem~\ref{M-thm:mainB}.} Mesoscopic boxes in which the local effective matrix deviates from the Born matrix are bad. By concentration, a unit box meets one with probability $O(\eps)$, and admissible cuts exist by Lemma~\ref{M-lem:cuts} (Figure~\ref{fig:alloy}). A reference operator is built by surgery: near the bad regions, the fluctuation is replaced by its homogenized limit $M_{\rm Born}$, without modifying the configuration (Figure~\ref{fig:surgery}). Theorem~\ref{M-thm:unifhom} applies to the reference operator, whose local effective matrices are all close to $M_{\rm Born}$. It is therefore gapped, and its index is $\BDI(P+M_{\rm Born},P)$. The difference between $H^\eps$ and the reference operator is a large potential supported near the bad regions, away from the cut, and Theorem~\ref{M-thm:robust} transfers the index.

\begin{figure}[ht]
\centering\includegraphics[width=0.8\textwidth]{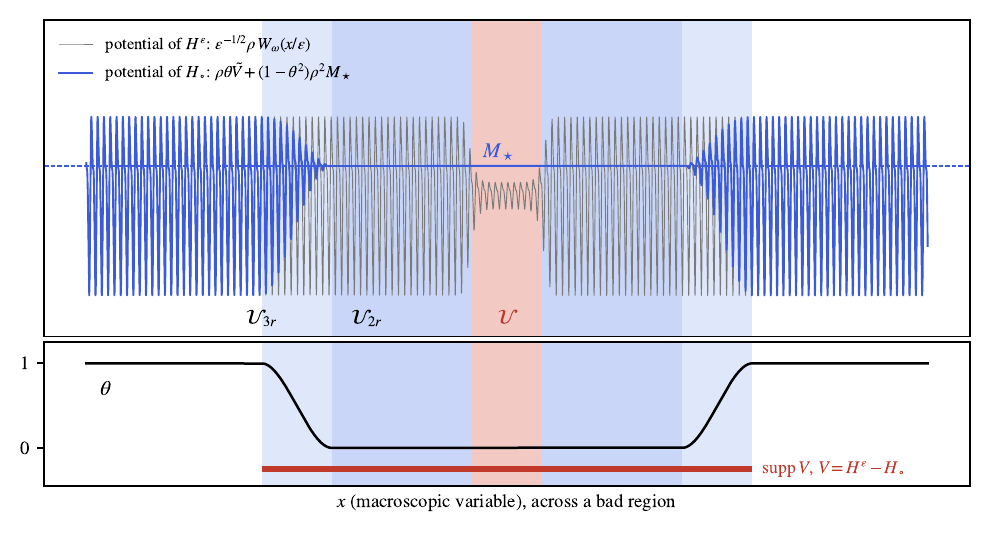}
\caption{Surgery near a bad region, along a cross-section in the macroscopic variable $x$ (schematic). The bad set $\mathcal U$ (red) is here a puddle where the couplings are small. Outside $\mathcal U_{3r}$ the potentials of $H^\eps$ (grey) and of the reference operator $H_\circ$ (blue) coincide. Inside $\mathcal U_{2r}$ the fluctuation is replaced by its homogenized limit $M_\star$ ($M_{\rm Born}$ in Theorem~\ref{M-thm:mainB}, $M_\eps$ in Theorem~\ref{M-thm:mainA}), and the cutoff $\theta$ interpolates in between. The configuration is not modified. The defect $V=H^\eps-H_\circ$ is supported in $\mathcal U_{3r}$. It is large, but it is kept at distance $\dd/2$ from the other side of the cut.}\label{fig:surgery}
\end{figure}

\paragraph{Outline of the proof of Theorem~\ref{M-thm:mainA}.} We follow the structure of Theorem~\ref{M-thm:mainB}, with bad unit boxes, defined by a large local size or a large gradient $\nabla Y_0$, in place of bad mesoscopic boxes, and Theorem~\ref{M-thm:det} in place of Theorem~\ref{M-thm:unifhom}. Theorem~\ref{M-thm:det}, like Theorem~\ref{M-thm:unifhom}, holds up to the real axis with a factor $1+|\Im z|^{-2}$, which the almost-analytic extension absorbs, so that the index of the reference operator follows directly from the Helffer--Sj\"ostrand formula, as in Theorem~\ref{M-thm:TAI}. One point differs: the decay [R1] of the reference operator, a second-order operator, follows from an $H^1$ resolvent bound of order $\sqrt{\log(1/\eps)}$, which degenerates slowly enough to be compensated by the polynomial smallness of the bad regions. The constants are fixed in the following order: the smallness threshold $\delta_*$ of the local size, which determines the probability exponent $p_0$ of bad boxes, then the exponents of the cuts, and finally $\eps_0$.

\section{Examples and numerical illustrations}\label{sec:Sexamples}

This section illustrates the formulas of section~\ref{sec:Ssetting} on first-order Dirac operators, regularized Dirac and BHZ operators with electromagnetic and matrix-valued fluctuations, gap-preserving alloys, and disordered superconductors, and compares the convergence observed numerically with \eqref{eq:Shom}. Unless stated otherwise, the periodic profiles depend on $y_1$ only.

\paragraph{First-order Dirac operators.} Let $P=D\cdot\sigma+m\sigma_3$ and $W=V\sigma_0+A_1\sigma_1+A_2\sigma_2$, with stream function $\psi$ of $A$, so that $M=2\avg{V\psi}\sigma_3$ by \eqref{eq:SEM}. The lower bulk $D\cdot\sigma+(m+2\avg{V\psi})\sigma_3$ is gapped when $m+2\avg{V\psi}\ne0$, and \eqref{eq:SBDI} with the Dirac value of $c$ gives, for $\eps$ small,
\begin{equation}\label{eq:SindDirac}
\Ind T(H^\eps)=\tfrac12\bigl(\sgn(m+2\avg{V\psi})-\sgn m\bigr).
\end{equation}
For $V(y)=a\cos(2\pi y_1)$ and $A(y)=(0,\,b\cos(2\pi y_1+\theta_0))$, the stream function is $\psi=b\sin(2\pi y_1+\theta_0)/(2\pi)$, and
\begin{equation}\label{eq:SMex1}
M=\frac{ab\sin\theta_0}{2\pi}\,\sigma_3 .
\end{equation}
Electric and magnetic fluctuations in phase produce no mass. For $m>0$, the index is $-1$ when $ab\sin\theta_0<-2\pi m$ and $0$ otherwise. Figure~\ref{fig:first}(b) shows this index map for $m=\frac12$. Figure~\ref{fig:first}(a) shows the interface spectrum for $a=\frac32$, $b=2\pi$ and $\theta_0=-\frac\pi2$, for which $M=-\frac32\sigma_3$, $m+M_3=-1$ and the index is $-1$. A single branch crosses the common gap $(-\frac12,\frac12)$, for $H^0$ and for $H^\eps$ with $\eps=\frac1{16}$. The computation uses a box periodic in $x_2$, which contains a second, fictitious interface. The states localized near it are discarded.

\begin{figure}[ht]
\centering\includegraphics[width=\textwidth]{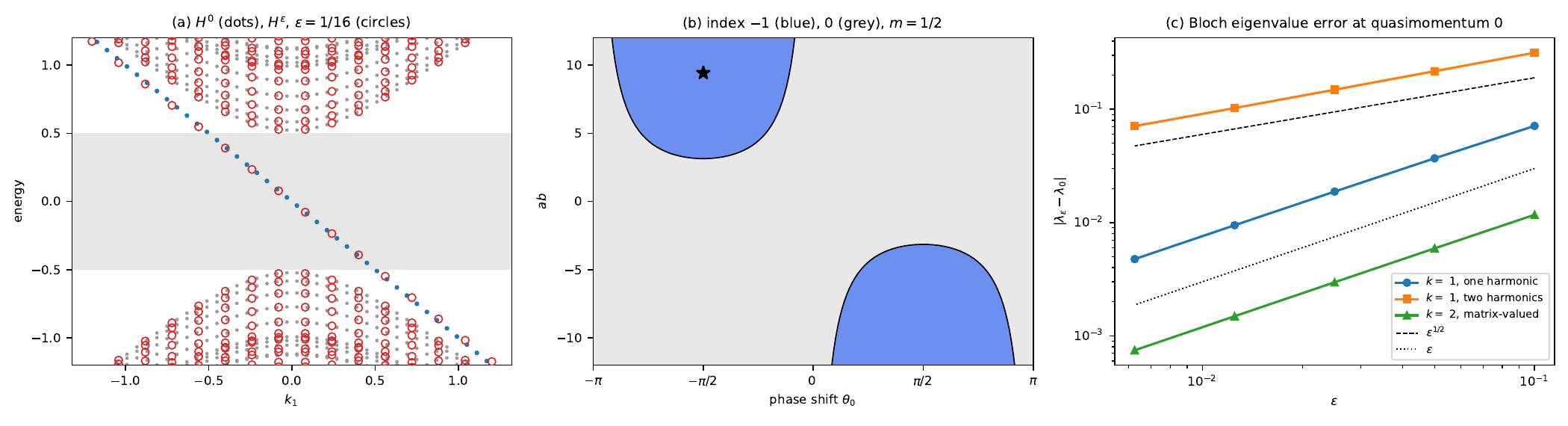}
\caption{First-order Dirac operator. (a) Interface spectrum as a function of the momentum $k_1$ along the interface: $H^0$ (dots, with interface states in blue) and $H^\eps$, $\eps=\frac1{16}$ (circles). The shaded band is the common gap. (b) Interface index \eqref{eq:SindDirac} as a function of the phase shift $\theta_0$ and of $ab$, for $m=\frac12$. The star marks the example of panel (a). (c) Error of the bulk Bloch eigenvalue closest to $0$, at quasimomentum $0$, for the examples of the paragraph \emph{Convergence of Bloch eigenvalues} of this section. This is not the norm-resolvent error of \eqref{eq:Shom}. See Remark~\ref{rem:sharp}.}\label{fig:first}
\end{figure}

\emph{Comparison with the Haldane model.} The mass $2\avg{V\psi}$ is odd under time reversal, which maps $A\mapsto-A$, hence $\psi\mapsto-\psi$, and fixes $V$. Moreover, the magnetic field $\partial_1A_2-\partial_2A_1=\Delta\psi$ has zero mean, so the fluctuation carries no net flux. This is the mechanism of the Haldane model \cite{S:Haldane88}, in which a staggered magnetic flux with zero average opens a time-reversal-breaking gap at a Dirac point. Here the flux pattern acts through its correlation with the electric potential.

\emph{General zeroth-order fluctuations.} Add a mass fluctuation, $W=V\sigma_0+A\cdot\sigma+V_3\sigma_3$ with $A=\nabla^\perp\psi+\nabla\zeta$, and let $-\Delta\phi_3=V_3$. Then
\begin{equation}\label{eq:SMgen}
M=2\avg{V\psi}\,\sigma_3+2\avg{V_3\psi}\,\sigma_0+M_\perp\cdot\sigma,\qquad M_\perp=2\avg{V\,\nabla^\perp\phi_3}.
\end{equation}
Indeed, fix $q$, write $\gamma_q=\hat q\cdot\sigma$, $\gamma_\perp=\hat q^\perp\cdot\sigma$, and $\hat W_q=v+s\gamma_q+t\gamma_\perp+v_3\sigma_3$, with $v=\hat V_q$, $v_3=(\hat V_3)_q$, $t=2\pi i|q|\hat\psi_q$, and $s=2\pi i|q|\hat\zeta_q$ the coefficient of the gauge part $\nabla\zeta$ of $A$. Here $G_1(\hat q)=-\gamma_q$, and by \eqref{eq:Spair} the pair $\pm q$ contributes $-(2\pi|q|)^{-1}(X^*\gamma_qX-X\gamma_qX^*)$ with $X=\hat W_q$. In \eqref{eq:Sbilin} with $(A_a)=(\sigma_0,\gamma_q,\gamma_\perp,\sigma_3)$ and $G=\gamma_q$, the only non-zero brackets are $[\gamma_q,\gamma_\perp]=2i\sigma_3$, $\gamma_q\sigma_3-\sigma_3\gamma_q=-2i\gamma_\perp$ and $\gamma_\perp\gamma_q\sigma_3-\sigma_3\gamma_q\gamma_\perp=-2i$, so that $s$ drops out and the pair contributes
\[
4\Re(\bar v\hat\psi_q)\,\sigma_3+4\Re(\overline{\hat\psi_q}v_3)\,\sigma_0-\frac{2}{\pi|q|}\Im(\bar vv_3)\,\gamma_\perp .
\]
Unfolding to all $q\ne0$ and using Parseval gives \eqref{eq:SMgen}, since $(\nabla^\perp\phi_3)^\wedge_q=iq^\perp v_3/(2\pi|q|^2)$. For the Dirac operator, $M_\perp$ is removed in the lower bulk by conjugation with the multiplication operator $e^{iM_\perp\cdot x}$, and $c$ in \eqref{eq:Sc} is invariant under the corresponding translation of $\xi$ by $M_\perp$. Hence only electric-magnetic correlations change the mass, and hence the topology. Mass-magnetic correlations shift the energy by $2\avg{V_3\psi}$, and close the gap at energy $0$ once $|2\avg{V_3\psi}|\ge|m+2\avg{V\psi}|$. Electric-mass correlations produce a pure gauge. Mass fluctuations alone give $M=0$ (Remark~\ref{M-rem:scalar} of the main paper). Indeed, for $W=f(y)B$ with $f$ scalar, both terms of \eqref{eq:Spair} equal $|\hat f_q|^2BG_1(\hat q)B$ and cancel for odd $k$.

\begin{remark}[Purely magnetic fluctuations]\label{rem:AC}
For periodic $a=\nabla^\perp\Psi+\nabla X$ of mean zero, the Aharonov--Casher zero modes $(e^{\Psi-iX},0)$ and $(0,e^{-\Psi-iX})$ of $(D+a)\cdot\sigma$ \cite{S:AharonovCasher79} are periodic Bloch eigenfunctions, at quasimomentum $0$, of $(D+a)\cdot\sigma+m\sigma_3$, with eigenvalues $\pm m$. Hence $\pm m$ belong to the spectrum at every amplitude, and the mass is pinned, in agreement with $M=0$ for $V=0$ in \eqref{eq:SEM} and with Theorem~\ref{M-thm:magno} of the main paper.
\end{remark}

\paragraph{Regularized Dirac and BHZ operators.} For the regularized Dirac operator and the electromagnetic fluctuation, \eqref{eq:SregDirac} gives $M=M_3\sigma_3$ with $M_3=\mathcal C_{00}/\beta-\beta\int|\nabla\psi|^2$, so that $M_\perp=M_0=0$ and, by \eqref{eq:Smeff}, the index is $\frac12(\sgn(m+M_3)-\sgn m)$. For $\beta>0$, the index is therefore $1$ when $m<0$ and $\mathcal C_{00}/\beta-\beta\int|\nabla\psi|^2>|m|$, and $-1$ when $m>0$ and $\beta\int|\nabla\psi|^2-\mathcal C_{00}/\beta>m$: electric fluctuations invert a trivial mass, and magnetic fluctuations a topological one. For the profiles $V$ and $A$ of \eqref{eq:SMex1}, $\mathcal C_{00}=a^2/(8\pi^2)$ and $\int|\nabla\psi|^2=b^2/2$. The phase $\theta_0$ plays no role.

\emph{BHZ operator: comparison with \cite{S:GWATB09}.} For the BHZ symbol $\xi\cdot\sigma+(m+b|\xi|^2)\sigma_3+c|\xi|^2$ with $|c|<|b|$ and scalar $W=V\sigma_0$, \eqref{eq:SBHZ} shows that the homogenized bulk has mass $\bar m=m-\mathcal E b/(b^2-c^2)$ and is shifted in energy by $\mathcal E c/(b^2-c^2)$, with $\mathcal E=\mathcal C_{00}=\int_{\T^2}|\nabla\phi_0|^2$. Scalar terms do not change spectral projectors as long as the energy $E$ lies in the gap. Hence, by \eqref{eq:SBDI} and \eqref{eq:Sdegree}, with $(d_1,d_2)=\xi$ and $\sgn d_3(\infty)=\sgn b$, the interface index at an energy $E$ in the common gap of $P$ and $P+M$ is $\frac12(\sgn\bar m-\sgn m)$.

Figure~\ref{fig:third}(a) shows the resulting phase diagram for $m=0.1$, $b=1$ and $c=\frac34$. The ratio $c/b$ is close to the ratio $\gamma/\beta=512/686$ of the HgTe parameters of \cite{S:GWATB09}. It has the structure of their phase diagram: line A marks the inversion of the mass, at $\mathcal E=m(b^2-c^2)/b$ (their curve A). The solid and dashed lines are the band edges of $P+M$ (their curve B) and $P$, which move with the energy shift. The topological region lies beyond A, between these band edges. Two differences should be noted. Here $\mathcal E$ is a finite periodic quantity, whereas in \cite{S:GWATB09} it is replaced by a logarithmically divergent lattice integral. In the main paper, the logarithm is absorbed in the renormalization $\lambda_\eps$ of \eqref{eq:Scrit}. Moreover, the band edges here are edges of a genuine spectral gap. Random on-site disorder is critical, and for it they become pseudo-gap edges. The corresponding index is defined along deformed cuts and given by the same formula, with the renormalized mass $\bar m=m-\ell s_0b/(2\pi(b^2-c^2))$ of Corollary~\ref{M-cor:BHZ} of the main paper.

With a magnetic field, \eqref{eq:SBHZ} adds $(b\sigma_3+c)\int|\nabla\psi|^2$ to $M$. This raises the mass by $b\int|\nabla\psi|^2$ and shifts the energy by $c\int|\nabla\psi|^2$. The mass-inversion line A moves to $\mathcal E=(m+b\int|\nabla\psi|^2)(b^2-c^2)/b$, and the topological region shrinks, as shown in Figure~\ref{fig:third}(b).

\begin{figure}[ht]
\centering\includegraphics[width=\textwidth]{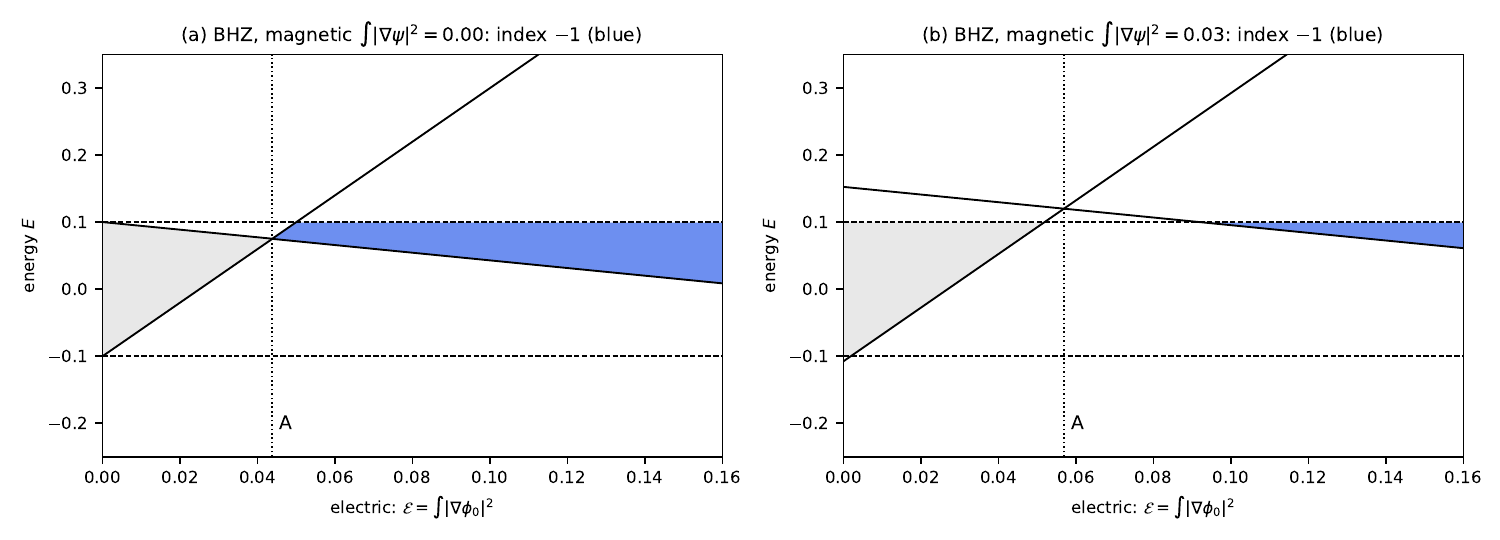}
\caption{BHZ operator with electromagnetic fluctuations, $m=0.1$, $b=1$, $c=\frac34$, (a) without and (b) with a magnetic field: index $-1$ (blue) or $0$ (grey) where the energy $E$ lies in the common gap. It is undefined (white) elsewhere. A: mass inversion. Solid: band edges of $P+M$. Dashed: band edges of $P$.}\label{fig:third}
\end{figure}

\paragraph{Matrix-valued zeroth-order potentials.} Let $P=D\cdot\sigma+(m+\beta\Delta)\sigma_3$, and let $W=\lambda v\,(c_0\sigma_0+c_1\sigma_1+c_3\sigma_3)$ with $c_0^2+c_1^2+c_3^2=1$ and $v$ a fixed real mean-zero profile, $-\Delta\phi_v=v$. Set $s=\lambda^2\int_{\T^2}|\nabla\phi_v|^2$. Then $\mathcal C_{ab}=s\,c_ac_b$, and \eqref{eq:SregDirac} gives
\[
M=\frac s\beta\Bigl[(c_0^2+c_3^2-c_1^2)\,\sigma_3+2c_0c_3\,\sigma_0+2c_1c_3\,\sigma_1\Bigr].
\]
By \eqref{eq:Smeff}, the index is $\frac12(\sgn m_{\rm eff}-\sgn m)$, with
\[
m_{\rm eff}=m+\frac s\beta(c_0^2+c_3^2-c_1^2)-\frac{4s^2}{\beta}\,c_1^2c_3^2,
\]
provided the energy shift satisfies $|2sc_0c_3/\beta|<\inf_\xi|d(\xi)|$. Three mechanisms are visible. Electric and mass components raise the mass by $s(c_0^2+c_3^2)/\beta$, while in-plane components lower it by $sc_1^2/\beta$. This is the dichotomy of Examples 1 and 2 of \cite{S:BalDang24}. Correlations between mass and in-plane components produce a vector potential $M_\perp=(2sc_1c_3/\beta,0)$, which, because of the regularization, lowers $m_{\rm eff}$ by a term quadratic in $s$. Correlations between electric and mass components shift the energy by $2sc_0c_3/\beta$. They leave $m_{\rm eff}$ unchanged, but they can close the gap at energy $0$.

Figure~\ref{fig:second} shows two slices for $m=-0.3$ and $\beta=1$. In the slice $c_0=0$, $(c_3,c_1)=(\cos\alpha,\sin\alpha)$, panel (a), $m_{\rm eff}=m+s\cos2\alpha-s^2\sin^22\alpha$ and the topological region is bounded. For $0<\alpha<\pi/4$ the transition is re-entrant in $s$, because the vector-potential term is quadratic in $s$ while the mass term is linear. The dashed curve is $m+M_3=0$. Between it and the blue region, $m+M_3>0$ while $m_{\rm eff}<0$, so that the vector potential $M_\perp$ alone keeps the interface trivial. At $\alpha=\pi/4$, $m+M_3=m$ and $m_{\rm eff}=m-s^2<0$. With $m>0$ instead, this is the fluctuation of Example~\ref{M-ex:counter} of the main paper, for which the topology changes through $M_\perp$ alone. In the slice $c_1=0$, $(c_0,c_3)=(\cos\alpha,\sin\alpha)$, panel (b), $m_{\rm eff}=m+s/\beta$ does not depend on $\alpha$, but the energy shift $s\sin2\alpha/\beta$ closes the gap at $0$ away from pure electric or pure mass fluctuations.

Compared with the electromagnetic case, matrix-valued potentials have two features that gauge-covariant fields lack. Off-diagonal components lower the mass linearly in $s$ and produce the vector-potential term in $m_{\rm eff}$. Correlations with the mass component produce energy shifts that close the gap.

\begin{figure}[ht]
\centering\includegraphics[width=\textwidth]{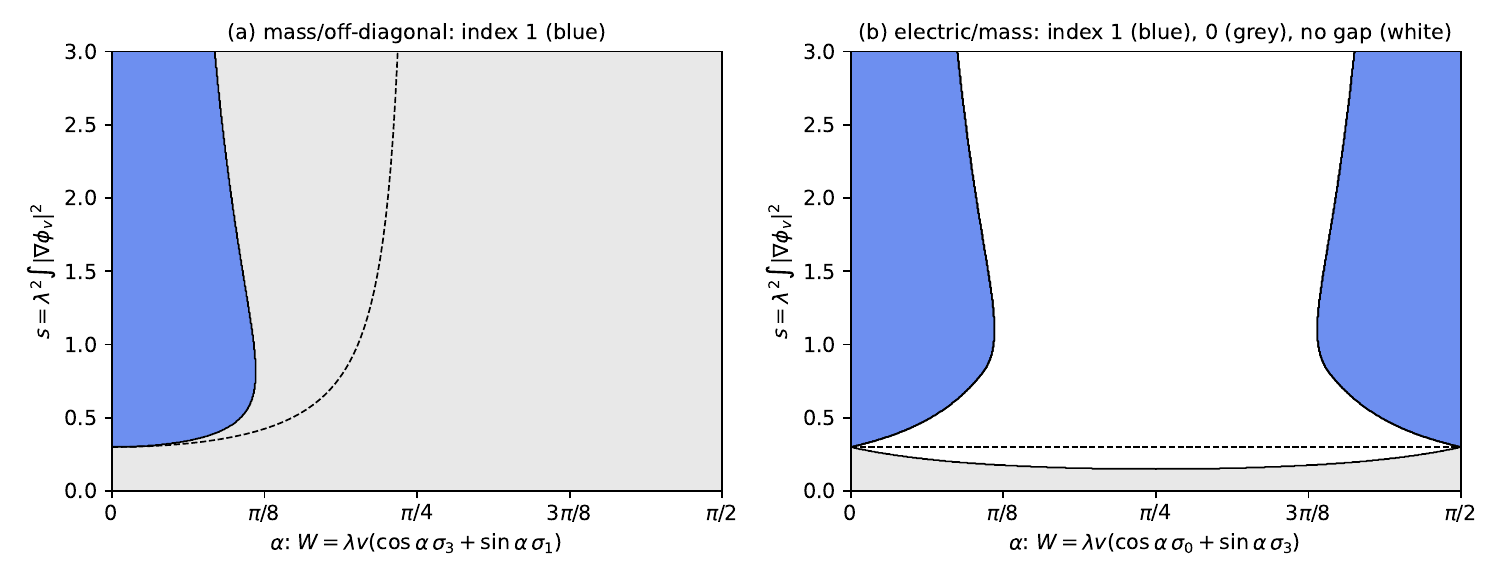}
\caption{Regularized Dirac operator with matrix-valued zeroth-order potentials (not gauge covariant), $m=-0.3$, $\beta=1$: interface index as a function of the direction $\alpha$ and of the strength $s$, at energy $0$. Index $1$ (blue), $0$ (grey), undefined because the gap closes (white). (a) $W=\lambda v(\cos\alpha\,\sigma_3+\sin\alpha\,\sigma_1)$, solid: $m_{\rm eff}=0$. (b) $W=\lambda v(\cos\alpha\,\sigma_0+\sin\alpha\,\sigma_3)$, solid: gap closing. Dashed: $m+M_3=0$.}\label{fig:second}
\end{figure}

\paragraph{Convergence of Bloch eigenvalues.} Figure~\ref{fig:first}(c) compares bulk Bloch eigenvalues of $H^\eps$ and $H^0$ at quasimomentum $0$ for three examples. The errors are of order $\eps$ for the Dirac example \eqref{eq:SMex1} and for the regularized Dirac operator, and of order $\eps^{1/2}$ once a harmonic is added that makes the third-order term $\avg{WL^{-1}WL^{-1}W}$, $L=P_1(D_y)$, non-zero, consistent with the rates of Theorem~\ref{M-thm:hom} and with Remark~\ref{rem:sharp}.

\paragraph{Gap-preserving random fluctuations.} We give explicit examples of the gap-preserving laws of Theorem~\ref{M-thm:phasepurek} of the main paper, recalled in section~\ref{sec:Ssetting} with conditions (A) and (B). For Theorem~\ref{thm:phasepure} and Proposition~\ref{prop:bernoulli}, let $k=1$, let $w$ satisfy \eqref{eq:Smoments}, that is $\int w=0$ and $\int y\,w=0$, let $P=D\cdot\sigma+m\sigma_3$ with $m>0$, let $\kappa_->0$, and assume the \emph{mass-dominance} condition
\begin{equation}\label{eq:massdom}
-(m+M_3)-|M_0|\ge g>0\qquad\text{for all }M=M_0\sigma_0+M_\perp\cdot\sigma+M_3\sigma_3\in\cK .
\end{equation}

\begin{theorem}[Gap-preserving fluctuations: a genuine gap]\label{thm:phasepure}
Under \eqref{eq:Smoments} and \eqref{eq:massdom}, there is $\eps_0>0$ such that for every $\eps\in(0,\eps_0]$ and \emph{every} configuration with $\kappa_-\le|\omega_j|\le\kappa_+$: \textup{(i)} $\spec\bigl(P+\eps^{-1/2}W_\omega(x/\eps)\bigr)\cap(-g/2,g/2)=\emptyset$; and \textup{(ii)} $T(H^\eps_\omega)$ is Fredholm with $\Ind T(H^\eps_\omega)=-1$, where $H^\eps_\omega=P+\eps^{-1/2}\rho\,W_\omega(x/\eps)$.
In particular, for i.i.d.\ couplings with law supported in $\{\kappa_-\le|t|\le\kappa_+\}$, these conclusions hold a.s.
\end{theorem}

\begin{proof}
This is case (A) with $J=\sigma_3$, $s_+=1$ and $s_-=-1$. For $M=M_0+M_\perp\cdot\sigma+M_3\sigma_3$ we have $M^{\rm e}=M_0+M_3\sigma_3$, and
\[
-\sigma_3(m\sigma_3+M^{\rm e})=-(m+M_3)-M_0\sigma_3\ge g
\]
by \eqref{eq:massdom}, while $\sigma_3(m\sigma_3)=m$. Apply Theorem~\ref{M-thm:phasepurek} of the main paper with $g$ replaced by $\min(g,m)$. Part (i) holds with $g$ itself, since its proof uses only the condition on $s_-$. For the index, take $\bar M\in\cK$. By \eqref{eq:SBDI}, $\BDI(P+\bar M,P)=c(h_-)-c(h_+)$. The symbol of $P+\bar M$ is $(\xi+\bar M_\perp)\cdot\sigma+(m+\bar M_3)\sigma_3+\bar M_0$, and since $|\bar M_0|<|m+\bar M_3|$, its negative spectral projector is that of $(\xi+\bar M_\perp)\cdot\sigma+(m+\bar M_3)\sigma_3$. The integral \eqref{eq:Sc} is invariant under the translation $\xi\mapsto\xi-\bar M_\perp$, so that $c(h_-)=\frac12\sgn(m+\bar M_3)=-\frac12$ by the Dirac value, while $c(h_+)=\frac12\sgn m=\frac12$. Hence $\BDI(P+\bar M,P)=-1$.
\end{proof}

\begin{proposition}[A Bernoulli example]\label{prop:bernoulli}
Let $\omega_j=\pm1$ and $w=v\sigma_0+\tau\nabla^\perp v\cdot\sigma$, with $v\in C_c^\infty(\R^2;\R)$, $\int v=0$ and $\int y\,v=0$ (e.g., $v=\Delta b$), and write $V=\sum_j\omega_jv(\cdot-j)$. Then $\mathbb K(q)$ has vanishing $\sigma_0$-component, $\sigma_3$-component $2\tau\sum_\ell|\hat v(q+\ell)|^2$, and a $(\sigma_1,\sigma_2)$-component that is odd in $q$, namely $-\sum_\ell|\hat v(p)|^2(1-4\pi^2\tau^2|p|^2)\,p/(2\pi|p|^2)$ with $p=q+\ell$. For every periodic configuration
\[
M(\omega)=2\tau\avg{V^2}\sigma_3,\qquad c_{\min}\le\avg{V^2}\le c_{\max},
\]
where $c_{\min}$ and $c_{\max}$ are the minimum and maximum over $q\in\T^2$ of $\sum_{\ell\in\Z^2}|\hat v(q+\ell)|^2$. If $c_{\min}>0$, $\tau<0$ and $2|\tau|c_{\min}\ge m+g$, then \eqref{eq:massdom} holds, so that Theorem~\ref{thm:phasepure} applies. The gap $(-g/2,g/2)$ and the index $-1$ then hold for every configuration.
\end{proposition}

\begin{proof}
The single site satisfies \eqref{eq:Smoments} with $k=1$, since $\int\nabla v=0$ and $\int y_i\partial_jv=-\delta_{ij}\int v=0$. The magnetic potential $A=\tau\nabla^\perp V$ has stream function $\psi=\tau V$, and \eqref{eq:SEM} gives $M=2\avg{V\psi}\sigma_3=2\tau\avg{V^2}\sigma_3$, with $M_0=M_\perp=0$, since $w$ has no $\sigma_3$ component. Parseval on the supercell, as in Lemma~\ref{M-lem:Kset} of the main paper, gives $\avg{V^2}=\int \mathfrak s_\omega(dq)\sum_\ell|\hat v(q+\ell)|^2$, where $\avg{V^2}$ is the supercell average and $\mathfrak s_\omega$ is a probability measure because $\omega_j^2=1$. In particular $\avg{V^2}\in[c_{\min},c_{\max}]$. For \eqref{eq:massdom}, which only involves $M_0$ and $M_3$, compute $\hat w(p)^*G(p)\hat w(p)$ directly, with $G(p)=-(2\pi)^{-1}(p\cdot\sigma)/|p|^2$ and $\hat w(p)=\hat v(p)(I+2\pi i\tau\,p^\perp\cdot\sigma)$. Using $(p\cdot\sigma)(p^\perp\cdot\sigma)=i|p|^2\sigma_3$, it equals $2\tau|\hat v(p)|^2\sigma_3-|\hat v(p)|^2(1-4\pi^2\tau^2|p|^2)(p\cdot\sigma)/(2\pi|p|^2)$. Hence every element of $\cK$ (here $\kappa_\pm=1$) has $M_0=0$ and $M_3\le2\tau c_{\min}$, so $-(m+M_3)-|M_0|\ge2|\tau|c_{\min}-m\ge g$.
\end{proof}

The same single site also illustrates the Born gap condition [BG] of Theorem~\ref{M-thm:mainB} of the main paper, for laws of couplings that charge small values and therefore produce pseudo-gaps. Recall that, for i.i.d.\ couplings with mean $m_1$ and variance $s^2$, the Born matrix is
\[
M_{\rm Born}=s^2\int_{\T^2}\mathbb K(q)\,dq+m_1^2\,\mathbb K(0),
\]
and that [BG] requires $\dist(0,\spec P)\ge g$ and $\dist(0,\spec(P+M_{\rm Born}))\ge g+\bar r$ for some $g>0$ and $\bar r\in(0,g]$. A sufficient condition is the graded criterion stated after Theorem~\ref{M-thm:mainB}: a unitary involution $J$ anticommuting with $\gamma_1,\gamma_2$ and commuting with $\Gamma$, and signs $s_\pm$ with $s_+J\Gamma\ge g$ and $s_-J(\Gamma+M_{\rm Born}^{\rm e})\ge g$, give [BG] with $g/2$ in place of $g$ and $\bar r=g/2$.

\begin{lemma}[Flat single sites]\label{lem:flat}
Let $v\in C_c^\infty(\R^2;\R)$ have vanishing moments of order $\le3$, and let $v_\delta=v(\cdot/\delta)$. Set
\[
\varphi_\delta(q)=\sum_{\ell}\frac{|\hat v_\delta(q+\ell)|^2}{4\pi^2|q+\ell|^2},\qquad \bar\varphi_\delta=\int_{\T^2}\varphi_\delta=\|(-\Delta)^{-1/2}v_\delta\|_{L^2}^2 .
\]
Then $\bar\varphi_\delta=\delta^4\bar\varphi_1$, and $\sup_q|\varphi_\delta(q)-\bar\varphi_\delta|\le C\delta^9$ for $\delta\le1$.
\end{lemma}
\begin{proof}
Let $u=(-\Delta)^{-1/2}v$. Then $\hat u=\hat v/(2\pi|p|)$ vanishes to order $3$ at $0$, and $|u(y)|\le C\avg y^{-5}$ by Lemma~\ref{M-lem:latsum} of the main paper (a symbol vanishing to order $s>0$ at the origin, smooth elsewhere and rapidly decaying, has a kernel bounded by $C\avg y^{-2-s}$). Its autocorrelation $a(x)=\int u(y)u(y+x)\,dy$ satisfies $|a(x)|\le C\avg x^{-5}$. By homogeneity, $(-\Delta)^{-1/2}v_\delta=\delta\,u(\cdot/\delta)$, whose autocorrelation is $a_\delta(x)=\delta^4a(x/\delta)$. Poisson summation gives $\varphi_\delta(q)=\sum_ma_\delta(m)e^{2\pi iq\cdot m}$. The term $m=0$ is $\bar\varphi_\delta=\delta^4\|u\|^2$, and $\sum_{m\ne0}|a_\delta(m)|\le C\delta^4\sum_{m\ne0}(|m|/\delta)^{-5}\le C\delta^9$.
\end{proof}

\begin{proposition}[Two second-order examples]\label{prop:highex}
Let $b\in C_c^\infty(\R^2;\R)$ with $b\ne0$, let $v=\Delta^2b$, and let $\omega_j=\pm1$.
\begin{enumerate}[label=\textup{(\alph*)},leftmargin=*,nosep]
\item \emph{A $J$-graded second-order operator.} Let $P=P_2(D)+m\sigma_2$ with $P_2(\xi)=-(\xi_1^2-\xi_2^2)\sigma_3-2\xi_1\xi_2\sigma_1$ and $m\ge g$. Let $w=v\sigma_2+\tau\bigl[(2\partial_1\partial_2v)\sigma_1+((\partial_1^2-\partial_2^2)v)\sigma_3\bigr]$, with $\tau>0$ and $\supp b\subset\{|y|<\frac12\}$. Then, for every $q$, $\mathbb K(q)$ has $\sigma_0$-component $0$ and $\sigma_2$-component $-2\tau\|v\|_{L^2}^2$. If $2\tau\|v\|^2_{L^2}\ge m+g$, then (A) holds with $J=\sigma_2$, $s_+=1$ and $s_-=-1$, and $\Ind T(H^\eps_\omega)=-2$ for every configuration and every $\eps\le\eps_0$.
\item \emph{The regularized Dirac operator, from its trivial phase.} Let $P=D\cdot\sigma+(m+\beta\Delta)\sigma_3$ with $m<0<\beta$, so that $c(P)=\frac12(\sgn m+\sgn\beta)=0$. Let $w=\lambda v_\delta\sigma_0$ with $v_\delta=v(\cdot/\delta)$, let $\bar\mu\in(0,\frac1{2\beta}]$, and choose $\lambda>0$ with $\lambda^2\|(-\Delta)^{-1/2}v_\delta\|^2_{L^2}=\beta(\bar\mu-m)$. There is $C$, independent of $\delta$, such that for $\delta$ small enough (B) holds with $\bar M=(\bar\mu-m)\sigma_3$, $\bar r=C\delta^5(\bar\mu-m)$ and $g=\min(|m|,\bar\mu-\bar r)$. Moreover $\Ind T(H^\eps_\omega)=1$ for every configuration and every $\eps\le\eps_0$.
\end{enumerate}
\end{proposition}
\noindent In (b) the disorder is scalar, as in \cite{S:GWATB09}. The unperturbed operator is trivial, the homogenized bulk $P+\bar M=D\cdot\sigma+(\bar\mu+\beta\Delta)\sigma_3$ is in the inverted (topological) phase, and the disordered bulk keeps a genuine spectral gap for every sign pattern. The relevant property is not the distribution of the couplings but the flatness of the periodized single-site kernel $\mathbb K$ (Lemma~\ref{lem:flat}), which makes all local effective matrices nearly equal.

\begin{proof}
Both single sites satisfy \eqref{eq:Smoments} with $k=2$, since $v=\Delta^2b$ has vanishing moments of order $\le3$, and so do its derivatives. With $\kappa_\pm=1$, $\cK=\overline\conv\{\mathbb K(q)\}$.

(a) The operator $P_2(D)$ satisfies [P2], since $P_2(\xi)^2=|\xi|^4I$, and $P_2$ is indefinite. With the convention $\widehat{\partial_jf}=2\pi ip_j\hat f$, we have $\hat w(p)=\hat v(p)X(p)$ with $X=\sigma_2+4\pi^2\tau P_2(p)$, and $G(p)=-P_2(p)/(4\pi^2|p|^4)$. The terms $\sigma_2P_2\sigma_2=-P_2$ and $P_2^3=|p|^4P_2$ anticommute with $\sigma_2$, while the cross terms give $X^*P_2X$ the $\sigma_2$-component $8\pi^2\tau|p|^4$ and no $\sigma_0$-component. Hence $\hat w^*G\hat w$ has $\sigma_2$-component $-2\tau|\hat v(p)|^2$ and $\sigma_0$-component $0$. Since $\supp v\subset\{|y|<\frac12\}$, the integer translates of $v$ have disjoint supports, and Poisson summation gives $\sum_\ell|\hat v(q+\ell)|^2=\|v\|_{L^2}^2$. This proves the formula for $\mathbb K$.

For $M\in\cK$, $M^{\rm e}=-2\tau\|v\|^2\sigma_2$, so $-\sigma_2(\Gamma+M^{\rm e})=2\tau\|v\|^2-m\ge g$ with $\Gamma=m\sigma_2$, while $\sigma_2\Gamma=m\ge g$. This gives (A). For the index, deform the $\sigma_2$-anticommuting part of $\bar M$ to $0$. By \eqref{eq:Sgrade} this keeps the gap. By \eqref{eq:SBDI}, $\BDI(P+\bar M,P)=c(P_2+\mu_-\sigma_2)-c(P_2+m\sigma_2)$ with $\mu_-=m-2\tau\|v\|^2<0$. The symbol $P_2(\xi)+\mu\sigma_2$ is $d\cdot\sigma$ with $d_2=\mu$ and $(d_3,d_1)$ the planar map $\xi_1+i\xi_2\mapsto-(\xi_1+i\xi_2)^2$, of degree $2$. Since $(\sigma_3,\sigma_1,\sigma_2)$ is a cyclic permutation of $(\sigma_1,\sigma_2,\sigma_3)$, \eqref{eq:Stwoband} gives $c=2\cdot\frac12\sgn\mu=\sgn\mu$, as for the Dirac value. Hence the index is $-1-1=-2$.

(b) Here $P_2(\xi)=-\beta|\xi|^2\sigma_3$ and $G(p)=\sigma_3/(4\pi^2\beta|p|^2)$, so $\mathbb K(q)=\beta^{-1}\lambda^2\varphi_\delta(q)\sigma_3$, and $\cK$ is the segment $\{t\sigma_3:\ \beta^{-1}\lambda^2\min\varphi_\delta\le t\le\beta^{-1}\lambda^2\max\varphi_\delta\}$. It contains $\bar M=\beta^{-1}\lambda^2\bar\varphi_\delta\sigma_3=(\bar\mu-m)\sigma_3$. By Lemma~\ref{lem:flat}, every point of $\cK$ lies within $\beta^{-1}\lambda^2C\delta^9=C'\delta^5(\bar\mu-m)=\bar r$ of $\bar M$, since $\lambda^2=\beta(\bar\mu-m)\delta^{-4}/\bar\varphi_1$.

The symbol of $P+\bar M$ has eigenvalues $\pm(t+(\bar\mu-\beta t)^2)^{1/2}$, with $t=|\xi|^2$. Since $1-2\beta(\bar\mu-\beta t)\ge1-2\beta\bar\mu\ge0$, the minimum is at $t=0$, and $\dist(0,\spec(P+\bar M))=\bar\mu\ge g+\bar r$. Similarly, $\dist(0,\spec P)\ge|m|\ge g$, because $m-\beta t\le m<0$. This gives (B). The index is $\frac12(\sgn\bar\mu-\sgn m)=1$ by \eqref{eq:Smeff}, with $M_\perp=0$, $M_0=0$ and $m_{\rm eff}=\bar\mu$.
\end{proof}

\paragraph{A pseudo-gap.} The next example illustrates Theorem~\ref{M-thm:mainB}: the Born-gap condition holds and the index along admissible cuts is non-trivial, while rare regions fill the gap with spectrum.
\begin{proposition}[A Born gap inside a pseudo-gap]\label{prop:uniform}
Let $P=D\cdot\sigma+m\sigma_3$ with $m>0$, and let $w=v\sigma_0+\tau\nabla^\perp v\cdot\sigma$ with $v=\Delta b$, $b\in C_c^\infty(\R^2;\R)$, $b\ne0$, $\supp b\subset\{|y|<\frac12\}$. Let the couplings be i.i.d.\ with mean $m_1$ and variance $s^2$. Then
\[
M_{\rm Born}=2\tau(s^2+m_1^2)\|v\|_{L^2}^2\,\sigma_3,\qquad \mathbb K(0)=2\tau\|v\|_{L^2}^2\,\sigma_3 .
\]
\begin{enumerate}[label=\textup{(\roman*)},leftmargin=*,nosep]
\item If $\tau<0$ and $2|\tau|(s^2+m_1^2)\|v\|_{L^2}^2\ge m+g$ with $0<g\le m$, the graded criterion holds with $J=\sigma_3$, $s_+=1$ and $s_-=-1$. Hence [BG] holds, and for $\eps$ small, a.s., $\ind\tilde T(H^\eps)=-1$ along every admissible cut.
\item If the couplings are uniformly distributed on $[-1,1]$, so that $m_1=0$ and $s^2=\frac13$, and $2|\tau|\|v\|_{L^2}^2\ge3(m+g)$, then, in addition, the hypotheses of Corollary~\ref{M-cor:pseudoB} of the main paper hold with $\kappa_0=1$. Thus, for every $g'>0$ and $\eps$ small, $\spec(P+\eps^{-1/2}W_\omega(x/\eps))\cap(-g',g')\ne\emptyset$ a.s.
\end{enumerate}
\end{proposition}
\noindent In (ii), the interval $(-g/2,g/2)$ is a pseudo-gap, filled by rare puddles where all couplings are small, while the index along admissible cuts is $-1$. This is a first-order analogue of the topological Anderson insulator of \cite{S:GWATB09}. For Bernoulli couplings $\pm1$, $M_{\rm Born}=2\tau\|v\|_{L^2}^2\sigma_3$ has the same form, with $\EE\omega^2=1$ in place of $\frac13$, but the law does not charge small couplings, and Proposition~\ref{prop:bernoulli} gives a genuine gap instead.

\begin{proof}
The single site satisfies \eqref{eq:Smoments} with $k=1$, as in the proof of Proposition~\ref{prop:bernoulli}, where $\hat w(p)^*G(p)\hat w(p)$ was shown to equal $2\tau|\hat v(p)|^2\sigma_3-|\hat v(p)|^2(1-4\pi^2\tau^2|p|^2)(p\cdot\sigma)/(2\pi|p|^2)$. Since $v$ is real, $|\hat v|$ is even, and the in-plane part is odd in $p$. Hence $\int_{\T^2}\mathbb K=\int_{\R^2}\hat w^*G\hat w\,dp=2\tau\|v\|_{L^2}^2\sigma_3$ by Plancherel, and the in-plane part of $\mathbb K(0)=\sum_{\ell\ne0}\hat w(\ell)^*G(\ell)\hat w(\ell)$ vanishes by the symmetry $\ell\mapsto-\ell$. Its $\sigma_3$ component is $2\tau\sum_{\ell}|\hat v(\ell)|^2$, since $\hat v(0)=0$, and Poisson summation gives $\sum_\ell|\hat v(\ell)|^2=\|v\|_{L^2}^2$, because the integer translates of $v$ have disjoint supports. This gives both formulas.

(i) $J=\sigma_3$ anticommutes with $\sigma_1,\sigma_2$ and commutes with $\Gamma=m\sigma_3$, and $\sigma_3\Gamma=m\ge g$. Since $M_{\rm Born}$ commutes with $\sigma_3$, $M^{\rm e}_{\rm Born}=M_{\rm Born}$ and $-\sigma_3(m\sigma_3+M_{\rm Born})=-m-2\tau(s^2+m_1^2)\|v\|^2\ge g$. Theorem~\ref{M-thm:mainB} of the main paper and the Dirac value of $c$ give $\ind\tilde T(H^\eps)=\BDI(P+M_{\rm Born},P)=\frac12(\sgn(m+2\tau(s^2+m_1^2)\|v\|^2)-\sgn m)=-1$.

(ii) The support of the law contains $[0,1]$. The mass of $P+\mathbb K(0)$ is $m+2\tau\|v\|^2\le m-3(m+g)<0$, so $P+\mathbb K(0)$ is gapped at $0$ with $\BDI(P+\mathbb K(0),P)=-1\ne0$, and Corollary~\ref{M-cor:pseudoB} applies.
\end{proof}

The second example uses single sites that are nearly flat in the following sense.
\paragraph{Disordered superconductors.} For scalar disorder at criticality, recalled in \eqref{eq:Scrit}, the correction is $M_{\rm crit}=-\ell s_0c_qA^{-1}$, where $A$ is the matrix of the principal symbol. Scalar disorder thus always shifts the zeroth-order part of $P$ in a direction fixed by the principal symbol alone. The topology can change only if the path $t\mapsto P+tM_{\rm crit}$, $t\in[0,1]$, closes the gap, and it changes when, moreover, $\BDI(P+M_{\rm crit},P)\not=0$. Where the phase boundaries lie depends on the operator. We illustrate this on a spinless $p$-wave superconductor, whose Bogoliubov-de Gennes operator is
\begin{equation}\label{eq:Spwave}
P=(|D|^2-\mu)\tau_3+\Delta_0(D_1\tau_1+D_2\tau_2),\qquad\Delta_0\in\R\setminus\{0\},
\end{equation}
with Pauli matrices $\tau_j$ acting in particle-hole space. On-site disorder enters as $\tilde V\tau_3$, which preserves the particle-hole symmetry. This is not scalar disorder, but it is of the form $\tilde VA$ with $A=\tau_3$, covered by Remark~\ref{rem:Adisorder} below. Here $q=|\xi|^2$, $A=\tau_3$ is indefinite, $A^2=I$ and $c_q=1/(2\pi)$, so that, by that remark, $M_{\rm crit}=-\ell s_0c_qA=-\frac{\ell s_0}{2\pi}\tau_3$, and the disorder shifts the chemical potential to $\bar\mu=\mu+\frac{\ell s_0}{2\pi}$.
\begin{corollary}[Disordered $p$-wave superconductors]\label{cor:pwave}
Let $\mu\not=0$ and $\bar\mu\not=0$, and let the disorder be the scalar alloy of \eqref{eq:Scrit}, multiplied by $\tau_3$. Then the conclusions of Theorem~\ref{M-thm:mainA} of the main paper hold, with index $-\frac12(\sgn\bar\mu-\sgn\mu)$ along admissible cuts. In particular, a trivial superconductor with $\mu<0$ becomes topological when $\ell s_0>2\pi|\mu|$.
\end{corollary}
\begin{proof}
The bulk symbols $d(\xi)\cdot\tau$ with $d=(\Delta_0\xi_1,\Delta_0\xi_2,|\xi|^2-\mu')$, $\mu'\in\{\mu,\bar\mu\}$, have eigenvalues $\pm|d(\xi)|$, which vanish nowhere when $\mu'\not=0$, so that both bulks are gapped. Since $(d_1,d_2)=\Delta_0\xi$ vanishes only at $0$, with $\det\partial_\xi(d_1,d_2)=\Delta_0^2>0$, and $d_3(\infty)>0$, \eqref{eq:Sdegree} gives $c(h)=-\frac12(\sgn\mu'+1)$, and \eqref{eq:SBDI} gives the index. Theorem~\ref{M-thm:mainA} of the main paper, in the form of Remark~\ref{rem:Adisorder} for disorder $\tilde V A$, applies since $P$ has the form recalled before \eqref{eq:Scrit}, with $A=\tau_3$, $B=I$, $c_j=\Delta_0\tau_j$ and $P_0=-\mu\tau_3$.
\end{proof}
\begin{remark}[Disorder proportional to the principal matrix]\label{rem:Adisorder}
Theorems~\ref{M-thm:det} and~\ref{M-thm:mainA} of the main paper remain valid if the scalar fluctuation $\rho\tilde V$ in \eqref{M-eq:Hcrit} of the main paper is replaced by $\rho\tilde V A$, with $A$ the matrix of [C1], provided $M_\eps$ and $M_{\rm crit}$ are replaced by $-c_\eps A$ and $-\ell s_0c_qA$. Indeed, with $\tilde\rho=\theta\rho$ and $Y=\tilde\rho Y_0$ as in section~\ref{M-sec:crit} of the main paper, set $X:=Y\,I$, which is scalar. The computation of Lemma~\ref{M-lem:conj}(i) gives $e^{-X}Aq(D)e^Xf=Aq(D)f-2A(B\nabla Y)\cdot\nabla f+A(q(D)Y)f-A|\nabla Y|_B^2f$, the term $A\,q(D)Y$ cancels the singular part $\tilde\rho\tilde V_{\rm high}A$, and the mean of the corrector energy produces $-\tilde\rho^2c_\eps A$. In Lemma~\ref{M-lem:conj}, $\cA_1$ and $\cA_2$ acquire a left factor $A$, $A^{-1}$ is replaced by $A$ in $\cA_3$ and $\cA_4$, the shift $(1-\theta^2)\rho^2M_\eps$ still commutes with $X$, the first part of $\cA_5$ and $\cA_6$ vanish since $X$ is scalar, and the second part of $\cA_5$ is $c_j(D_jX)$. Lemma~\ref{M-lem:prod}, part (ii) of Lemma~\ref{M-lem:conj} and the proof of Theorem~\ref{M-thm:det} apply verbatim, and the local size, the bad boxes and the probabilistic estimates are unchanged since they involve $Y_0$ only. The Born correction is consistent: $-\avg{(\tilde VA)(Aq(D))^{-1}(\tilde VA)}=-A\avg{\tilde Vq(D)^{-1}\tilde V}$. When $A^2=I$, the shift $-\ell s_0c_qA$ coincides with $-\ell s_0c_qA^{-1}$.
\end{remark}

\noindent Disorder thus favors the topological, weak-pairing phase, and the interface between a disordered and a clean region of a trivial superconductor carries a chiral Majorana mode for $\eps$ small.

\section{Extensions and refinements}\label{sec:Sremarks}

This section collects remarks on the hypotheses, the rates and the scope of the main results. In the remarks on the proof of Theorem~\ref{M-thm:hom} of the main paper, $S_\eps$ denotes the Fourier multiplier by $\one_{\{|\xi|\le1/\eps\}}$, the correctors are built on $v=S_\eps\psi^0$ with $\psi^0=(H^0-z)^{-1}f$, and the first corrector is $c_1=\eps^{k/2}\chi^\eps\rho v$ with $\chi=-P_k(D_y)^{-1}W$.

\begin{remark}[No coercivity or multiplier]\label{rem:J}
The proof of \eqref{eq:Shom} requires neither a uniform $H^1$ bound nor a multiplier $J$ with $J^*P_2+P_2J\ge c|\xi|^2$. The latter was used in \cite{S:BalDang24} (with $J=\pm\sigma_3$) to obtain $H^1$ bounds, which are needed only for $H^1$ corrector estimates. Such a $J$ may fail to exist: for $P_2(\xi)=-|\xi|^2(\cos2\theta\,\sigma_3+\sin2\theta\,\sigma_1)$, where $\xi=|\xi|(\cos\theta,\sin\theta)$, we have $P_2(\theta+\pi/2)=-P_2(\theta)$, which rules it out.
\end{remark}

\begin{remark}[Sharpness of the rates]\label{rem:sharp}
For $k=1$, the norm-resolvent rate $\eps^{1/2}$ in \eqref{eq:Shom} cannot be improved when $W\ne0$. For smooth $f$, the classical expansion, with one further corrector, gives
\[
\psi^\eps=\psi^0+\eps^{1/2}\bigl(\rho\chi^\eps\psi^0+\psi_1\bigr)+O(\eps)\quad\text{in }L^2,\qquad\psi_1=-(H^0-z)^{-1}\rho^3\avg{W\chi_2}\psi^0,
\]
with $\chi_2=-P_1(D_y)^{-1}(W\chi-M)$. Since $\chi$ has mean zero, $\chi^\eps\psi^0$ is asymptotically orthogonal to the non-oscillatory $\psi_1$, and two-scale convergence gives $\liminf_{\eps\to0}\eps^{-1/2}\|\psi^\eps-\psi^0\|\ge\|\rho\,\chi\,\psi^0\|_{L^2(\R^d\times\T^d)}$, which is positive unless $\chi\psi^0\equiv0$ on $\supp\rho$. The oscillatory corrector, and not $\avg{W\chi_2}$, thus fixes the norm-resolvent rate. The term $\avg{W\chi_2}$ governs instead the $O(\eps^{1/2})$ shift of eigenvalues discussed in section~\ref{sec:Sexamples}. For $k\ge2$, the first corrector has $L^2$ norm of order $\eps^{k/2}$, and the rate $\eps$ comes from the residual. For odd $k\ge3$, the improvement from $\eps^{1/2}$ to $\eps$ rests on the absence of mean terms of order $\frac12$ in the corrector hierarchy: these would arise only from potential terms of correctors of order $\frac12+\frac k2$, and there are none.
\end{remark}

\begin{remark}[Higher-order terms]\label{rem:cubic}
For odd $k$ and scalar disorder $W=V\,I$, the second-order term $M$ vanishes, by \eqref{eq:Spair}, and the leading effect comes from the cubic term $\avg{VL^{-1}VL^{-1}V}$, $L=P_k(D_y)$, relevant at amplitude $\eps^{-2/3}$. For the $2\times2$ Dirac operator this term is a multiple of the identity, so that it shifts energies but produces no mass. Indeed, $L^{-1}$ has symbol $(2\pi)^{-1}(p\cdot\sigma)/|p|^2$, and $(p\cdot\sigma)(p'\cdot\sigma)=p\cdot p'+i(p\times p')\sigma_3$, so that the $\sigma_3$-component is $i$ times a sum that is real by the symmetry $(p,p')\mapsto(-p,-p')$ of the real field $V$, while the cubic term is Hermitian. In larger non-Clifford systems the cubic term is generically not scalar, and whether it can drive a transition is open.
\end{remark}

\begin{remark}[Several fluctuation profiles]\label{rem:multi}
Let $H^\eps=P(D)+\eps^{-k/2}\sum_{i=1}^n\rho_i(x)W_i(x/\eps)$, with each pair $(\rho_i,W_i)$ as in \eqref{eq:SHeps}. Then \eqref{eq:Shom} holds with
\[
H^0=P(D)+\sum_{i,l}\rho_i\rho_l\,M_{il},\qquad M_{il}:=-\avg{W_i\,P_k(D_y)^{-1}W_l}.
\]
Since $P_k(D_y)^{-1}$ is symmetric on $L^2(\T^d)$, $M_{il}^*=M_{li}$, and the correction is Hermitian. In the proof of Theorem~\ref{M-thm:hom} for general $k$ in the main paper, replace $c_1$ by $\sum_l\eps^{k/2}\chi_l^\eps\rho_lv$, with $\chi_l=-P_k(D_y)^{-1}W_l$. Three terms change. The leading terms of $c_1$ combine with $(H^\eps-z)\psi^0$ into $(P-z)\psi^0+\eps^{-k/2}\sum_l\rho_lW_l^\eps(I-S_\eps)\psi^0$. The potential term of $c_1$ is $\sum_{i,l}\rho_i\rho_l(W_i\chi_l)^\eps v=\sum_{i,l}\rho_i\rho_lM_{il}v+\sum_{i,l}\rho_i\rho_l(W_i\chi_l-M_{il})^\eps v$, and the second sum consists of monomials of order $0$ with mean-zero periodic coefficients, cancelled by correctors as in the main paper. The effective term satisfies $(P-z)\psi^0+\sum_{i,l}\rho_i\rho_lM_{il}v=f-\sum_{i,l}\rho_i\rho_lM_{il}(I-S_\eps)\psi^0$.
Every other term is produced by $P(D)$ or by a potential $\eps^{-k/2}\rho_iW_i^\eps$ acting on a corrector, exactly as before, so the order-and-weight bookkeeping of the main paper applies without change. If $\rho_i\rho_l=0$ for $i\ne l$, the cross terms vanish. This also covers locally periodic fluctuations that are trigonometric polynomials in the fast variable, $W(x,y)=\sum_{0<|q|\le Q}\hat W_q(x)e^{2\pi iq\cdot y}$ with $\hat W_q\in C_b^\infty$ and $\hat W_{-q}=\hat W_q^*$: expanding $\hat W_q$ in a Hermitian basis writes $W(x,x/\eps)$ as a finite sum of real envelopes times mean-zero periodic Hermitian profiles.
\end{remark}

\begin{remark}[Quantization for first-order systems]\label{rem:quant}
Let $\gamma(\xi)=\xi_1\gamma_1+\xi_2\gamma_2$ with $\gamma_1,\gamma_2\in\Herm(n)$ be elliptic, i.e., $\gamma(e)$ is invertible for $|e|=1$. Since $\gamma(-e)=-\gamma(e)$ reverses the signature of $\gamma(e)$, which is constant on the circle, $n=2p$ is even and $\gamma(e)$ has signature $(p,p)$. Let $\cG$ be the set of $\Gamma\in\Herm(n)$ such that $\gamma(\xi)+\Gamma$ is invertible for all $\xi\in\R^2$. For $\Gamma\in\cG$, the bulk invariant \eqref{eq:Sc} of $\gamma(D)+\Gamma$ satisfies
\[
c(\Gamma)\in\tfrac p2+\Z,\qquad c(-\Gamma)=-c(\Gamma),
\]
and $c$ is locally constant on $\cG$. First-order bulk invariants are thus half-integers when $p$ is odd, and only their differences are integers. The proof writes $e^{2\pi ic(\Gamma)}$ as the determinant of the Berry holonomy of the negative spectral projector of $\gamma(Re_\theta)+\Gamma$, $e_\theta=(\cos\theta,\sin\theta)$, around a large circle. As $R\to\infty$, this projector converges to that of $\gamma(e_\theta)$, which the antipodal map $\theta\mapsto\theta+\pi$ exchanges with its complement, and the holonomy of such antipodally flipping loops has determinant $(-1)^p$. Consequently, if $n\equiv2\pmod4$ and $P=\gamma(D)+\Gamma$ with $\Gamma\in\cG$, the attainable-correction result recalled after \eqref{eq:Shom} gives a Hermitian trigonometric polynomial $W$ with $M(W)=-2\Gamma$, and Theorem~\ref{M-thm:TAI} of the main paper gives $\Ind T(H^\eps)=c(-\Gamma)-c(\Gamma)=-2c(\Gamma)$ for $\eps$ small, an odd integer. Every such first-order system thus admits a fluctuation-induced topological phase.
\end{remark}

\begin{remark}[Interfaces between two fluctuation-induced phases]\label{rem:twophase}
Let $\rho_\pm$ be smooth with $\rho_+=1$ for $x_2\ge1$, $\rho_+=0$ for $x_2\le0$, $\rho_-(x_2)=\rho_+(-x_2)$, and let $H^\eps=P+\eps^{-k/2}(\rho_+W_+^\eps+\rho_-W_-^\eps)$. By Remark~\ref{rem:multi}, the homogenized operator is $P+\rho_+^2M(W_+)+\rho_-^2M(W_-)$, since $\rho_+\rho_-=0$. If $P+M(W_\pm)$ are gapped on a common interval $(-g,g)$, the proof of Theorem~\ref{M-thm:TAI} of the main paper gives $\Ind T(H^\eps)=\BDI(P+M(W_-),P+M(W_+))$ for $\eps$ small. When $P_k$ is indefinite, the attainable-correction result recalled after \eqref{eq:Shom} realizes any pair of gapped constant corrections, and hence any bulk-difference invariant between two fluctuation-induced phases.
\end{remark}

\begin{remark}[Scope of the gap-closing theorem]\label{rem:gapscope}
Theorem~\ref{M-thm:gapclose} of the main paper concerns $H^\eps_a=P(D)+a_\eps\rho(x_2)W(x/\eps)$ and its lower bulk $H^\eps_{a,B}=P(D)+a_\eps W(x/\eps)$, with $W\in L^\infty(\R^2;\Herm(N))$, a continuous amplitude $\eps\mapsto a_\eps>0$, and $0\notin\spec P$: if $0\notin\spec H^{\eps'}_{a,B}\cup\spec H^{\eps''}_{a,B}$ and $\Ind T(H^{\eps'}_a)\ne\Ind T(H^{\eps''}_a)$, then $0\in\spec H^\eps_{a,B}$ for some $\eps\in(\eps',\eps'')$. The principle is that an integer index which is locally constant along norm-resolvent continuous paths with gapped bulks can only change if a bulk gap closes. The non-trivial input is that the dilated family $\eps^k\,\mathcal U_\eps H^\eps_a\,\mathcal U_\eps^{-1}$, with $(\mathcal U_\eps u)(y)=\eps u(\eps y)$, is norm-resolvent continuous, although $\eps\mapsto H^\eps_a$ is not.
The argument applies to every order $k$, every continuous amplitude $a_\eps$, and every bounded $W$, periodic, quasi-periodic or random, configuration by configuration. It also allows coefficients of $P$ that depend on $x_2$ only and are constant for $|x_2|\ge X_0$, since these survive dilation by the same estimate as $\rho$. It fails for generic $x$-dependent coefficients, for which $\eps\mapsto c(\eps\,\cdot)$ is not continuous in $L^\infty$. The result does not locate the critical value of $\eps$ nor identify the mechanism: for periodic $W$, a Bloch band reaches $0$, while for random $W$ the spectrum at $0$ may be carried by rare regions (Theorem~\ref{M-thm:pseudogap} of the main paper). Gaps may also close without any change of the index. For random operators with a mobility gap, the analogous statement, that a change of Chern number along a disorder parameter forces delocalization, is proved in \cite{S:PRR26}. In a pseudo-gap, the hypothesis $0\notin\spec H^\eps_{a,B}$ fails at small $\eps$, and the relevant index is the deformed-cut index of the main paper. The corresponding statement along $\eps$ is open.
\end{remark}

\begin{remark}[Alloys on $\R^d$]\label{rem:dimalloy}
Theorem~\ref{M-thm:unifhom} of the main paper, which compares $P+\eps^{-k/2}\rho W_\omega(x/\eps)$ with $P+\rho^2M^\eps_\omega$ for a smooth field of local effective matrices $M^\eps_\omega$ within $C\delta_k(\eps)$ of $\cK$, with $\delta_1=\eps^{1/4}$, $\delta_2=\eps^{1/3}\log(1/\eps)$ and $\delta_k=\eps^{1/3}$ for $k\ge3$, holds with the same rates for $\Z^d$-alloys on $\R^d$, with \eqref{eq:Smoments} imposed on $\R^d$. The dimension enters the proof only through the decay $\avg y^{-d-s}$ of the lattice sums of Lemma~\ref{M-lem:latsum}, the mesoscale kernel $K_{\varpi}(y)=\varpi^d\tilde K(\varpi y)$ of Lemma~\ref{M-lem:meso}, and the normalization of Lemma~\ref{M-lem:Kset}. Lemma~\ref{M-lem:mult} is dimension-free.
\end{remark}

\begin{remark}[Single sites with only $\int w=0$]\label{rem:log}
For $k=1$ and single sites with only $\int w=0$, so that \eqref{eq:Smoments} fails at order one, the single-site corrector $\chi_w=\mathcal F^{-1}[G\hat w]$ decays like $\Omega(\hat y)|y|^{-2}$ with an angular profile $\Omega$ of mean zero. For signs $\omega_j$ aligned with $\Omega$, the correctors $\sum_j\omega_j\chi_w(\cdot-j)$ grow like the logarithm of the number of sites involved, so that uniformly bounded correctors need not exist and \eqref{eq:Smoments} cannot simply be dropped.
\end{remark}

\begin{remark}[Perturbations away from the cut]\label{rem:awaycut}
Let $H_\circ=\gamma(D)+V_\circ$ be a first-order reference operator, e.g., the Dirac domain wall $D\cdot\sigma+m(x_2)\sigma_3$ with $m=\pm m_0$ for $\pm x_2\ge1$, and let $V\in L^\infty(\R^2;\Herm(N))$ with $\|V\|_\infty\le K_V$. For the Dirac operator, $V$ may be an arbitrary bounded change $\mu(x)\sigma_3$ of the mass, an electric potential, or a vector potential $a\cdot\sigma$. It need not be small, and $H_\circ+V$ need not be gapped anywhere. Since [R1] holds with $a=a_\gamma$ independently of the potential (Lemma~\ref{M-lem:CT}(i)), Theorem~\ref{M-thm:robust} gives $\dd_0$, depending only on $(\gamma,R_0,\varphi,K_V)$, such that $\tilde T(H_\circ+V)$ is Fredholm with index $\Ind T(H_\circ)$ along every cut that avoids $\supp V$ at distance $\dd_0$. Corollary~\ref{M-cor:farV} of the main paper is the case $H_\circ=H^\eps$ of this statement, and Corollary~\ref{M-cor:puddles} is another instance. The index is therefore stable under local variations of the location of the cut, as long as the cut stays at distance $\dd_0$ from $\supp V$. For straight cuts $\Pi_{x_0}=\one_{\{x_1\ge x_0\}}$, it does not depend on $x_0$ within any vertical band of width larger than $2\dd_0$ free of $V$. The condition is needed along the whole cut, far into both bulks: a perturbation that reaches the cut changes the bulks it crosses, and puddles meeting the cut carry states across it (Figure~\ref{M-fig:cuts}(b)). For relatively compact perturbations $V_c$, the index is independent of $x_0$ without restriction: $T_{x_0}(H+V_c)-T_{x_0}(H)$ is compact (Corollary~\ref{M-cor:compact}), and $\Ind T_{x_0}(H)$ does not depend on $x_0$, by Step~1(b) of the proof of Theorem~\ref{M-thm:robust}, since $\Pi_{x_0}-\Pi_0$ is supported in a vertical strip, bounded in every horizontal window. For second-order operators, $\dd_0$ depends in addition on the $H^1$ resolvent bound of Lemma~\ref{M-lem:CT}(ii).
\end{remark}

\section{Perspectives}
\label{sec:Sdiscussion}

We conclude with three directions suggested by the results.

\paragraph{Interface current.} The index of the main paper is computed along a deformed cut that depends on the configuration, whereas the physical interface conductance $2\pi\sigma_I(H)={\rm Tr}\,2\pi i[H,\Pi_s]\varphi'(H)$ uses a fixed cut. In a pseudo-gap, this operator is not trace-class, since infinitely many rare states meet any fixed cut. Along the deformed cut, the contribution of each rare region is small, of order $(a\dd)^{-n}$, but these contributions are not summable in general, because the distance from the cut to the rare regions does not grow along the cut. Since individual eigenstates carry no current and a single rare region does not change the index, we expect regularized traces to converge to the index once the interactions between rare regions are summable, e.g., under exponential localization of the rare states. Without such an assumption, the deformed-cut index is the index-theoretic definition of the quantized conductance. A statement along a fixed cut would be a mixed spectral and mobility-gap bulk--edge correspondence, in the spirit of \cite{S:EGS05,S:CombesGerminet05,S:BSS23}.

\paragraph{Localization.} Localization in the pseudo-gap would identify it with a mobility gap, as expected for an insulator. For monotone critical alloys with couplings of bounded density, a Wegner estimate follows from spectral averaging \cite{S:CombesHislop94}, and good regions provide decay of the resolvent at real energies at the rate $c/\sqrt{\log(1/\eps)}$. Since bad unit boxes have probability $O(\eps^{p_0})$, the polynomial losses of a multiscale analysis \cite{S:GK01} should be absorbed by the decay across a single good box. Combined with \cite{S:PRR26}, localization would also show that delocalized states cross the Fermi level at some intermediate $\eps$, where the index jumps from $0$ to $\BDI$. This is the topological Anderson transition of \cite{S:GWATB09}.

\paragraph{Higher dimensions and other settings.} Our proofs of Theorems~\ref{M-thm:TAI} and~\ref{M-thm:gapclose} use only three properties of the index: it is defined for elliptic interface operators with gapped bulks, it is invariant under gapped deformations, geometry-preserving conjugations and dilations, and it is given by \eqref{eq:SBDI} for constant-coefficient bulks. In even dimension $d$, a Dirac phase in the $d-1$ interface variables replaces $\Pi$ \cite{S:ProdanSchulzBaldes16}, and the first two properties should carry over; what is missing is a bulk-difference formula. Among other extensions, subcritical second-order alloys only lack an $H^1$ resolvent bound for the reference operator; Gaussian fields require Gaussian tails and hypercontractivity \cite{S:Janson97} in place of Hoeffding's and the Hanson--Wright inequalities; matrix-valued disorder at criticality requires the renormalization of the commutators produced by the exponential transform, possibly by paracontrolled methods \cite{S:AllezChouk15,S:Labbe19}; and, in the supercritical regime $d<k$, the homogenized limit is itself a random operator \cite{S:Bal10,S:BalGu15}.

\section{Index of notation}\label{sec:Snotation}

This section lists the notation of the main paper and of this supplementary material, with the place where each symbol is defined or first used. References without a prefix S are to the main paper; ``\S2(iv)'' refers to block (iv) of section~\ref{M-sec:setting}. Symbols marked $^\dagger$ have other meanings, collected in Table~\ref{tab:reuse}. Bound variables, generic constants $C,c,C_n$, and objects used only inside one proof and introduced there are listed only when they are used later or clash with a global symbol.

{\small
\renewcommand{\arraystretch}{1.12}
\begin{longtable}{@{}p{3.3cm}p{8.9cm}p{3.3cm}@{}}
\caption{Notation, by topic.}\label{tab:notation}\\
\toprule
Symbol & Meaning & Defined in\\
\midrule
\endfirsthead
\toprule
Symbol & Meaning & Defined in\\
\midrule
\endhead
\bottomrule
\endfoot
\multicolumn{3}{@{}l}{\emph{General conventions}}\\[2pt]
$D=-i\nabla$, $D_y$, $D^\alpha$ & derivatives in $x$, in the fast variable $y$, multi-index powers & \S2, Notation\\
$x$, $y=x/\eps$ & physical and fast variables; $f^\eps(x)=f(x/\eps)$ & \S2, Notation\\
$\avg\xi$; $\avg f$; $\ip\cdot\cdot$ & $(1+|\xi|^2)^{1/2}$; $\int_{\T^d}f\,dy$; $L^2$ inner product & \S2, Notation\\
$\T^d$, $\Herm(N)$ & torus $\R^d/\Z^d$; Hermitian $N\times N$ matrices with $\tr(AB)$ & \S2, Notation\\
$\hat f_q$ & Fourier coefficient $\int_{\T^d}fe^{-2\pi iq\cdot y}dy$, $q\in\Z^d$ & \S2, Notation\\
$\hat a(p)$ & $\int_{\R^2}a(y)e^{-2\pi ip\cdot y}dy$ (single sites, multipliers) & \S2(iv), \S\ref{M-sec:random}\\
$\hat f(\xi)$ & $\int f(x)e^{-ix\cdot\xi}dx$ (only in \S\ref{M-sec:crit} and App.~\ref{M-app:proba}) & \S\ref{M-sec:crit}\\
$R(z)$ & resolvent $(H-z)^{-1}$, superscripts as for $H$ & \S2, Notation\\
$\sigma_0,\dots,\sigma_3$; $\sigma$ & Pauli matrices; $\sigma=(\sigma_1,\sigma_2)$ & \S2, Notation\\
$\one_\Omega$, $\cU_r$ & indicator; $r$-neighbourhood $\{x:\dist(x,\cU)\le r\}$ of $\cU$ & \S2, Notation\\
$Q$, $c_Q$, $Q^*$, $\psi_Q$ & unit box $m+[0,1)^2$, its center, $m+[-1,2]^2$, partition of unity & \S2, Notation; \S\ref{M-sec:crit}\\
$e^\perp$, $\nabla^\perp$, $\hat q$ & $(-e_2,e_1)$, $(-\partial_2,\partial_1)$, $q/|q|$ & \S\ref{M-sec:num}, \S\ref{M-sec:univ}\\
$H^s$, $C^{k,\alpha}$, $C^\infty_b$ & Sobolev, H\"older, smooth bounded functions & standard\\
$C^s$, $B^s_{p,r}$ & H\"older--Zygmund $B^s_{\infty,\infty}$ and Besov spaces & \S\ref{M-sec:crit}\\
$\Delta_j$, $\phi_j$, $k_j$ & Littlewood--Paley blocks, their symbols and kernels & \S\ref{M-sec:crit}\\
$T_fg$, $R(f,g)$, $S_{j-1}$ & Bony paraproduct, resonant term, low-frequency cutoff & Lemma~\ref{M-lem:prod}\\
$\spec$, $\Ind$, $\mathfrak q$ & spectrum, Fredholm index, quotient map onto the Calkin algebra & Lemma~\ref{M-lem:fred}\\[4pt]
\multicolumn{3}{@{}l}{\emph{Unperturbed operators, gradings and invariants}}\\[2pt]
$d$, $N$, $k$ & dimension, number of components, order of $P$ & \S2(i)\\
$P(D)$, $A_\alpha$, $P_k$, $P_j$ & operator, coefficients, principal symbol, homogeneous part of degree $j$ & \S2(i); proof of Thm~\ref{M-thm:gapclose}\\
\textup{[P1]}, \textup{[P2]}, $c_P$ & self-adjointness, ellipticity, ellipticity constant & \S2(i)\\
$\Gamma^\dagger$, $\Gamma_\pm$ & constant zeroth-order (mass) part of $P$; constant bulk potentials & \S2(i); \eqref{M-eq:ch}\\
$\gamma(\xi)$, $\gamma_j$, $\|\gamma\|$ & first-order symbol $\xi_1\gamma_1+\xi_2\gamma_2$, its matrices and norm & \S2(vi); Lemma~\ref{M-lem:CT}(i)\\
$J$, $s^\dagger$, $P^{\rm o}$, $\Upsilon$, $q(\xi)^\dagger$ & grading, sign, $J$-odd part, $J$-even curvature, $\sum b_{ij}\xi_i\xi_j$ & \S2(iii), classes (a)--(b)\\
$H_\pm$ & bulks $P$ and $P+M$ of $H^0$ & \S2(iii)\\
\textup{[G]}, $g^\dagger$ & common-gap condition, gap half-width & \eqref{M-eq:gap}\\
$\varphi^\dagger$, $U(H)$, $S(H)^\dagger$ & energy switch, $e^{2\pi i\varphi(H)}$, $U(H)-I$ & \S2(iii)\\
$\Pi$, $T(H)$ & $\one_{\{x_1\ge0\}}$, Toeplitz operator & \eqref{M-eq:T}\\
$\Pi_s$, $\sigma_I(H)$ & smooth spatial switch, interface current & \S2(iii)\\
$c(h)$, $\Omega_h$, $\Pi_h$ & bulk invariant, curvature form, negative spectral projector of $h$ & \eqref{M-eq:ch}\\
$\BDI$ & bulk-difference invariant $c(P+\Gamma_-)-c(P+\Gamma_+)$ & \eqref{M-eq:ch}\\
$d(\xi)^\dagger$, $f^\dagger$, $n^\dagger$ & two-band symbol $d\cdot\sigma+f\sigma_0$, unit vector $d/|d|$ & Lemma~\ref{M-lem:BDIc}(ii); \eqref{M-eq:degree}\\
$c(m\sigma_3)$ & $c(D\cdot\sigma+m\sigma_3)=\frac12\sgn m$ & Lemma~\ref{M-lem:BDIc}(ii)\\
$D^a=D+a$, $P(D^a)$ & Peierls substitution & \S2(iii); Lemma~\ref{M-lem:peierls}\\[4pt]
\multicolumn{3}{@{}l}{\emph{Periodic fluctuations and homogenization}}\\[2pt]
$\rho^\dagger$ & envelope; in interface statements $\rho=\rho(x_2)$, $=0$ for $x_2\ge1$, $=1$ for $x_2\le-1$ & \S2(i); \S2(iii)\\
$\rho_{\alpha,c}$, $\rho_{\alpha,d}$ & independent envelopes of the individual terms & \S2(i)\\
$\mathrm{Sym}[cD^\alpha]$ & $\frac12(cD^\alpha+D^\alpha c^*)$ & \S2(i)\\
$c_\alpha$, $d_\alpha$, $\alpha_0$ & fluctuation coefficients, H\"older exponent & \eqref{M-eq:Heps}\\
$\mathcal V^\eps$, $\mathcal R^\eps$, $e^\dagger$, $f^\dagger$ & fluctuation, remainder and its coefficients & \eqref{M-eq:Heps}\\
$H^\eps$ & perturbed operator (\eqref{M-eq:Heps}; for alloys \S2(iv); critical \eqref{M-eq:Hcrit}) & \eqref{M-eq:Heps}\\
$W^\dagger$, $d_0^{\rm h}$ & effective zeroth-order profile, Hermitian part of $d_0$ & \eqref{M-eq:Weff}\\
$M^\dagger$, $M(W)$, $G(q)$, $H^0$ & correction, as a function of $W$, $-P_k(2\pi q)^{-1}$, homogenized operator & \eqref{M-eq:M}\\
$G_1(e)$, $\Lambda_+$ & $-P_k(e)^{-1}$ for $|e|=1$, half-lattice & \eqref{M-eq:pair}\\
$r_k$ & rates $r_1=\frac12$, $r_k=1$ ($k\ge2$) & \S2(i)\\
$A^\dagger$, $V^\dagger$, $a_{ij}$ & vector and electric potentials; coefficients of $P_2$ & \eqref{M-eq:peierls}; Lemma~\ref{M-lem:peierls}\\
$S_\eps^\dagger$ & Fourier multiplier by $\one_{\{|\xi|\le1/\eps\}}$ & \eqref{M-eq:smooth}\\
$\psi^0$, $\psi_{\rm app}$, $v^\dagger$ & homogenized solution, approximate solution, $S_\eps\psi^0$ & \S\ref{M-sec:hom}\\
$\chi^\dagger$, $\chi_2^\dagger$, $L^\dagger$ & first and second cell correctors; $L=\gamma(D_y)$ & proof of Thm~\ref{M-thm:hom}\\
$\mathfrak m$, $a^\dagger$, $w(\mathfrak m)^\dagger$ & monomial $\eps^ab^\eps Qv$, its order and weight & proof of Thm~\ref{M-thm:hom}\\
$c_1^\dagger$, $c_{\mathfrak m}$, $\mathcal T_y$ & first corrector, corrector of $\mathfrak m$, $\sum\mathrm{Sym}[c_\alpha D_y^\alpha]$ & proof of Thm~\ref{M-thm:hom}\\
$G^\perp$, $\mathcal S^\mp_G$ & $\{H:\tr(GH)=0\}$, single-direction cones & \S\ref{M-sec:univ}\\[4pt]
\multicolumn{3}{@{}l}{\emph{Examples}}\\[2pt]
$\psi^\dagger$, $\zeta$ & stream function and gauge part, $A=\nabla^\perp\psi+\nabla\zeta$ & \S\ref{M-sec:num}\\
$\mathcal Q(e)$ & $i[\gamma(e)^{-1},\gamma(e^\perp)]$ & \S\ref{M-sec:num}\\
$v^\dagger$, $\phi_v$ & real mean-zero periodic profile, $-\Delta\phi_v=v$ & \S\ref{M-sec:num}\\
$a^\dagger$, $b^\dagger$, $\theta_0$ & amplitudes and phase of $V=a\cos(2\pi y_1)$, $A_2=b\cos(2\pi y_1+\theta_0)$ & after Prop.~\ref{M-prop:EM}; \eqref{eq:SMex1}\\
$m$, $\beta$; $b^\dagger$, $c^\dagger$ & regularized Dirac mass and regularization; BHZ coefficients & \S\ref{M-sec:num}\\
$V_a$, $\phi_a$, $\mathcal C_{ab}$, $\mathcal E$ & components of $W$, $-\Delta\phi_a=V_a$, gradient correlations, $\mathcal C_{00}$ & before Prop.~\ref{M-prop:regDirac}\\
$M_0$, $M_\perp$, $M_3$, $m_{\rm eff}$ & components of $M$; effective mass $m+M_3-\beta|M_\perp|^2$ & Prop.~\ref{M-prop:regIndex}\\
$\bar m$, $e_0$ & renormalized BHZ mass and energy shift & before Cor.~\ref{M-cor:BHZ}\\
$\tau^\dagger$ & magnetic coupling in $w=v\sigma_0+\tau\nabla^\perp v\cdot\sigma$ & Cor.~\ref{M-cor:bernoulli}\\
$H_\circ^\dagger$, $m(x_2)$, $m_0$, $\cP_n$, $E_{\frac12}$ & Dirac domain wall, puddles, energy with $\varphi=\frac12$ & before Cor.~\ref{M-cor:puddles}\\
$c_0,c_1^\dagger,c_3$, $s^\dagger$, $\alpha^\dagger$ & direction and strength of matrix-valued $W$; slice angle & \S\ref{sec:Sexamples}\\
$c_{\min}$, $c_{\max}$; $m_1$, $s^2$ & extrema of $\sum_\ell|\hat v(q+\ell)|^2$; mean and variance of couplings & Props.~\ref{prop:bernoulli},~\ref{prop:uniform}\\
$v_\delta$, $\varphi_\delta^\dagger$, $\bar\varphi_\delta$, $u$, $a(x)^\dagger$ & flat single sites and their kernels & Lemma~\ref{lem:flat}\\
$\bar\mu^\dagger$ & target mass in Prop.~\ref{prop:highex}(b); renormalized chemical potential in \eqref{eq:Spwave} & \S\ref{sec:Sexamples}\\
$\mu^\dagger$, $\Delta_0$, $\tau_j^\dagger$ & chemical potential, pairing, particle-hole Pauli matrices & \eqref{eq:Spwave}\\
$M_{il}$; $\rho_\pm$, $W_\pm$ & cross corrections; two-phase envelopes and profiles & Rems.~\ref{rem:multi},~\ref{rem:twophase}\\
$n=2p$, $\cG$ & first-order quantization & Rem.~\ref{rem:quant}\\
$\mathsf U$, $R^\dagger$, $\Theta^\dagger$; $W_{\rm e}$, $W_{\rm o}$, $M_{\rm ee}$ etc. & symmetries; graded parts of $W$ and $M$ & \S\ref{sec:Stools-hom}\\[4pt]
\multicolumn{3}{@{}l}{\emph{Gap closing, defects, cuts and robustness}}\\[2pt]
$a_\eps$, $H^\eps_a$, $H^\eps_{a,B}$ & amplitude, interface and lower-bulk operators & before Thm~\ref{M-thm:gapclose}\\
$\mathcal U_\eps^\dagger$, $\tilde H(\eps)$ & dilation $(\mathcal U_\eps u)(y)=\eps u(\eps y)$, dilated family & proof of Thm~\ref{M-thm:gapclose}\\
$V_c$ & compactly supported defect & Cor.~\ref{M-cor:compact}\\
$(\cL,\cR)$, $\tilde\Pi$, $\tilde T(H)$ & cut, $\one_\cR$, Toeplitz operator along the cut & Def.~\ref{M-def:cut}\\
$\dd$, $r^\dagger$, $\nu_\dd$, $\nu_r$ & distance of the cut, surgery radius, their exponents & \eqref{M-eq:exponents}\\
$p_0$ & exponent: unit box meets $\cU_\omega$ with probability $O(\eps^{p_0})$; $p_0=1$ in \S2(vi), $p_0=\min(c\delta_*^2/\ell,1)$ in \S\ref{M-sec:crit} & \eqref{M-eq:exponents}; Lemma~\ref{M-lem:probboxA}\\
$\cU^\dagger$, $\cU_\omega$ & bad set; union of bad mesoboxes (\S2(vi), \S\ref{M-sec:pseudo}) or bad unit boxes (\S2(vii), \S\ref{M-sec:crit}) & \S2(vi)--(vii)\\
admissible cut & cut avoiding $(\cU_\omega)_{3r}$ at distance $\dd/2$ & \S2(vi)\\
$H_\circ^\dagger$, $V_\circ$, \textup{[R1]}, \textup{[R2]} & reference operator, its potential, decay and interface conditions & before Thm~\ref{M-thm:robust}\\
$C_0$, $a^\dagger$, $R_0$ & constants in \textup{[R1]} & before Thm~\ref{M-thm:robust}\\
$V^\dagger$, $K_V$, $\cO_\cL$, $\cO_\cR$, $\vartheta_1$, $\vartheta$ & defect, its bound, its two parts, robustness errors & Thm~\ref{M-thm:robust} and proof\\
$U_\circ$, $S_\circ$, $U_{\rm rel}$, $\cJ$, $\dloc$ & objects of the factorization \eqref{M-eq:factor} & proof of Thm~\ref{M-thm:robust}\\
$\phi^\dagger$, $b^\dagger$, $B_b$, $a_\gamma$, $\Theta^\dagger$ & Combes--Thomas weight, exponent, conjugation term, rates; $H^1$ bound & Lemma~\ref{M-lem:fred}; Lemma~\ref{M-lem:CT}\\
$B^\dagger$, $\mathcal C_i$, $\mathcal C(u)$, $p'$, $p''$ & coarse boxes, removed clusters, cluster of a bad unit box, probabilities & before and in the proof of Lemma~\ref{M-lem:cuts}\\
$\theta^\dagger=\theta_{\cU,r}$, $M_\star$, $H_\circ$, $H_{\circ,B}$ & surgery cutoff, inserted matrix, reference operator, its bulk & \eqref{M-eq:Hg}\\
$\tilde\rho=\theta\rho$, $V_{\rm s}$, $M_\circ$ & cut-off envelope, background potential, glued matrix & proof of Thm~\ref{M-thm:mainB}; Rem.~\ref{M-rem:Vs}\\[4pt]
\multicolumn{3}{@{}l}{\emph{Random alloys and mesoscale averages}}\\[2pt]
$W_\omega$, $\omega_j$, $w^\dagger$ & alloy, couplings, single site & \S2(iv)\\
$\kappa_\pm$, $\mu_0$, $\kappa_0$ & coupling bounds, law of couplings, $[0,\kappa_0]\subset\supp\mu_0$ & \S2(iv); Thm~\ref{M-thm:pseudogap}\\
$H^\eps_\omega$, $H^\eps_{\omega,B}$ & alloy interface operator, its lower bulk & \S2(iv)\\
$\mathbb K(q)$ & periodized single-site Born kernel, $q\in\T^2$ & \eqref{M-eq:Kq}\\
$\hat\omega_q$, $\mathfrak s_\omega$ & discrete transform, structure factor & Lemma~\ref{M-lem:Kset}\\
$\cK[\mu_1,\mu_2]$, $\cK^\dagger$ & sets of local effective matrices, $\cK=\cK[\kappa_-^2,\kappa_+^2]$ & \S2(iv)\\
$\delta_k(\eps)$ & rates $\eps^{1/4}$, $\eps^{1/3}\log(1/\eps)$, $\eps^{1/3}$ & \S2(iv)\\
(A), (B), $s_\pm$, $M^{\rm e}$, $M^{\rm o}$, $\bar M$, $\bar r$ & gap-preserving conditions, signs, $J$-even/odd parts, reference matrix, margin & \S2(iv)\\
$M^\eps_\omega$, $M_\omega$ & local effective matrix in $x$ and in $y$, $M^\eps_\omega(x)=M_\omega(x/\eps)$ & Thm~\ref{M-thm:unifhom}; \S\ref{M-sec:random}\\
$\tilde K$, $K_\varpi$, $K_j$ & mesoscale kernel, $\varpi^2\tilde K(\varpi\cdot)$, $K_\varpi(\cdot-j)$ & \S\ref{M-sec:random}; \S\ref{M-sec:pseudo}\\
$\varpi$, $\varpi_0$ & mesoscale, low-frequency cutoff of the quadratic corrector & \S\ref{M-sec:random}; Lemma~\ref{M-lem:mult}\\
$\chi_w$, $\chi^\dagger$, $\Xi$, $A^\dagger$ & single-site corrector, alloy corrector, correction density $W\chi$, its mesoscale average $K_\varpi*\Xi$; $M_\omega=\frac12(A+A^*)$ & \eqref{M-eq:Mloc}\\
$\mu(y)^\dagger$, $k_0(m)$, $a_j^\dagger$ & local mass $\sum_jK_j\omega_j^2$, lattice kernel, $K_j^{1/2}\omega_j$ & Lemma~\ref{M-lem:meso} and proof\\
$\cS^s$, degree & symbol class vanishing to order $s$ at $0$; degree of a linear coefficient & \S\ref{M-sec:random}; proof of Thm~\ref{M-thm:unifhom}\\
$m_\nu$, $B_\nu^\dagger$ & multipliers with $D_y^\nu\chi_2=-m_\nu(D_y)\Xi$, and their kernel bounds & Lemma~\ref{M-lem:mult}\\
$\chi_2^\dagger$, $c_{\chi_2}$ & quadratic corrector and its contribution to $\psi_{\rm app}$ & proof of Thm~\ref{M-thm:unifhom}\\
$M_{\rm Born}$, \textup{[BG]} & Born matrix, Born-gap condition & \eqref{M-eq:Mborn}; \S2(vi)\\
$v^\dagger$, $m_1$, $\tilde\omega_j$, $F(y)^\dagger$ & ${\rm Var}\,\omega$, $\EE\omega$, centered couplings, quadratic form & Lemma~\ref{M-lem:mean}\\
$Q'$, $M^{\rm tr}_{Q'}$ & mesobox of side $\eps\varpi^{-1}$, truncated local effective matrix & \S\ref{M-sec:pseudo}, bad mesoboxes\\[4pt]
\multicolumn{3}{@{}l}{\emph{Random fields and pseudo-gaps}}\\[2pt]
$W(y,\omega)$, $H^\eps_B$, $\Sigma_\eps$ & stationary field, bulk operator, almost-sure spectrum & \S2(v)\\
$W_*$, $W_t$ & charged pattern, path of patterns & Thm~\ref{M-thm:pseudogap}\\
$H^0_B$ & deterministic reference bulk of a pseudo-gap & Def.~\ref{M-def:pseudo}\\
$s(q)^\dagger$, $G\diamond s$ & spectral density, contraction & \S\ref{M-sec:random}\\
$W_{\rm per}$, $H^\eps_{\lambda'}$, $\omega_*$ & periodic part $\sum_jw(\cdot-j)-\int w$ of the lattice sum, comparison operator, constant coupling & proof of Thm~\ref{M-thm:pseudogap}(iii); Lemma~\ref{lem:globalsize}\\[4pt]
\multicolumn{3}{@{}l}{\emph{Criticality}}\\[2pt]
\textup{[C1]--[C3]} & assumptions at criticality & \S2(vii)\\
$A^\dagger$, $q(\xi)^\dagger=\xi\cdot B\xi$, $B^\dagger$ & principal matrix, quadratic form, its matrix & \textup{[C1]}\\
$P_1$, $c_j^\dagger$, $P_0$ & first- and zeroth-order parts in \textup{[C1]} & \textup{[C1]}\\
$s_0$, $\ell^\dagger$, $\lambda_\eps$, $\tilde V$ & $\EE\omega^2(\int w)^2$, disorder strength, $(\ell/\log(1/\eps))^{1/2}$, critical potential & \eqref{M-eq:Hcrit}\\
$c_q$, $M_{\rm crit}$ & angular constant, renormalized correction $-\ell s_0c_qA^{-1}$ & \S2(vii)\\
$\chi_0^\dagger$, $\tilde V_{\rm low}$, $\tilde V_{\rm high}$, $Y_0$ & cutoff, low/high-frequency parts, corrector $-q(D)^{-1}\tilde V_{\rm high}$ & \S\ref{M-sec:crit}, local size\\
$|v|_B$, $Y_\eps$, $c_\eps$, $\cW^\dagger$, $M_\eps$ & $(v\cdot Bv)^{1/2}$, cell $\eps[0,1)^2$, renormalization constant, fluctuation of $|\nabla Y_0|_B^2$, $-c_\eps A^{-1}$ & \eqref{M-eq:ceps}\\
$\kappa^\dagger$, $\delta_Q$, $\delta$ & regularity loss, local size, its supremum & \eqref{M-eq:deltaQ}\\
$H_\theta$, $\delta_\theta$, $N_\theta$, $\delta_*$, $R^\dagger$ & cut-off operator, local sizes and gradient bound near $\supp\theta$, threshold, bound on $|\Re z|$ & Thm~\ref{M-thm:det}\\
$\tilde\rho$, $Y$, $X^\dagger$, $\cA_i$, $C_0'$ & $\theta\rho$, $\tilde\rho Y_0$, $A^{-1}Y$, remainder terms, their constant & Lemma~\ref{M-lem:conj}\\
$R_M$, $\mathcal K^\dagger$, $C_1'$, $\tilde R$ & $(P+\rho^2M_\eps-z)^{-1}$, $\cA R_M$, bound on $\|\mathcal K\|/\delta_\theta$, conjugated resolvent & proof of Thm~\ref{M-thm:det}\\
$\delta_Q^{\rm tr}$, $N_Q^{\rm tr}$, $C_1$ & truncated $\delta_Q$ and $\|\nabla Y_0\|_{L^\infty(Q^*)}$, gradient threshold & Lemma~\ref{M-lem:probboxA}\\[4pt]
\multicolumn{3}{@{}l}{\emph{Probabilistic estimates (appendix~\ref{M-app:proba})}}\\[2pt]
$\mathsf S$, $\mathsf S_n$, $w_n$ & synthesis operator, its smoothed versions & \eqref{M-eq:synth}\\
$\tau_0,\tau_1,\tau_{\rm low}^\dagger$, $K_0,K_1,K_{\rm low}$, $c_\chi$ & multipliers, their kernels, radius where $\tau_0,\tau_1$ vanish & App.~\ref{M-app:proba}\\
$Z$ & $\nabla Y_0$ & App.~\ref{M-app:proba}\\
$\sigma^{2\dagger}$, $\Lambda^\dagger$, $\mathcal X$, $M_\Lambda$, $\mathcal E_{\rm HW}$ & variance proxy, test measure, quadratic kernel, multiplication, Hanson--Wright exponent & Lemma~\ref{M-lem:conc}\\
$\mathsf A$, $\mathsf A_{j,y}$ & matrices of quadratic forms in the couplings (Hanson--Wright) & Lemmas~\ref{M-lem:mean}, \ref{M-lem:conc}, \ref{M-lem:blocks}\\
$\Lambda_{j,y}$, $\Psi_j$, $\mathcal N_j$ & localized dyadic test functions, their envelopes & Lemma~\ref{M-lem:blocks}\\
$\delta_0$, $C_\Theta$, $\iota$ & mesh constant, gradient constant, sign of the weight exponent & Lemma~\ref{M-lem:sampling}\\
$\Gamma_j^\dagger$, $\gamma_0$, $c_1$ & sampling lattices, exponent of the dyadic sums, threshold constant & proof of Lemma~\ref{M-lem:prob}\\
$\cW_{\rm r}$, $\cW_{\rm d}$, $g_j^\dagger$ & random and deterministic parts of $\cW$, dyadic blocks & proof of Lemma~\ref{M-lem:prob}\\
$J^\dagger$, $Z_\le$, $Z_>$; $F^{\rm tr}$ & $\lceil\log_2(1/\eps)\rceil$, low- and high-frequency parts of $Z$; truncated field & proof of Lemma~\ref{M-lem:probboxA}\\
\end{longtable}
}

\noindent Table~\ref{tab:reuse} lists the letters used with several meanings. Within a section the meaning is fixed, except where noted.

{\small
\renewcommand{\arraystretch}{1.12}
\begin{longtable}{@{}p{1.8cm}p{13.7cm}@{}}
\caption{Symbols with several meanings.}\label{tab:reuse}\\
\toprule
Symbol & Meanings (where)\\
\midrule
\endfirsthead
\toprule
Symbol & Meanings (where)\\
\midrule
\endhead
\bottomrule
\endfoot
$A$ & coefficients $A_\alpha$ (\S2(i)); vector potential (\S2(i), \S\ref{M-sec:num}, \S\ref{M-sec:univ}); principal matrix in \textup{[C1]} (\S2(vii), Cor.~\ref{M-cor:BHZ}, \S\ref{M-sec:crit}, App.~\ref{M-app:proba}); $K_\varpi*\Xi$ (\S\ref{M-sec:random}); matrices $A_a$ in the bilinear identity (\S\ref{M-sec:num}); bounded sets $A_\pm$, $A_0$ (\S\ref{M-sec:topo}, \S\ref{M-sec:pseudo})\\
$a$ & Peierls vector potential (\S2(iii)); decay rate in \textup{[R1]} (\S\ref{M-sec:pseudo}); order of a monomial (\S\ref{M-sec:hom}); symbol in $\cS^s$ (Lemma~\ref{M-lem:latsum}); amplitude of $V$ (\S\ref{M-sec:num}, \S\ref{sec:Sexamples}); point of $\supp\mu_0$ (Lemma~\ref{lem:globalsize}); exponent (Lemma~\ref{M-lem:sampling}); autocorrelation $a(x)$ (Lemma~\ref{lem:flat}); also $a_\eps$, $a_j$, $a_\gamma$, $a_{ij}$\\
$B$ & matrix of $q$ in \textup{[C1]}; Combes--Thomas term $B_b$ (Lemma~\ref{M-lem:fred}); $B(y)$ (proof of Lemma~\ref{M-lem:meso}); $B_0$ (proof of Thm~\ref{M-thm:pseudogap}(iii)); kernel bounds $B_\nu$ (Lemma~\ref{M-lem:mult}); coarse boxes and $B_u$ (Lemma~\ref{M-lem:cuts}); $B\in\{S_\circ,S_\circ^*\}$ (proof of Thm~\ref{M-thm:robust}); balls $B(0,R)$, $B_R$; Besov $B^s_{p,r}$\\
$b$ & $b_{ij}$ of $q(\xi)$ in class (b) (\S2(iii)); BHZ coefficient (\S\ref{M-sec:num}); amplitude of $A$ (\S\ref{M-sec:num}); profile $b$ in $w=\Delta^kb$ (\S2(iv)); Combes--Thomas exponent (\S\ref{M-sec:topo}, \S\ref{M-sec:pseudo}); coefficient of a monomial (\S\ref{M-sec:hom}, \S\ref{M-sec:random})\\
$c$, $\mathcal C$ & invariant $c(h)$ (eq.~\ref{M-eq:ch}); BHZ coefficient; coefficients $c_\alpha$, $c_j$; corrector $c_1$, $c_{\mathfrak m}$ (\S\ref{M-sec:hom}); $c_q$, $c_\eps$, $c_P$, $c_Q$, $c_\chi$; directions $c_0,c_1,c_3$ (\S\ref{sec:Sexamples}); generic constants; correlations $\mathcal C_{ab}$ (\S\ref{M-sec:num}); clusters $\mathcal C_i$ (Lemma~\ref{M-lem:cuts})\\
$d$ & dimension; coefficients $d_\alpha$; two-band vector $d(\xi)$ (\S\ref{M-sec:num}, \S\ref{M-sec:topo}); distances $d$, $d_0$ (Lemma~\ref{M-lem:speccomp}, proof of Thm~\ref{M-thm:pseudogap}); exterior derivative\\
$D$ & derivative $-i\nabla$ and $D^a$, $D_y$; $\mathfrak d$ is the distance of the cut\\
$e$ & unit vectors $e$, $e^\perp$, $e_\theta$, basis vectors; remainder coefficient $e$ (eq.~\ref{M-eq:Heps}); right-hand side $e$ (Lemma~\ref{M-lem:cell}); energy shift $e_0$ (Cor.~\ref{M-cor:BHZ})\\
$E$, $\mathcal E$ & energy; error terms $E$, $E_1$ (proofs of Lemma~\ref{M-lem:meso}, Thm~\ref{M-thm:robust}); $E_1$ (proof of Thm~\ref{M-thm:pseudogap}(iii)); $\mathcal E=\mathcal C_{00}$ (\S\ref{M-sec:num}, \S\ref{sec:Sexamples}); event $\mathcal E_Q$ (Lemma~\ref{lem:globalsize}); $\mathcal E_{\rm HW}$; projectors $\mathsf E_j$ (Lemma~\ref{M-lem:BDIc})\\
$F$, $f$ & $F(z)$ (Lemma~\ref{M-lem:HSdecay}); $F_\pm(\omega,\xi)$ (Lemma~\ref{M-lem:BDIc}); $F(y)$, $F_1$, $F_2$ (Lemma~\ref{M-lem:mean}); $F^{\rm tr}$ (App.~\ref{M-app:proba}); right-hand side $f$; remainder coefficient $f$ (eq.~\ref{M-eq:Heps}); scalar part $f$ of a two-band symbol\\
$G$, $\Gamma$ & $G(q)$, $G_1$ (eqs.~\ref{M-eq:M}, \ref{M-eq:pair}); generic Hermitian $G$ (\S\ref{M-sec:univ}, Lemma~\ref{lem:grading}); $G_\pm$ (Lemma~\ref{M-lem:BDIc}); $G$ (proof of Lemma~\ref{M-lem:sampling}); mass $\Gamma$; lattices $\Gamma_j$ (App.~\ref{M-app:proba})\\
$g$ & gap half-width; $g_j$ dyadic blocks (App.~\ref{M-app:proba}); generic function $g$ (Lemma~\ref{M-lem:prod})\\
$h$, $\eta$ & symbols $h$, $h_\pm$ (\S2(iii), Lemma~\ref{M-lem:BDIc}); function $h$ (Lemma~\ref{M-lem:prod}); lattice mesh $h$ (Lemma~\ref{M-lem:sampling}); $\eta$ is $|\Im z|$ in comparison arguments and in Lemma~\ref{M-lem:CT}(ii); $\eta_1$ (Lemma~\ref{lem:globalsize})\\
$H_\circ$ & reference operator in general (\S\ref{M-sec:pseudo}); the Dirac domain wall (Cor.~\ref{M-cor:puddles}) and the surgery operator \eqref{M-eq:Hg} are instances\\
$\iota$ & sign of the weight exponent (Lemma~\ref{M-lem:sampling})\\
$J$ & grading (\S2(iii)--(iv), \S\ref{M-sec:topo}); open interval (Lemma~\ref{M-lem:fred}); $\lceil\log_2(1/\eps)\rceil$ (proof of Lemma~\ref{M-lem:probboxA})\\
$K$, $\mathcal K$, $\mathbb K$ & $\mathbb K(q)$ (eq.~\ref{M-eq:Kq}); set $\cK$ (\S2(iv)); operator $\mathcal K=\cA R_M$ (proof of Thm~\ref{M-thm:det}); kernels $\tilde K$, $K_\varpi$, $K_j$ (\S\ref{M-sec:random}--\ref{M-sec:pseudo}); $K_0$, $K_1$, $K_{\rm low}$ (App.~\ref{M-app:proba}); $K_V$ (Thm~\ref{M-thm:robust})\\
$k$ & order of $P$; dyadic kernels $k_j$; lattice kernel $k_0(m)$; dependence range $k'$; momentum $k_1$ (Fig.~\ref{fig:first})\\
$L$, $\cL$ & $L=\gamma(D_y)$ (proof of Thm~\ref{M-thm:hom}), $L(\xi)=\gamma(\xi)$ (\S\ref{sec:Stools-hom}), $L=P_1(D_y)$ (\S\ref{sec:Sexamples}, Bloch eigenvalues), $L=P_k(D_y)$ (Rem.~\ref{rem:cubic}); period $L$ (Lemma~\ref{M-lem:Kset}); truncation length (Lemma~\ref{M-lem:mesobad}); factor $L^TL=(b_{ij})$ (proof of Thm~\ref{M-thm:magno}); side $\cL$ of a cut\\
$\ell$ & disorder strength (eq.~\ref{M-eq:Hcrit}); lattice index in \eqref{M-eq:Kq} and in Poisson sums\\
$\lambda$ & $\lambda_\eps$; amplitudes $\lambda$ (Ex.~\ref{M-ex:counter}, \S\ref{sec:Sexamples}); eigenvalues $\lambda_j$, $\lambda_0$, $\lambda_\eps$; $\lambda'$ (proof of Thm~\ref{M-thm:pseudogap}(iii)); dilation factor (Lemma~\ref{M-lem:fred}(iii))\\
$M$ & correction and its variants ($M(W)$, $M(\omega)$, $M_{\rm Born}$, $M_{\rm crit}$, $M_\eps$, $M_\omega$, $M^\eps_\omega$, $M^{\rm tr}_{Q'}$, $M_\circ$, $M_\star$, $M_s$, $\bar M$); components $M_0,M_\perp,M_3$; $M_{il}$, $M_{\rm ee}$, $\dots$ (\S\ref{sec:Stools-hom}, \S\ref{sec:Sremarks}); multiplication $M_\Lambda$ (App.~\ref{M-app:proba})\\
$m$ & mass; lattice index ($\omega_m$, $Q=m+[0,1)^2$); $m_{\rm eff}$, $\bar m$, $m_0$, $m(x_2)$, $m'$; mean $m_1$ (Lemma~\ref{M-lem:mean}, Prop.~\ref{prop:uniform}); multipliers $m_\nu$; cluster sizes $m_n$ (Lemma~\ref{M-lem:cuts}); monomials $\mathfrak m$\\
$\mu$ & law $\mu_0$; bounds $\mu_1,\mu_2$ in $\cK[\mu_1,\mu_2]$; local mass $\mu(y)$ (Lemma~\ref{M-lem:meso}); chemical potential $\mu$, $\mu'$ (eq.~\ref{eq:Spwave}); $\mu_-$, and $\bar\mu$ with two meanings (Prop.~\ref{prop:highex}(b), Cor.~\ref{cor:pwave})\\
$n$ & integer exponents ($C_n$); unit vector $d/|d|$; matrix size $n=2p$ (Rem.~\ref{rem:quant})\\
$\nu$ & multi-indices (\S\ref{M-sec:hom}, \S\ref{M-sec:random}); exponents $\nu_r$, $\nu_\dd$\\
$\omega$ & couplings $\omega_j$, configuration $\omega$; frequency variable in $F_\pm(\omega,\xi)$ (Lemma~\ref{M-lem:BDIc})\\
$\Omega$ & curvature $\Omega_h$; measurable sets in Combes--Thomas bounds; angular profile $\Omega(\hat y)$ (Rem.~\ref{rem:log})\\
$p$ & Fourier variable on $\R^2$ (\S2(iv), \S\ref{M-sec:random}, \S\ref{sec:Sexamples}); exponent $p_0$; probabilities $p'$, $p''$ (Lemma~\ref{M-lem:cuts}); half-size $n=2p$ (Rem.~\ref{rem:quant})\\
$\varphi$, $\phi$ & energy switch $\varphi$; flat-site kernel $\varphi_\delta$ (Lemma~\ref{lem:flat}); Combes--Thomas weight $\phi$ (\S\ref{M-sec:topo}, \S\ref{M-sec:pseudo}); Littlewood--Paley symbols $\phi_j$ (\S\ref{M-sec:crit}, App.~\ref{M-app:proba}); potentials $\phi_a$, $\phi_v$, $\phi_3$ (\S\ref{M-sec:num}, \S\ref{sec:Sexamples}); cell solution $\phi$ (Lemma~\ref{M-lem:cell})\\
$\psi$, $\Psi$ & stream function $\psi$; $\psi^0$, $\psi_{\rm app}$, $\psi_1$; quasimodes $\psi$, $\psi_j$ (\S\ref{M-sec:num}, \S\ref{M-sec:pseudo}); partition $\psi$, $\psi_Q$ (\S\ref{M-sec:crit}); $\Psi$ in Rem.~\ref{rem:AC}; envelopes $\Psi_j$ (App.~\ref{M-app:proba})\\
$\Pi$ & $\one_{\{x_1\ge0\}}$; $\tilde\Pi$, $\Pi_s$; spectral projectors $\Pi_h$\\
$q$ & frequencies $q\in\Z^d$ or $q\in\T^2$; quadratic forms $q(\xi)$ (class (b) in \S2(iii), \textup{[C1]}); order of the operator $Q$ in a monomial (\S\ref{M-sec:hom})\\
$Q$ & unit boxes; mesoboxes $Q'$; differential operator in a monomial (\S\ref{M-sec:hom}); rectangle (proof of Lemma~\ref{M-lem:fred}); $\mathcal Q(e)$ (\S\ref{M-sec:num})\\
$R$, $\cR$ & resolvents; radius $R$ (Thms~\ref{M-thm:pseudogap},~\ref{M-thm:det}; Lemma~\ref{M-lem:sampling}); $R_0$, $R_1$; $R(f,g)$; $R_M$, $\tilde R$; orthogonal map $R$ (\S\ref{sec:Stools-hom}); side $\cR$ of a cut; remainder $\mathcal R^\eps$\\
$r$ & rates $r_k$; surgery radius $r$, also the range $2r$ of bad events; margin $\bar r$\\
$\rho$ & envelope; $\tilde\rho=\theta\rho$; $\rho_{\alpha,c}$, $\rho_{\alpha,d}$, $\rho_i$, $\rho_\pm$\\
$S$, $\cS$, $\mathsf S$ & $S(H)$, $S_\circ$; smoothing $S_\eps$; paraproduct cutoff $S_{j-1}$; spheres $S^{d-1}$, $\Sm^1$; symbol classes $\cS^s$; cones $\cS^\mp_G$; synthesis $\mathsf S$\\
$s$ & signs $s$, $s_\pm$; spectral density $s(q)$, $s_0$; strength $s$ and variance $s^2$ (\S\ref{sec:Sexamples}); homotopy parameter $s$; regularity index $s$ ($\cS^s$, $C^s$, $H^s$); coefficients $s,t$ (proof of Prop.~\ref{M-prop:EM}); thresholds $s_j,s'_j,s''_j$ (App.~\ref{M-app:proba})\\
$\sigma$ & Pauli matrices; interface current $\sigma_I$; $\sgn\int w$ (proof of Thm~\ref{M-thm:pseudogap}(iii)); variance proxies $\sigma^2$, $\sigma_j^2$ (App.~\ref{M-app:proba})\\
$T$, $\cT$ & Toeplitz operators $T$, $\tilde T$; paraproduct $T_fg$; $\mathcal T_y$ (\S\ref{M-sec:hom}); constant $T_0$ (proof of Thm~\ref{M-thm:pseudogap}(iii))\\
$t$ & parameter of paths $W_t$; thresholds (Lemma~\ref{M-lem:mean}(ii), Lemma~\ref{M-lem:prob}); $t=|\xi|^2$ (\S\ref{sec:Sexamples}); Sobolev index in \eqref{M-eq:smooth}\\
$\tau$ & single-site magnetic coupling (Cor.~\ref{M-cor:bernoulli}, \S\ref{sec:Sexamples}); multipliers $\tau_0,\tau_1,\tau_{\rm low}$ (App.~\ref{M-app:proba}); Pauli matrices $\tau_j$ in particle-hole space (eq.~\ref{eq:Spwave})\\
$\theta$, $\Theta$ & surgery cutoff $\theta$; phase $\theta_0$; angle $\theta$ (Rem.~\ref{rem:J}, Rem.~\ref{rem:quant}); $H^1$ bound $\Theta$ (Lemma~\ref{M-lem:CT}(ii)); Schwartz function $\Theta$ and constant $C_\Theta$ (Lemma~\ref{M-lem:sampling}); antiunitary $\Theta$ (\S\ref{sec:Stools-hom})\\
$U$, $\cU$ & $U(H)$, $U_\circ$, relative unitary $U_{\rm rel}=U_\circ^*U(H)$ (Thm~\ref{M-thm:robust}); bad sets $\cU$, $\cU_\omega$ and neighbourhoods $\cU_r$; dilation $\mathcal U_\eps$ (proof of Thm~\ref{M-thm:gapclose}); unitary $\mathsf U$ (\S\ref{sec:Stools-hom})\\
$V$, $\mathcal V$ & generic potential (Prop.~\ref{M-thm:QB}); generic differential perturbations $\mathcal V$, $\mathcal V_\pm$ (Lemma~\ref{M-lem:fred}, Cor.~\ref{M-cor:compact}); electric potential (eq.~\ref{M-eq:peierls}); puddle potential (Cor.~\ref{M-cor:puddles}); defect (Thm~\ref{M-thm:robust}, eq.~\ref{M-eq:Hg}); scalar alloy $\sum_j\omega_jv(\cdot-j)$ (Prop.~\ref{prop:bernoulli}); also $V_\pm$, $V_c$, $V_\circ$, $V_{\rm s}$, $V_a$, $V_\cL$, $V_\cR$, $\tilde V$\\
$v$ & $S_\eps\psi^0$ and $S_\eps u$ (\S\ref{M-sec:hom}, \S\ref{M-sec:random}); periodic profile (\S\ref{M-sec:num}); compactly supported profile of a single site (Cor.~\ref{M-cor:bernoulli}, \S\ref{sec:Sexamples}); null vectors (Lemma~\ref{M-lem:null}); $\hat V_q$ (proof of Prop.~\ref{M-prop:EM}); ${\rm Var}\,\omega$ (Lemma~\ref{M-lem:mean})\\
$W$, $\cW$ & effective profile (eq.~\ref{M-eq:Weff}); alloy $W_\omega$; random field $W(y,\omega)$; patterns $W_*$, $W_t$; periodic part $W_{\rm per}$; $W_\pm$, $W_i$, $W_{\rm e}$, $W_{\rm o}$ (\S\ref{sec:Stools-hom}, \S\ref{sec:Sremarks}); fluctuation $\cW$ of the corrector energy (\S\ref{M-sec:crit})\\
$w$ & single site; weight $w(\mathfrak m)$ of a monomial (\S\ref{M-sec:hom}); $w_n$ (App.~\ref{M-app:proba})\\
$X$ & window heights $X$, $X_0$ (Lemmas~\ref{M-lem:fred},~\ref{M-lem:cuts}); $X=A^{-1}Y$ (\S\ref{M-sec:crit}); matrices $X$ (\S\ref{M-sec:univ}, \S\ref{M-sec:num}); gauge part $X$ (Rem.~\ref{rem:AC})\\
$Y$ & $Y_0$, $Y$, $Y_\eps$ (\S\ref{M-sec:crit}); width $Y$ (proof of Lemma~\ref{M-lem:fred})\\
$\chi$ & correctors $\chi$, $\chi_2$ (with different definitions in \S\ref{M-sec:hom} and \S\ref{M-sec:random}), $\chi_w$; cutoff $\chi_0$ (separately in \S\ref{M-sec:random} and \S\ref{M-sec:crit}); cutoffs $\chi_A$, $\chi'$ (\S\ref{M-sec:topo}, Lemma~\ref{M-lem:conj})\\
$\delta$ & $\delta_k(\eps)$, $\delta_\eps$; local sizes $\delta_Q$, $\delta$, $\delta_\theta$, $\delta_*$, $\delta_Q^{\rm tr}$; $\dloc$; tolerance $\delta$ in the charging condition (\S2(v)); flatness $\delta$ (Lemma~\ref{lem:flat}); $\delta'$; mesh constant $\delta_0$ (Lemma~\ref{M-lem:sampling}); Kronecker $\delta_{ij}$; Dirac masses $\delta_q$, $\delta_x$\\
$\kappa$ & coupling bounds $\kappa_\pm$, $\kappa_0$; regularity loss $\kappa\in(0,\frac14)$ (\S\ref{M-sec:crit})\\
$\Lambda$, $\Sigma$ & half-lattice $\Lambda_+$; distance parameter $\Lambda$ (Lemma~\ref{M-lem:HSdecay}, eq.~\ref{M-eq:HSa}); test measures $\Lambda$, $\Lambda_{j,y}$ (App.~\ref{M-app:proba}); a.s.\ spectrum $\Sigma_\eps$; parts $\Sigma_\cL$, $\Sigma_\cR$ of $S(H)-S_\circ$ (proof of Thm~\ref{M-thm:robust})\\
\end{longtable}
}

\end{document}